\documentclass{article}
\usepackage{graphicx}

\usepackage{hyperref}
\usepackage{amssymb}
\usepackage{setspace}
\usepackage{amsmath}
\usepackage{url}
\usepackage{amsbsy}
\usepackage{amsfonts}
\usepackage{tikz-cd}
\usepackage{caption}
\numberwithin{equation}{section}
\usepackage{epigraph}
\usepackage{amsthm}
\newtheorem{theorem}{Theorem}[section]
\newtheorem{define}{Definition}[section]

\newtheorem{corollary}{Corollary}[section]
\newtheorem{principle}{Principle}[section]

\newtheorem{axioms}{Axioms}[section]
\newtheorem{axiom}{Axiom}[section]

\usepackage{float}
\usepackage{graphicx}
\usepackage[outdir=./]{epstopdf}
\usepackage{tikz-cd}
\usepackage{pgfplots}
\pgfplotsset{colormap/viridis}
\usepackage{tikz-3dplot}
\makeatletter
\tikzset{use path/.code=\tikz@addmode{\pgfsyssoftpath@setcurrentpath#1}}
\makeatother
\tdplotsetmaincoords{70}{100}
\pgfplotsset{compat=newest}

\usepackage{amsmath}
\usepackage{amssymb}
\usepackage{amsbsy}
\usepackage{amsfonts}
\usepackage{amsthm}

\usepackage{faktor}
\usepackage{xfrac}

\usepackage{setspace}
\usepackage[utf8]{inputenc}
\usepackage[english]{babel}
\usepackage[T1]{fontenc}
\usepackage{lmodern}
\usepackage{enumerate}
\usepackage{mathtools} 
\usepackage{mathrsfs} 
\usepackage{accsupp} 
\usepackage{bm}
\usepackage{adjustbox}
\usepackage{multirow}
\usepackage{enumitem}

\usepackage{epigraph}
\usepackage[font=small,labelfont=bf]{caption}
\usepackage{hyperref}
\usepackage{url}

\usepackage{graphicx}
\usepackage{amssymb}
\usepackage{float}
\usepackage{url}
\usepackage[utf8]{inputenc}
\usepackage[english]{babel}
\usepackage{fancyhdr}
\usepackage{enumerate}
\usepackage{amssymb}
\usepackage{amsfonts}
\usepackage{amsmath}
\usepackage{mathrsfs}  
\usepackage{faktor}
\usepackage{xfrac}
\usepackage[font=small,labelfont=bf]{caption}
\DeclareFontFamily{U}{MnSymbolC}{}
\DeclareSymbolFont{MnSyC}{U}{MnSymbolC}{m}{n}
\DeclareMathSymbol{\diamondplus}{\mathbin}{MnSyC}{"7C}
\DeclareMathSymbol{\diamonddot}{\mathbin}{MnSyC}{"7E}
\DeclareFontShape{U}{MnSymbolC}{m}{n}{
	<-6>  MnSymbolC5
	<6-7>  MnSymbolC6
	<7-8>  MnSymbolC7
	<8-9>  MnSymbolC8
	<9-10> MnSymbolC9
	<10-12> MnSymbolC10
	<12->   MnSymbolC12}{}
\usepackage{accsupp}

\newcommand{\tru}{\mathbf{T}}
\newcommand{\fals}{\mathbf{F}}

\newcommand{\notleft}{\mathrel{\ooalign{$\Leftarrow$\cr\hidewidth$/$\hidewidth}}}
\newcommand{\notright}{\mathrel{\ooalign{$\Rightarrow$\cr\hidewidth$/$\hidewidth}}}
\usepackage{adjustbox}

\makeatletter
\providecommand*{\dashv}{%
	\mathrel{%
		\mathpalette\@dashv\vdash
	}%
}
\newcommand*{\@dashv}[2]{%
	\reflectbox{$\m@th#1#2$}%
}
\makeatother

\newtheorem{goal}{Design Goal}
\newtheorem{criteria}{Design Criteria}
\usetikzlibrary{arrows}
\usepackage{pgf}
\usepackage{tikz}
\usetikzlibrary{arrows,automata}
\usetikzlibrary{arrows.meta,positioning}
\tikzset{
	treenode/.style = {align=center, inner sep=0pt, text centered,
		font=\sffamily},
	arn_n/.style = {treenode, circle, white, font=\sffamily\bfseries, draw=black,
		fill=black, text width=1.5em},
	arn_r/.style = {treenode, circle, red, draw=red, 
		text width=1.5em, very thick},
	arn_x/.style = {treenode, rectangle, draw=black,
		minimum width=0.5em, minimum height=0.5em},
	arn_w/.style = {treenode, circle, black, draw=black,
		text width=1.5em, very thick}
}

\DeclareMathSymbol{\diamondplus}{\mathbin}{MnSyC}{"7C}
\DeclareMathSymbol{\diamonddot}{\mathbin}{MnSyC}{"7E}

\newcommand{\sm}{\mathbf{P}}

\begin{document}

	\date{
		\today
	}
	
	\title{The Inference Framework}
	
	\author{Nicholas Carrara\thanks{nmcarrara@ucdavis.edu}\\Dept. of Physics\\University at Albany\\Albany, NY 12244}
	
	\maketitle
	
\begin{abstract}
    The following three sections and appendices are taken from my thesis ``The Foundations of Inference and its Application to Fundamental Physics'' from 2021, in which I construct a theory of entropic inference from first principles.  The majority of these chapters are not original, but are a collection of various sources through the history of the subject.  The first section deals with deductive reasoning, which is inference in the presence of \textit{complete} information.  The second section expands on the deductive system by constructing a theory of inductive inference, a theory of probabilities, which is inference in the presence of \textit{incomplete} information.  Finally, section three develops a means of updating these probabilities in the presence of new information that comes in the form of constraints. 
\end{abstract}
\section{Deductive inference: A prelude to induction}\label{chapter2}
\epigraph{Logic issues in tautologies, mathematics in identities, philosophy in definitions; all trivial, but all part of the vital work of clarifying and organising our thought.}{\textit{Frank Plumpton Ramsey}, 1929}

In this Chapter we will develop a theory of \textit{deductive inference} which will allow us to analyze various statements, or \textit{propositions}, using formal logic and determine whether or not they are \textit{valid}.  We will begin by briefly discussing the basic construction of a formal system\footnote{For the interested reader, the details of propositional and first order predicate logic are given in the \hyperref[logicappendix]{Appendix}.}.  Then, we will discuss the various issues with modern \textit{conditional logic}, that will help to motivate the construction of an \textit{extended proposition space} and the inductive framework of \hyperref[chapter3]{Chapter three}.

Deductive systems are constructed from a more basic structure called a \textit{formal language}.  A formal language, which we denote by $\mathcal{L}$, is defined with respect to a particular \textit{alphabet}, usually labeled $\Sigma$, which contains a collection of objects that act as the basis for the entire language\footnote{For example, the alphabet for the English language $\mathcal{L}_{\mathrm{eng.}}$ would consist of the set of latin letters together with punctuation marks, $\Sigma_{\mathrm{eng.}} = \left\{a,b,c,\dots,z,.,?,\dots\right\}$, whereas the alphabet for a language which consists of equations for adding natural numbers would be $\Sigma = \left\{0,1,2,3,4,5,6,7,9,+,=\right\}$}.  Elements of the alphabet $\Sigma$ are called \textit{letters}, and by joining them through juxtaposition one can form \textit{words}, or \textit{strings}.  The collection of all words associated to an alphabet $\Sigma$ is usually denoted $\Sigma^*$, which is called the \textit{Kleene star}\footnote{Considering the set $\Sigma_0 = \{\varepsilon\}$ containing the empty string and $\Sigma_1 \stackrel{\mathrm{def}}{=}\Sigma$ as the alphabet, the Kleene star is defined as,
	\begin{equation}
	\Sigma^* = \bigcup_{i\geq 0} \Sigma_i = \Sigma_0\cup\Sigma_1\cup\Sigma_2\cup\dots,
	\end{equation}
	where,
	\begin{equation}\Sigma_{i+1} = \left\{\omega\sigma\middle|\, \omega \in \Sigma_{i},\sigma \in \Sigma\right\}.\end{equation}} \cite{Kleene}.  Not all of the possible words contained in $\Sigma^*$ will necessarily be of interest, or have any well defined meaning, in a formal language\footnote{Using the two examples from before, a valid string from $\Sigma_{\mathrm{eng.}}$ could be ``aoheijf.?ej'' or ``abbbbb.''  Likewise for $\Sigma_+$ we could write something like ``234=4==234+9'' or ``=2425.''  Neither of these examples however have a well defined meaning in either language.}.  Only those words which are \textit{well-formed formulas} in $\Sigma^*$ are considered to be part of the formal language.  The rules that decide which elements of $\Sigma^*$ are considered well-formed formulas are called grammatical rules, or the \textit{syntax} \cite{Chomsky} of the language.  

In the formal language of logic we add an additional feature called a \textit{semantics} \cite{Davidson}, which appends an interpretation to the symbols of well-formed formulas.  In particular, it adds an interpretation of ``truth'' \cite{Tarski} to the elements of the language, which we then call \textit{propositions}\footnote{There is the possibility of identifying words in $\Sigma^*$ which obey the syntax of formal logic, but which do not have an interpretation (i.e. do not have a well defined truth value) and thus cannot be identified as propositions.  This situation will become important when we consider \hyperref[conditionallogic]{conditional logic}.}.

From a formal language we can construct a formal system (which is given the generic symbol $\mathcal{FS}$) or \textit{deductive system}, which consists of the language, $\mathcal{L}$, together with a set of \textit{axioms}, $\mathcal{I}$, and \textit{rules of inference}, $\mathcal{Z}$, that allow us to state and prove theorems through logical argument.  Theorems are well-formed formulas which can be deduced from the rules and axioms of the system.  A proof is a sequence of well formed formulas which are either axioms, or follow from the rules of inference, and conclude with the intended theorem.

Deductive inference is mainly concerned with determining the \textit{validity} and \textit{soundness} of logical arguments \cite{Tomassi}.  Any logical argument is composed of \textit{premises}, which may or may not be true, and a \textit{conclusion}, whose ``truth'' we wish to determine.  A valid argument is one in which a conclusion, or theorem, follows from the premises, i.e. if the premises are true then the conclusion has to be true.  The premises however need not be ``actually true'' for an argument to be valid, so long as they are \textit{well-formed formulas} of the system.  Only when the premises are ``actually true'' is the argument said to be \textit{sound}. 

The foundations of logical inference are propositional calculus and first-order \hyperref[predicatelogic]{predicate calculus}, both of which we will develop in this chapter.  One of the key results to keep in mind is that deductive inference is only capable of assigning ``truth'' to valid arguments, i.e. arguments which are invalid do not exist in the deductive system.  An argument can be invalid for several reasons, but the one that will be of importance, is when the premises contain \textit{incomplete information}.  This is to say that the conclusion cannot be completely determined from the premises.  How then can we reason about the ``truth'' of such conclusions?  This is of course the experience we encounter in everyday life, in science, in medicine, in the stock market, etc.  In order to deal with this situation one needs to develop an \textit{inductive inference}, which will be the topic of the next chapter.

Propositional and first-order predicate logic serve as part of the foundations for all of mathematics \cite{Barwise}.  They are the necessary ingredients that one needs in order to develop an axiomatic, or formal, system of logic \cite{Tomassi}, which then allows one to state and prove mathematical theorems.  To motivate the construction of the standard propositional calculus, which we label $\mathcal{D}$, we will begin by discussing the three laws of thought \cite{Aristotle}.

\subsection{The laws of logic}\label{lawsoflogic}
\epigraph{It depends on what the meaning of the word ‘is’ is.}{\textit{President William J. Clinton}, 1998}
\epigraph{Nothing unreal exists.}{\textit{Spock}, 1984}
As physicists, and in all manner of being pragmatic, we take the classical approach and accept the classical laws of logic as defining the basis of our model of inference.  These are the Aristotelian laws of thought: \textit{identity, non-contradition} and \textit{excluded middle} \cite{Hamilton,Aristotle}.  The three logical laws act as a foundation for how we should approach deductive inference.  They are stated as laws because they cannot be proved, however they appear natural, or obvious, and so we accept them as being always ``true.''  To quote Russell \cite{Russell}, the first law is perhaps the most obvious, yet the most subtle,
\begin{quotation}\label{identity}
	\textbf{Law of identity} - ``Whatever is, is.''
\end{quotation}
The second law is that of non-contradiction, which again quoting Russell, can be stated as
\begin{quotation}\label{noncontradiction}
	\textbf{Law of non-contradiction} - ``Nothing can both be and not be.''
\end{quotation}
This law is incredibly subtle since it assumes something about the nature of time.  To be more precise, the law should read -- ``Nothing can both be and not be at the same time.''  While there are no logical systems which attempt to relax the first law, there are systems which attempt to relax the law of non-contradiction.  These systems adopt the view of \textit{dialetheism}, a belief which proposes that there exist some true contradictions \cite{Priest,GarfieldPriest,Priest2}, which can lead to the \textit{principle of explosion} (See \hyperref[materialimplication]{Section 5}).  A system in which true contradictions are allowed to exist, but the principle of explosion is rejected, is called a \textit{paraconsistent logic} \cite{paraconsistent}.

The third law, which is perhaps the weakest, is the law of excluded middle, which is, again quoting Russell
\begin{quotation}\label{excludedmiddle}
	\textbf{Law of excluded middle} - ``Everything must either be or not be.''
\end{quotation}
The law of excluded middle states that propositions must either be true or false, but not neither and not both\footnote{There is no shortage of paradoxes that can be invented which seem to violate the law of excluded middle.  Take for example the statement, ``the King of France is bald.''  Since France no longer has Kings, it would seem logical to assign a truth value of false to this statement, however the law of excluded middle then suggests that the counter statement ``the King of France is not bald'' is true, but this suffers from the same problem therefore throwing into question the validity of the law of excluded middle.  Russell had his own ideas for how to deal with such statements \cite{Russell2}, but one could always conclude that such statements are invalid and do not qualify as propositions.}.  This axiom can be dropped in systems such as \textit{fuzzy logic} \cite{Zadeh}, \textit{autoepistemology} \cite{Marek} or in \textit{intuitionism} \cite{Brouwer,Heyting,Kleene}.  Fuzzy logic and autoepistemology are attempts at extending bivalent logic to a continuum where true and false are the extreme positions.  The autoepistemists call this a ``degree of knowledge.''    

The three laws of thought provide a basis for our formal system.  All other statements which are a priori assumed true will be added as axioms, from which theorems can be proven.  The laws are themselves \textit{tautologies}, meaning that they are always true independent of the propositions that define them.

Statements which follow logically from a formal system are \textit{consequences} of the system.  There are two notions of logical consequence that are employed.  The weaker of the two is called syntactic consequence which introduces the symbol $\vdash$ and is often called a \textit{turnstile}.  It is meant to represent the words ``I know.''  Thus, something like, $\vdash a$, means, ``I know $a$ to be true.''  Given a sequence of well-formed formulas, $\Gamma \Leftrightarrow \{a_i\}_{i=1}^n$, from the formal system $\mathcal{D}$, we say that another well-formed formula $b$ is a \textit{syntactic consequence} of $\Gamma$ if there is a formal proof of $b$ within $\Gamma$.  This is written with the \textit{sequent} notation
\begin{equation}
\Gamma \vdash b,\label{syntacticconsequence}
\end{equation}
which means ``from $\Gamma$ I know $a$'' or, ``$\Gamma$ proves $a$.''  This is also often called \textit{entailment}, i.e. ``$\Gamma$ entails $a$.''   

A stronger version of entailment is called \textit{semantic consequence}, which states in addition to (\ref{syntacticconsequence}) that there are no interpretations in which each well-formed formula in $\Gamma$ is true and $b$ is not true, i.e. $\Gamma$ necessarily entails $b$.  This is written as,
\begin{equation}
\Gamma \vDash b.
\end{equation}

\paragraph{Why Deductive Inference Is Incomplete ---}
With the formal system developed in the \hyperref[logicappendix]{Appendix} we can evaluate a large range of possible statements from the general set of statements $\Sigma^*$.  These include all mathematical theorems, propositional sentences, counterfactual propositions, modal statements, matter-of-fact conditionals, subjunctive conditionals and others.  While this list is perhaps suitable for most formal situations, it does not include one important situation which occurs most often in everyday life, as well as in science.  This is the situation in which one is trying to access the truth of some proposition, in the presence of incomplete information.  

To illustrate the idea of incomplete information, consider a proof $\Gamma \Leftrightarrow a_1\wedge a_2\wedge\dots\wedge a_n$ for a theorem $b$ which is both valid and sound so that, $\Gamma \vdash b$.  We can identify $\Gamma \vdash b$ as the conditional statement $b |\Gamma$, which is necessarily true since $\Gamma$ is a proof for $b$.  But now consider that we remove one of the steps, $a_i$, from $\Gamma$,
\begin{equation}
\Gamma\rightarrow \Gamma' \Leftrightarrow a_1\wedge \dots a_{i-1}\wedge a_{i+1}\wedge\dots \wedge a_n.
\end{equation}
While $\Gamma'$ is a well defined proposition, it is no longer a valid argument for $b$ since it is \textit{incomplete}.  Thus, the conditional $b|\Gamma'$ does not have a well defined truth value according to the standard rules of propositional logic.  It is still reasonable however to ask whether one believes the proposition $b$ given the information in $\Gamma'$ even though it is incomplete.  This is a question that deductive reasoning cannot answer, in part because it is too rigid.  

Another situation logic is ill-equipped to handle is when $\Gamma$ contains irrelevant information.  If in $a|\Gamma$, $\Gamma$ tells us nothing about the validity of $a$, then $\Gamma$ is irrelevant.  In this case, our beliefs about $a|\Gamma$ should be equivalent to our beliefs about $a$, i.e. $a|\Gamma \rightarrow a$, however, deductive logic has no way to handle this situation, since it cannot evaluate the statement $a|\Gamma$, when $a$ is independent of $\Gamma$.  An argument of this sort is simply invalid according to the laws of deductive reasoning.  


In order to evaluate statements about which we have incomplete information, we will need to develop an inductive inference, which is the subject of Chapter two.  Before that however, we will enlarge our universe of discourse to contain statements such as $a|\Gamma'$ when $\Gamma'$ is incomplete, irrelevant or does not reduce to some element generated by $\mathcal{D}_c$.

\subsection{The extended proposition space}\label{section26}
With the idea of incomplete information introduced, we can define an extension to the proposition space $\mathcal{A}$ in which one can construct statements such as ``$b$ given $a$'' when $a$ is \textit{incomplete}.  Such a space will be labeled $\tilde{\mathcal{A}}$.
\begin{define}
	We define the extended proposition space $\tilde{\mathcal{A}}$ of some underlying proposition space $\mathcal{A}$ as the collection of all statements of the form, $\Delta|\Gamma$, where $\Delta,\Gamma \in \mathcal{A}$, i.e.,
	\begin{equation}
	\tilde{\mathcal{A}} \stackrel{\mathrm{def}}{=}\left\{\frac{}{}\Delta|\Gamma\,\middle|\, \Delta,\Gamma \in \mathcal{A},\Gamma\not\Leftrightarrow\fals\right\},
	\end{equation}
	with the restriction that $\Gamma \not\Leftrightarrow \mathbf{F}$\footnote{We do not allow the second argument to contain contradictions such as $a \wedge \neg a$.}.
\end{define}
Elements $\Delta|\Gamma\in\tilde{\mathcal{A}}$ will generally be referred to as ``statements.''  Some statements are themselves propositions, and thus have a well defined truth value, while others do not.  The size of $\tilde{\mathcal{A}}$ depends on how many statements in $\mathcal{A}$ are equivalent with $\fals$.  Call the subspace of statements in $\mathcal{A}$ which are equivalent with false, $\mathcal{A}_{\fals} = \{a \in \mathcal{A}\, |\, a \Leftrightarrow \fals\}$.  Then, the size of the set $|\mathcal{A}\backslash \mathcal{A}_{\fals}|$ determines the number of possible arguments for $\Gamma$.  Since $\Delta$ can be any element of $\mathcal{A}$, the size of $\tilde{\mathcal{A}}$ is,
\begin{equation}
|\tilde{\mathcal{A}}| = |\mathcal{A}|\cdot|\mathcal{A}\backslash\mathcal{A}_{\fals}|.
\end{equation}
We take the convention that when $\Gamma = \tru$, the statement $\Delta|\Gamma$ reduces to the proposition $\Delta$.  From here, the subspace of propositions $\mathcal{A}\subset\tilde{\mathcal{A}}$ is defined as the \textit{truth} conditioned subspace, i.e.
\begin{equation}
\mathcal{A} \stackrel{\mathrm{def}}{=} \left\{\Delta|\Phi \in \tilde{\mathcal{A}}\,\middle|\,\Phi = \tru\right\}
\end{equation}

There are various special cases of statements in $\tilde{\mathcal{A}}$, such as when $\Delta|\Gamma$ represents a \textit{syntactic} or \textit{semantic} consequence.  In either of those cases we identify,
\begin{equation}
\Delta |\Gamma \Leftrightarrow \Gamma \vdash \Delta \quad \mathrm{or} \quad \Delta|\Gamma \Leftrightarrow \Gamma \vDash \Gamma.
\end{equation}   
\paragraph{Deductive conditionals ---}
We can also represent logical arguments, such as those in \hyperref[inferencerules]{Section 6}, as conditionals.  Since any valid logical argument can be written as a truth-functional tautology, and by assuming that the argument is sound, we can construct a conditional proposition.  Consider for example, a syllogism $\Gamma$ which consists of $n$ propositional steps,
\begin{equation}
\Gamma \Leftrightarrow a_1\wedge a_2 \wedge \dots \wedge a_n.
\end{equation}
A theorem $b \in \mathcal{A}$ which is proved by $\Gamma$ can be written in semantic notation as,
\begin{equation}
a_1,a_2,\dots,a_n \vdash b.
\end{equation}
As a truth-functional tautology we can write the above as,
\begin{equation}
(a_1\wedge a_2\wedge\dots\wedge a_n) \Rightarrow b.
\end{equation}
Whenever we assume that the premises are true, we can write the logical argument as a special case of a conditional,
\begin{equation}
b|\Gamma \Leftrightarrow \Gamma \vdash b.
\end{equation}
Logical arguments are then represented as a subset $\mathcal{O}\subset\tilde{\mathcal{A}}$ of the extended proposition space.

There are many ways of splitting the space $\tilde{\mathcal{A}}$ into subsets.  Perhaps the simplest splitting is into the two disjoint subsets $\tilde{\mathcal{A}}_{\nu}$ and $\tilde{\mathcal{A}}_{\bar{\nu}}$, which are the set of statements in $\tilde{\mathcal{A}}$ which have a well-defined truth value and those which do not respectively.  Obviously we have that, $\tilde{\mathcal{A}}_{\nu}\cap\tilde{\mathcal{A}}_{\bar{\nu}} = \emptyset$ and $\tilde{\mathcal{A}}_{\nu}\cup\tilde{\mathcal{A}}_{\bar{\nu}} = \tilde{\mathcal{A}}$.  We also have that the original proposition space $\mathcal{A} \subset \tilde{\mathcal{A}}_{\nu}$.

There are other special subsets of $\tilde{\mathcal{A}}$ that will become important later.  The first type is called a $\Gamma$-contextual subset, $\tilde{\mathcal{A}}_{\Gamma} \subset\tilde{\mathcal{A}}$, which contains all statements whose second argument is $\Gamma$,
\begin{equation}
\tilde{\mathcal{A}}_{\Gamma} \stackrel{\mathrm{def}}{=} \left\{\Delta|\Phi \in \tilde{\mathcal{A}}\,\middle|\, \Phi = \Gamma\right\}.
\end{equation}
The $\Gamma$-contextual subspaces will have important properties that we will discuss in the rest of this section.  The number of possible $\Gamma$-contextual subsets is equal to the size of $|\mathcal{A}\backslash\mathcal{A}_{\fals}|$, and each of them are almost disjoint,
\begin{equation}
\forall i\neq j : \tilde{\mathcal{A}}_{\Gamma_i}\cap\tilde{\mathcal{A}}_{\Gamma_j} = \{\tru,\fals\} \quad \mathrm{and} \quad \bigcup_{i=1}^{|\mathcal{A}\backslash\mathcal{A}_{\fals}|}\tilde{\mathcal{A}}_{\Gamma_i} = \tilde{\mathcal{A}}.
\end{equation}

\paragraph{The Boolean algebra over $\tilde{\mathcal{A}}_{\Gamma}$ ---}
We will inherit the Boolean algebra of the connectives $\{\wedge,\vee,\neg\}$ over particular subspaces of $\tilde{\mathcal{A}}$.  The simplest construction concerns the $\Gamma$-contextual subspaces $\tilde{\mathcal{A}}_{\Gamma}$, for which we define the following axioms,
\begin{axioms}[Contextual Boolean algebra]
	Let $\tilde{\mathcal{A}}_{\Gamma}$ be a $\Gamma$-contextual subspace of $\tilde{\mathcal{A}}$.  We define the context dependent Boolean operations on $\tilde{\mathcal{A}}_{\Gamma}$ as the following,
	\begin{enumerate}
		\item \textbf{Negation} - Context dependent negation has the following property,
		\begin{equation}
		\forall a|\Gamma \in \tilde{\mathcal{A}}_{\Gamma} : \neg(a|\Gamma) = \neg a|\Gamma.
		\end{equation}
		\item \textbf{Conjunction} - Context dependent conjunction is defined as,
		\begin{equation}
		\forall a|\Gamma, b|\Gamma \in \tilde{\mathcal{A}}_{\Gamma} : a|\Gamma \wedge b|\Gamma = [a \wedge b]|\Gamma.
		\end{equation}
		\item \textbf{Disjunction} - Likewise, context dependent disjunction is,
		\begin{equation}
		\forall a|\Gamma, b|\Gamma \in \tilde{\mathcal{A}}_{\Gamma} : a|\Gamma \vee b|\Gamma = [a \vee b]|\Gamma.
		\end{equation}
	\end{enumerate}
	The collection $(\tilde{\mathcal{A}}_{\Gamma},\wedge,\vee,\neg)$ forms a Boolean algebra over $\tilde{\mathcal{A}}_{\Gamma}$.  
\end{axioms}
Each of these operators preserves context, i.e. $\neg:\tilde{\mathcal{A}}_{\Gamma}\rightarrow \tilde{\mathcal{A}}_{\Gamma}$ and $\{\wedge,\vee\}:\tilde{\mathcal{A}}_{\Gamma}\times\tilde{\mathcal{A}}_{\Gamma}\rightarrow\tilde{\mathcal{A}}_{\Gamma}$.  The conjunction and disjunction between two different context dependent subspaces is not well defined, e.g. statements such as $a|\Gamma \wedge b|\Phi \notin \tilde{\mathcal{A}}$.  While some elements, $a|\Gamma, b|\Gamma \in \tilde{\mathcal{A}}$, may not have a well defined truth value, their conjunction or disjunction could.  For example, if $a|\Gamma$ and $b|\Gamma$ are mutually exclusive, then their conjunction is false independent of $\Gamma$,
\begin{equation}
a|\Gamma \wedge b|\Gamma = [a\wedge b]|\Gamma = \fals|\Gamma \Leftrightarrow \fals.
\end{equation}
So, even in the situation where $a|\Gamma, b|\Gamma \in \tilde{\mathcal{A}}_{\bar{\nu}}$, their conjunction is in $\tilde{\mathcal{A}}_{\nu}$.  If they also happen to be exhaustive then their disjunction is true,
\begin{equation}
a|\Gamma \vee b|\Gamma = [a \vee b]|\Gamma = \tru|\Gamma \Leftrightarrow \tru.
\end{equation}
Thus, context dependent conjunction and disjunction can take two elements from $\tilde{\mathcal{A}}_{\Gamma}\cap\tilde{\mathcal{A}}_{\bar{\nu}}$ and send them to $\tilde{\mathcal{A}}_{\Gamma}\cap\tilde{\mathcal{A}}_{\nu}$.    

\paragraph{Contextualization ---}
The conditioning, or \textit{contextualizing} of two elements in $\tilde{\mathcal{A}}$ has important properties.  First, for generic atomic propositions in $\mathcal{A}$, the conditional operator ``$|$'' sends, $\mathcal{A}\times\mathcal{A}\rightarrow \tilde{\mathcal{A}}_{\pi_2}$ where $\pi_2$ is the projection onto the second argument, i.e.
\begin{equation}
\forall a,b\in \mathcal{A} : a|b \in \tilde{\mathcal{A}}_{b}, \quad \mathrm{and} \quad b|a \in \tilde{\mathcal{A}}_{a}.
\end{equation}
In general we have that,
\begin{axiom}[Contextualization is distributive]
	Let $\tilde{\mathcal{A}}_{\Gamma}$ be a $\Gamma$-contextual subspace of $\tilde{\mathcal{A}}$.  Given a generic element $a|\Gamma \in \tilde{\mathcal{A}}_{\Gamma}$ and an element $\Delta \in \mathcal{A}$ which is not equivalent to $\fals$, the conditioning of $a|\Gamma$ on $\Delta$ is defined as $|:\tilde{\mathcal{A}}_{\Gamma}\times\mathcal{A}\rightarrow\tilde{\mathcal{A}}_{\Gamma\times \pi_2}$,
	\begin{equation}
	\forall a|\Gamma \in \tilde{\mathcal{A}}_{\Gamma} : \forall \Delta \not\Leftrightarrow \fals \in \mathcal{A} : [a|\Gamma]|\Delta = a|\Gamma\wedge \Delta \in \tilde{\mathcal{A}}_{\Gamma\wedge \Delta},
	\end{equation}
\end{axiom}
This is similar in behavior to the distributive property of material implication \cite{LewisLangford}.  

\subsection{Partial order and disjunction}\label{partialordersection}
The disjunction connection has the special property of generating a \textit{partial order} on subsets of $\tilde{\mathcal{A}}$ which contain mutually exclusive statements.  A mutually exclusive subset $\tilde{\mathcal{A}}_{\vee} \subseteq \tilde{\mathcal{A}}$ is such that,
\begin{equation}
\tilde{\mathcal{A}}_{\vee} = \left\{a|\Gamma\in \tilde{\mathcal{A}}_{\Gamma}\,\middle|\,\forall b|\Gamma \in \tilde{\mathcal{A}}_{\Gamma} : a|\Gamma\wedge b|\Gamma \Leftrightarrow \fals\right\}.
\end{equation}
If the subset $\tilde{\mathcal{A}}_{\vee}$ is \textit{exhaustive} then, 
\begin{equation}
\bigvee_{a|\Gamma\in\tilde{\mathcal{A}}_{\vee}}\tilde{\mathcal{A}}_{\vee} \Leftrightarrow \tru.\label{exhaustive}
\end{equation}
A partial order on $\tilde{\mathcal{A}}_{\vee}$ is defined as,
\begin{define}\label{partialorder}
	Let $\tilde{\mathcal{A}}_{\vee} \subseteq \tilde{\mathcal{A}}$ be a mutually exclusive subset of $\tilde{\mathcal{A}}$.  Then, the collection $(\tilde{\mathcal{A}}_{\vee},\vee,\preceq)$ is a partially ordered set where the binary relation $\preceq$ has the following properties,
	\begin{enumerate}
		\item \textbf{Reflexive} - $\forall a|\Gamma \in \tilde{\mathcal{A}}_{\vee} : a|\Gamma \preceq a|\Gamma$.
		\item \textbf{Antisymmetry} - $\forall a|\Gamma,b|\Gamma \in \tilde{\mathcal{A}}_{\vee} : [a|\Gamma \preceq b|\Gamma]\wedge [b|\Gamma \preceq a|\Gamma] \Rightarrow [a|\Gamma =_{\preceq} b|\Gamma]$.
		\item \textbf{Transitive} - $\forall a|\Gamma,b|\Gamma,c|\Gamma \in \tilde{\mathcal{A}}_{\vee} : [a|\Gamma \preceq b|\Gamma]\wedge[b|\Gamma \preceq c|\Gamma] \Rightarrow [a|\Gamma \preceq c|\Gamma]$.
	\end{enumerate}
\end{define}
The disjunction $\vee$ on $\tilde{\mathcal{A}}_{\vee}$ establishes a hierarchy such that,
\begin{equation}
\forall a|\Gamma,b|\Gamma \in \tilde{\mathcal{A}}_{\vee} : a|\Gamma \preceq [a|\Gamma \vee b|\Gamma].
\end{equation}
If the propositions $a\preceq b$ and $b\preceq a$ are both false, then we say that $a$ and $b$ are \textit{incomparable}.  An upper bound $\top_{\vee} \in \tilde{\mathcal{A}}_{\vee}$ is an element,
\begin{equation}
\forall a|\Gamma : [\exists c|\Gamma : a|\Gamma \preceq c|\Gamma]\Rightarrow c|\Gamma = \top_{\vee} 
\end{equation}
A lower bound $\bot_{\vee} \in \tilde{\mathcal{A}}_{\vee}$ is analogously defined,
\begin{equation}
\forall a|\Gamma : [\exists d|\Gamma : d|\Gamma \preceq a|\Gamma]\Rightarrow d|\Gamma = \bot_{\vee}
\end{equation}
If the elements of $\tilde{\mathcal{A}}_{\vee}$ are exhaustive (\ref{exhaustive}), then every element $a|\Gamma \in \tilde{\mathcal{A}}_{\vee}$ has the same \textit{upper bound}, which is $\top_{\vee} = \tru$,
\begin{equation}
\forall a|\Gamma \in \tilde{\mathcal{A}}_{\vee} : a|\Gamma \preceq \tru.\label{upperbound}
\end{equation}
Likewise, there is a \textit{lower bound}, which is $\bot_{\vee} = \fals$,
\begin{equation}
\forall a|\Gamma \in \tilde{\mathcal{A}}_{\vee} : \fals \preceq a|\Gamma.\label{lowerbound}
\end{equation}   
Whenever these two elements (\ref{upperbound}) and(\ref{lowerbound}) exist, then the partially ordered set $\tilde{\mathcal{A}}_{\vee}$ forms a \textit{lattice}, $\mathfrak{L} = (\tilde{\mathcal{A}}_{\vee},\vee_{\mathfrak{L}},\wedge_{\mathfrak{L}},\preceq)$.  The conjunction and disjunction connectives are identified as the \textit{meet} and \textit{join} relations in the lattice respectively.  The \textit{meet}, $\vee_{\mathfrak{L}}$, of any two elements $a|\Gamma$ and $b|\Gamma$ of $\tilde{\mathcal{A}}_{\vee}$ is their \textit{greatest lower bound}, i.e. the element $c|\Gamma$,
\begin{align}
\forall a|\Gamma, b|\Gamma \in \tilde{\mathcal{A}}_{\vee} &: \forall d|\Gamma \in \tilde{\mathcal{A}}_{\vee} : d|\Gamma \preceq a|\Gamma\nonumber\\
&: \forall e|\Gamma \in \tilde{\mathcal{A}}_{\vee} : e|\Gamma \preceq b|\Gamma\nonumber\\
&: \exists c|\Gamma \in \tilde{\mathcal{A}}_{\vee}\nonumber\\
&: [d|\Gamma \preceq c|\Gamma]\wedge[e|\Gamma\preceq c|\Gamma] \Rightarrow a|\Gamma\wedge_{\mathfrak{L}} b|\Gamma = c|\Gamma.
\end{align}
Likewise, the join, $\vee_{\mathfrak{L}}$, of any two elements $a|\Gamma$ and $b|\Gamma$ is their \textit{least upper bound},
\begin{align}
\forall a|\Gamma, b|\Gamma \in \tilde{\mathcal{A}}_{\vee} &: \forall d|\Gamma \in \tilde{\mathcal{A}}_{\vee} : a|\Gamma \preceq d|\Gamma\nonumber\\
&: \forall e|\Gamma \in \tilde{\mathcal{A}}_{\vee} : b|\Gamma \preceq d|\Gamma\nonumber\\
&: \exists c|\Gamma \in \tilde{\mathcal{A}}_{\vee}\nonumber\\
&: [c|\Gamma \preceq d|\Gamma]\wedge[c|\Gamma\preceq e|\Gamma] \Rightarrow a|\Gamma\vee_{\mathfrak{L}} b|\Gamma = c|\Gamma.
\end{align}
Given that the elements of $\tilde{\mathcal{A}}_{\vee}$ are all within the same context, $\tilde{\mathcal{A}}_{\vee} \subseteq \tilde{\mathcal{A}}_{\Gamma}$, then the meet and join are simply identified as,
\begin{align}
\forall a|\Gamma,b|\Gamma \in \tilde{\mathcal{A}}_{\vee} : a|\Gamma \wedge_{\mathfrak{L}} b|\Gamma &\Leftrightarrow a \wedge b|\Gamma\\
\forall a|\Gamma,b|\Gamma \in \tilde{\mathcal{A}}_{\vee} : a|\Gamma \vee_{\mathfrak{L}} b|\Gamma &\Leftrightarrow a \vee b|\Gamma
\end{align}
Whenever the distributive properties of $\wedge_{\mathfrak{L}}$ and $\vee_{\mathfrak{L}}$ hold, the lattice $\mathfrak{L}$ is called a \textit{distributive lattice}.  If for every element $a|\Gamma \in \tilde{\mathcal{A}}_{\vee}$ we also have its negation, $\neg a|\Gamma \in \tilde{\mathcal{A}}_{\vee}$, such that,
\begin{equation}
\forall a|\Gamma \in \tilde{\mathcal{A}}_{\vee} : \exists \neg a|\Gamma \in \tilde{\mathcal{A}}_{\vee} : [a|\Gamma \vee_{\mathfrak{L}} \neg a|\Gamma = \tru] \wedge [a|\Gamma \wedge_{\mathfrak{L}} \neg a|\Gamma = \fals],
\end{equation}
then the lattice $\mathfrak{L}$ is called a \textit{complemented lattice}.  Thus, we can identify the complete Boolean algebra over $\tilde{\mathcal{A}}_{\vee}$, when the atomic elements of $\tilde{\mathcal{A}}_{\vee}$ are mutually exclusive and exhaustive, as a \textit{complemented distributive lattice}, $(\tilde{\mathcal{A}}_{\vee},\wedge,\vee,\neg) \rightarrow (\tilde{\mathcal{A}}_{\vee},\wedge_{\mathfrak{L}},\vee_{\mathfrak{L}},\preceq)$.  The lattice $\mathfrak{L}$ has the additional structure of the $\preceq$ relation, which is identical to the material implication, which allows one to connect $\wedge$ and $\vee$ with $\neg$. 

The disjunction of two statements in $\tilde{\mathcal{A}}_{\vee}$ is always at higher order than the statements themselves.  Why exactly is this useful and what does it mean?  If one examines the properties in (\ref{partialorder}) they will realize that the binary relation $\preceq$ is equivalent to the material implication $\Rightarrow$.  This allows us to prove the useful result that disjunction in $\tilde{\mathcal{A}}_{\vee}$ preserves order.
\begin{theorem}
	Let $(\tilde{\mathcal{A}}_{\vee},\vee,\preceq)$ be a partially ordered set with respect to $\vee$.  For any $a|\Gamma,b|\Gamma$ and $c|\Gamma$ in $\tilde{\mathcal{A}}_{\vee}$, the disjunction preserves order,
	\begin{equation}
	\forall a|\Gamma, b|\Gamma, c|\Gamma \in \tilde{\mathcal{A}}_{\vee} : [a|\Gamma \preceq b|\Gamma] \Rightarrow [a|\Gamma \vee c|\Gamma] \preceq [b|\Gamma \vee c|\Gamma].\label{disprec}
	\end{equation}
\end{theorem}
\begin{proof}
	Assuming that $a|\Gamma, b|\Gamma$ and $c|\Gamma$ are mutually exclusive, then it is true that,
	\begin{align}
	\forall a|\Gamma, b|\Gamma \in \tilde{\mathcal{A}}_{\vee} : a|\Gamma \wedge b|\Gamma &\Leftrightarrow \fals\nonumber\\
	\neg a|\Gamma \vee \neg b|\Gamma &\Leftrightarrow \tru.
	\end{align}
	Let $a|\Gamma \rightarrow a, b|\Gamma \rightarrow b, c|\Gamma\rightarrow c$ for ease of notation.  The theorem (\ref{disprec}) can be written using material implication as,
	\begin{equation}
	\forall a, b, c \in \tilde{\mathcal{A}}_{\vee} : [a \Rightarrow b] \Rightarrow \left([a \vee c] \Rightarrow [b \vee c]\right).
	\end{equation}
	Or semantically as,
	\begin{equation}
	(a\Rightarrow b),\neg(a\wedge b),\neg(a\wedge c),\neg(b\wedge c) \vdash (a\vee c)\Rightarrow (b\vee c).
	\end{equation}
	The second term can be written,
	\begin{align}
	\left([a \vee c] \Rightarrow [b \vee c]\right) &\Leftrightarrow \neg \left([a \vee c] \wedge \neg[b \vee c]\right)\nonumber\\
	&\Leftrightarrow \neg\left([a\vee c]\wedge [\neg b \wedge \neg c]\right)\nonumber\\
	&\Leftrightarrow \neg\left([\neg b\wedge \neg c \wedge a]\vee[c \wedge \neg c \wedge b]\right)\nonumber\\
	&\Leftrightarrow \neg\left([a\wedge \neg b]\wedge \neg c\right)\nonumber\\
	&\Leftrightarrow [a\Rightarrow b]\vee c.
	\end{align}
	Therefore, the theorem can be written,
	\begin{equation}
	\forall a,b,c\in\tilde{\mathcal{A}}_{\vee} : [a\Rightarrow b] \Rightarrow \left([a\Rightarrow b]\vee c\right),
	\end{equation}
	which is just disjunction introduction, which is always a valid rule of inference.  Using the partial order notation, the theorem can also be written as,
	\begin{equation}
	\forall a|\Gamma,b|\Gamma,c|\Gamma\in\tilde{\mathcal{A}}_{\vee} : [a|\Gamma\preceq b|\Gamma] \Rightarrow \left([a|\Gamma\preceq b|\Gamma]\vee c|\Gamma\right),
	\end{equation}
\end{proof}
One way of visualizing partially ordered sets is with a \textit{Hasse diagram} \cite{Hasse}.  The diagram represents the partial order as a graph which is oriented vertically.  As one moves upwards along the graph, one moves upward along the partial order.  Considering three mutually exclusive statements $a, b$ and $c$, the corresponding Hasse diagram with respect to disjunction is,
\begin{figure}[H]
	\centering
	\begin{tikzpicture}[scale=.7]
	\node (true) at (-5,4) {$a\vee_{\mathfrak{L}} b\vee_{\mathfrak{L}} c$};
	\node (ab) at (-7,2) {$a\vee_{\mathfrak{L}} b$};
	\node (ac) at (-5,2) {$a\vee_{\mathfrak{L}} c$};
	\node (bc) at (-3,2) {$b \vee_{\mathfrak{L}} c$};
	\node (a) at (-7,0) {$a$};
	\node (b) at (-5,0) {$b$};
	\node (c) at (-3,0) {$c$};
	\node (false) at (-5,-2) {$\bot$};
	\draw (false) -- (a);
	\draw (false) -- (b);
	\draw (false) -- (c);
	\draw (a) -- (ab);
	\draw (a) -- (ac);
	\draw (b) -- (ab);
	\draw (b) -- (bc);
	\draw (c) -- (ac);
	\draw (c) -- (bc);
	\draw (ac) -- (true);
	\draw (ab) -- (true);
	\draw (bc) -- (true);
	\end{tikzpicture}
	\caption{Hasse diagram for the distributive lattice $\mathfrak{L}_{\vee}$ of $\vee$ over a set of three mutually exclusive and exhaustive propositions $\{a,b,c\}\in\tilde{\mathcal{A}}_{\vee}$.}
\end{figure}
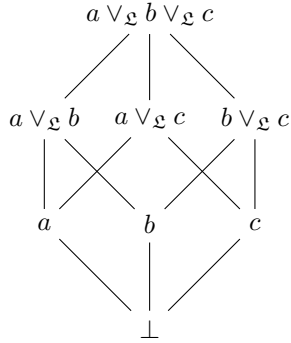
The Hasse diagram for any set of $n$-mutually exclusive statements forms an $n$-hypercube with respect to disjunction.  Interestingly enough, the Hasse diagram for two statements $a$ and $b$ forms a tesseract with respect to the binary connectives,
\begin{figure}[H]
	\centering
	\begin{tikzpicture}[scale=.7]
	\node (true) at (0,6) {$\top$};
	\node (nor) at (-4.5,-3) {$a \downarrow b$};
	\node (notimp) at (-1.5,-3) {$a \not\Rightarrow b$};
	\node (notimp2) at (1.5,-3) {$b\not\Rightarrow a$};
	\node (and) at (4.5,-3) {$a\wedge b$};
	\node (nota) at (-6.25,0) {$\neg b$};
	\node (notb) at (-3.75,0) {$\neg a$};
	\node (xor) at (-1.25,0) {$a \dot{\vee} b$};
	\node (bi) at (1.25,0) {$a \Leftrightarrow b$};
	\node (a) at (3.75,0) {$a$};
	\node (b) at (6.25,0) {$b$};
	\node (nand) at (-4.5,3) {$a \uparrow b$};
	\node (impb) at (-1.5,3) {$b\Rightarrow a$};
	\node (impa) at (1.5,3) {$a\Rightarrow b$};
	\node (or) at (4.5,3) {$a\vee b$};
	\node (false) at (0,-6) {$\bot$};
	\draw (false) -- (nor);
	\draw (false) -- (notimp);
	\draw (false) -- (notimp2);
	\draw (false) -- (and);
	\draw (nor) -- (nota);
	\draw (nor) -- (notb);
	\draw[green] (nor) -- (bi);
	\draw (and) -- (a);
	\draw (and) -- (b);
	\draw[green] (and) -- (bi);
	\draw (notimp) -- (nota);
	\draw[green] (notimp) -- (xor);
	\draw (notimp) -- (a);
	\draw (notimp2) -- (notb);
	\draw (notimp2) -- (b);
	\draw[green] (notimp2) -- (xor);
	\draw (nota) -- (nand);
	\draw (nota) -- (impb);
	\draw (notb) -- (nand);
	\draw (notb) -- (impa);
	\draw[green] (xor) -- (nand);
	\draw[green] (xor) -- (or);
	\draw[green] (bi) -- (impa);
	\draw[green] (bi) -- (impb);
	\draw (a) -- (impb);
	\draw (a) -- (or);
	\draw (b) -- (impa);
	\draw (b) -- (or);
	\draw (nand) -- (true);
	\draw (impb) -- (true);
	\draw (impa) -- (true);
	\draw (or) -- (true);
	\end{tikzpicture}
	\caption{Hasse diagram for the partially ordered set of binary connectives on two elements $a, b \in \mathcal{A}$.  The diagram corresponds to the graph of a tesseract.  Each level in the ordering corresponds to the number of possible values of truth there are in the image of each map, e.g. $\neg b$ is true for two out of four possible inputs, while $a\uparrow b$ is true for three our of four.}
\end{figure}
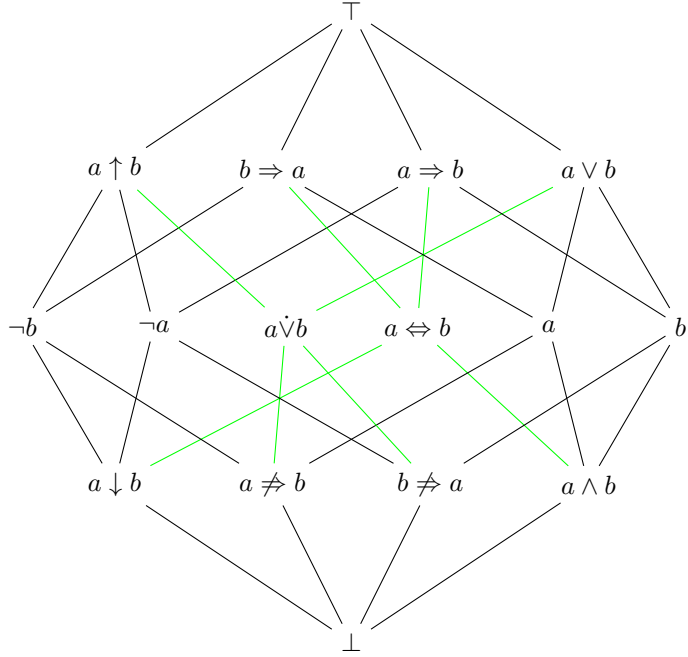
Now that we have extended Boolean algebra to include generic conditional statements we will use this as our universe of discourse for developing an inductive inference.
\subsection{Summary}

In this chapter we reviewed the basic structure of propositional and predicate logic, as well as the common extensions of modal and conditional logic.  Despite numerous efforts by Stalnaker \cite{Stalnaker}, Chellas \cite{Chellas} and Lewis \cite{Lewis} we saw that these standard logics are incapable of addressing situations in which one has incomplete information, which requires additional structure that comes from inductive inference.  As an intermediate step to the inductive framework, we constructed a Boolean logic over the \textit{extended proposition space}, which allows the introduction of conditional statements which do not contain complete information.  In the next section we will construct a proper theory of inductive inference which will allow us to address questions in which there is some uncertainty in the given propositions.

\section{Inductive inference}\label{chapter3}
\epigraph{What is there then that can be taken as true? Perhaps only this one thing, that nothing at all is certain.}{\textit{Ren\'{e} Descartes}}
\epigraph{Probability theory is nothing but common sense reduced to calculation.}{\textit{Pierre-Simon Laplace}, 1819}
Our goal in this chapter is to develop an inductive inference for reasoning when one has incomplete information.  Having a lack of complete information about our premises means that we are uncertain about our conclusions, and thus, inductive inference is sometimes developed as a means to deal with this uncertainty.  In this sense, inductive inference is often called \textit{uncertainty quantification}, or \textit{uncertain reasoning} \cite{Paris}.  To deal with these uncertainties, we will develop an inductive system, $\mathcal{P}$ (not to be confused with the space of predicates), in parallel to the deductive system $\mathcal{D}_c$ of the previous chapter.  Since the system $\mathcal{D}_c$ is so successful in its domain of applicability, it is instructive to use as a guide for the construction of $\mathcal{P}$.  The system $\mathcal{P}$ will have several analogous structures to that in $\mathcal{D}_c$ and we will attempt to preserve as much of the structure of $\mathcal{D}_c$ as possible.

The first object we must specify is the universe of discourse, or the subject matter of $\mathcal{P}$.  Since the truth values of elements $a|\Gamma \in \tilde{\mathcal{A}}$ are not known, or necessarily well defined, we instead seek to quantify \textit{beliefs} about their possible truth value.  A belief is a \textit{judgement} that one assigns to a particular statement, which qualifies how strongly one accepts that the statement is true.  Our first goal will be to quantify these beliefs.  

To quantify beliefs, we will assign a \textit{degree of belief} \cite{CatichaBook} to arguments or statements in $\tilde{\mathcal{A}}$.  These degrees of belief can be represented by a space $\Xi$ which will somehow be connected to the statements $\tilde{\mathcal{A}}$\footnote{As we will see the space $\Xi$ is simply a space of functions over $\tilde{\mathcal{A}}$.}.  The quantitative analysis to be developed for $\Xi$ will depend on criteria that will be imposed on the framework $\mathcal{P}$ through various constraints.  These criteria define a \textit{rationale}, or system of reasoning, which we label $\mathcal{R}$.  The rationale is the inductive analogy to the axioms $\mathcal{I}$ from $\mathcal{D}_c$.  In the very least, the rationale should determine how the space $\Xi$ is related to statements in $\tilde{\mathcal{A}}$ which are built through the connectives $\Omega$, and as a requirement, these relationships must be internally consistent.   

There are other approaches to inductive inference which suggest assigning various \textit{degrees} other than \textit{beliefs} to some underlying universe of discourse.  These include a \textit{degree of plausibility} \cite{Jaynes}, a \textit{degree of truth} \cite{Zadeh}\footnote{The use of the word \textit{truth} is perhaps a bit too suggestive in that our inductive framework is a metatheory which assigns truth, and so this terminology should be avoided.} or a \textit{degree of reasonable expectation} \cite{Cox}.  Some others are a \textit{degree of implication} \cite{KnuthSkilling}, or a \textit{degree of propensity} \cite{Popper}.  What all of these suggestions have in common is that one should assign a \textit{degree} of ``something'' to statements about which the \textit{truth} is either unknown, or undefined.  While the language seems to indicate a difference in interpretation among all these suggestions -- as we will see -- any reasonable assignment of a \textit{degree} to a statement of partial knowledge must be manipulated according to the rules of probability theory \cite{CatichaBook}.

A first question one might ask is, ``what precisely do we mean by \textit{rational}?''  As discussed in \cite{CatichaBook}, the precise meaning of the word \textit{rational} is difficult to pin down.  We take the pragmatic suggestion in \cite{CatichaBook} and define \textit{rational} as a compliment to the modes of reasoning which seem to lead to reliable conclusions.  If our inference framework is working out there in the real world, then we claim it must be rational.

The rationale $\mathcal{R}$ can be developed by imposing a set of constraints.  For example, in order to be as general as possible, the rationale must be able to be used in any situation of inductive reasoning.  We would not want to adopt a rationale that only applies to the stock market, or to medical trials, or to quantum mechanics.  Thus, the rationale must be independent of the particular subject matter within $\tilde{\mathcal{A}}$, so that it adheres to a principle of universal applicability, 
\begin{principle}[Universal Applicability]\label{universal}
	The inference framework should have wide appeal and universal applicability.
\end{principle}

According to the deductive system $\mathcal{D}_c$, any universe of discourse $\tilde{\mathcal{A}}$ necessarily forms a collection of interconnected statements which cannot all behave independently.  For example, the truth value of the proposition $a \vee b \in \mathcal{A}$ necessarily depends on the truth of $a$ and $b$, and cannot be assigned arbitrarily.  The same must be true for any beliefs we assign to statements in $\tilde{\mathcal{A}}$.  Thus, the relationship between $\Xi$ and $\tilde{\mathcal{A}}$ can be understood as an \textit{interconnected web of beliefs}, which we will label $\Xi(\tilde{\mathcal{A}})$.

\begin{figure}[H]
	\centering
	\begin{tikzpicture}[scale=.70]
	\node[style={draw,shape=circle,fill=black,scale=.4}] (0) at (-0.24642748898702138,4.903397528516398) {};
	\node[style={draw,shape=circle,fill=black,scale=.4}] (1) at (-8.171180732664602,2.5093670533012284) {};
	\node[style={draw,shape=circle,fill=black,scale=.4}] (2) at (-5.353194174725346,-2.8076283463238143) {};
	\node[style={draw,shape=circle,fill=black,scale=.4}] (3) at (-0.6086619503266277,4.363160038862827) {};
	\node[style={draw,shape=circle,fill=black,scale=.4}] (4) at (4.731072287184425,-0.19375862872988436) {};
	\node[style={draw,shape=circle,fill=black,scale=.4}] (5) at (0.9084807818566407,3.4792404301187334) {};
	\node[style={draw,shape=circle,fill=black,scale=.4}] (6) at (1.3900183154146837,6.772476205851982) {};
	\node[style={draw,shape=circle,fill=black,scale=.4}] (7) at (-0.6870722816541687,2.026552927220406) {};
	\node[style={draw,shape=circle,fill=black,scale=.4}] (8) at (-2.2923914969829813,-5.805322957691579) {};
	\node[style={draw,shape=circle,fill=black,scale=.4}] (9) at (0.9871942849760509,-4.166153886799797) {};
	\node[style={draw,shape=circle,fill=black,scale=.4}] (10) at (0.5874950040975324,-0.33304194964253925) {};
	\node[style={draw,shape=circle,fill=black,scale=.4}] (11) at (-0.8808804394247889,0.37594074080131734) {};
	\node[style={draw,shape=circle,fill=black,scale=.4}] (12) at (-3.684130547652636,-4.526860127123821) {};
	\node[style={draw,shape=circle,fill=black,scale=.4}] (13) at (-1.5096550551325638,-0.83816516812518) {};
	\node[style={draw,shape=circle,fill=black,scale=.4}] (14) at (-5.630447572400615,-0.9395857037195199) {};
	\node[style={draw,shape=circle,fill=black,scale=.4}] (15) at (3.962194344156046,0.3092399997754269) {};
	\node[style={draw,shape=circle,fill=black,scale=.4}] (16) at (-0.17840629652140746,2.0901424740776697) {};
	\node[style={draw,shape=circle,fill=black,scale=.4}] (17) at (-0.07082686917595846,3.983254305885911) {};
	\node[style={draw,shape=circle,fill=black,scale=.4}] (18) at (2.725974422812101,-0.9068692317132425) {};
	\node[style={draw,shape=circle,fill=black,scale=.4}] (19) at (0.014756285628729673,-1.2523744424117127) {};
	\node[style={draw,shape=circle,fill=black,scale=.4}] (20) at (-2.2652344530308475,-2.018828759865789) {};
	\node[style={draw,shape=circle,fill=black,scale=.4}] (21) at (3.2854318081267078,1.070137363373548) {};
	\node[style={draw,shape=circle,fill=black,scale=.4}] (22) at (-1.0401532648112524,-2.6287201108212432) {};
	\node[style={draw,shape=circle,fill=black,scale=.4}] (23) at (1.9158312707643175,-1.3809274189836827) {};
	\node[style={draw,shape=circle,fill=black,scale=.4}] (24) at (-1.4101317374861542,-0.8051841178173667) {};
	\node[style={draw,shape=circle,fill=black,scale=.4}] (25) at (1.0458911632406482,1.0397533472396865) {};
	\node[style={draw,shape=circle,fill=black,scale=.4}] (26) at (-4.4906513599465505,-2.137387201088091) {};
	\node[style={draw,shape=circle,fill=black,scale=.4}] (27) at (-3.848004516112627,6.575770742741382) {};
	\node[style={draw,shape=circle,fill=black,scale=.4}] (28) at (-7.4770335521904085,0.2830583084842103) {};
	\node[style={draw,shape=circle,fill=black,scale=.4}] (29) at (0.39560577111813094,1.942234424909784) {};
	\node[style={draw,shape=circle,fill=black,scale=.4}] (30) at (2.0841158528781536,-5.591104916533054) {};
	\node[style={draw,shape=circle,fill=black,scale=.4}] (31) at (-0.7505820007585167,3.7665352301716117) {};
	\node[style={draw,shape=circle,fill=black,scale=.4}] (32) at (0.2868688300785899,-2.379731983394283) {};
	\node[style={draw,shape=circle,fill=black,scale=.4}] (33) at (-2.575455317395105,-0.9115445114635408) {};
	\node[style={draw,shape=circle,fill=black,scale=.4}] (34) at (3.213891311887558,-2.6523419516075926) {};
	\node[style={draw,shape=circle,fill=black,scale=.4}] (35) at (0.4661000928022127,-1.4053713812376594) {};
	\node[style={draw,shape=circle,fill=black,scale=.4}] (36) at (3.5175898296674255,2.351975701012196) {};
	\node[style={draw,shape=circle,fill=black,scale=.4}] (37) at (0.5364999179142421,1.0237361516865717) {};
	\node[style={draw,shape=circle,fill=black,scale=.4}] (38) at (-1.4471472225645008,-2.2230932580794507) {};
	\node[style={draw,shape=circle,fill=black,scale=.4}] (39) at (0.02140414591845549,0.6398041301603318) {};
	\node[style={draw,shape=circle,fill=black,scale=.4}] (40) at (-2.0544421754171602,-0.0038667309978214306) {};
	\node[style={draw,shape=circle,fill=black,scale=.4}] (41) at (2.2356293174919197,-3.7620447436495423) {};
	\node[style={draw,shape=circle,fill=black,scale=.4}] (42) at (3.2200739056799055,0.40473759284692856) {};
	\node[style={draw,shape=circle,fill=black,scale=.4}] (43) at (-1.5402675040259226,1.9566680368892406) {};
	\node[style={draw,shape=circle,fill=black,scale=.4}] (44) at (-1.0526482559224162,-0.31744613848674863) {};
	\node[style={draw,shape=circle,fill=black,scale=.4}] (45) at (0.43325673994630864,1.4510793347610764) {};
	\node[style={draw,shape=circle,fill=black,scale=.4}] (46) at (-0.1528101879195447,1.3149389529849163) {};
	\node[style={draw,shape=circle,fill=black,scale=.4}] (47) at (0.6037343178034267,2.0648722406458258) {};
	\node[style={draw,shape=circle,fill=black,scale=.4}] (48) at (0.5941288221464623,-4.2169794674192005) {};
	\node[style={draw,shape=circle,fill=black,scale=.4}] (49) at (5.425212010296713,-1.7361870182538193) {};
	\node[style={draw,shape=circle,fill=black,scale=.4}] (50) at (-0.2967204515242212,1.0737856214859534) {};
	\node[style={draw,shape=circle,fill=black,scale=.4}] (51) at (-2.5221121316607316,1.1952783516520245) {};
	\node[style={draw,shape=circle,fill=black,scale=.4}] (52) at (-1.733940597650447,2.2567904870649658) {};
	\node[style={draw,shape=circle,fill=black,scale=.4}] (53) at (4.077673838335878,2.1628750596657937) {};
	\node[style={draw,shape=circle,fill=black,scale=.4}] (54) at (-1.2959060795654098,1.698319185231084) {};
	\node[style={draw,shape=circle,fill=black,scale=.4}] (55) at (1.3574992376391868,-1.191767254160161) {};
	\node[style={draw,shape=circle,fill=black,scale=.4}] (56) at (0.5557832115006083,0.38003182493510285) {};
	\node[style={draw,shape=circle,fill=black,scale=.4}] (57) at (-1.6705149296321729,1.5386336207889344) {};
	\node[style={draw,shape=circle,fill=black,scale=.4}] (58) at (-2.766592415235434,0.26032227501809957) {};
	\node[style={draw,shape=circle,fill=black,scale=.4}] (59) at (0.7965376730819496,0.9795720658574023) {};
	\node[style={draw,shape=circle,fill=black,scale=.4}] (60) at (-3.3132662916270803,-7.122524025223802) {};
	\node[style={draw,shape=circle,fill=black,scale=.4}] (61) at (-2.3739199375037066,-4.755434518758451) {};
	\node[style={draw,shape=circle,fill=black,scale=.4}] (62) at (-3.537790472939938,0.6001426422008841) {};
	\node[style={draw,shape=circle,fill=black,scale=.4}] (63) at (0.15545284386970692,-0.5239250432319298) {};
	\node[style={draw,shape=circle,fill=black,scale=.4}] (64) at (0.19845099210555922,-3.3820719866650233) {};
	\node[style={draw,shape=circle,fill=black,scale=.4}] (65) at (3.1622415113708455,0.41960941823062514) {};
	\node[style={draw,shape=circle,fill=black,scale=.4}] (66) at (0.5947344304963367,1.1598187577119874) {};
	\node[style={draw,shape=circle,fill=black,scale=.4}] (67) at (-0.8040513518612743,2.358606253334604) {};
	\node[style={draw,shape=circle,fill=black,scale=.4}] (68) at (-3.0409378940712495,4.384835658451392) {};
	\node[style={draw,shape=circle,fill=black,scale=.4}] (69) at (1.342584753149798,-2.189471079143325) {};
	\node[style={draw,shape=circle,fill=black,scale=.4}] (70) at (6.885244705020319,-0.6099417650092436) {};
	\node[style={draw,shape=circle,fill=black,scale=.4}] (71) at (-4.101403049155358,1.142661384377401) {};
	\node[style={draw,shape=circle,fill=black,scale=.4}] (72) at (3.4579318520951676,-3.163084917686362) {};
	\node[style={draw,shape=circle,fill=black,scale=.4}] (73) at (0.9474381200467105,0.02551245416659136) {};
	\node[style={draw,shape=circle,fill=black,scale=.4}] (74) at (-3.125036500914456,1.4717231610413013) {};
	\draw[densely dotted] (0) -- (1);
	\draw[densely dotted] (0) -- (10);
	\draw[densely dotted] (0) -- (21);
	\draw[densely dotted] (0) -- (22);
	\draw[densely dotted] (0) -- (34);
	\draw[densely dotted] (0) -- (55);
	\draw[densely dotted] (0) -- (59);
	\draw[densely dotted] (0) -- (61);
	\draw[densely dotted] (0) -- (68);
	\draw[densely dotted] (1) -- (64);
	\draw[densely dotted] (2) -- (14);
	\draw[densely dotted] (2) -- (49);
	\draw[densely dotted] (2) -- (66);
	\draw[densely dotted] (3) -- (10);
	\draw[densely dotted] (3) -- (35);
	\draw[densely dotted] (4) -- (32);
	\draw[densely dotted] (4) -- (34);
	\draw[densely dotted] (4) -- (66);
	\draw[densely dotted] (5) -- (38);
	\draw[densely dotted] (6) -- (3);
	\draw[densely dotted] (6) -- (32);
	\draw[densely dotted] (6) -- (35);
	\draw[densely dotted] (7) -- (2);
	\draw[densely dotted] (7) -- (3);
	\draw[densely dotted] (8) -- (63);
	\draw[densely dotted] (9) -- (66);
	\draw[densely dotted] (10) -- (43);
	\draw[densely dotted] (10) -- (72);
	\draw[densely dotted] (11) -- (4);
	\draw[densely dotted] (11) -- (37);
	\draw[densely dotted] (11) -- (48);
	\draw[densely dotted] (11) -- (56);
	\draw[densely dotted] (11) -- (60);
	\draw[densely dotted] (12) -- (0);
	\draw[densely dotted] (12) -- (11);
	\draw[densely dotted] (12) -- (20);
	\draw[densely dotted] (13) -- (0);
	\draw[densely dotted] (13) -- (4);
	\draw[densely dotted] (13) -- (9);
	\draw[densely dotted] (13) -- (65);
	\draw[densely dotted] (13) -- (66);
	\draw[densely dotted] (14) -- (22);
	\draw[densely dotted] (14) -- (39);
	\draw[densely dotted] (14) -- (55);
	\draw[densely dotted] (15) -- (5);
	\draw[densely dotted] (15) -- (39);
	\draw[densely dotted] (15) -- (47);
	\draw[densely dotted] (15) -- (49);
	\draw[densely dotted] (15) -- (61);
	\draw[densely dotted] (16) -- (27);
	\draw[densely dotted] (16) -- (32);
	\draw[densely dotted] (17) -- (15);
	\draw[densely dotted] (17) -- (73);
	\draw[densely dotted] (18) -- (38);
	\draw[densely dotted] (18) -- (39);
	\draw[densely dotted] (18) -- (60);
	\draw[densely dotted] (18) -- (64);
	\draw[densely dotted] (19) -- (12);
	\draw[densely dotted] (19) -- (24);
	\draw[densely dotted] (19) -- (25);
	\draw[densely dotted] (20) -- (5);
	\draw[densely dotted] (20) -- (12);
	\draw[densely dotted] (20) -- (25);
	\draw[densely dotted] (20) -- (51);
	\draw[densely dotted] (20) -- (69);
	\draw[densely dotted] (20) -- (71);
	\draw[densely dotted] (21) -- (4);
	\draw[densely dotted] (21) -- (42);
	\draw[densely dotted] (21) -- (64);
	\draw[densely dotted] (21) -- (68);
	\draw[densely dotted] (22) -- (29);
	\draw[densely dotted] (22) -- (33);
	\draw[densely dotted] (23) -- (15);
	\draw[densely dotted] (23) -- (18);
	\draw[densely dotted] (23) -- (31);
	\draw[densely dotted] (23) -- (41);
	\draw[densely dotted] (23) -- (42);
	\draw[densely dotted] (23) -- (45);
	\draw[densely dotted] (23) -- (47);
	\draw[densely dotted] (24) -- (7);
	\draw[densely dotted] (24) -- (13);
	\draw[densely dotted] (24) -- (21);
	\draw[densely dotted] (24) -- (41);
	\draw[densely dotted] (24) -- (59);
	\draw[densely dotted] (24) -- (63);
	\draw[densely dotted] (25) -- (0);
	\draw[densely dotted] (25) -- (18);
	\draw[densely dotted] (25) -- (55);
	\draw[densely dotted] (25) -- (65);
	\draw[densely dotted] (26) -- (0);
	\draw[densely dotted] (26) -- (31);
	\draw[densely dotted] (26) -- (34);
	\draw[densely dotted] (26) -- (39);
	\draw[densely dotted] (26) -- (47);
	\draw[densely dotted] (26) -- (65);
	\draw[densely dotted] (26) -- (73);
	\draw[densely dotted] (27) -- (5);
	\draw[densely dotted] (27) -- (24);
	\draw[densely dotted] (28) -- (53);
	\draw[densely dotted] (28) -- (63);
	\draw[densely dotted] (29) -- (19);
	\draw[densely dotted] (29) -- (67);
	\draw[densely dotted] (30) -- (0);
	\draw[densely dotted] (30) -- (36);
	\draw[densely dotted] (30) -- (61);
	\draw[densely dotted] (30) -- (65);
	\draw[densely dotted] (30) -- (73);
	\draw[densely dotted] (30) -- (74);
	\draw[densely dotted] (31) -- (11);
	\draw[densely dotted] (31) -- (17);
	\draw[densely dotted] (31) -- (37);
	\draw[densely dotted] (31) -- (44);
	\draw[densely dotted] (32) -- (72);
	\draw[densely dotted] (33) -- (2);
	\draw[densely dotted] (33) -- (7);
	\draw[densely dotted] (33) -- (18);
	\draw[densely dotted] (33) -- (48);
	\draw[densely dotted] (33) -- (56);
	\draw[densely dotted] (33) -- (59);
	\draw[densely dotted] (33) -- (71);
	\draw[densely dotted] (34) -- (16);
	\draw[densely dotted] (34) -- (23);
	\draw[densely dotted] (34) -- (36);
	\draw[densely dotted] (34) -- (37);
	\draw[densely dotted] (34) -- (43);
	\draw[densely dotted] (34) -- (52);
	\draw[densely dotted] (34) -- (61);
	\draw[densely dotted] (34) -- (63);
	\draw[densely dotted] (35) -- (11);
	\draw[densely dotted] (35) -- (18);
	\draw[densely dotted] (35) -- (23);
	\draw[densely dotted] (35) -- (25);
	\draw[densely dotted] (35) -- (51);
	\draw[densely dotted] (35) -- (61);
	\draw[densely dotted] (36) -- (22);
	\draw[densely dotted] (36) -- (26);
	\draw[densely dotted] (36) -- (50);
	\draw[densely dotted] (37) -- (15);
	\draw[densely dotted] (37) -- (16);
	\draw[densely dotted] (37) -- (27);
	\draw[densely dotted] (37) -- (40);
	\draw[densely dotted] (38) -- (22);
	\draw[densely dotted] (38) -- (47);
	\draw[densely dotted] (38) -- (49);
	\draw[densely dotted] (38) -- (68);
	\draw[densely dotted] (39) -- (20);
	\draw[densely dotted] (39) -- (46);
	\draw[densely dotted] (39) -- (52);
	\draw[densely dotted] (40) -- (30);
	\draw[densely dotted] (40) -- (33);
	\draw[densely dotted] (40) -- (56);
	\draw[densely dotted] (41) -- (10);
	\draw[densely dotted] (41) -- (32);
	\draw[densely dotted] (41) -- (44);
	\draw[densely dotted] (41) -- (45);
	\draw[densely dotted] (42) -- (15);
	\draw[densely dotted] (42) -- (17);
	\draw[densely dotted] (42) -- (34);
	\draw[densely dotted] (42) -- (35);
	\draw[densely dotted] (42) -- (73);
	\draw[densely dotted] (43) -- (4);
	\draw[densely dotted] (43) -- (22);
	\draw[densely dotted] (43) -- (62);
	\draw[densely dotted] (44) -- (14);
	\draw[densely dotted] (44) -- (28);
	\draw[densely dotted] (44) -- (30);
	\draw[densely dotted] (44) -- (39);
	\draw[densely dotted] (44) -- (64);
	\draw[densely dotted] (45) -- (8);
	\draw[densely dotted] (46) -- (18);
	\draw[densely dotted] (46) -- (24);
	\draw[densely dotted] (47) -- (1);
	\draw[densely dotted] (47) -- (27);
	\draw[densely dotted] (47) -- (48);
	\draw[densely dotted] (47) -- (61);
	\draw[densely dotted] (47) -- (64);
	\draw[densely dotted] (47) -- (65);
	\draw[densely dotted] (47) -- (70);
	\draw[densely dotted] (48) -- (14);
	\draw[densely dotted] (48) -- (35);
	\draw[densely dotted] (49) -- (5);
	\draw[densely dotted] (49) -- (9);
	\draw[densely dotted] (49) -- (18);
	\draw[densely dotted] (49) -- (41);
	\draw[densely dotted] (51) -- (21);
	\draw[densely dotted] (51) -- (44);
	\draw[densely dotted] (51) -- (59);
	\draw[densely dotted] (52) -- (51);
	\draw[densely dotted] (52) -- (64);
	\draw[densely dotted] (53) -- (2);
	\draw[densely dotted] (53) -- (6);
	\draw[densely dotted] (53) -- (50);
	\draw[densely dotted] (54) -- (9);
	\draw[densely dotted] (54) -- (17);
	\draw[densely dotted] (55) -- (8);
	\draw[densely dotted] (55) -- (25);
	\draw[densely dotted] (55) -- (45);\node[style={draw,shape=circle,fill=black,scale=.4}] (0) at (-1.4136070157515375,-1.4437729087961484) {};
	\node[style={draw,shape=circle,fill=black,scale=.4}] (1) at (5.33036347019625,1.0810815913819583) {};
	\node[style={draw,shape=circle,fill=black,scale=.4}] (2) at (-1.5569287312879103,2.0950956603069564) {};
	\node[style={draw,shape=circle,fill=black,scale=.4}] (3) at (-2.3586824825503507,-0.6913324624670901) {};
	\node[style={draw,shape=circle,fill=black,scale=.4}] (4) at (-0.4862000208610338,2.260834104656081) {};
	\node[style={draw,shape=circle,fill=black,scale=.4}] (5) at (-2.5413003230445934,-3.531747503306613) {};
	\node[style={draw,shape=circle,fill=black,scale=.4}] (6) at (5.270978215488773,0.49901129017931634) {};
	\node[style={draw,shape=circle,fill=black,scale=.4}] (7) at (6.526885488556107,-7.445053567293323) {};
	\node[style={draw,shape=circle,fill=black,scale=.4}] (8) at (2.624090641570932,2.8493652674970966) {};
	\node[style={draw,shape=circle,fill=black,scale=.4}] (9) at (-2.3278079322504768,-0.7992619283003786) {};
	\draw[densely dotted] (5) -- (0);
	\draw[densely dotted] (7) -- (3);
	\draw[densely dotted] (9) -- (4);
	\draw[densely dotted] (9) -- (6);
	\draw[densely dotted] (56) -- (9);
	\draw[densely dotted] (56) -- (31);
	\draw[densely dotted] (56) -- (72);
	\draw[densely dotted] (57) -- (13);
	\draw[densely dotted] (57) -- (23);
	\draw[densely dotted] (57) -- (26);
	\draw[densely dotted] (57) -- (29);
	\draw[densely dotted] (57) -- (36);
	\draw[densely dotted] (57) -- (67);
	\draw[densely dotted] (58) -- (7);
	\draw[densely dotted] (58) -- (47);
	\draw[densely dotted] (58) -- (49);
	\draw[densely dotted] (58) -- (60);
	\draw[densely dotted] (59) -- (2);
	\draw[densely dotted] (59) -- (71);
	\draw[densely dotted] (60) -- (3);
	\draw[densely dotted] (60) -- (9);
	\draw[densely dotted] (60) -- (41);
	\draw[densely dotted] (60) -- (62);
	\draw[densely dotted] (60) -- (64);
	\draw[densely dotted] (61) -- (20);
	\draw[densely dotted] (61) -- (31);
	\draw[densely dotted] (61) -- (53);
	\draw[densely dotted] (61) -- (66);
	\draw[densely dotted] (62) -- (1);
	\draw[densely dotted] (62) -- (5);
	\draw[densely dotted] (62) -- (22);
	\draw[densely dotted] (62) -- (23);
	\draw[densely dotted] (62) -- (36);
	\draw[densely dotted] (62) -- (44);
	\draw[densely dotted] (63) -- (16);
	\draw[densely dotted] (63) -- (17);
	\draw[densely dotted] (63) -- (29);
	\draw[densely dotted] (64) -- (4);
	\draw[densely dotted] (64) -- (57);
	\draw[densely dotted] (64) -- (67);
	\draw[densely dotted] (64) -- (73);
	\draw[densely dotted] (65) -- (4);
	\draw[densely dotted] (65) -- (5);
	\draw[densely dotted] (65) -- (7);
	\draw[densely dotted] (65) -- (48);
	\draw[densely dotted] (65) -- (73);
	\draw[densely dotted] (66) -- (39);
	\draw[densely dotted] (66) -- (55);
	\draw[densely dotted] (67) -- (9);
	\draw[densely dotted] (67) -- (53);
	\draw[densely dotted] (67) -- (54);
	\draw[densely dotted] (68) -- (25);
	\draw[densely dotted] (68) -- (48);
	\draw[densely dotted] (70) -- (4);
	\draw[densely dotted] (70) -- (16);
	\draw[densely dotted] (70) -- (24);
	\draw[densely dotted] (70) -- (39);
	\draw[densely dotted] (70) -- (47);
	\draw[densely dotted] (71) -- (3);
	\draw[densely dotted] (71) -- (7);
	\draw[densely dotted] (71) -- (39);
	\draw[densely dotted] (72) -- (5);
	\draw[densely dotted] (72) -- (23);
	\draw[densely dotted] (72) -- (30);
	\draw[densely dotted] (72) -- (38);
	\draw[densely dotted] (73) -- (26);
	\draw[densely dotted] (73) -- (40);
	\draw[densely dotted] (73) -- (47);
	\draw[densely dotted] (73) -- (61);
	\draw[densely dotted] (74) -- (50);
	
	\node[shape=circle] (0) at (-8.9136070157515375,2.7437729087961484) {$\xi(a|\Gamma)$};
	\node[shape=circle] (0) at (7.4136070157515375,-7.5437729087961484) {$\xi(b|\Gamma)$};
	\node[shape=circle] (0) at (2.5136070157515375,7.1437729087961484) {$\xi(a\vee b|\Gamma)$};
	\node[shape=circle] (0) at (-4.9136070157515375,-5.0437729087961484) {$\xi(a|\Gamma\wedge b)$};
	\end{tikzpicture}
	\caption{An interconnected web of beliefs $\Xi(\tilde{\mathcal{A}})$.}
\end{figure}
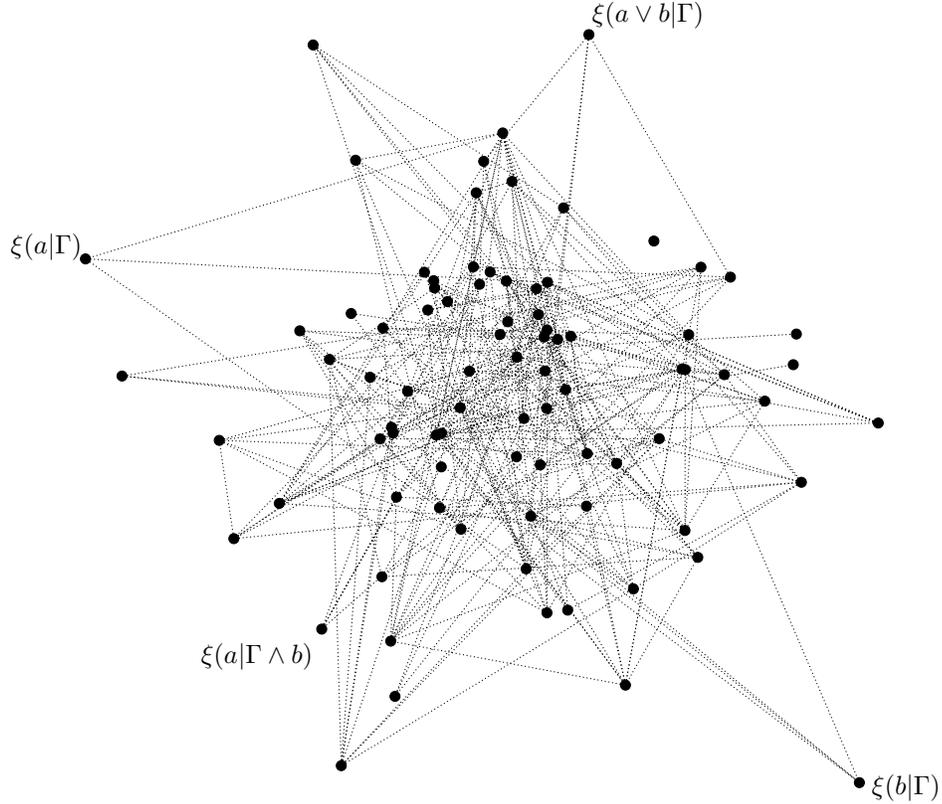
In the same way that the deductive system $\mathcal{D}_c$ is constructed to be internally consistent, we require that the web of beliefs $\Xi(\tilde{\mathcal{A}})$ also be internally consistent\footnote{As mentioned in \cite{Horvitz}, the inter-dependency of the web of beliefs is sometimes referred to as \textit{context dependency}.}.  This constraint can be expressed as another principle,    
\begin{principle}[Consistency]\label{consistent}
	The inference framework should not be self-refuting.
\end{principle}

A few comments about the web of beliefs $\Xi(\tilde{\mathcal{A}})$ are in order.  First, the beliefs $\Xi(\tilde{\mathcal{A}})$ that will be manipulated here are not the personalistic beliefs of any individual, but instead represent the beliefs of an \textit{ideally rational agent}.  An ideally rational agent (who we shall call \textit{Ira} \cite{CatichaBook}) is an agent who strictly adheres to the rationale $\mathcal{R}$, and is not subject to the practical limitations of human beings.  In this way, the rationale $\mathcal{R}$ and the web of beliefs $\Xi(\tilde{\mathcal{A}})$ become explicitly dependent.  As we will see, the beliefs of compound statements in $\Xi(\tilde{\mathcal{A}})$ are necessarily constrained by the rationale in order to impose consistency.  By specifying $\mathcal{R}$ fully, we in turn fully specify the allowed behavior of the web $\Xi(\tilde{\mathcal{A}})$.  In this way, the web represents degrees of \textit{rational} belief (DoRB) \cite{VansletteThesis}.

While the construction of the formal system $\mathcal{P}$ is certainly a philosophically satisfying endeavor, it would be a waste if $\mathcal{P}$ was unable to be used in practice.  Thus, we state the rather obvious criteria,
\begin{principle}[Practicality]\label{practical}
	The inference framework must be useful in practice.  It must allow for quantitative analysis.
\end{principle}
This principle suggests that an inductive inference should be able to be represented and manipulated by standard mathematical tools, such as calculus, geometry and topology, algebra, and etc.  We will eventually make use of all of these tools in subsequent chapters.  Practicality suggests another constraint which is used throughout the entirety of science.  If there are two or more ways of constructing a viable theory, we prefer the simplest.  This is just Occam's criteria, which should also be applied to our inductive framework.  After all, the inductive framework must be useful for science. 
\begin{principle}[Parsimony]\label{parsimony}
	The inference framework should represent the simplest possible system for conducting rational inference.
\end{principle}
In other words, superfluous structure should be avoided in order to reduce confusion and irrationality.  Both the principle of practicality and the principle of parsimony are not necessary in order to specify a consistent abstract inductive framework $\mathcal{P}$, however our ultimately goal is to conduct science, for which the tools of mathematics have been known to operate with great success.  So, why not use them?

Whatever the particular subject matter of interest $\tilde{\mathcal{A}}$ may be, our inductive framework must be able to apply to its entirety, otherwise it will not be useful or of universal applicability.  In this sense we have the following corollary with \hyperref[universal]{principle 1.1} and \hyperref[practical]{1.3},
\begin{corollary}[Completeness]\label{complete}
	Degrees of belief should be assignable to any well defined statement.	
\end{corollary}

We will construct the inductive system $\mathcal{P} = \left\{\tilde{\mathcal{A}},\Xi,\mathcal{R},\mathcal{Z}\right\}$ according to the principles just stated, together with any inductive inferences rules $\mathcal{Z}$.  The rationale $\mathcal{R}$ determines how the web of beliefs $\Xi(\tilde{\mathcal{A}})$ \textit{ought} to be connected.  We begin by determining the form of the space of degrees of belief, $\Xi$.

\subsection{Degrees of belief $\Xi$}
\epigraph{By degree of probability, we really mean, or ought to mean, degree of belief.}{\textit{De Morgan 1847}}
In order to be useful, the inference framework must allow us to compare our beliefs among different statements.  For example, if we are given two statements, $a|\Gamma$ and $b|\Gamma$, do we believe one \textit{more} or \textit{less} than the other?  Consider statements $a|\Gamma, b|\Gamma$ and $c|\Gamma$, which all belong to the same universe of discourse.  If we believe $a|\Gamma$ more than $b|\Gamma$, but also believe $b|\Gamma$ more than $c|\Gamma$, then to be rational we must believe $a|\Gamma$ more than $c|\Gamma$.  An obvious way to capture this information is to assign degrees of belief real numbers\footnote{Jaynes makes this assumption in \cite{Jaynes} where he states ``Degrees of plausibility are represented by real numbers.''  In Caticha \cite{CatichaBook} the statement is slightly different, Degrees of rational belief (or, as we shall later call them, probabilities) are represented by real numbers.''  The property that belief functions should be a single real number is often called \textit{scalar continuity} \cite{Horvitz}.},
\begin{axiom}[Degrees of belief]\label{realnumbers}
	The degree of belief, $\xi$\footnote{Throughout the literature it is common for people to use the notation $Bel(a)$ to denote a belief function.}, assigned to any statement $\Delta|\Gamma \in \tilde{\mathcal{A}}$ will be represented by a real number,
	\begin{align}
	\xi : \tilde{\mathcal{A}} &\rightarrow \mathbb{R}\nonumber\\
	\Delta|\Gamma &\mapsto \xi(\Delta|\Gamma).
	\end{align}
\end{axiom}
As pointed out by Norton \cite{Norton}, it is not necessarily obvious that a single real number should be enough to capture our degree of belief about a statement $\Delta|\Gamma \in \tilde{\mathcal{A}}$.  On the contrary, according to our principle of parsimony, it is certainly reasonable to require that whatever representation for degrees of belief one chooses, that it be the simplest possible in order to avoid confusion.  

In order for the inference framework to be useful, we must be able to compare our beliefs in a consistent and rational way.  If one were to take Norton's advice so that, for example, our assignment of degrees of belief were to map to a pair of numbers, $\xi(a|\Gamma) \rightarrow (\xi_1(a|\Gamma),\xi_2(a|\Gamma))$, it would no longer be obvious how one could compare two degrees of belief.  How does one determine if they believe $a|\Gamma$ more or less than $b|\Gamma$?  Does one pick the pair of numbers which has the largest element?  What about the Euclidean distance from the origin?  How about the larger of the two absolute differences $|\xi_2(a|\Gamma)-\xi_1(a|\Gamma)|$ and $|\xi_2(b|\Gamma)-\xi_1(b|\Gamma)|$?  There is no clear preferred choice that accomplishes the goal, and hence the assignment of multiple ``degrees'' leads to compounding assumptions, which is undesirable. 

Norton points out another potential problem with the assignment of real numbers to degrees of belief \cite{Norton}.  This has to do with the consequence of \textit{universal comparability} \cite{Horvitz}, i.e. that two real numbers are always comparable even if the statements they are quantifying are unrelated.  For example, if our degree of belief in the statement $a = $ ``tomorrow the Dow will gain 3 points'' is larger than the statement, $b = $ ``a supersymmetric particle is the best candidate for dark matter'', it is certainly true that $\xi(b) < \xi(a)$, however this comparison is not necessarily meaningful since the two statements have nothing a priori to do with one another.  Thus, the assignment of a single real number to a degree of belief introduces the potential for committing inductive fallacies.  But these fallacies are no more a problem for inductive inference as logical fallacies are for deductive inference.  While the compound statement $a\wedge b = $ ``tomorrow the Dow will gain 3 points \textit{and} a supersymmetric particle is the best candidate for dark matter,'' is a perfectly well defined proposition, the deductive system $\mathcal{D}_c$ does not suggest that it can be used as a step in a valid argument.  Despite this, we would certainly not be inclined to throw away all of deductive inference on the basis that the algebra may lead us to an invalid argument.  If this were an inevitable conclusion, mathematical logic would not exist.  To summarize, the comparison $\xi(a) < \xi(b)$ is between the degree of belief in the truth of $a$ and $b$, however what $a$ and $b$ might mean is not relevant\footnote{e.g. $2a + 3a = 5a$ is true regardless of whether $a$ is ``apples'' or ``oranges.''}.

Like with the truth value of propositions, we do not claim to know how one is \textit{supposed} to assign a degree of belief to a generic atomic statement $a \in \tilde{\mathcal{A}}$.  The inference framework is mostly agnostic with respect to this question, except for some special circumstances which we will discuss in later chapters.  It is up to the rational agent to determine these values.  It is also assumed that the inference framework is agnostic to the assignment of a degree of belief to the conditional statement $a|\Gamma \in \tilde{\mathcal{A}}$, when $a$ is atomic.  In this way, the degrees of belief $\xi(a|\Gamma)$ and $\xi(a)$ serve as the atomic elements of $\Xi(\tilde{\mathcal{A}})$.  It is assumed that all other statements, which can be formed from $n$-ary connectives in $\Omega$, shall be functions of these basic degrees of belief\footnote{Usually the conditional probability $p(a|b)$ is not recognized as a \textit{true} probability, but rather as the ratio of two probabilities via the product rule, $p(a|b) = p(a\wedge b)/p(a)$.  For us however, the probability $p(a|b)$ is just as much a probability as any other, there are no issues with its interpretation since we are not starting a priori from the Kolmogorov picture.}.

Without loss of generality, we will continue with our discussion of degrees of belief of statements in $\tilde{\mathcal{A}}$ by only considering the special case in which the first argument is an atomic proposition, i.e. $\xi(a|\Gamma)$.  The goal is to determine the rationale $\mathcal{R}$ that will allow us to quantify the functional form of $\xi(\Delta|\Gamma)$, when $\Delta$ is an arbitrary statement in $\tilde{\mathcal{A}}$.  To do this, we will first need to determine the functional forms that represent conjunction, disjunction and negation.  Due to $\vee,\wedge$ and $\neg$ forming a set of functionally complete operators, all other subsets of $\tilde{\mathcal{A}}$ can be generated by these and hence we would have found a general form for $\xi(\Delta|\Gamma)$\footnote{By using the definition of the extended proposition space $\tilde{\mathcal{A}}$ as our universe of discourse, we have already incorporated another axiom that is usually imposed when designing belief functions.  This can be called an axiom of \textit{context dependency} \cite{Horvitz,KnuthSkilling} where the context is the quantity that appears on the right hand side of the solidus in $a|\Gamma$.}.

There are several ways to proceed depending on the assumed topological and structural properties for the space of degrees of belief, $\Xi$.  In the very least, we assume that the maps $\xi \in \Xi$ are continuous, i.e. $\Xi \subseteq C^0(\tilde{\mathcal{A}})$\footnote{In another Section we will formulate our inductive system assuming that the maps $\xi \in \Xi$ are smooth, i.e. $\Xi \subset C^{\infty}(\tilde{\mathcal{A}})$.}.  Without any restrictions on the behavior of $\Xi(\tilde{\mathcal{A}})$, the elements $\xi$ form an algebraic field.

\subsection{Algebraic properties of $(\Xi(\mathcal{A}),+,\cdot)$}\label{algebraicxi}
Without restriction on the $\Xi(\tilde{\mathcal{A}})$, such that the image of $\Xi$ is the reals, $\mathrm{im}(\Xi) = \mathbb{R}$, the space of $\xi$'s has the algebraic properties of a field, $(\Xi(\mathcal{A}),+,\cdot)$, i.e. for all $a|\Gamma,b|\Gamma$ and $c|\Gamma \in \tilde{\mathcal{A}}$,  
\begin{enumerate}
	\item \textbf{Commutativity of $+$} - The sum in $\Xi(\mathcal{A})$ is commutative,
	\begin{equation}
	\xi(a|\Gamma) + \xi(b|\Gamma) = \xi(b|\Gamma) + \xi(a|\Gamma).
	\end{equation}
	\item \textbf{Associativity of $+$} - The sum in $\Xi(\mathcal{A})$ is associative,
	\begin{align}
	\xi(a|\Gamma) + (\xi(b|\Gamma) + \xi(c|\Gamma)) &= (\xi(a|\Gamma) + \xi(b|\Gamma)) + \xi(c|\Gamma)\nonumber\\
	&= \xi(a|\Gamma) + \xi(b|\Gamma) + \xi(c|\Gamma).
	\end{align}
	\item \textbf{Neutral element of $+$} - The sum in $\Xi(\mathcal{A})$ contains an identity,
	\begin{equation}
	\exists! 0 \in \Xi(\mathcal{A}) : \xi(a|\Gamma) + 0 = \xi(a|\Gamma).
	\end{equation}
	\item \textbf{Additive inverse} - There exists additive inverses for each element in $\Xi(\mathcal{A})$,
	\begin{equation}
	\forall \xi \in \Xi(\mathcal{A}) : \exists! \xi^{-1}_{+} \in \Xi(\mathcal{A}) : \xi + \xi^{-1}_{+} = 0.
	\end{equation}
	\item \textbf{Commutativity of $\cdot$} - The product in $\Xi(\mathcal{A})$ is commutative,
	\begin{equation}
	\xi(a|\Gamma)\cdot\xi(b|\Gamma) = \xi(b|\Gamma)\cdot\xi(a|\Gamma).
	\end{equation}
	\item \textbf{Associativity of $\cdot$} - The product in $\Xi(\mathcal{A})$ is associative,
	\begin{align}
	\xi(a|\Gamma)\cdot[\xi(b|\Gamma)\cdot\xi(c|\Gamma)] &= [\xi(a|\Gamma)\cdot\xi(b|\Gamma)]\cdot\xi(c|\Gamma)\nonumber\\
	&= \xi(a|\Gamma)\cdot\xi(b|\Gamma)\cdot\xi(c|\Gamma).
	\end{align}
	\item \textbf{Neutral element of $\cdot$} - There exists a neutral element for the product in $\Xi(\mathcal{A})$,
	\begin{equation}
	\exists! 1 \in \Xi(\mathcal{A}) : \forall\xi \in \Xi(\mathcal{A}) : \xi \cdot 1 = 1 \cdot \xi = \xi.
	\end{equation}
	\item \textbf{Multiplicative inverse} - Provided that the function $\xi$ does not have any values which are zero, there exists a unique multiplicative inverse in $\Xi(\mathcal{A})$,
	\begin{equation}
	\forall \xi \in \Xi(\mathcal{A}) : \forall a \in \mathcal{A} : \xi(a) \neq 0 : \exists! \xi^{-1}_{\cdot} \in \Xi(\mathcal{A}) : \xi\cdot\xi^{-1}_{\cdot} = \xi^{-1}_{\cdot}\cdot\xi = 1.
	\end{equation}
\end{enumerate}
For brevity we will simply write the product in $\Xi(\mathcal{A})$ as a juxtaposition, $\xi_1 \cdot \xi_2 \stackrel{\mathrm{def}}{=} \xi_1\xi_2$.

In order to determine a correspondence between the algebra of propositions formed by the $n$-ary connectives $\Omega$ and their associated degrees of belief, we will construct a \textit{representation} for a functionally complete set.  The set of interest will be $\{\wedge,\vee\,\neg\}$, which is not minimal, however the relations $\wedge$ and $\vee$ have several useful properties, such as associativity and commutativity, that will make finding the representation easier. 

Despite the current agnosticism with respect to generic propositions, there are two propositions which should be given a consistent value for their degree of belief by any IRA.  These are the \textit{absolute true} and \textit{absolute false} propositions\footnote{The axiom (\ref{absolutes}) is actually unnecessary once one adopts a transitivity criterion for degrees of belief.  We will comment on the role and necessity of different axioms at the end of this Section.}.

\begin{axiom}[Absolutes]\label{absolutes}
	The degree of belief assigned to the propositions $\tru$ and $\fals$ should be independent of any Ira, i.e. they should be constants,
	\begin{equation}
	\forall \Gamma \in \tilde{\mathcal{A}} : \xi(\tru|\Gamma) = \xi(\tru) \stackrel{\mathrm{def}}{=} \xi_{\tru} \qquad \mathrm{and} \qquad \xi(\fals|\Gamma) = \xi(\fals) \stackrel{\mathrm{def}}{=} \xi_{\fals}.
	\end{equation}
\end{axiom}

An immediate consequence of this axiom is that the value $\xi_{\tru}$ gives an upper bound on the allowed degrees of belief for any statement, since nothing could be believed \textit{more} than absolute truth, i.e.,
\begin{equation}
\forall a|\Gamma \in \tilde{\mathcal{A}} : \xi(a|\Gamma) \leq \xi_{\tru}.
\end{equation}
Likewise, since nothing can be believed \textit{less} than absolute false, there is a unique lower bound,
\begin{equation}
\forall a|\Gamma \in \tilde{\mathcal{A}} : \xi(a|\Gamma) \geq \xi_{\fals}.
\end{equation}
We will find that the appropriate bounds for the map $\xi: \tilde{\mathcal{A}} \rightarrow \langle\xi_{\fals},\xi_{\tru}\rangle$ destroys its \hyperref[algebraicxi]{algebraic field} properties\footnote{The angular brackets $\langle \cdot,\cdot\rangle$ are placeholders for whether or not the left and right sides are open or closed.}.  In particular, it reduces the algebra over $\Xi(\tilde{\mathcal{A}})$ to a unital algebra\footnote{A unital algebra is an algebra that contains a multiplicative identity element, e.g. $1x = x1 = x$.}, which, perhaps to no surprise, is exactly what the \hyperref[algebrapropositions]{algebra of propositions} is.  Thus, by constraining the algebra of probable inference to be compatible with the algebra of propositions, we find that they actually have the same algebraic structure.

Since statements of the form $a|\fals$ are not defined in $\tilde{\mathcal{A}}$, the conditional degrees of belief $\xi(a|\fals)$ are necessarily also undefined.  The statement $a|\tru$ however is defined in $\tilde{\mathcal{A}}$, for which we assign the rule,
\begin{equation}
\forall a \in \mathcal{A} : \xi(a|\tru) = \xi(a).\label{conditionaltrue}
\end{equation}
There are two other peculiarities associated with absolute true and false.  These have to do with applying the opposite absolute in the unital formulas (\ref{unitalor}) and (\ref{unitaland}).  Starting with (\ref{unitalor}) we have the identity,
\begin{equation}
\forall a \in \mathcal{A} : (a \vee \tru) \Leftrightarrow \neg(\neg a \wedge \fals) \Leftrightarrow \neg(\fals) \Leftrightarrow \tru.\label{ortrue}
\end{equation}
Likewise for (\ref{unitaland}),
\begin{equation}
\forall a \in \mathcal{A} : (a \wedge \fals) \Leftrightarrow \neg(\neg a \vee \neg \fals) \Leftrightarrow \neg(\tru) \Leftrightarrow \fals.\label{andfalse}
\end{equation}
Before we determine the upper and lower bounds of $\xi$, we shall state some more axioms which associate binary relations between propositions to their corresponding degrees of belief.  
\subsection{The web of beliefs $\Xi(\tilde{\mathcal{A}})$}\label{alg}
In the following sections we seek to construct a \textit{representation} of the degrees of belief according to conjunction, disjunction and negation.  No one representation is necessarily more correct than any other, since the principle of consistency tells us that if there are multiple ways of conducting an inference, they better agree.  However, by adopting a parsimonious attitude, we can attempt to find a set of \textit{simple} representations.

\subsection{Representations in $\Xi(\tilde{\mathcal{A}})$}\label{section321}
A \textit{representation} for an element $\xi \in \Xi$ is defined as a function $f:\mathscr{P}(\Xi(\tilde{\mathcal{A}})) \rightarrow \Xi(\tilde{\mathcal{A}})$ whose arguments are determined as part of the definition of the representation.  For example, consider a generic composite statement $a \star b$, where $\star:\tilde{\mathcal{A}}\times\tilde{\mathcal{A}}\rightarrow \tilde{\mathcal{A}}$ is some binary operator in $\tilde{\mathcal{A}}$.  A possible representation for $\xi(a \star b)$ could be,
\begin{equation}
\xi^{(1)}(a\star b) = f^{(1)}_{\star}\left(\xi(a),\xi(b)\right),
\end{equation}   
which depends on the atomic degrees of belief $\xi(a)$ and $\xi(b)$.  Another choice of representation could be,
\begin{equation}
\xi^{(2)}(a \star b) = f^{(2)}_{\star}\left(\xi(a),\xi(b|a),\xi(a \vee b)\right).
\end{equation}
In principle, the number of possible representations for any generic element $\xi \in \Xi$ is equal to the cardinality $2^{|\Xi(\tilde{\mathcal{A}})|}$.
\begin{define}[Representations]
	Let $\Xi(\tilde{\mathcal{A}})$ be a web of beliefs.  A representation of a generic element $\xi \in \Xi(\tilde{\mathcal{A}})$ is defined as a function,
	\begin{equation}
	f:\mathscr{P}(\Xi(\tilde{\mathcal{A}}))\rightarrow\Xi(\tilde{\mathcal{A}}),
	\end{equation} 
	whose arguments are a potential subset of $\Xi(\tilde{\mathcal{A}})$.
\end{define}
For any generic element $\xi \in \Xi(\tilde{\mathcal{A}})$, we call $F[\Xi(\tilde{\mathcal{A}})]$ a \textit{maximal} representation, since it includes the entire web.  Throughout the rest of this section we seek to find suitable representations for the conjunction, disjunction and negation.  From the criterion of consistency, we can state the obvious axiom,
\begin{axiom}[Consistency of representations]\label{consistencyconnections}
	Let $\Xi(\tilde{\mathcal{A}})$ be a web of beliefs.  Let $\Omega \in \mathcal{D}_c$ be the set of connectives.  Due to the principle of consistency, we require that any set of representations of the connectives, $\{f_{\Omega}\}$, must be internally consistent with respect to $\Omega$.
\end{axiom}
This axiom perhaps goes without saying.  If we find a representation for each of $f_{\neg}, f_{\wedge}$ and $f_{\vee}$, then because $\neg, \wedge$ and $\vee$ are necessarily connected with one another, the representations must be consistent with respect to those connections.  At this point we can collect axioms \hyperref[realnumbers]{1.1}, \hyperref[absolutes]{1.2} and \hyperref[consistencyconnections]{1.3} into our rationale $\mathcal{R}$.  Our next goal is to specify the representations $f_{\wedge}, f_{\vee}$ and $f_{\neg}$.

\paragraph{The negation connective ---}
The simplest connective one can consider for a statement $a|\Gamma$ is the negation, $\neg a|\Gamma$.  Following Cox \cite{Cox}, we make the following assumption,
\begin{quotation}
	``The probability of an inference on given evidence determines the probability of its contradictory on the same evidence.'' (Cox 1961)
\end{quotation}
This is certainly just an application of common sense, together with the acceptance of the law of excluded middle.  In both \cite{Cox} and \cite{Cox2}, Cox made explicit the functional dependence of a contradictory as being an involution on $\Xi(\tilde{\mathcal{A}})$ which is also at least twice differentiable.  We will not assume differentiability in this part of the proof, and will only state the general behavior of the negation to be used later.  The common sense employed by Cox can be imposed as a necessary condition for $\mathcal{R}$ to adhere to the principle of parsimony.
\begin{axiom}[Negation]\label{negation}
	The degree of belief assigned to any proposition determines the degree of belief of its negation, which is represented by a strictly decreasing function, $f_{\neg}:\Xi\rightarrow\Xi$
	\begin{equation}
	\forall a|\Gamma \in \tilde{\mathcal{A}} : \xi(\neg a|\Gamma) = f_{\neg}(\xi(a|\Gamma)).\label{negationfunction}
	\end{equation}
\end{axiom}
As pointed out in \cite{Cox,CatichaBook,Jaynes}, the axiom\footnote{In \cite{CatichaBook}, Caticha shows that axiom (\ref{negation}) is not needed and can be derived from the sum and product rules.} (\ref{negation}) represents the intuition that the \textit{more} one believes in $a|\Gamma$, the \textit{less} one should believe in $\neg a|\Gamma$, and vice versa.  Given that the negation, $\neg a|\Gamma$, only depends on the proposition $a$ itself, the axiom (\ref{negation}) certainly obeys the principle of parsimony, i.e. it would be irrational to assume that the function $f_{\neg}$ necessarily depended on propositions other than $a|\Gamma$.   

An obvious requirement from the double negation property (i.e. $\neg(\neg a) = a$), is that the degree of belief $\xi(a|\Gamma)$ should be invariant if the negation function $f_{\neg}$ is applied twice,
\begin{equation}
\forall a|\Gamma\in\tilde{\mathcal{A}} : \xi(a|\Gamma) = f_{\neg}(\xi(\neg a|\Gamma)) = f_{\neg}(f_{\neg}(a|\Gamma)) .
\end{equation}
In general, we have the following relationship,
\begin{figure}[H]
	\centering
	\begin{tikzcd}
	\tilde{\mathcal{A}}\arrow[r,"\xi"]\arrow[d,shift left=.5ex]&\Xi\arrow[d,"f_{\neg}",shift left=.5ex]\\
	\tilde{\mathcal{A}}\arrow[r,"\xi"']\arrow[u,"\neg",shift left=.5ex]&\Xi\arrow[u, shift left=.5ex]
	\end{tikzcd}
	\caption{Commuting diagram for the negation connective and its inductive counterpart.  In general we have that, $\xi\circ\neg = f_{\neg}\circ\xi$.}
\end{figure}
We could also argue axiom (\ref{negation}) on the basis of consistency.  Since $\neg$ is an involution in $(\tilde{\mathcal{A}},\wedge,\vee,\neg)$, the principle of parsimony requires $f_{\neg}$ be an involution on $\Xi(\tilde{\mathcal{A}})$.  

\subsection{The disjunction connective}
In order to complete the algebra $\{\wedge,\vee,\neg\}$, we will need to determine degrees of belief for the disjunction and conjunction relations.  It is perhaps less obvious how the representation of conjunction and disjunction should depend on the atomic degrees of belief associated to their constituents.  In the very least, we adopt the view of Caticha in \cite{CatichaBook},
\begin{quotation}
	``In order to be rational our beliefs in $a\vee b$ and $a\wedge b$ must be somehow related to our separate beliefs in $a$ and $b$ (Caticha, 2021 \cite{CatichaBook}).''
\end{quotation}
Since our beliefs form an interconnected web, they cannot all be independent unless all of the individual statements $a|\Gamma \in \tilde{\mathcal{A}}$ are also independent.  In general our beliefs in $a|\Gamma$ may depend on our beliefs in some other statements $b|\Gamma$.  Thus, whatever our beliefs in the individual statements $a|\Gamma$ and $b|\Gamma$ are, they must constrain our beliefs in $a\wedge b|\Gamma$ and $a\vee b|\Gamma$ in a way which is consistent with the interconnected web.  But our beliefs in $a\wedge b|\Gamma$ shouldn't in general depend on the entire web $\Xi(\tilde{\mathcal{A}})$ but only some subset\footnote{Given our practicality requirement, one could argue that, much in the same way that the space of propositions can be constructed from a set of elementary propositions using the binary connectives, the entire web of beliefs should be able to be constructed on the basis of the elementary degrees of belief $\xi(a|\Gamma)$.}.  Exactly which subset one chooses determines the representation.   

In the least parsimonious representation, the degree of belief $\xi(a\vee b|\Gamma)$ is necessarily a function of every other degree of belief in the web,
\begin{align}
\xi(a\vee b|\Gamma) \xrightarrow{\mathrm{consistency}}& F_{\vee}\left[\frac{}{}\xi(a|\Gamma),\xi(b|\Gamma),\xi(c|\Gamma),\xi(a\wedge c|\Gamma),\dots\right]\nonumber\\
\stackrel{\mathrm{def}}{=}& F_{\vee}[\Xi(\tilde{\mathcal{A}})],
\end{align}    
which is a maximal representation.  While it is certainly possible to define such a function, the principle of parsimony, as well as practicality, tells us that this representation is undesirable.  We should choose a representation which is as simple as possible, but not too simple as to not be of universal applicability.  Certainly we expect that the functions which quantify $\xi(a\wedge b|\Gamma)$ and $\xi(a\vee b|\Gamma)$ should depend on $\xi(a|\Gamma)$ and $\xi(b|\Gamma)$, but is that all?  If the two statements $a,b\in\tilde{\mathcal{A}}$ are not independent, i.e. if $a|b \neq a$ and $b|a \neq b$, then the information about their dependence is seemingly lost if we choose either $\xi(a\wedge b) = f(\xi(a),\xi(b))$ or $\xi(a \vee b) = g(\xi(a),\xi(b))$.  Thus, the simple choice,
\begin{equation}
F_{\vee}[\Xi(\tilde{\mathcal{A}})] \xrightarrow{\mathrm{parsimony}} f_{\vee}\left(\frac{}{}\xi(a|\Gamma),\xi(b|\Gamma)\right),\label{nope}
\end{equation} 
is too restrictive.

In order to be rational, we should at least assume that the dependence of $a|\Gamma$ on $b|\Gamma$, and vice versa, will have an impact on our judgement about either $a\vee b|\Gamma$ or $a \wedge b|\Gamma$.  Take for example the extreme case when $b = \neg a$.  Had we used the representation in (\ref{nope}), it would seemingly not account for situations of this type since one would only make a judgement based on the degrees of belief $\xi(a|\Gamma)$ and $\xi(\neg a|\Gamma)$.  What other degrees of belief could we include in (\ref{nope}) that would account for these specific situations?  Since the degrees of belief $\xi(a|\Gamma)$ form the atomic elements of the theory, it is reasonable to suggest that the extra information must be contained in the conditionals $[a|\Gamma]|b$ and $[b|\Gamma]|a$, i.e. the degrees of belief $\xi(a|\Gamma\wedge b)$ and $\xi(b|\Gamma\wedge a)$.  These conditional degrees of belief capture the inter-dependency of $a|\Gamma$ and $b|\Gamma$ in the simplest way possible.  Thus, we have at least the four degrees of belief, $\xi(a|\Gamma), \xi(b|\Gamma), \xi(a|\Gamma\wedge b)$ and $\xi(b|\Gamma\wedge a)$ which all convey meaningful information about $\xi(a\vee b|\Gamma)$ and $\xi(a \wedge b|\Gamma)$.

We will find that these four degrees of belief are enough to specify a consistent representation for $\xi(a\vee b|\Gamma)$.
\begin{equation}
F_{\vee}[\Xi(\tilde{\mathcal{A}})] \xrightarrow[\mathrm{parsimony}]{\mathrm{consistency}} f_{\vee}\left(\frac{}{}\xi(a|\Gamma),\xi(b|\Gamma),\xi(a|\Gamma\wedge b),\xi(b|\Gamma\wedge a)\right).
\end{equation}
More formally, we have the axiom
\begin{axiom}[The disjunction connective]\label{disjunction}
	The degree of belief assigned to the statement $a \vee b|\Gamma \in \tilde{\mathcal{A}}_{\Gamma}$ must be represented by a function of the degrees of belief in the elementary propositions $a|\Gamma$, $b|\Gamma$, $a|\Gamma\wedge b$ and $b|\Gamma \wedge a$,
	\begin{equation}
	\forall a|\Gamma,b|\Gamma \in \tilde{\mathcal{A}}: \xi(a\vee b|\Gamma) = f_{\vee}\left(\frac{}{}\xi(a|\Gamma),\xi(b|\Gamma),\xi(a|\Gamma\wedge b),\xi(b|\Gamma\wedge a)\right).\label{orfunction}
	\end{equation}
\end{axiom}
Any follower of foundational probability theory would claim that we are positioning ourselves to stumble upon the sum and product rules by ``happening'' to pick the right constraints.  On the contrary, the goal here is to find a consistent, simple, yet sufficiently general, representative for $\xi(a\vee b|\Gamma)$ in terms of the elementary degrees of belief $\xi(a|\Gamma), \xi(b|\Gamma), \xi(a|\Gamma\wedge b)$ and $\xi(b|\Gamma\wedge a)$, that preserves as much of the algebraic structure on $\tilde{\mathcal{A}}$ as possible.  The axiom (\ref{orfunction}) is a really a statement of desire for that particular representative, i.e. this is the functional relationship we wish to determine.  In fact, the constraints used throughout our derivation lead to an infinite number of possible solutions, a subset of which will have the form of the sum rule.  As we will see, there are other solutions to (\ref{orfunction}) which do not look like a sum rule at all.  

\subsection{The conjunction connective}
We must also determine a representation for the degree of belief in a conjunction of two elements in $\tilde{\mathcal{A}}$.  Originally, Cox made the following suggestion (inspired by Venn) for the representation of the conjunction,  
\begin{quotation}
	``The probability on the given evidence that both of two inferences are true is determined by their separate probabilities, one on the given evidence, the other on this evidence with the additional assumption that the first inference is true.'' (Cox 1961)
\end{quotation}
While we could heed his advice and declare that the conjunction must necessarily be of the form,
\begin{equation}
\xi(a\wedge b|\Gamma)\xrightarrow{\mathrm{Cox}}f_{\wedge}\left(\xi(a|\Gamma),\xi(b|\Gamma\wedge a)\right),
\end{equation}
we will argue this form from a more modest assumption.  Following the logic that we used to determine the representation of the disjunction, we can apply exactly the same criteria for the conjunction.  Thus we have,
\begin{axiom}[The conjunction connective]\label{conjunction}
	The degree of belief assigned to the statement $a \wedge b|\Gamma \in \tilde{\mathcal{A}}$ must be represented by a function of the degrees of belief in elementary propositions $a|\Gamma$,$b|\Gamma$,$a|\Gamma\wedge b$ and $b|\Gamma \wedge a$,
	\begin{equation}
	\forall a|\Gamma,b|\Gamma \in \tilde{\mathcal{A}}: \xi(a\wedge b|\Gamma) = f_{\wedge}\left(\frac{}{}\xi(a|\Gamma),\xi(b|\Gamma),\xi(a|\Gamma\wedge b),\xi(b|\Gamma\wedge a)\right).\label{andfunction}
	\end{equation}
\end{axiom}
The axioms \hyperref[negation]{1.4}, \hyperref[disjunction]{1.5} and \hyperref[conjunction]{1.6} can be appended to our rationale $\mathcal{R}$.  According to axiom \hyperref[consistencyconnections]{1.3}, an additional requirement will be that the definitions in (\ref{negationfunction}), (\ref{orfunction}) and (\ref{andfunction}) be compatible with respect to the algebra of propositions.  Due to (\ref{orand}), and \hyperref[consistencyconnections]{consistency}, we should also have,
\begin{align}
\xi(a\wedge b|\Gamma) &= \xi\left(\neg(\neg a \vee \neg b)|\Gamma\right) = f_{\neg}\left(\xi(\neg a \vee \neg b|\Gamma)\right)\nonumber\\
&= f_{\neg}\left(f_{\vee}\left(\frac{}{}\xi(\neg a|\Gamma),\xi(\neg b|\Gamma),\xi(\neg a|\Gamma\wedge\neg b),\xi(\neg b|\Gamma\wedge\neg a)\right)\right)\nonumber\\
&= f_{\neg}\left(f_{\vee}\left(\frac{}{}f_{\neg}(\xi(a|\Gamma)),f_{\neg}(\xi(b|\Gamma)),f_{\neg}(\xi(a|\Gamma\wedge b)),f_{\neg}(\xi(b|\Gamma\wedge\neg a))\right)\right).\label{andorrelation1}
\end{align}
Likewise for $\xi(a \vee b|\Gamma)$ we have,
\begin{align}
\xi(a\vee b|\Gamma) &= \xi\left(\neg(\neg a \wedge \neg b)|\Gamma\right) = f_{\neg}\left(\xi(\neg a \wedge \neg b|\Gamma)\right)\nonumber\\
&= f_{\neg}\left(f_{\wedge}\left(\frac{}{}\xi(\neg a|\Gamma),\xi(\neg b|\Gamma),\xi(\neg a|\Gamma\wedge\neg b),\xi(\neg b|\Gamma\wedge\neg a)\right)\right)\nonumber\\
&= f_{\neg}\left(f_{\wedge}\left(\frac{}{}f_{\neg}(\xi(a|\Gamma)),f_{\neg}(\xi(b|\Gamma)),f_{\neg}(\xi(a|\Gamma\wedge\neg b)),f_{\neg}(\xi(b|\Gamma\wedge\neg a))\right)\right).\label{andorrelation2}
\end{align}
We will find that the class of functions that satisfy (\ref{negationfunction}), (\ref{orfunction}) and (\ref{andfunction}) will necessarily satisfy the above relations (\ref{andorrelation1}) and (\ref{andorrelation2}).
\paragraph{Simplification of $f_{\wedge}$ ---}
The functional form of $f_{\wedge}$ can actually be widdled down from four to two arguments, as was shown in \cite{CatichaBook}.  In general, $f_{\wedge}$ could take $\sum_{n=1}^{4}\begin{pmatrix}4 \\ n\end{pmatrix} = 15$ different combinations of inputs, but only one is consistent with what we expect.  First, we can eliminate the functions which only depend on a single argument (e.g. $\xi(a\wedge b|\Gamma) = f_{\wedge}(\xi(a|\Gamma))$) since this violates common sense.  Four more combinations can be eliminated through commutativity, $a|\Gamma\wedge b|\Gamma = b|\Gamma \wedge a|\Gamma$.  Following \cite{CatichaBook}, the seven remaining functions are the following,
\begin{align}
\xi^{(1)}(a\wedge b|\Gamma) &= f_{\wedge}^{(1)}\left(\frac{}{}\xi(a|\Gamma),\xi(b|\Gamma)\right)\label{fwedge1}\\
\xi^{(2)}(a\wedge b|\Gamma) &= f_{\wedge}^{(2)}\left(\frac{}{}\xi(a|\Gamma),\xi(a|\Gamma\wedge b)\right)\label{fwedge2}\\
\xi^{(3)}(a\wedge b|\Gamma) &= f_{\wedge}^{(3)}\left(\frac{}{}\xi(a|\Gamma),\xi(b|\Gamma \wedge a)\right)\label{fwedge3}\\
\xi^{(4)}(a\wedge b|\Gamma) &= f_{\wedge}^{(4)}\left(\frac{}{}\xi(a|\Gamma\wedge b),\xi(b|\Gamma\wedge a)\right)\label{fwedge4}\\
\xi^{(5)}(a\wedge b|\Gamma) &= f_{\wedge}^{(5)}\left(\frac{}{}\xi(a|\Gamma),\xi(b|\Gamma),\xi(a|\Gamma\wedge b)\right)\label{fwedge5}\\
\xi^{(6)}(a\wedge b|\Gamma) &= f_{\wedge}^{(6)}\left(\frac{}{}\xi(a|\Gamma),\xi(a|\Gamma\wedge b),\xi(b|\Gamma\wedge a)\right)\label{fwedge6}\\
\xi^{(7)}(a\wedge b|\Gamma) &= f_{\wedge}^{(7)}\left(\frac{}{}\xi(a|\Gamma),\xi(b|\Gamma),\xi(a|\Gamma\wedge b),\xi(b|\Gamma\wedge a)\right)\label{fwedge7}
\end{align}
The first function (\ref{fwedge1}) can be rejected for the same reasons as (\ref{nope}), i.e. since it is independent of any correlations between $a$ and $b$.  For example, if $b = \neg a$, then $a\wedge \neg a = \fals$, however the function $\xi^{(1)}(a\wedge\neg a|\Gamma) = f_{\wedge}^{(1)}(\xi(a|\Gamma),\xi(\neg a|\Gamma)) = \xi_{\fals}$ will arbitrarily depend on $\xi(a|\Gamma)$ and $\xi(\neg a|\Gamma)$.  The second function (\ref{fwedge2}) can also be thrown out by examining the special case of $a = \tru$, for which $b\wedge\tru = b$.  The function $f_{\wedge}^{(2)}$ however becomes a constant, $\xi^{(2)}(\tru\wedge b|\Gamma) = \xi^{(2)}(b|\Gamma) = f_{\wedge}^{(2)}(\xi_{\tru},\xi_{\tru})$, which is unacceptable.    The fourth function (\ref{fwedge4}) is easily discarded since for the special case $b = a$, we have $\xi^{(4)}(a\wedge a|\Gamma) = \xi^{(4)}(a|\Gamma) = f_{\wedge}^{(4)}(\xi_{\tru},\xi_{\tru})$.

Functions $f_{\wedge}^{(5)}$ (\ref{fwedge5}), $f_{\wedge}^{(6)}$ (\ref{fwedge6}) and $f_{\wedge}^{(7)}$ (\ref{fwedge7}) can all be ruled out using associativity\footnote{For details on the arguments see \cite{CatichaBook} pg. 24.}, $a\wedge(b\wedge c) \Leftrightarrow (a\wedge b)\wedge c$, which shows that they are identical to one of the other four.  The winning function is $f_{\wedge}^{(3)}$ (\ref{fwedge3}), which we will simply rename as $f_{\wedge}$\footnote{This property is sometimes called \textit{hypothetical conditioning} \cite{Horvitz}.},
\begin{equation}
\xi(a\wedge b|\Gamma) \xrightarrow{Caticha} f_{\wedge}\left(\frac{}{}\xi(a|\Gamma),\xi(b|\Gamma\wedge a)\right).
\end{equation}
This constrains the reciprocal relations defined in (\ref{andorrelation1}) and (\ref{andorrelation2}) so that,
\begin{align}
\xi(a\vee b|\Gamma) &= f_{\neg}\left(\frac{}{}f_{\wedge}\left(\frac{}{}f_{\neg}(\xi(a|\Gamma)),f_{\neg}(\xi(b|\Gamma\wedge a))\right)\right).
\end{align}
Thus, the dependence of $\xi(a \vee b|\Gamma)$, and hence of $f_{\vee}$, on $\xi(b|\Gamma)$ and $\xi(a|\Gamma\wedge b)$ in (\ref{orfunction}) must somehow manifest through the composition of the negation map, $f_{\neg}$ with $f_{\wedge}$.

Due to the commutativity of $\wedge$, we have the following commutative diagram with respect to the representation $f_{\wedge}$,
\begin{figure}[H]
	\centering
	\begin{tikzcd}
	\tilde{\mathcal{A}}_{\Gamma}\times\tilde{\mathcal{A}}_{\Gamma\times \pi_1}\arrow[r,"{(\xi,\xi)}"]&\Xi\times\Xi\arrow[dr,"f_{\wedge}"]&\\
	\tilde{\mathcal{A}}_{\Gamma}\times\tilde{\mathcal{A}}_{\Gamma}\arrow[u,"{(\pi_2,|\pi_1)}"]\arrow[d,"{(\pi_1,|\pi_2)}"']\arrow[r,"\wedge"]&\tilde{\mathcal{A}}_{\Gamma}\arrow[r,"\xi"]&\Xi\\
	\tilde{\mathcal{A}}_{\Gamma}\times\tilde{\mathcal{A}}_{\Gamma\times\pi_2}\arrow[r,"{(\xi,\xi)}"']&\Xi\times\Xi\arrow[ur,"f_{\wedge}"']&
	\end{tikzcd}
	\caption{Commutative diagram for the representation $f_{\wedge}$.  The projections $\pi_1$ and $\pi_2$ select the first and second element respectively.  The space $\tilde{\mathcal{A}}_{\Gamma\times \pi_i}$ is the $\Gamma\times\pi_i$ conditioned subspace.}
\end{figure}
It is now up to us to determine the class of functions $[f_{\neg}]$, $[f_{\wedge}]$ and $[f_{\vee}]$ which satisfy axioms \hyperref[consistencyconnections]{1.3}, \hyperref[negation]{1.4}, \hyperref[disjunction]{1.5} and \hyperref[conjunction]{1.6}.  We will first focus on the disjunction representation $f_{\vee}$.  In order to motivate the derivation of $f_{\vee}$, we will first examine some basic properties of the disjunction.

\paragraph{A natural order on $\Xi$ ---}
Since the disjunction preserves order in $\tilde{\mathcal{A}}_{\vee}$, we expect that it should also preserve order under the map $\xi$.  Since $\Xi$ is a number field, the partial order $\preceq$ on $\Xi$ is simply the \textit{natural order} $\leq$. If $\xi(a)$ and $\xi(b)$ represent our degrees of belief in $a$ and $b$, then the statement $\xi(a) \leq \xi(b)$ means that our degree of belief in $b$ is at least the same as in $a$.  If we consider another statement $c$, which is mutually exclusive from $a$ or $b$, in disjunction with $a$ or $b$, then the order of $\xi(a)$ and $\xi(b)$ should not change,
\begin{axiom}[Disjunction preserves order]\label{transitive}
	Let $(\tilde{\mathcal{A}}_{\vee},\vee,\preceq)$ be a partially ordered set with respect to $\vee$.  Let $\Xi(\tilde{\mathcal{A}}_{\vee}) \subset C^0(\tilde{\mathcal{A}}_{\vee})$ represent our degrees of belief in the statements of $\tilde{\mathcal{A}}_{\vee}$.  Since disjunction in $\tilde{\mathcal{A}}_{\vee}$ preserves order in $\tilde{\mathcal{A}}_{\vee}$, then we impose that the map $\xi:\tilde{\mathcal{A}}_{\vee}\rightarrow\langle\xi_{\fals},\xi_{\tru}\rangle$ preserves natural order, 
	\begin{equation}
	\forall a|\Gamma,b|\Gamma,c|\Gamma\in \tilde{\mathcal{A}}_{\vee} : \xi(a|\Gamma) \leq \xi(b|\Gamma) \Rightarrow \xi(a\vee c|\Gamma) \leq \xi(b \vee c|\Gamma).
	\end{equation}
\end{axiom}
This axiom guarantees that the function $f_{\vee}$ be strictly increasing in both arguments whenever they are mutually exclusive.  Below is a representation of the distributive lattice for three mutually exclusive events along with the corresponding natural ordering in $\Xi$,
\begin{figure}[H]
	\centering
	\begin{tikzpicture}[scale=.7]
	\node (true) at (-5,4) {$a\vee b\vee c$};
	\node (ab) at (-7,2) {$a\vee b$};
	\node (ac) at (-5,2) {$a\vee c$};
	\node (bc) at (-3,2) {$b \vee c$};
	\node (a) at (-7,0) {$a$};
	\node (b) at (-5,0) {$b$};
	\node (c) at (-3,0) {$c$};
	\node (false) at (-5,-2) {$\emptyset$};
	\draw (false) -- (a);
	\draw (false) -- (b);
	\draw (false) -- (c);
	\draw (a) -- (ab);
	\draw (a) -- (ac);
	\draw (b) -- (ab);
	\draw (b) -- (bc);
	\draw (c) -- (ac);
	\draw (c) -- (bc);
	\draw (ac) -- (true);
	\draw (ab) -- (true);
	\draw (bc) -- (true);
	\node (true2) at (5,4) {$\xi_{\tru}$};
	\node (ab2) at (3,2) {$\xi(a\vee b)$};
	\node (ac2) at (5,2) {$\xi(a\vee c)$};
	\node (bc2) at (7,2) {$\xi(b \vee c)$};
	\node (a2) at (3,0) {$\xi(a)$};
	\node (b2) at (5,0) {$\xi(b)$};
	\node (c2) at (7,0) {$\xi(c)$};
	\node (false2) at (5,-2) {$\xi_{\fals}$};
	\draw (false2) -- (a2);
	\draw (false2) -- (b2);
	\draw (false2) -- (c2);
	\draw (a2) -- (ab2);
	\draw (a2) -- (ac2);
	\draw (b2) -- (ab2);
	\draw (b2) -- (bc2);
	\draw (c2) -- (ac2);
	\draw (c2) -- (bc2);
	\draw (ac2) -- (true2);
	\draw (ab2) -- (true2);
	\draw (bc2) -- (true2);
	\end{tikzpicture}
	\caption{Hasse diagrams for the distributive lattice $\mathfrak{L}_{\vee}$ of $\vee$ over a set of three mutually exclusive and exhaustive propositions $\{a,b,c\}\in\tilde{\mathcal{A}}_{\vee}$ and the corresponding lattice $\mathfrak{L}_{\xi}$ for the degrees of belief $\Xi(\{a,b,c\})$.}
\end{figure}
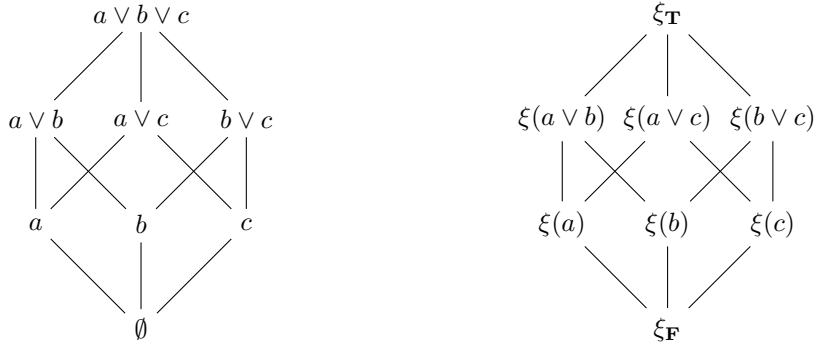
The lattice $\mathfrak{L}_{\xi}$ forms a semi-group, where $\xi_{\fals}$ is the identity.  Such semi-groups where $\xi(a\vee b)$ obeys an associative law are called \textit{Acz\'{e}lian} \cite{Ling}.  The use of the partially ordered set structure, as well as lattice structure, of Boolean algebras has been used to derive probability rules in the past \cite{Knuth1,Knuth2,KnuthSkilling,KlainRota}.

We take the approach given in \cite{CatichaBook} in which the rules of inference are determined by finding the representations $f_{\wedge}$ and $f_{\vee}$.  All other rules can be derived, including the negation, from knowing the general functional form of these as well as imposing transitivity.  In other words, we only need to adopt axioms (\ref{disjunction}) and (\ref{conjunction}) in order to derive (\ref{negation}).  

\subsection{Representation of the disjunction $f_{\vee}$}\label{disjun}
Before specifying the functional form of $f_{\vee}$ for generic statements, we will first determine its form for the special case when the two arguments $a$ and $b$ are mutually exclusive, i.e. when $a|b \Leftrightarrow b|a \Leftrightarrow\fals$.  This leads to
\begin{equation}
\xi(a\vee b|\Gamma) = f_{\vee}\left(\xi(a|\Gamma),\xi(b|\Gamma),\xi_{\fals},\xi_{\fals}\right) = f_{\vee}(\xi(a|\Gamma),\xi(b|\Gamma)),
\end{equation}
which reduces $f_{\vee}$ to a function of two arguments.  
To constrain $f_{\vee}$ further, we appeal to the associativity properties of $\vee$.  Given three mutually exclusive statements $a|\Gamma, b|\Gamma$ and $c|\Gamma$, we have 
\begin{align}
\xi(a\vee b\vee c|\Gamma) &= f_{\vee}\left(\frac{}{}f_{\vee}\left(\frac{}{}\xi(a|\Gamma),\xi(b|\Gamma)\right),\xi(c|\Gamma)\right)\nonumber\\
&= f_{\vee}\left(\frac{}{}\xi(a|\Gamma),f_{\vee}\left(\frac{}{}\xi(b|\Gamma),\xi(c|\Gamma)\right)\right).\label{associativeor}
\end{align}  
From axiom (\ref{transitive}) the function $f_{\vee}$ in (\ref{associativeor}) must be strictly increasing in both arguments.  Acz\'{e}l showed that by assuming that the function $f_{\vee}$ is continuous and strictly increasing in both arguments \cite{Aczel}, the general solution to (\ref{associativeor}) in terms of an arbitrary invertable function $\phi$ is given by the following theorem\footnote{While the function $f_{\wedge}$ also obeys an associative law, it is strictly \textit{decreasing} in both of its arguments.},
\begin{theorem}[Acz\'{e}l]\label{aczel}
	Let $\mathrm{im}(\Xi) = \langle\xi_{\fals},\xi_{\tru}\rangle \subseteq \mathbb{R}$ be an interval on $\mathbb{R}$.  Any function $f_{\vee}:\Xi\times\Xi\rightarrow\Xi$ is continuous, strictly increasing in both arguments and associative if and only if there exists a continuous and strictly monotonic function $\phi:\Xi\rightarrow \mathbb{R}$ such that,
	\begin{equation}
	\forall \xi(a|\Gamma),\xi(b|\Gamma) \in \Xi(\tilde{\mathcal{A}}_{\vee}) : f_{\vee}\left(\xi(a|\Gamma),\xi(b|\Gamma)\right) = \phi^{-1}\left(\frac{}{}\phi(\xi(a|\Gamma)) + \phi(\xi(b|\Gamma))\right).\label{solutionor}
	\end{equation}
\end{theorem}
A proof of (\ref{solutionor}) is given in \hyperref[proofaczel]{Appendix A}.  Applying the arbitrary function $\phi$ to both sides,
\begin{equation}
\phi\left(\xi(a\vee b|\Gamma)\right) = \phi(\xi(a|\Gamma)) + \phi(\xi(b|\Gamma)).\label{answertwo}
\end{equation}
Since $\phi$ is strictly monotonic, it simply provides a regraduation of the degrees of belief $\xi$.  Any legitimate inference which satisfies the associativity constraint can be represented by the equation above.  Thus, one can then adopt the simplest representative by choosing $\phi$ equal to the identity, $\phi = \mathrm{Id}_{\Xi}$, so that,
\begin{equation}
\xi^{(\mathrm{Id}_{\Xi})}(a\vee b|\Gamma) = \xi(a|\Gamma) + \xi(b|\Gamma),
\end{equation}
which makes $f_{\vee}$ linear in its first two arguments.  Of course this choice gives the impression that we are really after the sum rule, however we could have also picked the monotonic function $\phi = \log$, which gives,
\begin{equation}
\xi^{(\log)}(a\vee b|\Gamma) = \xi(a|\Gamma)\xi(b|\Gamma).
\end{equation}
This representation also constrains $f_{\vee}$ to be linear in its first two arguments\footnote{Any scaling of $\phi = \log$ to $\phi = \alpha\log$ where $\alpha \in \mathbb{R}$, results in
	\begin{equation}
	\xi^{(\alpha\log)}(a\vee b|\Gamma)^{\alpha} = \xi(a|\Gamma)^{\alpha}\xi(b|\Gamma)^{\alpha},
	\end{equation}
	which is the general result obtained in \cite{Cox2} for the product rule, but which is no longer linear.}. 

Before settling on a choice of $\phi$, it will be useful to list a set of special cases.  First, recalling $a\vee\fals \Leftrightarrow a$ (\ref{unitalor}) we have,
\begin{align}
\xi(a \vee \fals|\Gamma) &= f_{\vee}\left(\frac{}{}\xi(a|\Gamma),\xi(\fals|\Gamma),\xi_{\fals},\xi_{\fals}\right)\nonumber\\
&= f_{\vee}(\xi(a|\Gamma),\xi_{\fals}) = \xi(a|\Gamma).
\end{align}
Due to commutativity of $\vee$ we also have that,
\begin{equation}
f_{\vee}(\xi(a|\Gamma),\xi_{\fals}) = f_{\vee}(\xi_{\fals},\xi(a|\Gamma)) = \xi(a|\Gamma),\label{orcommutative}
\end{equation}
Inserting (\ref{orcommutative}) into (\ref{answertwo}) we find,
\begin{equation}
\phi(\xi(a\vee \fals|\Gamma)) = \phi(\xi(a|\Gamma)) + \phi(\xi_{\fals}),
\end{equation}
however $\xi(a\vee\fals|\Gamma) = \xi(a|\Gamma)$, hence we must have,
\begin{equation}
\phi(\xi_\fals) = 0.\label{phizero}
\end{equation}
Thus, while the function $\phi$ may be arbitrary, its evaluation of $\xi_{\fals}$ is not\footnote{This can also be traced back to the definition of $\psi$ (\ref{psizero}) which explicitly defined $\psi(0) = \phi^{-1}(0) = \xi_{\fals}$.}.   Thus, for any choice of $\phi$, the value of $\xi_{\fals}$ is constrained so that (\ref{phizero}) holds.  This can also be demonstrated from the identity $\fals\vee\tru = \tru$,
\begin{equation}
\phi(\xi(\fals\vee\tru|\Gamma)) = \phi(\xi_\fals) + \phi(\xi_\tru) = \phi(\xi_{\tru}).
\end{equation}
\paragraph{A convenient regraduation ---}
Without loss of generality, one can always recast the function $\xi$ as the composition, $\xi' = \phi \circ \xi$.  Since $\xi$ is continuous, as well as $\phi$, then the composition of maps $\phi \circ \xi$ is also continuous.  Likewise, since $\phi$ is monotonic, then the transitivity property (\ref{transitive}) is also preserved under the composition of $\xi$ with $\phi$.  Thus, we can define the \textit{regraduated} degree of belief,
\begin{equation}
\xi'(a\vee b|\Gamma) \stackrel{\mathrm{def}}{=} \phi(\xi(a\vee b|\Gamma)),\label{orregraduate}
\end{equation}
so that we get explicitly the sum rule,
\begin{equation}
\xi'(a\vee b|\Gamma) = \xi'(a|\Gamma) + \xi'(b|\Gamma).
\end{equation}
With this choice, we have left open the functional form of $\xi'$ as it applies to statements in $\tilde{\mathcal{A}}$.  The sum above does not inform one as to what function $\xi'$, or equivalently $\phi$, in $C^0(\tilde{\mathcal{A}})$ that one should use as a degree of belief.  It is strictly a rule for manipulating degrees of belief $\xi'$.  If one accepts the axioms \hyperref[disjunction]{1.5}, \hyperref[transitive]{1.7} and \hyperref[realnumbers]{1.1}, then the general solution must be (\ref{orregraduate}).  The space $\Xi'$ then, is simply the result of a homeomorphism $\phi:\Xi\rightarrow\Xi'$, so that $\Xi' \subset C^0(\tilde{\mathcal{A}})$ and
\begin{equation}
\Xi' \cong_{\mathrm{homeo}}\Xi,
\end{equation}   
From here we can identify $\phi \circ f_{\vee} = f_{\vee}'$, where $f_{\vee}'$ is simply addition on $\Xi'$.  We then have the following commutative diagram,  
\begin{figure}[H]
	\centering
	\begin{tikzcd}
	\tilde{\mathcal{A}}_{\vee}\times\tilde{\mathcal{A}}_{\vee}\arrow[r,"\vee"]\arrow[d,"{(\xi',\xi')}"']&\tilde{\mathcal{A}}_{\vee}\arrow[d,"\xi'"]\\
	\Xi'\times\Xi'\arrow[r,"f_{\vee}'"']&\Xi'
	\end{tikzcd}
	\caption{Commutative diagram for the sum rule of two mutually exclusive propositions.}
\end{figure}
For mutually exclusive propositions, the diagram above commutes, 
\begin{equation}
f_{\vee}'\circ(\xi',\xi') = \xi'\circ \vee
\end{equation}

\paragraph{The negation function $f_{\neg}$ ---}
Using the fact that $a\wedge \neg a|\Gamma \Leftrightarrow \fals$, and $a\vee \neg a|\Gamma \Leftrightarrow \tru$, we also have
\begin{equation}
\forall a|\Gamma \in \tilde{\mathcal{A}} : \xi'(a\vee \neg a|\Gamma) = \xi'(a|\Gamma) + \xi'(\neg a|\Gamma) = \xi_{\tru},
\end{equation}   
which, according to (\ref{negation}), means that for any degree of belief $\xi'(a|\Gamma)$,
\begin{equation}
\xi'(\neg a|\Gamma) = f_{\neg}'(\xi'(a|\Gamma)) = \xi'_{\tru} - \xi'(a|\Gamma).
\end{equation}
At this point it will be beneficial to drop the primes from the degrees of belief $\xi'$ for brevity, although one should not confuse the $\xi$'s from here on out with the $\xi$'s from before the regraduation.

\paragraph{The general sum rule ---}
To determine the general sum rule for arbitrary statements, $a|\Gamma$ and $b|\Gamma$, we can use the distributive property of conjunction.  Since any proposition can be written as the disjunction of two mutually exclusive statements,
\begin{equation}
a|\Gamma \Leftrightarrow (a|\Gamma \wedge b|\Gamma)\vee(a|\Gamma\wedge\neg b|\Gamma),
\end{equation}
we have that,
\begin{align}
a|\Gamma \vee b|\Gamma &\Leftrightarrow \left[\frac{}{}(a|\Gamma\wedge b|\Gamma)\vee (a|\Gamma\wedge \neg b|\Gamma)\right]\vee\left[\frac{}{}(a|\Gamma\wedge b|\Gamma)\vee (\neg a|\Gamma \wedge b|\Gamma)\right]\nonumber\\
&\Leftrightarrow \left[\frac{}{}(a|\Gamma\wedge b|\Gamma)\vee (\neg a|\Gamma \wedge b|\Gamma) \vee (a|\Gamma \wedge \neg b|\Gamma)\right]\nonumber\\
&\Leftrightarrow a|\Gamma\vee[\neg a|\Gamma\wedge b|\Gamma].
\end{align}
Since $a|\Gamma$ and $[\neg a|\Gamma \wedge b|\Gamma]$ are mutually exclusive, we have,
\begin{equation}
\xi(a\vee b|\Gamma) = \xi(a\vee(\neg a\wedge b)|\Gamma) = \xi(a|\Gamma) + \xi(\neg a\wedge b|\Gamma)
\end{equation}
Adding and subtracting $\xi(a \wedge b|\Gamma)$ to both sides yields,
\begin{align}
\xi(a\vee b|\Gamma) &= \xi(a|\Gamma) + \xi(\neg a\wedge b|\Gamma) + \xi(a\wedge b|\Gamma) - \xi(a \wedge b|\Gamma)\nonumber\\
&= \xi(a|\Gamma) + \left[\frac{}{}\xi(a\wedge b|\Gamma) + \xi(\neg a\wedge b|\Gamma)\right] - \xi(a \wedge b|\Gamma).
\end{align}
Then, using the fact that $(a|\Gamma\wedge b|\Gamma)$ and $(\neg a|\Gamma \wedge b|\Gamma)$ are mutually exclusive we get,
\begin{align}
\xi(a\vee b|\Gamma) &= \xi(a|\Gamma) + \xi((a\wedge b)\vee(\neg a\wedge b)|\Gamma) - \xi(a\wedge b|\Gamma)\nonumber\\
&= \xi(a|\Gamma) + \xi(b|\Gamma) - \xi(a\wedge b|\Gamma),\label{generalsum}
\end{align}
which is the general sum rule.  The general sum rule must also be associative, which one can easily check by evaluating the disjunction in two ways.  First,
\begin{align}
\forall a|\Gamma,b|\Gamma,c|\Gamma \in \tilde{\mathcal{A}} : \xi(a\vee(b\vee c)|\Gamma) &= \xi(a|\Gamma) + \xi(b\vee c|\Gamma) - \xi(a\wedge(b\vee c)|\Gamma)\nonumber\\
&= \xi(a|\Gamma) + \xi(b|\Gamma) + \xi(c|\Gamma) - \xi(b\wedge c|\Gamma)\nonumber\\
&- \xi(a\wedge b|\Gamma) - \xi(a\wedge c|\Gamma) + \xi(a\wedge b\wedge c|\Gamma),
\end{align}
where we used the distributive property for $\xi(a\wedge(b\vee c)|\Gamma) = \xi([a\wedge b]\vee[a\wedge c]|\Gamma)$.  We also have that,
\begin{align}
\forall a|\Gamma,b|\Gamma,c|\Gamma \in \tilde{\mathcal{A}} : \xi((a\vee b)\vee c|\Gamma) &= \xi(a\vee b|\Gamma) + \xi(c|\Gamma) - \xi((a\vee b)\wedge c)|\Gamma)\nonumber\\
&= \xi(a|\Gamma) + \xi(b|\Gamma) + \xi(c|\Gamma) - \xi(a\wedge b|\Gamma)\nonumber\\
&- \xi(a\wedge c|\Gamma) - \xi(b\wedge c|\Gamma) + \xi(a\wedge b\wedge c|\Gamma),
\end{align}
so that in general,
\begin{equation}
\xi(a\vee(b\vee c)|\Gamma) = \xi((a\vee b)\vee c|\Gamma),
\end{equation}
as expected.

\subsection{Representation of the conjunction $f_{\wedge}$}
There are several ways in which we could determine the functional form of $f_{\wedge}$.  No matter which way we choose however, the result must be consistent with what we already know about the form of $f_{\vee}$ and $f_{\neg}$, since the three are related through (\ref{andorrelation1}) and (\ref{andorrelation2}).  In Cox's original derivation \cite{Cox2}, he used the associativity constraint applied to $f_{\wedge}$, which of course leads to the same solution as in (\ref{answertwo}).  It is perfectly reasonable to take this approach, however one must constrain the choice of the function $\phi$ so as not to conflict with the choice for $f_{\vee}$.  In other words, one cannot independently determine that the functional form of $f_{\wedge}$ is a sum, otherwise they violate (\ref{andorrelation1}) and (\ref{andorrelation2}).  We will discuss more about these results in Section \hyperref[rep2]{1.5}.

In order to be consistent with the chosen regraduation of (\ref{answertwo}), we will use the distributive property of $\wedge$ and $\vee$ (\ref{distributiveone}).  This will incorporate directly the already chosen regraduation of $\phi$.  To see this, consider the statement,
\begin{equation}
\forall a|\Gamma,b|\Gamma,c|\Gamma \in \tilde{\mathcal{A}} : a|\Gamma\wedge[b|\Gamma\vee c|\Gamma] \Leftrightarrow [a|\Gamma\wedge b|\Gamma]\vee[a|\Gamma\wedge c|\Gamma].
\end{equation}
Written using the properties (\ref{distributiveone}) we have,
\begin{equation}
\forall a|\Gamma,b|\Gamma,c|\Gamma \in \tilde{\mathcal{A}}_{\Gamma} : a\wedge (b\vee c)|\Gamma \Leftrightarrow [a\wedge b]\vee[a\wedge c]|\Gamma.
\end{equation}
Now consider the degree of belief in the conjunction,
\begin{equation}
\xi(a\wedge(b\vee c)|\Gamma) = f_{\wedge}\left(\frac{}{}\xi(a|\Gamma),\xi(b\vee c|\Gamma\wedge a)\right).
\end{equation}
Due to the distributivity property this is also equal to,
\begin{align}
f_{\wedge}\left(\frac{}{}\xi(a|\Gamma),\xi(b\vee c|\Gamma\wedge a)\right) &= f_{\wedge}\left(\frac{}{}\xi(a|\Gamma),\xi(b|\Gamma\wedge a) + \xi(c|\Gamma\wedge a)\right)\nonumber\\
&= f_{\wedge}\left(\frac{}{}\xi(a|\Gamma),\xi(b|\Gamma\wedge a)\right)\nonumber\\
&+ f_{\wedge}\left(\frac{}{}\xi(a|\Gamma),\xi(c|\Gamma\wedge a)\right).\label{distributiveand}
\end{align}
This is a functional equation for the three quantities $\xi(a|\Gamma), \xi(b|\Gamma\wedge a)$ and $\xi(c|\Gamma \wedge a)$.  This functional equation, which is a multi-variate version of Cauchy's functional equation \cite{Aczel}, demonstrates that $f_{\wedge}$ must be linear in its second argument,
\begin{theorem}[Cauchy's functional equation]\label{Cauchy}
	Let $f_{\wedge}:\Xi\times\Xi\rightarrow\Xi$ be an additive function in its second argument (\ref{distributiveand}).  Then, $f_{\wedge}$ must be linear in its second argument.  The general solution is,
	\begin{equation}
	\forall \xi(a|\Gamma),\xi(b|\Gamma\wedge a) \in \Xi : f_{\wedge}(\xi(a|\Gamma),\xi(b|\Gamma\wedge a)) = \lambda g_{\wedge}(\xi(a|\Gamma))\xi(b|\Gamma\wedge a),\label{gfunction}
	\end{equation}
	where $\lambda \in \mathbb{R}$.
\end{theorem}
A proof for (\ref{gfunction}) is given in \hyperref[cauchyproof]{Appendix B}.  From here we can consider the commutative property of $\wedge$, $a\wedge b = b \wedge a$, so that 
\begin{equation}
f_{\wedge}\left(\frac{}{}\xi(a|\Gamma),\xi(b|\Gamma\wedge a)\right) = f_{\wedge}\left(\frac{}{}\xi(b|\Gamma),\xi(a|\Gamma\wedge b)\right),
\end{equation}
which using (\ref{gfunction}) is written,
\begin{equation}
g_{\wedge}(\xi(a|\Gamma))\xi(b|\Gamma\wedge a) = g_{\wedge}(\xi(b|\Gamma))\xi(a|\Gamma\wedge b).\label{gtogether}
\end{equation}
If we consider the special case for when $a|\Gamma$ and $b|\Gamma$ are independent so that
\begin{equation}
\xi(a|\Gamma\wedge b) = \xi(a|\Gamma) \quad \mathrm{and}\quad \xi(b|\Gamma\wedge a) = \xi(b|\Gamma),
\end{equation}
then we have for (\ref{gtogether}),
\begin{equation}
g_{\wedge}(\xi(a|\Gamma))\xi(b|\Gamma) = g_{\wedge}(\xi(b|\Gamma))\xi(a|\Gamma).
\end{equation}
Thus the function $g_{\wedge}$ must also be linear so that in general we have,
\begin{equation}
f_{\wedge}\left(\frac{}{}\xi(a|\Gamma),\xi(b|\Gamma\wedge a)\right) = \lambda'\xi(a|\Gamma)\xi(b|\Gamma\wedge a).
\end{equation}
The arbitrary constant $\lambda'$ can be set by regraduating $\xi$ again using the property $\xi(a\wedge \tru|\Gamma) = \xi(a|\Gamma)$,
\begin{equation}
\xi(a\wedge \tru|\Gamma) = \lambda\xi(a|\Gamma)\xi_{\tru} = \xi(a|\Gamma).
\end{equation}
Thus we must have $\lambda = \xi_{\tru}^{-1}$,
\begin{equation}
\xi(a\wedge b|\Gamma) = \frac{1}{\xi_{\tru}}\xi(a|\Gamma)\xi(b|\Gamma\wedge a).
\end{equation}
Multiplying through by $\xi_{\tru}^{-1}$,
\begin{equation}
\frac{\xi(a\wedge b|\Gamma)}{\xi_{\tru}} = \frac{\xi(a|\Gamma)}{\xi_{\tru}}\frac{\xi(b|\Gamma\wedge a)}{\xi_{\tru}}.
\end{equation}
Thus, we can define $\xi' = \xi/\xi_{\tru}$ so that we recover the product rule for conjunction,
\begin{equation}
\xi'(a\wedge b|\Gamma) = \xi'(a|\Gamma)\xi'(b|\Gamma\wedge a).
\end{equation}

\subsection{Representations using $C^{\infty}(\tilde{\mathcal{A}})$}\label{rep2}

Let us assume instead that the space of functions $\Xi$ belongs to the other extreme, $\Xi \subset C^{\infty}(\tilde{\mathcal{A}})$.  Using the tools of calculus we can derive the same sum and product rules from the previous sections.  The derivation here follows along the same lines as the one first achieved by Cox \cite{Cox} and later improved by others \cite{Jaynes,CatichaBook}.  We begin again with the derivation of the sum rule for mutually exclusive propositions,
\begin{align}
\xi(a\vee b\vee c|\Gamma) &= f_{\vee}\left(\frac{}{}f_{\vee}\left(\frac{}{}\xi(a|\Gamma),\xi(b|\Gamma)\right),\xi(c|\Gamma)\right)\nonumber\\
&= f_{\vee}\left(\frac{}{}\xi(a|\Gamma),f_{\vee}\left(\frac{}{}\xi(b|\Gamma),\xi(c|\Gamma)\right)\right).
\end{align}  
To simplify the derivation, we adopt the notation,
\begin{equation}
\xi(a|\Gamma) \stackrel{\mathrm{def}}{=} a, \quad \xi(b|\Gamma) \stackrel{\mathrm{def}}{=} b, \quad \mathrm{and} \quad \xi(c|\Gamma) \stackrel{\mathrm{def}}{=} c.
\end{equation}
Then, the associativity equation becomes,
\begin{equation}
f_{\vee}(f_{\vee}(a,b),c) = f_{\vee}(a,f_{\vee}(b,c)).
\end{equation}
To probe the properties of this equation, we will take various derivatives up to order two, although we assume that the functions $\xi$ are smooth.  Since there are four different versions of the function $f_{\vee}$ above, we will introduce an additional notation to simplify future equations,
\begin{equation}
f_{\vee}(a,b) \stackrel{\mathrm{def}}{=} \mathfrak{ab} \quad \mathrm{and} \quad f_{\vee}(b,c) \stackrel{\mathrm{def}}{=} \mathfrak{bc}.
\end{equation}
Then the associativity equation simplifies further,
\begin{equation}
f_{\vee}(\mathfrak{ab},c) = f_{\vee}(a,\mathfrak{bc}).
\end{equation}
Taking the derivative with respect to $a, b$ and $c$ we find,
\begin{align}
\partial_{\mathfrak{ab}}f_{\vee}(\mathfrak{ab},c)\partial_af_{\vee}(a,b) &= \partial_af_{\vee}(a,\mathfrak{bc}),\label{first}\\
\partial_{\mathfrak{ab}}f_{\vee}(\mathfrak{ab},c)\partial_bf_{\vee}(a,b) &= \partial_{\mathfrak{bc}}f_{\vee}(a,\mathfrak{bc})\partial_bf_{\vee}(b,c),\label{second}\\
\partial_cf_{\vee}(\mathfrak{ab},c) &= \partial_{\mathfrak{bc}}f_{\vee}(a,\mathfrak{bc})\partial_cf_{\vee}(b,c).\label{third}
\end{align}
One observation is that the derivatives of $f_{\vee}$ with respect to either the first or second argument is always positive.  The reason is due to axiom (\ref{transitive}), which tells us that
\begin{equation}
\partial_af_{\vee}(a,b) \geq 0 \quad \mathrm{and} \quad \partial_bf_{\vee}(a,b) \geq 0.\label{positive}
\end{equation}
Dividing the second equation (\ref{second}) by the first (\ref{first}) we find,
\begin{equation}
\frac{\partial_bf_{\vee}(a,b)}{\partial_af_{\vee}(a,b)} = \frac{\partial_{\mathfrak{bc}}f_{\vee}(a,\mathfrak{bc})}{\partial_af_{\vee}(a,\mathfrak{bc})}\partial_bf_{\vee}(b,c).\label{firstpart}
\end{equation}
Likewise, dividing the second equation by the third (\ref{third}) gives,
\begin{equation}
\partial_bf_{\vee}(a,b)\frac{\partial_{\mathfrak{ab}}f_{\vee}(\mathfrak{ab},c)}{\partial_cf_{\vee}(\mathfrak{ab},c)} = \frac{\partial_bf_{\vee}(b,c)}{\partial_cf_{\vee}(b,c)}.
\end{equation}
We can introduce another convenient notation to simplify the ratios in the following equations.  Consider the definition,
\begin{equation}
K(a,b) \stackrel{\mathrm{def}}{=} \frac{\partial_b f_{\vee}(a,b)}{\partial_a f_{\vee}(a,b)},
\end{equation}
so that the ratios,
\begin{equation}
\frac{\partial_bf_{\vee}(a,b)}{\partial_af_{\vee}(a,b)} = K(a,b), \quad \frac{\partial_{\mathfrak{bc}}f_{\vee}(a,\mathfrak{bc})}{\partial_af_{\vee}(a,\mathfrak{bc})} = K(a,\mathfrak{bc}), \quad \frac{\partial_cf_{\vee}(b,c)}{\partial_bf_{\vee}(b,c)} = K(b,c).
\end{equation}
The right hand side of (\ref{firstpart}) contains a term proportional to the derivative of $f_{\vee}(b,c)$ with respect to $b$.  If we multiply both sides of (\ref{firstpart}) by $K(b,c)$, the right hand side of (\ref{firstpart}) will then be proportional to the derivative of $f_{\vee}(b,c)$, but now with respect to $c$,
\begin{align}
K(a,b) &= K(a,\mathfrak{bc})\partial_bf_{\vee}(b,c),\label{firstsecondpart}\\
K(a,b)K(b,c) &= K(a,\mathfrak{bc})\partial_cf_{\vee}(b,c).\label{secondpart}
\end{align}
One can easily verify that taking the derivative of (\ref{secondpart}) with respect to $b$ is equivalent to taking the derivative of (\ref{firstsecondpart}) with respect to $c$,
\begin{equation}
\frac{\partial}{\partial b}\left[K(a,\mathfrak{bc})\partial_cf_{\vee}(b,c)\right] = \frac{\partial}{\partial c}\left[K(a,\mathfrak{bc})\partial_bf_{\vee}(b,c)\right]\label{derives}
\end{equation}
But, the right hand side of (\ref{derives}) is zero due to the left hand side of (\ref{firstsecondpart}),
\begin{equation}
\frac{\partial}{\partial c}\left[K(a,b)\right] = 0.
\end{equation}
Thus, the derivative of the left hand side of (\ref{secondpart}) with respect to $b$ must also be zero,
\begin{equation}
\frac{\partial}{\partial b}\left[K(a,b)K(b,c)\right] = 0,
\end{equation}
which leads to,
\begin{equation}
\frac{1}{K(a,b)}\frac{\partial K(a,b)}{\partial b} = -\frac{1}{K(b,c)}\frac{\partial K(b,c)}{\partial b}.\label{equals}
\end{equation}
Since the left hand side is independent of $c$ while the right hand side is independent of $a$, both sides must only depend on $b$, i.e.,
\begin{equation}
\frac{1}{K(a,b)}\frac{\partial K(a,b)}{\partial b} = k(b), \quad \mathrm{and} \quad \frac{1}{K(b,c)}\frac{\partial K(b,c)}{\partial b} = -k(b).\label{thirdpart}
\end{equation}
Since both $\partial_bf_{\vee}(a,b)$ and $\partial_af_{\vee}(a,b)$ are positive (\ref{positive}), we can integrate both equations in (\ref{thirdpart}) to get,
\begin{align}
K(a,b) &= K(a,0)\exp\left[\int_0^bdb'\,k(b')\right],\label{intone}\\
K(b,c) &= K(0,c)\exp\left[-\int_0^bdb'\, k(b')\right].\label{inttwo}
\end{align}
Defining the term,
\begin{equation}
\mathcal{K}(b) = \exp\left[-\int_0^bdb'\,k(b')\right],\label{intdef1}
\end{equation}
we find from (\ref{intone}) and (\ref{inttwo}) that,
\begin{equation}
K(a,b) = K(a,0)\mathcal{K}(b)^{-1} = K(0,b)\mathcal{K}(a).\label{arg}
\end{equation}
So, $K(a,b)$ can be expressed either with respect to its first argument $K(a,0)$ or its second $K(0,b)$.  Rearranging (\ref{arg}) we get,
\begin{equation}
K(a,0)\mathcal{K}(a)^{-1} = K(0,b)\mathcal{K}(b) = \mathfrak{K},
\end{equation}
which are equal to a constant $\mathfrak{K}$ since the left hand side is independent of $b$ and the right hand side is independent of $a$.  We can then rewrite (\ref{arg}) as,
\begin{equation}
K(a,b) = \mathfrak{K}\frac{\mathcal{K}(a)}{\mathcal{K}(b)}.
\end{equation}
Substituting this expression back into (\ref{firstsecondpart}) and (\ref{secondpart}) we find,
\begin{equation}
\partial_bf_{\vee}(b,c) = \frac{\mathcal{K}(\mathfrak{bc})}{\mathcal{K}(b)} \quad \mathrm{and}\quad \partial_cf_{\vee}(b,c) = \mathfrak{K}\frac{\mathcal{K}(\mathfrak{bc})}{\mathcal{K}(c)}.\label{stuff1}
\end{equation}
Then, using the definition $f_{\vee}(b,c) = \mathfrak{bc}$, and (\ref{stuff1}) we find,
\begin{align}
df_{\vee}(b,c) &= \partial_bf_{\vee}(b,c)db + \partial_cf_{\vee}(b,c)dc\nonumber\\
&= \frac{\mathcal{K}(\mathfrak{bc})}{\mathcal{K}(b)}db + \mathfrak{K}\frac{\mathcal{K}(\mathfrak{bc})}{\mathcal{K}(c)}dc,
\end{align}
which is rearranged to give,
\begin{equation}
\frac{df_{\vee}(b,c)}{\mathcal{K}(\mathfrak{bc})} = \frac{db}{\mathcal{K}(b)} + \mathfrak{K}\frac{dc}{\mathcal{K}(c)}.\label{almost}
\end{equation}
Using (\ref{intdef1}) we can define the integral,
\begin{equation}
\phi(a) \stackrel{\mathrm{def}}{=} \int_0^ada'\, \mathcal{K}(a')^{-1} = \int_0^ada'\,\exp\left[\int_0^{a'}da''\,k(a'')\right],\label{intdef2}
\end{equation}
which is necessarily positive.  Then, the solution to (\ref{almost}) is,
\begin{equation}
\phi(f_{\vee}(b,c)) = \phi(b) + \mathfrak{K}\phi(c) + \alpha,
\end{equation} 
where $\alpha$ is an integration constant.  Since $\vee$ is commutative, i.e. $f_{\vee}(b,c) = f_{\vee}(c,b)$, the constant $\mathfrak{K} = 1$.  We then have the solution for $f_{\vee}(b,c)$,
\begin{equation}
f_{\vee}(b,c) = \phi^{-1}\left(\frac{}{}\phi(b) + \phi(c) + \alpha\right).\label{orsolution}
\end{equation}
In terms of the degrees of belief $\xi(a|\Gamma)$, the above becomes,
\begin{equation}
\phi(\xi(a\vee b|\Gamma)) = \phi(\xi(a|\Gamma)) + \phi(\xi(b|\Gamma)) + \alpha.
\end{equation}
We can always regraduate our degrees of belief so that $\xi(a|\Gamma) \rightarrow \xi'(a|\Gamma) = \phi(\xi(a|\Gamma)) + \alpha$.  Then the mutually exclusive sum rule becomes,
\begin{equation}
\xi'(a\vee b|\Gamma) = \xi'(a|\Gamma) + \xi'(b|\Gamma),
\end{equation}
which is the familiar result from Section \hyperref[disjun]{1.3}.  Application of the identity $a\vee \fals = a$ suggests that $\xi'_{\fals} = 0$, so that,
\begin{equation}
\phi(\xi_{\fals}) = -\alpha.\label{alphafalse}
\end{equation} 
\paragraph{The distributivity constraint ---}
To determine the product rule, we can make similar arguments concerning derivatives.  The approach in \cite{CatichaBook} makes use of the distributivity constraint (\ref{distributivetwo}), while in \cite{Cox} we again only need associativity.  The argument from \cite{CatichaBook} is as follows.  Consider the distributivity condition for $\wedge$ and $\vee$,
\begin{equation}
\xi(a\wedge(b\vee c)|\Gamma) = \xi\left((a\wedge b|\Gamma)\vee (a\wedge c|\Gamma)\right).
\end{equation}
Using the definition of $f_{\wedge}$ and $f_{\vee}$ when $b$ and $c$ are mutually exclusive we have,
\begin{align}
f_{\wedge}\left(\xi(a|\Gamma),\xi(b\vee c|\Gamma\wedge a)\right) &= f_{\wedge}\left(\xi(a|\Gamma),\xi(b|\Gamma\wedge a)\right)\nonumber\\
&+ f_{\wedge}\left(\xi(a|\Gamma),\xi(c|\Gamma\wedge a)\right).\label{firstand}
\end{align}
We can simplify the above by identifying,
\begin{equation}
\xi(a|\Gamma) \stackrel{\mathrm{def}}{=} a, \quad \xi(b|\Gamma\wedge a) \stackrel{\mathrm{def}}{=} b, \quad \mathrm{and}\quad \xi(c|\Gamma\wedge a) \stackrel{\mathrm{def}}{=} c,
\end{equation}
so that, $\xi(b\vee c|\Gamma\wedge a) = b+c$.  Substituting these definitions into (\ref{firstand}) gives,
\begin{equation}
f_{\wedge}(a,b+c) = f_{\wedge}(a,b) + f_{\wedge}(a,c).\label{secondand}
\end{equation}
Since the functions on the right hand side are independent of $c$ and $b$ resepectively, differentiating with respect to $b$ and $c$ we find,
\begin{equation}
\frac{\partial^2 f_{\wedge}(a,b+c)}{\partial b \partial c} = 0.
\end{equation}
Likewise, letting $b + c = d$, we find,
\begin{equation}
\frac{\partial^2f_{\wedge}(a,d)}{\partial d^2} = 0,
\end{equation}
which shows that $f_{\wedge}$ is linear in its second argument.  The above is easily integrated to find,
\begin{equation}
f_{\wedge}(a,b) = f_1(a)b + f_2(a).
\end{equation}
Using this solution in (\ref{secondand}) shows that the term $f_2(a) = 0$, thus we have,
\begin{equation}
f_{\wedge}(a,b) = f_1(a)b.\label{thirdand}
\end{equation}
Finally, consider $\xi(a|d) = \xi(a\wedge d|d)$ to evaluate,
\begin{equation}
\xi(a|d) = \xi(a\wedge d|d) = f_{\wedge}(\xi(a|d),\xi(d|a\wedge d)) = f_{\wedge}(\xi(a|d),\xi_{\tru}),
\end{equation}
which according to (\ref{thirdand}) is,
\begin{equation}
\xi(a|d) = f_1(a|d)\xi_{\tru}.
\end{equation}
We can then just regraduate the function $f_1(a|d)$ by dividing $\xi(a|d)$ by $\xi_{\tru}$ so that,
\begin{equation}
\frac{\xi(a\wedge b|\Gamma)}{\xi_{\tru}} = \frac{\xi(a|\Gamma)}{\xi_{\tru}}\frac{\xi(b|\Gamma \wedge a)}{\xi_{\tru}}.\label{andsolutioncaticha}
\end{equation}
\paragraph{The associativity constraint revisited ---}
The derivation of the conjunction relation in \cite{Cox} uses the associativity constraint instead of distributivity,
\begin{equation}
f_{\wedge}(a,f_{\wedge}(b,c)) = f_{\wedge}(f_{\wedge}(a,b),c).\label{andassociative}
\end{equation}
In fact, the derivation is exactly the same as for the disjunction, up through step (\ref{almost}),
\begin{equation}
\frac{df_{\wedge}(b,c)}{\mathcal{K}(\mathfrak{bc})} = \frac{db}{\mathcal{K}(b)} + \mathfrak{K}\frac{dc}{\mathcal{K}(c)}.\label{almost2}
\end{equation}
The difference is in the definition of the solution. Define the function,
\begin{equation}
\psi(a) \stackrel{\mathrm{def}}{=} e^{\phi(a)} = \exp\left[\int_0^ada'\,\mathcal{K}(a')^{-1}\right].\label{psidef}
\end{equation}
Then, the solution to (\ref{almost2}) is,
\begin{equation}
\alpha \psi(f_{\wedge}(a,b)) = \psi(a)\psi(b)^{\mathfrak{K}}.
\end{equation}
Inserting this into (\ref{andassociative}) demonstrates that again $\mathfrak{K} = 1$,
\begin{equation}
\alpha\psi(f_{\wedge}(\xi(a|\Gamma),\xi(b|\Gamma\wedge a))) = \psi(\xi(a|\Gamma))\psi(\xi(b|\Gamma\wedge a)).\label{andsolution}
\end{equation}
To determine $\alpha$, we appeal to the special case of $\xi(a\wedge a|\Gamma) = \xi(a|\Gamma)$ so that (\ref{andsolution}) becomes,
\begin{align}
\alpha\psi(\xi(a\wedge a|\Gamma)) &= \psi(\xi(a|\Gamma))\psi(\xi(a|\Gamma\wedge a))\nonumber\\
&= \psi(\xi(a|\Gamma))\psi(\xi_{\tru}),
\end{align}
so that,
\begin{equation}
\alpha = \psi(\xi_{\tru}),
\end{equation}
and therefore,
\begin{equation}
\psi(\xi_{\tru})\psi(\xi(a\wedge b|\Gamma)) = \psi(\xi(a|\Gamma))\psi(\xi(b|\Gamma\wedge a)).\label{andsolution2}
\end{equation}
One can then choose a regraduation so that $\xi'(a|\Gamma) = \psi(\xi(a|\Gamma))/\psi(\xi_{\tru})$, which then leads to the recognizable form,
\begin{equation}
\xi'(a\wedge b|\Gamma) = \xi'(a|\Gamma)\xi'(b|\Gamma\wedge a).\label{andsolution3}
\end{equation}
One may recognize that the distinction between the solutions (\ref{orsolution}) and (\ref{andsolution}) only differs by application of the exponential in (\ref{psidef}), which is monotonic and hence only applies a rescaling.  Nowhere in either derivation did we distinguish any defining characteristics of $\wedge$ and $\vee$.  Therefore, it seems that we cheated in picking (\ref{intdef2}) over (\ref{psidef}) in the derivation of $f_{\vee}(a,b)$ by peeking ahead at the answer we wanted, when the other solution is just as good.  The key is to understand that one cannot use the same arguments for independently deriving the functional form of $f_{\vee}$ and $f_{\wedge}$, since they are linked through the negation $f_{\neg}$.  If we had instead chose (\ref{psidef}) for the solution of $f_{\vee}(a,b)$, we would have arrived at the following,
\begin{equation}
\alpha\psi(\xi(a \vee b|\Gamma)) = \psi(\xi(a|\Gamma))\psi(\xi(b|\Gamma)).
\end{equation}
Then using the special case, $\xi(a\vee \fals|\Gamma) = \xi(a|\Gamma)$, we find that instead the constant $\alpha = \psi(\xi_{\fals})$,
\begin{equation}
\psi(\xi_{\fals})\psi(\xi(a\vee b|\Gamma)) = \psi(\xi(a|\Gamma))\psi(\xi(b|\Gamma)).\label{alphafalse2}
\end{equation}
Applying the log to both sides yields,
\begin{equation}
\phi(\xi_{\fals}) + \phi(\xi(a\vee b|\Gamma)) = \phi(\xi(a|\Gamma)) + \phi(\xi(b|\Gamma)).
\end{equation}
Using the result from (\ref{alphafalse}), the regraduated $\phi(\xi(a|\Gamma))$ is defined as,
\begin{equation}
\xi'(a|\Gamma) = \phi(\xi(a|\Gamma)) - \phi(\xi_{\fals}),
\end{equation}
which gives the familiar result from (\ref{orsolution}).  The regraduated form of (\ref{alphafalse2}) on the other hand is,
\begin{equation}
\frac{\psi(\xi(a\vee b|\Gamma))}{\psi(\xi_{\fals})} =  \frac{\psi(\xi(a|\Gamma))}{\psi(\xi_{\fals})}\frac{\psi(\xi(b|\Gamma))}{\psi(\xi_{\fals})}.
\end{equation}
Call $\xi''(a|\Gamma) = \psi(\xi(a|\Gamma))/\psi(\xi_{\fals})$.  Then we have,
\begin{equation}
\xi''(a\vee b|\Gamma) = \xi''(a|\Gamma)\xi''(b|\Gamma).
\end{equation}
Going back to the form (\ref{alphafalse2}), if we appeal to the special case $a\vee \neg a = \tru$, we find,
\begin{equation}
f_{\neg}(\xi''(a|\Gamma)) = \xi''_{\tru}\xi''(a|\Gamma)^{-1}
\end{equation}
Using the same arguments that lead to (\ref{generalsum}), we find for the general sum rule,
\begin{equation}
\xi''(a\vee b|\Gamma) = \frac{\xi''(a|\Gamma)\xi''(b|\Gamma)}{\xi''(a\wedge b|\Gamma)},
\end{equation}
and thus the function $f_{\wedge}(a,b)$ is defined as,
\begin{equation}
\xi''(a\wedge b|\Gamma) = \frac{\xi''(a|\Gamma)\xi''(b|\Gamma)}{\xi''(a\vee b|\Gamma)}.
\end{equation}
This gives the sum rule in terms of a product and a corresponding product rule in terms of a ratio.  If one exponentiates the solution for the product rule from (\ref{andsolutioncaticha}) we find,
\begin{equation}
\xi''(a\wedge b|\Gamma) = \xi''(a|\Gamma)^{\log\xi''(b|\Gamma \wedge a)},
\end{equation}
which is certainly unappealing and not intuitive, however when $b = \neg a$, we recover $\xi''(a\wedge b|\Gamma) = \xi''_{\fals}$ as expected.  Since $\phi(\xi_{\fals}) = 0$, the value $\psi(\xi_{\fals}) = \exp[\phi(\xi_{\fals})] = 1$.  Thus in this logarithmic scaling, total disbelief is given the value $\xi''_{\fals} = 1$.  Complete certainty on the other hand corresponds to $\xi''_{\tru} = \psi(\xi_{\tru}) = \exp[\phi(\xi_{\tru})]$, which has not yet been specified.  Given that the image of $\phi$ is now restricted to the positive reals, $\mathrm{im}(\phi) = \mathbb{R}_+$, one can always regraduate the solutions so that $\phi(\xi_{\tru}) = 1$ which sends the image to,
\begin{equation}
\mathrm{im}(\phi) = [0,1].
\end{equation}
In this scale, $\mathrm{im}(\psi) = [1,e]$.  

Both solution (\ref{andsolution}) and (\ref{orsolution}) are simply the result of the associativity constraint, which both $\wedge$ and $\vee$ satisfy.  However, once one determines the form of either $\{f_{\vee},f_{\neg}\}$ or $\{f_{\wedge},f_{\neg}\}$, the remaining function is constrained through the relationships in (\ref{andorrelation1}) and (\ref{andorrelation2}) as well as with the definitions of the lower and upper bounds given from (\ref{unitaland}) and (\ref{unitalor}).

\paragraph{Remarks on the results ---}
We have identified that the rules for manipulating degrees of belief are the rules for probabilities.  Hence, we make the identification $\xi \rightarrow P$ and $\xi_{\tru}\rightarrow P_{\tru} = 1$ so that the sum and product rules for all $a|\Gamma, b|\Gamma \in \tilde{\mathcal{A}}$ are written in the standard notation,
\begin{align}
P(a\vee b|\Gamma) &= P(a|\Gamma) + P(b|\Gamma) - P(ab|\Gamma)\label{sumrule}\\
P(ab|\Gamma) &= P(a)P(b|\Gamma a) = P(b)P(a|\Gamma b),\label{productrule}
\end{align}
where for ease of notation we have adopted that conjunction is written as a juxtaposition, $a\wedge b \rightarrow ab$.  The negation is written,
\begin{equation}
P(\neg a|\Gamma) = 1 - P(a|\Gamma).
\end{equation}
The space of probabilities is written $\Xi \rightarrow \mathbf{P}$, which is the set of maps from $\tilde{\mathcal{A}}$ to the interval $[0,1]$ which are consistent with the properties of $f_{\vee},f_{\wedge}$ and $f_{\neg}$.  Our inductive inference framework now consists of the web of beliefs $\mathbf{P}(\tilde{\mathcal{A}})$ constrained according to the rationale $\mathcal{R}$ which contains the representations of $f_{\neg}, f_{\wedge}$ and $f_{\vee}$.

Several authors have argued that the Cox derivation given above cannot apply generally to all universes of discourse.  Specifically, Halpern \cite{Halpern} claims that it fails for sets $\tilde{\mathcal{A}}$ which are finite.  His argument is that, for special cases in which the space $\Xi$ is discrete, one cannot assume that $\Xi = [\xi_{\fals},\xi_{\tru}] \subseteq \mathbb{R}$ is an interval.  Thus, the results from Cox and more importantly from Acz\'{e}l do not hold.  This is certainly true since the space $\Xi$ must be continuous according to the equivalence in theorem (\ref{aczel}) to hold, however there is no requirement that discrete spaces $\tilde{\mathcal{A}}$ must be represented by discrete degrees of belief $\Xi$.  Assuming that $\Xi$ is continuous does not preclude it from representing degrees of belief for discrete spaces $\tilde{\mathcal{A}}$.  In other words, continuity in $\Xi$ does not require that $\Xi(\tilde{\mathcal{A}})$ be surjective in $\Xi$.

Another argument, that the results of Cox do not have universal applicability, comes from Colyvan \cite{Colyvan}.  The author correctly points out the subtleties of calling Cox's results an extension of propositional logic, specifically the fact that some statements of interest are not propositions at all, like the situation of incomplete arguments from chapter one.  He also suggests that Cox's results cannot be general, since they rely on preservation of the law of excluded middle, which some logics abandon.  Neither of these objections are problems for us since we have already dealt with the problem of incomplete information in the definition of $\tilde{\mathcal{A}}$.  Colyvan's objection only goes to show the uselessness of the intuitionist program, since the rules of probability are already known to work in the real world.    

\paragraph{Consistency between $f_{\vee}$, $f_{\wedge}$ and $f_{\neg}$ ---}
It is easy to check that the relationships between the connectives $\vee$, $\wedge$ and $\neg$ are satisfied by the regraduated solutions (\ref{sumrule}) and (\ref{productrule}).  
\begin{align}
P(a \vee b) &= P(\neg(\neg a \neg b))\nonumber\\
&= 1 - P(\neg a\neg b) = 1 - P(\neg a)P(\neg b|\neg a)\nonumber\\
&= 1 - [P(\neg a)P(\neg b|\neg a) + P(a)P(\neg b|a)] + P(a)P(\neg b|a)\nonumber\\
&= 1 - P(\neg b) + P(a)(1 - P(b|a))\nonumber\\
&= P(a) + P(b) - P(a)P(b|a).
\end{align}
Likewise we have,
\begin{align}
P(ab) &= P(\neg(\neg a \vee \neg b))\nonumber\\
&= 1 - P(\neg a \vee \neg b) = 1 - [P(\neg a) + P(\neg b) - P(\neg a \neg b)]\nonumber\\
&= 1 - P(\neg a) - P(\neg b) + [P(\neg a)P(\neg b|\neg a) + P(a)P(\neg b|a)] - P(a)P(\neg b|a)\nonumber\\
&= P(a)(1 - P(\neg b|a)) = P(ab).
\end{align}
From here one can form the degree of belief for any $n$-ary proposition from the fact that $\wedge$, $\vee$ and $\neg$ form a functionally complete set.  

\paragraph{Modal degrees of belief? ---}
Now that we've determined a set of representations for the connectives $\Omega$, can we do the same for the modal operators?  Essentially we would like to determine if some representation of the form $f_{\square}$ and $f_{\diamondsuit}$ can make sense in our system.  The existence of such functions must be compatible with the other representations $\{f_{\Omega}\}$ and must obey the simple axioms,
\begin{align}
P(\square a) = f_{\square}(P(a)) &= f_{\neg}(f_{\diamondsuit}(f_{\neg}(P(a))))\nonumber\\
&= 1 - f_{\diamondsuit}(1 - P(a)),
\end{align}
and the dual,
\begin{align}
P(\diamondsuit a) = f_{\diamondsuit}(P(a)) &= f_{\neg}(f_{\square}(f_{\neg}(P(a))))\nonumber\\
&= 1 - f_{\square}(1 - P(a)).
\end{align}
Depending on the axioms that one includes in $\mathcal{I}_{\square,\diamondsuit}$, one could potentially constrain the functional form of $f_{\square}$ and $f_{\diamondsuit}$ further.  For example, if one adopts axioms $\mathbf{NC}$ and $\mathbf{CN}$ so that $\square (a \wedge b) \Leftrightarrow (\square a)\wedge (\square b)$, then the form of $f_{\square}$ is constrained to obey,
\begin{align}
P(\square(a \wedge b)) &= f_{\square}(P(a \wedge b))\nonumber\\
&= f_{\square}(P(a))f_{\square}(P(b|\square a))
\end{align}  
We will not pursue these ideas any further in this thesis and instead suggest that it be considered for future work.

\paragraph{Degrees of beliefs for predicates? ---}
An obvious question at this point is wether we can assign a meaningful degree of belief to some predicate $A(x)$.  Certainly whenever $x$ is bound by quantification the term $A(x)$ just becomes a standard statement, e.g. something like $P(\exists xA(x))$ is well defined.  But what about the statement $P(A(x))$ where $x$ remains free?  Such a degree of belief is not well defined, and hence the notation $P(A(x))$ is not meaningful.  What we can do however, is adopt a different notation for predicates which agrees with the standard meaning in probability.  

Whenever an argument is a predicate in which at least one of the variables remains free, we write a lowercase $p$, so that $p(A(x))$ refers to a \textit{probability distribution function} over the predicate $A(x)$.  Often times we take a shorthand notation and simply write the distribution as $p(x)$, leaving out the fact that $x$ refers to a predicate $A(x)$.  It is often the case that this notation can obscure meaning.  When we write something like $p(x)$, we are not really writing a probability but are instead referring to a probability valued function over the variable $x$\footnote{The most common example of this situation is whenever $x$ is some subset of $\mathbb{R}$, i.e. $x \in X \subseteq \mathbb{R}$.  For example, $x$ could refer to the position of a particle in one dimension.  Then, the predicate $A(x)$ is identified as $A(x) = $``the particle is at position $x$.''  Because the universe of discourse is typically mentioned ahead of time, there is seldom any confusion when we write $p(x)$ to refer to the predicate $p(A(x))$.}.

Whenever a value of $x$ is specified, so that $x = x_0$, we identify,
\begin{equation}
p(x=x_0) = P(x_0) = P(A(x_0)).
\end{equation}
When $x$ is continuous, then the function $p(x)$ is called a \textit{probability density function}, indeed because it transforms as a density.  The probability associated with some continuous interval $\Delta x \subset X$ is given by the integral,
\begin{equation}
P(\Delta x) = \int_{\Delta x}dx\, p(x),
\end{equation} 
where $p(x)dx$ is the probability assigned to the infinitesimal volume $dx$.  This is nothing more than a sum of all the probabilities $P(a)$ where $a \in \Delta x$.  For predicates of several variables, we write the free variables separated by a comma so that $p(A(x,y))=p(x,y)$.  Since the sum and product rule are valid for any statement, they can be written in terms of distributions,
\begin{align}
p(x,y) &= p(x)p(y|x) = p(y)p(x|y),\\
p(x\vee y) &= p(x) + p(y) - p(x,y).
\end{align}

\subsection{Inductive inference rules $\mathcal{Z}$}\label{section36}
%
%
%
We will briefly survey some important inference rules.  Inductive inference rules can either be direct consequences of the axioms, or exist as rules of thumb.  We will mainly focus on the former case and leave the latter to another chapter.
\paragraph{Bayes theorem ---}
One of the most important results of probability theory is that of Bayes' theorem \cite{Bayes}, which was arguably first identified by Laplace.  Consider the commutativity of $\wedge$ so that the product rule can be written in two equivalent ways,
\begin{equation}
p(a b|\Gamma) = p(a|\Gamma)p(b|\Gamma a) = p(b|\Gamma)p(a|\Gamma b).
\end{equation}
Thus, it must be the case that,
\begin{equation}
p(a|\Gamma b) = \frac{p(a|\Gamma)p(b|\Gamma  a)}{p(b|\Gamma)} \quad \mathrm{and} \quad p(b|\Gamma  a) = \frac{p(b|\Gamma)p(a|\Gamma  b)}{p(a|\Gamma)}.\label{Bayes}
\end{equation}
This is an incredibly powerful statement, for several reasons.  First, this demonstrates that the web of beliefs $\mathbf{P}(\tilde{\mathcal{A}})$ is not just constrained to be consistent with $\wedge,\vee$ and $\neg$, but that once one determines $p(a|\Gamma)$ and $p(b|\Gamma)$, the conditional degrees of belief $p(a|\Gamma b)$ and $p(b|\Gamma a)$ cannot be assigned independently.  Because of this, one is able to make inferences about statements $a|\Gamma b$ when $b|\Gamma a$ is known and vice versa.

One can write the conditional $p(a|\Gamma b)$ as a function $f_{\phi}$ of the other degrees of belief,
\begin{equation}
p(a|\Gamma b) \stackrel{\mathrm{def}}{=} f_{\phi}\left(\frac{}{}p(a|\Gamma),p(b|\Gamma),p(b|\Gamma a)\right),
\end{equation}  
or equivalently in terms of the conjunction,
\begin{equation}
p(a|\Gamma b) \stackrel{\mathrm{def}}{=} f_{\phi}'\left(\frac{}{}p(ab|\Gamma),p(a|\Gamma)\right),
\end{equation}
so that either $f_{\phi}$ or $f_{\phi}'$ represent the application of Bayes theorem.

\paragraph{Marginalization ---}
Sometimes one would like to conduct a joint inference on $a|\Gamma$ and $b|\Gamma$, when not much is known about one of the statements.  Consider that we have the joint distribution $p(ab|\Gamma)$, but we have little information about $b|\Gamma$. A convenient procedure then, is to \textit{marginalize} over $b|\Gamma$ by considering instead the probability $p(a[b\vee \neg b]|\Gamma)$,
\begin{align}
p(a|\Gamma) = p(a[b\vee \neg b]|\Gamma) &= p(ab|\Gamma) + p(a\neg b|\Gamma)\nonumber\\
&= \sum_b p(a|\Gamma)p(b|\Gamma a).\label{totalprob} 
\end{align}
The distribution $p(a|\Gamma)$ is called the \textit{marginal} and equation (\ref{totalprob}) is often called the \textit{law of total probability}.  When the universe of discourse $X\times Y$ is continuous, the marginal is an integral,
\begin{equation}
p(x) = \int dy\, p(x,y) = \int dy\, p(y)p(x|y).
\end{equation}

\paragraph{Expected values ---}
An extremely useful concept is that of the \textit{expected value}.  For a discrete universe of discourse $A$, it is simply the weighted average over all $a \in A$,
\begin{equation}
\langle a\rangle \stackrel{\mathrm{def}}{=} \sum_{a\in A} a P(a).
\end{equation}
The expected value is sometimes written $\mathbb{E}[a] = \langle a \rangle \stackrel{\mathrm{def}}{=} \mu_{a}$, which is also called the \textit{mean}.  Typically when we use the words \textit{expected value} we are referring to the weighted average over each element of the subject matter, however the words expected value can apply more generally for any power of $a$,
\begin{equation}
\langle a^n\rangle = \sum_{a\in X}a^nP(a).
\end{equation}
More general still, one can define the expected value for any function of $a$,
\begin{equation}
\langle f(a) \rangle = \sum_{a\in A} f(a)P(a).
\end{equation}
For continuous universes of discourse $X$ the sum gets replaced by an integral,
\begin{equation}
\langle f(x) \rangle = \int_{X} dx\, f(x)p(x).
\end{equation}

\paragraph{Bayes' rule ---}\label{bayesrules}
Everything we have discussed so far represents a calculus for manipulating beliefs when no new information is present.  Now we ask the question of how to \textit{update} our beliefs once new information becomes available in the form of data.  We can formulate this problem in terms of deciding how one should assign a new web of beliefs $\mathbf{P}(\tilde{\mathcal{A}})$ given some new information about $\tilde{\mathcal{A}}$.  We will represent a generic change in the web of beliefs by some function $\star:\mathbf{Q}\rightarrow \mathbf{P}$ so that $\mathbf{Q}(\tilde{\mathcal{A}})\xrightarrow{\star}\mathbf{P}(\tilde{\mathcal{A}})$.

Consider that we have some prior information about the relationship between two statements $a|\Gamma$ and $b|\Gamma$, whose joint distribution is given by,
\begin{equation}
q(ab|\Gamma) = q(a|\Gamma)q(b|\Gamma a) \in \mathbf{Q},
\end{equation} 
which we call the \textit{prior}.  Given new information about $b|\Gamma$, we would like to determine a suitable \textit{posterior} $p(a|\Gamma) \in \mathbf{P}$.  It is important to emphasize at this point that the universe of discourse does not consist of just the statements $a|\Gamma$ or $b|\Gamma$, but of both.  This will be important later when we construct a generalization to Bayes' rule.  

Consider now that we have collected new information about $b|\Gamma$, such that the value of $b|\Gamma$ is known, i.e. $b|\Gamma \xrightarrow{\star}b'|\Gamma$ so that the posterior web of beliefs is constrained to satisfy,
\begin{equation}
q(b|\Gamma) \xrightarrow{\star} p(b|\Gamma) = \sum_{a}p(ab|\Gamma) = \delta_{bb'}.
\end{equation}
This is not enough to constrain the form of $p(a|\Gamma) \in \mathbf{P}(\tilde{\mathcal{A}})$, since many posteriors will satisfy this rule and be internally consistent with respect to the rationale.  In order to constrain the form of $p(a|\Gamma)$, one needs to invoke some new principle \cite{CatichaBook} inspired by parsimony,
\begin{principle}[Principle of Minimal Updating]
	The web of beliefs $\mathbf{Q}(\tilde{\mathcal{A}})$ should only be revised to the extent required by the new information.
\end{principle}
Like we have said, such a constraint is not necessary for an updating scheme to be consistent, but is merely imposed to satisfy some ethical desire for how we \textit{should} conduct inference \cite{CatichaBook}.  In the case of the joint posterior distribution,
\begin{equation}
p(ab|\Gamma) = p(a|\Gamma)p(b|\Gamma a) = p(a|\Gamma b)\delta_{bb'},
\end{equation}  
the conditional $p(a|\Gamma b)$ is at this point still arbitrary.  If we invoke the principle of minimal updating (PMU), then nothing requires us to change the distribution $p(a|\Gamma b)$ since the change is already invoked in $\delta_{bb'}$.  Thus we identify,
\begin{equation}
p(a|\Gamma b) = q(a|\Gamma b),
\end{equation}  
so that the web of posterior beliefs $\mathbf{P}(\tilde{\mathcal{A}})$ is given by,
\begin{equation}
p(ab'|\Gamma) = q(a|\Gamma b')\delta_{bb'}
\end{equation}
The posterior distribution $p(a|\Gamma)$ is then given by marginalizing over $b|\Gamma$,
\begin{equation}
p(a|\Gamma) = \sum_{b}q(a|\Gamma b')\delta_{bb'} = q(a|\Gamma b').
\end{equation}
This identifies the updated posterior distribution $p(a|\Gamma)$ as the \textit{prior conditional distribution} $q(a|\Gamma b')$, which is called \textit{Bayes' rule} and is often written as,
\begin{equation}
q(a|\Gamma) \xrightarrow{\star} p(a|\Gamma) = q(a|\Gamma)\frac{q(b'|\Gamma a)}{q(b'|\Gamma)}.\label{bayesrule}
\end{equation}
Typically the conditional distribution $q(b'|\Gamma a)$ is called the \textit{likelihood} and the factor,
\begin{equation}
q(b'|\Gamma) = \sum_{a}q(ab'|\Gamma),
\end{equation}
is called the \textit{evidence}.  Some may point out that (\ref{bayesrule}) is simply the same thing as Bayes' theorem in (\ref{Bayes}), however the theorem refers only to a consequence of the product rule in $\mathbf{Q}(\tilde{\mathcal{A}})$.  The action of \textit{assigning} to the degree of belief $p(a|\Gamma)$ the conditional $q(a|\Gamma b)$ is something entirely different since we changing the web $\mathbf{Q}(\tilde{\mathcal{A}}) \xrightarrow{\star}\mathbf{P}(\tilde{\mathcal{A}})$.

\subsection{Other approaches}
There exist a plethora of other approaches for constructing the rules of probability in the literature.  We have mentioned several already that have been inspired by Cox, such as Caticha \cite{CatichaBook,Ariel1,Ariel2}, Fine \cite{Fine}, Jaynes \cite{Jaynes,Jaynes1}, Knuth \cite{Knuth1,Knuth2}, Paris \cite{Paris}, Skilling \cite{KnuthSkilling}, Smith and Erickson \cite{SmithErickson}, Tribus \cite{Tribus}, Van Horn \cite{VanHorn,VanHorn2} and Vanslette \cite{VansletteThesis}.  We will briefly discuss some of the other popular approaches, such as the \textit{subjective} and \textit{frequentist} approaches and their similarities, as well as their differences, with the approach in this thesis.

\paragraph{De Finetti ---}
\epigraph{My thesis, paradoxically, and a little provocatively, but nonetheless genuinely, is simply this: PROBABILITY DOES NOT EXIST. The abandonment of superstitious beliefs about the existence of Phlogiston, the Cosmic Ether, Absolute Space and Time,..., or Fairies and Witches, was an essential step along the road to scientific thinking. Probability, too, if regarded as something endowed with some kind of objective existence, is no less a misleading misconception, an illusory attempt to exteriorize or materialize our true probabilistic beliefs.}{\textit{Bruno de Finetti}, 1970}
Bruno de Finetti championed the idea that probabilities are inherently subjective and do not \textit{exist} in any objective sense.  This idea is the overarching theme of what is called the \textit{subjective approach} to probability theory\footnote{Perhaps the main difficulty with de Finetti's approach is the confusion over the dichotomies  \textit{subjective vs. objective} and \textit{ontic vs. epistemic}.  The problem is that often times when one says ``subjective'' (personal) what they really mean is ``epistemic'' (having to do with knowledge).  While we take the view that probabilities are inherently epistemic, they can still be either subjective or objective depending on the subject matter.  For example, as we will see in Chapters six, seven and eight, in quantum mechanics probabilities play a central role and are entirely epistemic quantities, however they are not subjective.  The probabilities which appear in QM are objective in the sense that they are constrained to agree with the results of experiment.}.  De Finetti's idea was that probabilities are assignments of personal belief which represent how one would \textit{bet} on a particular outcome of some event given their current information.  In this way, probabilities are intimately tied to \textit{actions}, which is quite different from what we have developed in this chapter.  In our treatment, the probability calculus determines what one ought to \textit{believe}, but does not suggest any method for determining what one ought to \textit{do} with those beliefs.  Those ideas belong to a different subject altogether, such as \textit{decision theory} or \textit{game theory}.

De Finetti's contemporaries, Ramsey \cite{Ramsey} and Savage \cite{Savage}, developed a similar approach in which the specific scheme of betting with money was replaced with a more general concept of \textit{utility}.  This idea later culminated in utility theories developed by von Neumann and Morgenstern \cite{vonNeumannMorgenstern}.  In this scheme, one is required to specify a \textit{utility function} $u:X\rightarrow\mathbb{R}$, where $X$ is the subject matter.  Depending on the treatment, the utility function is constrained to obey a set of axioms \cite{vonNeumannMorgenstern}.  The central quantity of interest is the expected utility,
\begin{equation}
\langle u(x) \rangle = \sum_{x\in X}u(x)p(x).\label{exputil}
\end{equation}        
As was shown by von Neumann and Morgenstern \cite{vonNeumannMorgenstern}, by assuming that the function $u(x)$ obeys certain properties, such as transitivity and continuity, then a rational agent will make decisions which maximize (\ref{exputil}).

The problem with utilitarian approaches is that they single out universes of discourse which only apply to the behavior of the rational agent.  Utilitarianism is certainly useful for the stock market and even for the design of physical experiments, but it is not useful for physics as a whole.  Even so, the idea that probabilities are inherently subjective was an important step for bringing inference into the 20th century.  One use of the de Finetti approach in physics is \textit{quantum Bayesianism}, or \textit{QBism} \cite{Fuchs}.  

\paragraph{Kolmogorov ---}
Perhaps the most popular approach to probability theory is the one developed by Kolmogorov \cite{Kolmogorov} in which probabilities are \textit{measures} associated to some set $X$.  The formalism was developed as a means for assigning probabilities, or measures, to volumes in spaces such as $\mathbb{R}^n$.  A measure is a function $\mu$ from a $\sigma$-algebra over $X$ to the extended real numbers $\mathbb{R}\cup[-\infty,\infty]$.  A $\sigma$-algebra for $X$ is a subset of the power set $\Sigma \subseteq \mathscr{P}(X)$ which has the following properties,
\begin{axioms}[$\sigma$-algebra]\label{sigmaalgebra}
	Let $X$ be a set.  A collection of subsets $\Sigma \subseteq \mathscr{P}(X)$ is called a $\sigma$-algebra if,
	\begin{enumerate}
		\item The entire set $X$ and the empty set $\emptyset$ are in $\Sigma$,
		\begin{equation}
		X,\emptyset \in \Sigma.
		\end{equation}
		\item $\Sigma$ is closed under complementation,
		\begin{equation}
		\forall A \in \Sigma : X\backslash A \in \Sigma.\label{sigma2}
		\end{equation}
		\item $\Sigma$ is closed under countable unions.
		\begin{equation}
		\forall n \in \mathbb{N} : \forall A_1,\dots,A_n \in \Sigma : \bigcup_{k=1}^nA_k \in \Sigma.\label{sigma3}
		\end{equation}
	\end{enumerate}
\end{axioms}
Due to De Morgans laws axioms (\ref{sigma2}) and (\ref{sigma3}) are equivalent to,
\begin{equation}
\forall n \in \mathbb{N} : \forall A_1,\dots,A_n \in \Sigma : \bigcap_{k=1}^nA_k \in \Sigma.
\end{equation}
These axioms demonstrate that a $\sigma$-algebra is essentially an algebra of sets with the additional requirement that countable unions be in $\Sigma$.  In special cases, $\sigma$-algebras $(\Sigma,\cup,\cap,\backslash)$ are isomorphic to the Boolean algebra $(\mathcal{A},\vee,\wedge,\neg)$ whenever $\mathcal{A}$ is finite\footnote{There are some subtleties with the differences between $\sigma$-algebras and Boolean algebras.  Due to Stone's representation theorem, Boolean algebras are isomorphic to some field of sets, which may only require closure under \textit{finite} unions, i.e. given a finite sequence of subsets $A_i,\dots,A_n$ of the algebra $X$, we have $\forall n < \infty : \bigcup_{i=1}^n A_i \in X$.  Any $\sigma$-algebra however, is \textit{completed} with respect to the union, so that it is closed under a countable number of unions (axiom (\ref{sigma3})).  The reason for axiom (\ref{sigma3}) actually comes from a different direction that is related to the Banach-Tarski paradox \cite{BanachTarski}.  In order to assign a meaningful volume to $n$-dimensional spaces, Banach and Tarski showed that one of four assumptions must be true, either
	\begin{enumerate}
		\item The volume of a set could change under rotation.
		\item One must abandon the axiom of choice.
		\item The volume of the union of two disjoint sets is not necessarily equal to the sum of their volumes.
		\item There are some sets which are \textit{non-measureable}.
	\end{enumerate}  
	Obviously solutions (1) and (3) are undesirable.  The opted for solution has been number (4), for which we need to develop the notion of a \textit{measurable set}.  Axiom (\ref{sigma3}) is explicitly included to avoid the problems of the Banach-Tarski paradox.}.  If the Boolean algebra $(\mathcal{A},\wedge,\vee,\neg)$ is complete and the set $\mathcal{A}$ is countable, then the properties of $(\mathcal{A},\wedge,\vee,\neg)$ are \textit{stronger} than the axioms in \hyperref[sigmaalgebra]{1.1}.  In this sense, a $\sigma$-algebra is a sort of in-between for these two situations.

A space $X$ together with a sigma algebra $\Sigma$ is called a \textit{measurable space} $\langle X,\Sigma\rangle$.  A measure $\mu$ on $\Sigma$ is then defined as,
\begin{axioms}
	Let $\Sigma$ be a $\sigma$-algebra over $X$.  A function $\mu:\Sigma\rightarrow\mathbb{R}\cup[-\infty,\infty]$ is called a measure if,
	\begin{enumerate}
		\item The function $\mu$ is non-negative,
		\begin{equation}
		\forall A \in \Sigma : \mu(A) \geq 0.
		\end{equation}
		\item The empty set has zero measure,
		\begin{equation}
		\mu(\emptyset) = 0.
		\end{equation}
		\item For all countable collections $\{A_k\}_{k=1}^{n}$, the function $\mu$ is countably additive over disjoint subsets,
		\begin{equation}
		\forall n \in \mathbb{N} : \forall \{A_k\}_{k=1}^{n} : \forall k \neq \ell A_k\cap A_{\ell} = \emptyset : \mu\left(\bigcup_{k=1}^{n}A_k\right) = \sum_{k=1}^n\mu(A_k).
		\end{equation}
	\end{enumerate}
\end{axioms}
Finally, whenever $\mu(X) = 1$ we say that $\mu$ is a \textit{probability measure} over $\Sigma$.  While the axioms for $\sigma$-algebras and measures are reasonable, they have the problem of lacking an interpretation.  The set $X$ could be anything, such as an interval $[a,b] \subset \mathbb{R}$.  By assigning a probability measure to subsets of $[c,d] \subset [a,b]$, what exactly is the meaning of $\mu([c,d])$?  What is $\mu([c,d])$ the probability of?  By lacking an interpretation, this approach can lead to confusion.  

There is also the additional problem of not having a definition of a conditional probability a priori.  It has to be put in by hand after the fact by defining,
\begin{equation}
\mu(a|b) \stackrel{\mathrm{def}}{=}\frac{\mu(ab)}{\mu(b)}.
\end{equation}
What's worse about this situation is that the Kolmogorov approach is ill-equipped to handle an updating procedure.  Because there is no interpretation as to what the measure $\mu(a|b)$ is supposed to represent, the procedure in \hyperref[bayesrules]{Section 1.6} is not well defined.   

In our approach we fix the problem of the interpretation from the beginning.  The subject matter is not some generic set $X$, but the set of statements $\tilde{\mathcal{A}}$.  Some may argue that this is less \textit{general} than Kolmogorov because we do not allow the subject matter to be arbitrary sets $X$, but so be it.  The problem with Kolmogorov is that it is \textit{too} general and completely lacks an interpretation, which has the potential to allow fallacious reasoning.  

\paragraph{Frequentism ---}
\epigraph{In the long run we shall all be dead.}{\textit{Keynes}, 1923}
There are many ``interpretations'' of probability theory which claim that probabilities are frequencies, or \textit{expectations of chance} given some large ensemble of repeated events.  The appeal of this interpretation is that one can determine frequencies of random events by simply counting a number of positive occurrences in the ensemble.  Such an assignment of probability is independent of any agent, be they ideally rational or not, and hence one feels that they have constructed some objective notion of probability.

There are several philosophical and practical objections to this idea.  The biggest problem is that frequentism relies on there being a large number of repeated trials, or that the situation in question be repeatable in principle.  Thus, inferences about which there are no such ensembles cannot be defined in the frequentist interpretation.  A question such as, ``what is the probability that there is life on Mars?'' is not even a legitimate question, since there is only one Mars and not an ensemble.  To combat these issues a frequentist may simply agree, that such a question is nonsensical, and that probabilities are \textit{only} defined in the presence of repeatable events.  Such an admission greatly reduces the power of any inference framework, which is undesirable.  On the contrary, they may attempt to justify such questions on the basis of the possible worlds semantics so that one constructs a \textit{virtual ensemble}, but this destroys any notion of objectivity that one may have gained from the interpretation in the first place.

The frequentist interpretation also tends to be rather vague, despite its claim of being objective.  One source of vagueness is the meaning of the word \textit{random}.  The popular choice for defining random variables is to define a \textit{probability space} $P$, which is a triple $P = \{\Omega, \mathcal{F}, \mu\}$, where $\Omega$ is a set of mutually exclusive and exhaustive outcomes and $\mathcal{F} \subseteq \mathscr{P}(\Omega)$ is a collection of subsets of $\Omega$ called the \textit{event space}.  The event space $\mathcal{F}$ is a \hyperref[sigmaalgebra]{$\sigma$-algebra} generated by $\Omega$ and hence is equivalent to a Boolean algebra in special cases.  One then defines a random variable $X:(\Omega,\Sigma)\rightarrow (E,T)$ as a \textit{measureable function}\footnote{A measureable function $f:X\rightarrow Y$ between two measurable spaces $(X,\Sigma)$ and $(Y,T)$ preserves the structure of the underlying $\sigma$-algebras, i.e. the preimage of any $B\in T$ is in $\Sigma$.} from the measurable space $(\Omega,\Sigma)$ to some other measureable space $(E,T)$.  Then, the measure associated to whether the function $X$ assigns values to some subset $S \subseteq E$ in the target is defined by,
\begin{equation}
\mu\left(X \in S\right) \stackrel{\mathrm{def}}{=} \mu\left(\left\{\omega \in \Omega\,\middle|\, X(\omega) \in S\right\}\right).
\end{equation} 
In this sense, $X$ is just a function which groups the elements of $\Omega$ into particular subsets.  An analogous construction in $(\tilde{\mathcal{A}},\wedge,\vee,\neg)$ would be a homomorphism $\mathcal{X}:\tilde{\mathcal{A}}\rightarrow\tilde{\mathcal{B}}$ to another Boolean algebra $(\tilde{\mathcal{B}},\wedge,\vee,\neg)$ such that for every element $b \in \tilde{\mathcal{B}}$, the preimage $\mathcal{X}^{-1}(b) \in \tilde{\mathcal{A}}$.  While all of this is perfectly well defined, it has already obscured the subject matter enough to make the effort seem not worthwhile.    

One of the biggest proponents of the frequentist view was Fisher \cite{Fisher}, who originally started out as a Bayesian but was unable to deal with the problem of the assignment of the prior \cite{Jaynes2}.  Fisher then became obsessed with trying to remove all subjectivity from inductive inference, which is practically impossible.  No matter what one does, there will always be some subjective component left over in any inference where one has incomplete information.  Savage refers to Fisher's attempt as ``a bold attempt to make the Bayesian omelet without breaking the Bayesian egg'' \cite{Savage1}.

\paragraph{Quantum probabilities? ---}
There have been attempts at arguing that classical inductive inference is incomplete, since quantum theory seems to demonstrate a violation of the sum and product rules.  As was shown in \cite{CatichaBook}, these arguments are simply a misunderstanding of the problem.  We will demonstrate the results of the argument here.

The prototypical example for why classical probability theory is incompatible with quantum mechanics will make use of the superposition principle and the role of interference effects.  Consider the gedanken two slit experiment in which a particle $q$ is prepared at a source $s$ with some velocity in the direction of two slits $a$ and $b$.  The particle then has some probability of going through either slit $a$ or slit $b$ and then landing somewhere on the screen on the other side at location $x$, 
\begin{figure}[H]
	\centering
	\includegraphics[width=.6\linewidth]{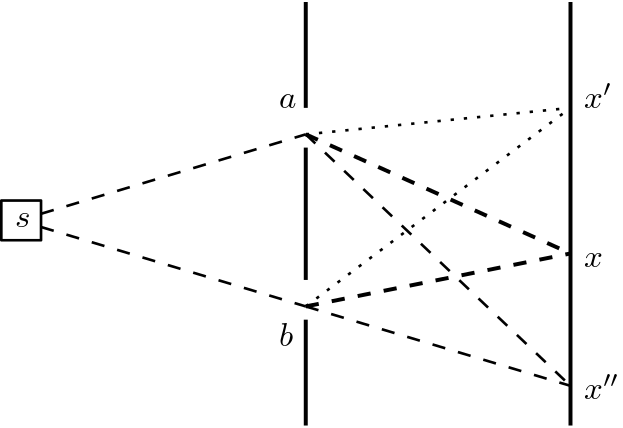}
	\caption{A two slit experiment where a particle is prepared at source $s$ and then passes through either slit $a$ or slit $b$ and arrives at some location on the screen.}
\end{figure}
In quantum mechanics, the states that describe the particle going through slit $a$ and $b$ are represented by the wave functions, $\psi_a$ and $\psi_b$ respectively.  The state that describes the particle arriving at $x$ is then given by the superposition $\psi_{ab} \propto \psi_a + \psi_b$, where the probability of detection at $x$ is,
\begin{equation}
p_{ab}(x) = |\psi_{ab}|^2 \propto |\psi_a + \psi_b|^2 = |\psi_a|^2 + |\psi_b|^2 + 2\mathrm{Re}[\psi_a^*\psi_b].
\end{equation}
The term $|\psi_a|^2$ corresponds to the probability of detection at $x$ when the particle passes through slit $a$ \textit{and} when slit $b$ is closed.  Likewise $|\psi_b|^2$ is the probability that the particle passes through slit $b$ when slit $a$ is closed.  Since the particle must pass through either slit $a$ or slit $b$, it looks like quantum mechanics violates the sum rule since for these mutually exclusive events we find,
\begin{equation}
p_{ab}(x) \neq |\psi_a|^2 + |\psi_b|^2.\label{wrongq}
\end{equation}
The problem, which was originally suggested by Feynman \cite{Feynman}, is actually more subtle.  A counter argument was given by Koopman \cite{Koopman}, as well as in Ballentine \cite{Ballentine} and Caticha \cite{CatichaBook}.  Following Caticha, consider first that we enumerate the various statements in the problem as follows,
\begin{itemize}
	\item $s = $``particle $q$ is generated at the source $s$.
	\item $a = $``slit $a$ is open'' and $\neg a = $``slit $a$ is closed.''
	\item $b = $``slit $b$ is open'' and $\neg b = $``slit $b$ is closed.''
	\item $\alpha = $``particle $q$ goes through slit $a$.''
	\item $\beta = $``particle $q$ goes through slit $b$.''
	\item $x = $``particle $q$ is detected at $x$.''
\end{itemize}
The statement of interest is then, $x | s a b$, or,
\begin{quotation}``The particle $q$ is detected at $x$, given that it was generated at source $s$ and slit $a$ and slit $b$ are both open.''
\end{quotation}
A conjunction of interest will be $[\alpha \vee \beta] x|s a  b$, whose probability can be determined as,
\begin{equation}
p([\alpha \vee \beta] x|s a  b) = p(\alpha  x|s a  b) + p(\beta  x|s a  b) - p(\alpha  \beta  x|s a  b).
\end{equation}
Since the particle cannot go through both slits at the same time, the conjunction vanishes,
\begin{equation}
p(\alpha\beta x|s a  b) = 0.
\end{equation}
Thus we have,
\begin{equation}
p([\alpha \vee \beta] x|s a  b) = p(\alpha  x|s a  b) + p(\beta  x|s a  b).
\end{equation}
Likewise, we can also write $p([\alpha\vee\beta] x|s a  b)$ using the product rule,
\begin{equation}
p([\alpha\vee\beta] x|s a  b) = p(x|s a  b)p([\alpha \vee \beta]|s  a  b  x).
\end{equation}
However, since the particle must go through either slit $a$ or $b$ if it was detected at $x$, then the second term $p([\alpha \vee \beta]|s  a  b  x) = $.  Thus we have,
\begin{equation}
p(x|s a  b) = p(\alpha  x|s a  b) + p(\beta  x|s a  b).\label{rightq}
\end{equation}
Finally, the corresponding equation for (\ref{wrongq}) is,
\begin{equation}
p(x|s a  b) \neq p(\alpha x|s a  \neg b) + p(\beta x|s \neg a  b).
\end{equation}
But, as we have shown in (\ref{rightq}), this discrepancy is perfectly consistent with the results of probability theory, hence there is no contradiction.  However, one may be tempted to use the product rule in (\ref{rightq}) so that we have,
\begin{equation}
p(x|s a  b) = p(\alpha|s a  b)p(x|s a  b  \alpha) + p(\beta|s a  b)p(s a  b  \beta).\label{rightq2}
\end{equation}  
Our classical intuition tells us that, if the particle passes through slit $a$, then the fact that slit $b$ is open makes no difference and hence one can make the substitution,
\begin{equation}
p(x|s a  b  \alpha) \rightarrow p(x|s a  \neg b  \alpha),
\end{equation}
and likewise for the second term in (\ref{rightq2}).  Thus, one would be led to the inference,
\begin{align}
p(x|s a  b) &= p(\alpha|s a  \neg b)p(x|s a  \neg b  \alpha) + p(\beta|s \neg a  b)p(s \neg a  b  \beta)\nonumber\\
&= p(\alpha x|s a  \neg b) + p(\beta x|s \neg a  b),\label{wrong2}
\end{align}
which is incorrect.  This demonstrates that our classical intuition about whether distant objects can influence quantum systems is wrong, quantum mechanics can be non-local.

\subsection{Summary}

In this chapter we constructed a theory of inference by designing the rules for probabilities as degrees of rational belief.  By constraining representations (\hyperref[section321]{Section 3.2.1}) of elements of the Boolean algebra over extended proposition space (\hyperref[section26]{Section 2.6}) to adhere to certain criteria, we arrived at the sum and product rules for probabilities.  We introduced the \textit{principle of minimal updating} (\hyperref[section36]{Section 3.6}) in order to derive Bayes' rule.  In the next chapter we will generalize Bayes' rule to situations in which one wishes to update their beliefs in the presence of constraints.

\section{Entropic inference}\label{chapter4}
\epigraph{The theory of probability, if it is to be useful at all, demands a theory for updating probabilities.}{\textit{Ariel Caticha}, 2020}
At the end of the previous chapter, we had constructed an inductive inference for dealing with states of partial knowledge and a calculus for manipulating them in a way consistent with some rationale, $\mathcal{R}$.  This answered the first of the two questions posed in the Introduction, the second of which -- \textit{how does one update their beliefs in the presence of new information} -- we will address in this chapter.  As the quotation above suggets, any inference framework would be remiss if it didn't also contain a method for updating ones beliefs whenever new information becomes available.  

The rationale for inductive inference was developed to express several desirable properties including \textit{Universal Applicability}, \textit{Consistency}, \textit{Practicality} and \textit{Parsimony}.  We proceed in much the same way by invoking a similar rationale, $\mathcal{R}_{S}$, for the purposes of updating.  We have already seen some instances of this with the \textit{Principle of Minimal Updating} which was imposed in the derivation of Bayes' rule at the end of chapter two.  As systematic as Bayes' rule is, it only provides a way for updating beliefs when new information comes in the form of data.  This raises the central question of this chapter: How should we update our beliefs when new information does not come in the form of data?  Before we can address that question, we must first clarify what is meant after all by \textit{information}.

The tool for updating that we will construct is commonly referred to as the \textit{Maximum Entropy method} (ME) \cite{CatichaBook} which has its roots in \textit{MaxEnt} \cite{Jaynes}.  Jaynes developed MaxEnt for the purposes of assigning priors from arbitrary constraints \cite{Jaynes1,Jaynes2}.  While this helps to identify constraints as being a \textit{type of information} different from data, MaxEnt is ultimately not a tool for updating, hence it cannot take into account arbitrary priors \cite{Jaynes3}.  On the other hand, Bayes' rule allows for arbitrary priors but does not allow for arbitrary constraints.  The goal of constructing the ME method will be to unify both Bayes' rule and MaxEnt into a complete method for inductive inference -- which we call \textit{entropic inference} -- so that one can update arbitrary priors on the basis of arbitrary constraints.  This is done by generalizing the PMU.      

In most circles, the words \textit{entropy} and \textit{information} are invoked in the context of \textit{Information Theory} \cite{CoverThomas_Book}.  Originally developed by Shannon \cite{Shannon_Book}, Information Theory is mainly concerned with sending bits of information across communication channels and analysing the degrading effects of noise.  Through the use of various axioms, Shannon constructed his entropy as a means of quantifying an \textit{amount} of information that is missing from a probability distribution.  While certainly useful in the context of communications theory, the Shannon entropy -- which was later adopted in the context of MaxEnt -- can only refer to discrete variables; for continuous variables it is undefined.  One can get around this problem by introducing an invariant measure into the definition so that the entropy takes the form of the \textit{Kullback-Leibler divergence}\footnote{The KL divergence was actually first written down by Gibbs \cite{Gibbs}.} (KL).  Despite this, the Shannon entropy was constructed using axioms which have nothing to do with updating beliefs.  As pointed out by Caticha \cite{CatichaBook}, this presents a serious problem since one could have introduced different axioms which would lead to different entropies (such as the Renyi \cite{Renyi} or Tsallis \cite{Tsallis1} types).  Given all of the possibilities, which functional should we choose?

The revolution in approaches to using entropy -- and the answer to the previous question -- begins with a paper by Shore and Johnson \cite{ShoreJohnson} in which they derive the functional form of the relative entropy for the purposes of updating.  Their approach was to define a set of axioms that constrain the updating procedure, which they then use to determine the functional form of the entropy via \textit{eliminative induction}.  This approach of axiomitizing the updating process, rather than the measure itself, gives incredible clarity to the subject of inference -- while previous justifications for using entropy relied on discussing the properties of the functional, Shore and Johnson instead construct the entropy as the tool for incorporating new information.  This point is made clear in the introduction to their paper,

\begin{quotation}
	``Most [approaches to entropy] are based on a formal description of what is required of an information measure; none are based on a formal description of what is required of a method for taking information into account.'' (Shore and Johnson, 1980)
\end{quotation}

While Shore and Johnson's approach provided a systematic way of constructing an entropy functional for the purposes of updating, they had not addressed the question of why one should use a variational principle in the first place, i.e. \textit{why the max in MaxEnt}.  The explanation for this came in Skilling's paper ``The Axioms of Maximum Entropy'' \cite{Skilling} in which he makes explicit the desire for a tool which \textit{ranks} posterior distributions with respect to some \textit{order of preference}.  In \cite{Skilling} he states\footnote{In his paper \cite{Skilling}, Skilling was mainly concerned with reconstructing images with incomplete data.},
\begin{quotation}
	``We aim to provide an image [posterior distribution] which is the `best' according to an agreed criterion. This involves setting up a ranking procedure which determines which of two images is `better'. To avoid circularity, and to ensure that there is always some image which is not `bettered' by any other, we impose the transitivity requirement,
	(f better than g) and (g better than h) $\Rightarrow$ (f better than h).'' (Skilling, 1988)
\end{quotation}  
Once we specify that what we desire is a tool for ranking probabilities -- which in turn is the relative entropy -- then maximizing it to find the best posterior is completely natural, intuitive and no longer requires any ad hoc justification.  

Another set of advancements for entropic methods of inference comes from Caticha \cite{CatichaBook}.  His contributions to the subject include answering questions such as, \textit{what constitutes information} and \textit{are entropic and Bayesian methods compatible}, both of which we will discuss in this chapter.  Despite showing that Bayes' rule is a special case of the ME method, his work on the subject \cite{Caticha_Giffin_1,Caticha_Giffin_2} is perhaps still largely unappreciated.  Before he and his student A. Giffin published their results, it was largely believed that Bayesian and entropic methods referred to two different and unconnected (or even incompatible) ideas. 

Caticha's original paper \cite{Caticha_Entropy_3} introduced the \textit{Principle of Minimal Updating} (PMU) which we had used for constraining Bayes' rule in the previous chapter.  The idea, which is central to the Bayesian philosophy, allowed Caticha to design the relative entropy using only three design criteria, whereas previous approaches (such as Shore and Johnson, Skilling) required at least four.  As we already mentioned, the design derivation in this thesis follows closely the work of Caticha by adopting the PMU and the same design criteria found in \cite{CatichaBook,Caticha_Entropy_3,Caticha_Entropy_4,Caticha_Entropy_1,Caticha_Entropy_2}.

Apart from answering the question of how we update from a prior to a posterior, we also wish to construct a tool for the purposes of \textit{ranking} joint distributions with respect to their \textit{correlations}.  Correlations are a central component to any study involving inference, so the desire for quantitative measures is felt universally.  Such measures are plentiful in information theory, however here our goal is to design one according to first principles.  A design derivation for correlation functionals had not been attempted before K. Vanslette and myself published our results on the matter in the journal Entropy \cite{CarraraVanslette}.  We will follow the design derivation of that paper in this thesis.

As we will see, the \textit{type} of correlation functional $I[\mathbf{P}(\tilde{\mathcal{A}})]$ one arrives at depends on a choice of the \textit{splits} within the proposition space $\tilde{\mathcal{A}}$, and thus the functional we seek is $I[\mathbf{P}(\tilde{\mathcal{A}})]\rightarrow I[\mathbf{P}(\tilde{\mathcal{A}}),\tilde{\mathcal{A}}]$.  For example, if one has a proposition space $\tilde{\mathcal{A}} = \tilde{\mathcal{A}}_1\times\dots\times\tilde{\mathcal{A}}_N$, consisting of $N$ variables, then one must specify \textit{which} correlations the functional $I[\mathbf{P}(\tilde{\mathcal{A}}),\tilde{\mathcal{A}}]$ should quantify.  Do we wish to quantify how the variable $\tilde{\mathcal{A}}_1$ is correlated with the other $N-1$ variables?  Or do we want to study the correlations \textit{between} all of the variables?  In our design derivation, each of these questions represent the extremal cases of the family of quantifiers $I[\mathbf{P}(\tilde{\mathcal{A}}),\tilde{\mathcal{A}}]$, the former being a \textit{bi-partite} correlation (or mutual information) functional and the latter being a \textit{total} correlation functional.

In the main design derivation we will focus on the the case of \textit{total correlation} which is designed to quantify the correlations between every variable subspace $\tilde{\mathcal{A}}_i$ in a set of variables $\tilde{\mathcal{A}} = \tilde{\mathcal{A}}_1\times\dots\times\tilde{\mathcal{A}}_N$.  We suggest a set of design criteria (DC) for the purpose of designing such a tool. These DC are guided by the \textit{Principle of Constant Correlations} (PCC), which states that ``the amount of correlations in $\mathbf{P}(\tilde{\mathcal{A}})$ should not change unless required by the transformation, $(\mathbf{P}(\tilde{\mathcal{A}}),\tilde{\mathcal{A}})\stackrel{*}{\rightarrow}(\mathbf{P}(\tilde{\mathcal{A}})',\tilde{\mathcal{A}}')$.'' This implies our design derivation requires us to study equivalence classes of $[\mathbf{P}(\tilde{\mathcal{A}})]$ within statistical manifolds $\mathbf{\Delta}$ under the various transformations of distributions $\mathbf{P}(\tilde{\mathcal{A}})$ that are typically performed in inference tasks. We will find, according to our design criteria, that the global quantifier of correlations we desire in this special case is equivalent to the total correlation \cite{Watanabe}.

Once one arrives at the TC as the solution to the design problem, one can then form special cases such as the \textit{mutual information} \cite{CoverThomas_Book} or, as we will call them, any \textit{$n$-partite information} NPI which measures the correlations shared between generic $n$-partitions of the proposition space.  The NPI can also be derived using the same principles as the TC except with one modification, as we will discuss in \hyperref[npartite]{Section V}.

The special case of NPI when $n=2$ is the \textit{bipartite} (or mutual) information, which quantifies the \textit{amount} of correlations present between two subsets of some proposition space $\tilde{\mathcal{A}}$.  Mutual information (MI) as a measure of correlation has a long history, beginning with Shannon's seminal work on communication theory \cite{Shannon_Book} in which he first defines it.  While Shannon provided arguments for the functional form of his entropy \cite{Shannon_Book}, he did not provide an axiomatization of MI.

\paragraph{Information?}
\epigraph{Information is what information does, it changes your mind.}{\textit{Ariel Caticha}, 2020}
Information can mean many different things depending on the context in which it is used.  For example, the subject of \textit{Information Theory} considers information to be something that is transmitted from a \textit{source} to some \textit{destination} through a \textit{channel} and this information can be measured, i.e. we can quantify an \textit{amount of information}.  Information theory however has little to do with the beliefs of rational agents, much less with updating them.  Hence, this view of information is unsatisfactory for our purposes.   

We also often hear that information is ``physical,'' that physical systems ``carry'' information.  This view was likely inspired by information theory in the interests of understanding the limitations of computation.  Questions such as, \textit{what is the minimum amount of energy required to store a single bit of information}, or, \textit{how much energy is expended when we throw away information}, suggest that a piece of information must, in the very least, be represented by something physical.  This idea was pushed further by Landauer who first argued it in his 1991 article ``Information is physical'' \cite{Landauer1} that has inspired many others in the field.   He later wrote \cite{Landauer2},	
\begin{quotation}
	``Information is not an abstract entity but exists only through a physical representation, thus tying it to all the restrictions and possibilities of our real physical universe.'' (Landauer, 1999) 
\end{quotation}
Essentially, all information must somehow be coupled to the physical world in order to be transmitted from one agent to another, whether it be words inscribed on rock or the electromagnetic waves sent through an ethernet cable.  Therefore, the distinction between some deeper understanding of what constitutes information and its physical representation become unimportant -- one simply adopts the pragmatic view that information is physical itself so that no further abstraction is needed.

John Wheeler suggested a similar idea in his paper ``Information, Physics, Quantum: The Search for Links,'' in which he coined the term ``it from bit'' \cite{Wheeler}.  The idea of ``it from bit'' is that every constituent of existence can be described by a series of yes/no questions, such as \textit{is the position of the particle `x'}.  By answering either ``yes'' or ``no'' to such a question, one is associating a single classical bit of information to the description of a physical object.  Consequently, every physical entity is defined by a set of unique answers of ``yes'' and ``no.''  While this topic is certainly interesting, the idea that information is physical turns out to be too restrictive for the purposes of inference.

Here we adopt the pragmatic view that information is an epistemic quantity.  In this sense, information is what causes us to change our beliefs -- it is what we seek when we ask a question.  Indeed, since the inductive inference framework was developed on the basis of reasoning with \textit{incomplete} information, it is only natural to seek more complete information, with the goal of improving our ability to reason.  Following Caticha \cite{CatichaBook}, we state the defining characteristics of information, 
\begin{quotation}
	``Information is what affects and therefore constraints rational beliefs.  Information is what forces a change in beliefs.'' (Caticha, 2020)	
\end{quotation}	
In this epistemic definition, there is no notion of an \textit{amount} of information that can be measured.  New information is simply a set of constraints on our beliefs and the constraints specify what we should believe given that we were ideally rational.  

There is an instructive analogy of our definition of information with that of physics, which was presented in \cite{CatichaBook}.  In classical mechanics, a system is defined by its \textit{quantity of motion}, its momentum.  The systems state of motion changes whenever it is acted upon by forces -- an initial state $\vec{p}_i$ changes to some final state $\vec{p}_f$ under the influence of some external force $\vec{F}_{\mathrm{ext.}}$.  In inference a state is defined by \textit{degrees of belief}, a set of probabilities $\mathbf{P}(\tilde{\mathcal{A}})$.  Information then, is the \textit{force} which causes a change from some initial set of beliefs, $\mathbf{Q}(\tilde{\mathcal{A}})$, to a final set of beliefs, $\mathbf{P}(\tilde{\mathcal{A}})$.  In this sense, updating beliefs is a form of dynamics.  A generalization to this dynamics is discussed in chapter four, and the application of its generalization to physics as an \textit{Entropic Dynamics} is discussed in part two of this thesis.

\paragraph{Correlations? ---}
The term correlation may be defined colloquially as being a relation between two or more ``things". 
While we have a sense of what correlations are, how do we quantify this notion more precisely?  If correlations have to do with ``things'' in the real world, are correlations themselves ``real?''  Can correlations be ``physical?''  In the context of inference, correlations are broadly defined as being statistical relationships between propositions.  In this thesis we adopt the view that whatever correlations may be, their effect is to influence our beliefs about the natural world.  In this sense, correlations are just a particular type of information.  More specifically,\\

\textit{Correlations are the pieces of information that constrain conditional beliefs}.\\

The design for correlation quantifiers can be motivated by first examining the more intuitive extreme cases.  Perhaps the most intuitive situation occurs when two systems are actually independent -- the joint system contains no correlations whatsoever.  We would expect in this case that our joint beliefs would factor, $p(a,b) = p(a)p(b)$, since knowing either the value of $a$ or $b$ tells us nothing about the other.  The other extreme happens when $a$ and $b$ are completely determined by each other -- the joint system is maximally correlated.  In this case, we would expect our joint beliefs to behave as $p(a,b) = p(a)\delta(b - f_a(a)) = p(b)p(a - f_b(b))$ where the delta functions signify that $b$ is some function of $a$, $b = f_a(a)$, and that this function is invertible, $f_a = f_b^{-1}$.  While we understand both of these edge cases, it is not entirely clear how one quantifies anything that happens in the middle.  If we are given two joint distributions $p_1(a,b) = p_1(a)p_1(b|a)$ and $p_2(a,b) = p_2(a)p_2(b|a)$, how can we decide if one represents a more correlated state than the other?

Quantifying correlations as part of entropic inference is fundamentally different from the procedure of updating.  While they will both involve objects which take the form of relative entropies, the correlation quantifiers are not designed as tools for updating.  They serve the weaker function of simply providing a \textit{ranking} of joint distributions with respect to their correlations.  There is no variational principle which guides the construction of the correlation quantifiers, and hence more details are needed in order to construct them.

In Section 2 we will outline the basic principles for constructing a tool for updating, as well as a tool for ranking correlations.  In Section 3 we will enumerate various design criteria inspired by these principles and then go on to demonstrate the solution to the updating problem.  We will also breifly discuss the special case of Bayes' rule from the ME method.  Section 4 will be similar to Section 3 in that we will specify design criteria for the purposes of deriving a set of correlation quantifiers.  We conclude with Section 5 which discusses various results surrounding \textit{sufficiency} that were developed in \cite{CarraraVanslette}.

\subsection{The design of entropic inference}
Now that we have settled the question as to what constitutes information, we need to determine a systematic way of incorporating that new information into our beliefs -- we must answer the question of how does one update their beliefs in the presence of new information?  In the spirit of chapter two, we will define a set of reasonable criteria that our entropic inference must adhere to.  This defines a rationale for entropic inference, which we label $\mathcal{R}_S$.  

Since entropic inference is concerened with updating beliefs, we can take some suggestions for criteria of rationality from the rationale of inductive inference $\mathcal{R}$.  For one, since the design of inductive inference is agnostic to the underlying subject matter for which it is used, so too should be the method of updating.  The updating scheme must apply to every situation, otherwise how would one know which updating scheme to use?

\begin{principle}[Universality]\label{universality}
	The method for updating beliefs should be independent of the subject matter, it must have universal applicability.
\end{principle}

The criterion of universal applicability has profound implications.  If our method was required to behave differently if we were updating beliefs about the stock market as opposed to physics, then we would have broken the rationale set forth in the inductive inference framework, which is unacceptable.  This means that our beliefs in the inner workings of the natural world must be subject to the same dynamical rules as our beliefs in anything else.  

Influenced by the discussion of section ... in chapter two, we adopt a principle of parsimony.  Specifically, we adopt the view that past information is valuable and that one should not throw away what was once learned unless the new information renders it obsolete.  This view is expressed in a slightly altered version of the principle of minimal updating that was given in chapter two.  

\begin{principle}[Parsimony: Minimal Updating]
	Beliefs should be updated only to the extent required by the new information.
\end{principle}

The principle of minimal updating (PMU) automatically imposes a level of objectivity into the inference framework.  For example, consider the special case for when there is no new information available.  The PMU suggests that any ideally rational agent (Ira) should not change their beliefs.  To do the opposite would violate common sense -- rational agents should not change their minds arbitrarily.  Another objective component of the PMU derives from the same situation.  Because we are stipulating when an Ira is \textit{not} supposed to update their beliefs, the act of not updating is unique, i.e. there are many ways to change something, but only one way to keep it the same.  

Like with the rationale for inductive inference, we impose a principle of consistency.  This principle is inspired by one of the central aspects of scientific inquiry, namely independence.  A basic requirement for any scientific model to be successful is that all relevant variables are taken into account and that whatever was left out, the rest of the universe, does not matter.  In other words, in order to do science we must be able to understand parts of the universe without understanding its entirety.  The central idea is that we can focus on a particular system of interest while neglecting everything else because the two systems are \textit{statistically independent}, i.e. they are uncorrelated.  Following Caticha \cite{Caticha_towards}, we adopt the requirement that,
\begin{quotation}
	``Whenever two systems are a priori believed to
	be independent and we receive information about one it should not matter if
	the other is included in the analysis or not.'' (Caticha, 2014)
\end{quotation}
This idea can be stated in the form of a consistency requirement.
\begin{principle}[Consistency: Independence]\label{consistency}
	Whenever systems are known to be independent, it should not matter whether the analysis treats them jointly or separately.
\end{principle}
In other words, if there exist multiple ways to conduct the same inference, the results better agree.  This sentiment was declared in Shore and Johnson \cite{ShoreJohnson} whose axioms were based on a fundamental principle of consistency,

\begin{quotation}
	``Since the maximum entropy principle is asserted as a general method of inductive inference, it is reasonable to require that different ways of using it to take the same information into account should lead to consistent results.'' (Shore and Johnson, 1980)
\end{quotation} 

These three principles, \textit{Universality}, \textit{Parismony} and \textit{Consistency}, form the basis of our rationale $\mathcal{R}_S$ which we will use to constrain our design of entropic inference.  

\paragraph{Concerning correlations ---} To design our correlation quantifiers, we must also decide on a set of criteria in the form of a rationale, $\mathcal{R}_I$.  Certainly we would expect any tool for quantifying correlations to be of universal applicability, so that the \textit{Universality} principle (\ref{universality}) is also included in $\mathcal{R}_I$.  It is also reasonable to expect that the \textit{Consistency} principle be included in $\mathcal{R}_I$ since properties of independence play a crucial role in the study of correlations.  

Since correlation quantifiers are not tools for updating, their design does not warrant the PMU.  Instead of focusing on what happens when we update our beliefs, we can consider a broader class of transformations on probability distributions that was discussed in \cite{CarraraVanslette}.  These include \textit{coordinate transformations}, \textit{marginalization}, \textit{products} and entropic updating as a special case.  In certain special cases of these transformations, we expect correlations to remain unchanged, and hence we impose that they are not.  This notion can be summarized in the following principle which is similar to the PMU,  

\begin{principle}[Constant Correlations]
	The amount of correlations present in a state of knowledge should not change unless required by a particular transformation.
\end{principle}

Much like the PMU, the \textit{Principle of Constant Correlations} (PCC) imposes a level of objectivity in the design of correlation quantifiers since it instructs us when \textit{not} to make changes to our beliefs about correlations.  In one sense, it is more general than the PMU since it concerns transformations which are of a more general form than entropic updating.  On the other hand, correlations occupy only a subset of the kind of information that is present in any set of beliefs and in this way it refers to a special case.  

In conclusion, the rationale $\mathcal{R}_I$ is identical to $\mathcal{R}_S$ except for swapping the PMU with the PCC.  This is expected, since each tool is designed with a different purpose in mind.  

\paragraph{The universe of discourse ---}
The arena of any inference task consists of two ingredients, the first of which is the subject matter, or what is often called \textit{the universe of discourse}.  This refers to the actual propositions that one is interested in making inferences about.  Propositions tend to come in two classes, either \textit{discrete} or \textit{continuous}.  Discrete proposition spaces will be denoted by caligraphic uppercase latin letters, $\mathcal{X}$, and the individual propositions will be lowercase latin letters $x_i \in \mathcal{X}$ indexed by some variable $i = \{1,\dots,|\mathcal{X}|\}$, where $|\mathcal{X}|$ is the number of distinct propositions in $\mathcal{X}$.  

From this chapter onward we will generally be interested in universes of discourse which are represented by continuous spaces $\mathbf{X}$, such as the position of a particle in $\mathbb{R}^3$.  Thus, the spaces of statements $\tilde{\mathcal{A}}$ as well as the web of beliefs $\mathbf{P}(\tilde{\mathcal{A}})$ become uncountable sets.  In accordance with the notation of Chapter 2 Section ..., we will denote probability distributions over $\mathbf{X}$ as $p(x)$ and priors as $q(x)$.  For generic distributions we will use the Greek $\rho$ for posteriors and $\varphi$ for priors.

The second ingredient that one needs to define for general inference tasks is the space of models, or the space of probability distributions which one wishes to assign to the underlying proposition space.  These spaces can often be given the structure of a manifold, which in the literature is called a \textit{statistical manifold} \cite{CatichaBook}\footnote{We will discuss in more detail the structure of statistical manifolds in the next chapter, but for now it suffices to think of $\sm$ as simply a set of distributions.}.A statistical manifold $\sm$, is a manifold in which each point $p(x) \in \sm$ is an entire probability distribution, i.e. $\sm$ is a space of maps from subsets of $\mathbf{X}$ to the interval $[0,1]$, $p(x):\mathcal{P}(\mathbf{X})\rightarrow[0,1]$.  The notation $\mathcal{P}(\mathbf{X})$ denotes the \textit{power set} of $\mathbf{X}$, which is the set of all subsets of $\mathbf{X}$, and has cadinality equal to $|\mathcal{P}(\mathbf{X})| = 2^{|\mathbf{X}|}$.

The \textit{correlations} present in any distribution $p(x)$ necessarily depend on the conditional relationships between various propositions.  For instance, consider the binary case of just two proposition spaces $\mathbf{X}$ and $\mathbf{Y}$, so that the joint distribution factors,
\begin{equation}
p(x,y) = p(x)p(y|x) = p(y)p(x|y).\label{firstjoint}
\end{equation}
The correlations present in $p(x,y)$ will necessarily depend on the form of $p(x|y)$ and $p(y|x)$ since the conditional relationships tell us how one variable is statistically dependent on the other.  As we will see, the correlations defined in the above eq. are quantified by the \textit{mutual information}.  For situations of many variables however, the global correlations are defined by the \textit{total correlation}, which we will design first.  All other measures which break up the joint space into conditional distributions (including (\ref{firstjoint})) are special cases of the total correlation.  

\paragraph{A tool for updating ---}
We wish to design a tool for the purposes of updating an entire probability distribution $\varphi\xrightarrow{\star} \rho$ when new information comes in the form of constraints.  This will be done through the process of \textit{eliminative induction}, which was discussed in the Introduction of Part I.  The goal is to construct a tool that will allow us to systematically search for a preferred posterior distribution.  The procedure, which was first proposed by Skilling \cite{Skilling}, is quite simple: first, rank the possible posterior distributions according to some \textit{order of preference}, and then pick the one which ranks highest.  

While simple, the design specification suggested by Skilling is highly constraining.  For one, it implies that our preferences in posterior distributions should be transitive, i.e. if distribution $\rho_1$ is preferred to distribution $\rho_2$ and distribution $\rho_2$ is preferred to $\rho_3$ then $\rho_1$ must also be preferred to $\rho_3$.  Similar to the construction of probabilities as degrees of belief, we implement a ranking by assigning a real number to each posterior distribution $S[p]$ in such a way that if $\rho_1$ is preferred to $\rho_2$ then $S[\rho_1] > S[\rho_2]$.

As pointed out by Caticha \cite{CatichaBook,Caticha_Entropy_3}, the choice to rank distributions implies that the updating scheme will take the form of a variational principle.  This principle, which we call the \textit{Method of Maximum Entropy} (ME), allows us to determine the preferred posterior distribution $\rho$, by maximizing the entropy functional $S[\rho]$.  Since the ranking of posteriors must be done with respect to a prior $\varphi$, the desired entropy functional must depend also on the prior, $S[\rho] \rightarrow S[\rho,\varphi]$.  In this sense, the entropy functional ranks the posterior $\rho$ \textit{relative} to the prior $\varphi$, and hence $S[\rho,\varphi]$ is typically called the \textit{relative entropy}.
\begin{define}
	Let $\varphi$ be a web of prior beliefs and $\rho$ be a posterior web.  The order of preference given to $\rho$ is quantified by the real number,
	\begin{align}
	S:\Delta\times\Delta &\rightarrow \mathbb{R}\nonumber\\
	(\rho,\varphi) &\mapsto S[\rho,\varphi],
	\end{align}
	which is called the relative entropy.
\end{define}

To find the desired entropy functional, we use the process of eliminative induction by specifying certain \textit{design criteria}.  These design criteria (DC) specify the behavior of the ME method in certain special cases in which we know the answer.  Amazingly, by specifying a small set of DC one can arrive at a general functional form for the entropy.  This procedure has been exercized many times in the past \cite{ShoreJohnson,Skilling,Jaynes,CatichaBook,VansletteThesis} and each derivation has become increasingly simpler.  Shore and Johnson \cite{ShoreJohnson} and Skilling \cite{Skilling} each give four design critera.  Caticha \cite{CatichaBook} reduced the required DC to three, while Vanslette \cite{VansletteThesis,Vanslette} used only two\footnote{For a review on some of the various approaches to designing entropy functionals see Csiszar \cite{Csiszar}.}.  In the following sections we will mostly follow the procedures of Caticha and Vanslette, while also pointing out the more historical versions.

\paragraph{A tool for ranking correlations ---}
In order to better motivate the construction of correlation functionals, it will be instructive to consider what happens when we change coordinates through some kind of generic coordinate transformation.  A coordinate transformation $f:\mathbf{X}\rightarrow \mathbf{X}'$, is a special type of transformation of the proposition space $\mathbf{X}$ that respects certain properties.  It is essentially a continuous version of a reparameterization\footnote{A reparameterization is an isomorphism between discrete proposition spaces, $g:\mathcal{X}\rightarrow\mathcal{Y}$ which identifies for each proposition $x_i \in \mathcal{X}$, a unique proposition $y_i \in \mathcal{Y}$ so that the map $g$ is a bijection.}.  For one, each proposition $x \in \mathbf{X}$ must be identified with one and only one proposition $x'\in \mathbf{X}'$ and vice versa.  This means that coordinate transformations must be bijections on proposition space -- i.e. an isomorphism between sets.  The reason for this is simply by design, we would like to study the transformations that leave the proposition space invariant.  A general transformation of this type on $\sm$ which takes $\mathbf{X}$ to $\mathbf{X}' = f(\mathbf{X})$, is met with the following transformation of the densities,
\begin{equation}
p(x) \rightarrow p'(x') \quad \mathrm{where} \quad p(x)dx= p'(x')dx'.\label{i}
\end{equation}
Like we already mentioned, the coordinate transforming function $f:\mathbf{X}\rightarrow\mathbf{X}'$ must be a bijection in order for (\ref{i}) to hold, i.e. the map $f^{-1}:\mathbf{X}'\rightarrow\mathbf{X}$ is such that $f\circ f^{-1} = \mathrm{id}_{\mathbf{X}'}$ and $f^{-1}\circ f = \mathrm{id}_{\mathbf{X}}$.  While the densities $p(x)$ and $p'(x')$ are not necessarily equal, the probabilities defined in (\ref{i}) must be (according to the rules of probability theory, see the Appendix B of \cite{CarraraVanslette}). This indicates that $\rho\rightarrow\rho'=\rho$ is in the same location in the statistical manifold. That is, the global state of knowledge has not changed -- what has changed is the way in which the local information in $\rho$ has been expressed, which must be invertible in general.

For a coordinate transformation (\ref{i}) involving two variables, $x \in \mathbf{X}$ and $y \in \mathbf{Y}$, we also have,
\begin{equation}
p(x,y) \rightarrow p'(x',y') \quad \mathrm{where}\quad p(x,y)dxdy = p'(x',y')dx'dy'.
\end{equation}
A few general properties of these transformations are as follows:  First, 
the density $p(x,y)$ can be expressed in terms of the density $p'(x',y')$,
\begin{equation}
p(x,y) = p'(x',y')\gamma(x',y'),
\end{equation}
where $\gamma(x',y')$ is the Jacobian \cite{CatichaBook} that defines the transformation,
\begin{equation}
\gamma(x',y') = |\det\left(J(x',y')\right)|, \quad \mathrm{where}\quad J(x',y') = \begin{bmatrix}
\frac{\partial x'}{\partial x} & \frac{\partial x'}{\partial y} \\ \frac{\partial y'}{\partial x} & \frac{\partial y'}{\partial y}
\end{bmatrix}.
\end{equation}
For a finite number of variables $x = (x_1,\dots,x_N)$, general coordinate transformations $p(x_1,\dots,x_N) \rightarrow p'(x_1',\dots,x_N')$ are written,
\begin{equation}
p(x_1,\dots,x_N)\prod_{i=1}^Ndx_i = p'(x_1',\dots,x_N')\prod_{i=1}^Ndx_i',
\end{equation} 
and the Jacobian becomes,
\begin{equation}
J(x'_1,\dots,x'_N) = \begin{bmatrix}
\frac{\partial x_1'}{\partial x_1} &\cdots &\frac{\partial x_1'}{\partial x_N} \\ \vdots & \ddots & \vdots \\ \frac{\partial x_N'}{\partial x_1} & \cdots & \frac{\partial x_N'}{\partial x_N}
\end{bmatrix}.
\end{equation}
One can also express the density $p'(x')$ in terms of the original density $p(x)$ by using the inverse transform,
\begin{equation}
p'(x') = p(f^{-1}(x'))\gamma(x) = p(x)\gamma(x).
\end{equation}
\paragraph{Split Invariant Coordinate Transformations ---}
Consider a class of coordinate transformations that result in a diagonal Jacobian matrix, i.e.,
\begin{equation}
\gamma(x'_1,\dots,x'_N) = \prod_{i=1}^N\frac{\partial x'_i}{\partial x_i} = \prod_{i=1}^N\gamma(x_i').\label{diag}
\end{equation}
These transformations act within each of the variable spaces independently, and hence they are guaranteed to preserve the definition of the split between any $n$-partitions of the propositions -- and because they are coordinate transformations, they are invertible and do not change our state of knowledge, $(\rho,\mathbf{X})\rightarrow(\rho',\mathbf{X}') = (\rho,\mathbf{X})$.  We call such special types of transformations (\ref{diag}) \textit{split invariant coordinate transformations}.  The marginal distributions of $p(x)$ are preserved under split invariant coordinate transformations, 
\begin{equation}
\forall\, \mathbf{X}_i \subset \mathbf{X} : p(x_i)dx_i = p'(x_i')dx_i'.\label{marginal}
\end{equation}
If one allows generic coordinate transformations of the joint space, then the marginal distributions may depend on variables outside of their original split. Thus, if one redefines the split after a coordinate transformation to new variables $\mathbf{X}\rightarrow\mathbf{X}'$, the original problem statement changes as to what variables we are considering correlations \emph{between} and thus eq. (\ref{marginal}) no longer holds. This is apparent in the case of two variables $(x,y)$.  Let $x' = f_{x'}(x,y)$ so that
\begin{equation}
dx' = df_{x'} = \frac{\partial f_{x'}}{\partial x}dx + \frac{\partial f_{x'}}{\partial y}dy,
\end{equation}  
which depends on $y$.  Redefining the split after this coordinate transformation breaks the original independence since a distribution which originally factors, $p(x,y) = p(x)p(y)$, would be made to have conditional dependence in the new coordinates, i.e. if $x' = f_{x'}(x,y)$ and $y' = f_{y'}(x,y)$, then,
\begin{equation}
p(x,y) = p(x)p(y) \rightarrow p'(x',y') = p'(x')p'(y'|x').
\end{equation}
So, even though the above transformation satisfies (\ref{i}), this type of transformation may change the correlations in $p(x)$ by allowing for the potential redefinition of the split $\mathbf{X}\rightarrow\mathbf{X}'$.  Hence, when designing our functional, we identify split invariant coordinate transformations as those which preserve correlations. These restricted coordinate transformations help isolate a single functional form for our global correlation quantifier.

A split in the proposition space can be defined by an index set $K \in \mathbf{K}$ which is a set of sets of indicies corresponding to the groupings of subspaces in proposition space.  Each value in $K$ corresponds to an indexing of the powerset of subspaces of $\mathbf{X}$\footnote{For example, the space of three variables, $\mathbf{X} = \mathbf{X}_1\times\mathbf{X}_2\times\mathbf{X}_3$, has four possible subspace splits,
	\begin{align}
	\{\{\mathbf{X}_1\},\{\mathbf{X}_2\},\{\mathbf{X}_3\}\}, \quad \{\{\mathbf{X}_1\},\{\mathbf{X}_2,\mathbf{X}_3\}\}, \quad \{\{\mathbf{X}_1,\mathbf{X}_2\},\{\mathbf{X}_2\}\}, \quad \{\{\mathbf{X}_1,\mathbf{X}_3\},\{\mathbf{X}_2\}\},\nonumber
	\end{align}
	so that their is a corresponding space $\mathbf{K}$ with the following elements,
	\begin{equation}
	\{\{1\},\{2\},\{3\}\}, \quad \{\{1\},\{2,3\}\}, \quad \{\{1,2\},\{3\}\}, \quad \{\{1,3\},\{2\}\}.
	\end{equation}
}. A splitting of $\mathbf{X}$ is then given by the collection of subsets, $\{\mathbf{X}^{(i)}\}_{(k)} = \left\{\mathbf{X}^{(1)},\mathbf{X}^{(2)},\dots,\mathbf{X}^{(k)}\right\}$.  Given a particular split in the proposition space, a correlation quantifer can be defined as,
\begin{define}
	Let $\mathbf{X}$ be a joint proposition space and let $\mathbf{K}$ denote the set of possible splittings of $\mathbf{X}$.  The correlations present in $\rho$ with respect to a splitting $K \in \mathbf{K}$ is quantified by the real number,
	\begin{align}
	I:\sm\times \mathbf{K}&\rightarrow \mathbb{R}\nonumber\\
	(\rho,\{\mathbf{X}^{(i)}\}_{(k)}) &\mapsto I[\rho,\mathbf{X}^{(1)};\dots;\mathbf{X}^{(k)}].
	\end{align} 
\end{define}
Whatever the value of our desired quantifier $I[\rho]$ gives for a particular distribution $\rho$, we expect that if we change $\rho \xrightarrow{*} \rho' = \rho + \delta \rho$, that our quantifier also changes $I[\rho] \rightarrow I[\rho'] = I[\rho] + \delta I$, and that this change of $I[\rho]$ reflects the change in the correlations, i.e. if $\rho$ changes in a way that increases the correlations, then $I[\rho]$ should also increase.  Thus, our quantifier should be an increasing functional of the correlations, i.e. it should provide a \textit{ranking} of $\rho$'s.

\subsection{The relative entropy}\label{section42}
Before we begin, it is perhaps instructive to specify explicitly the design goal.
\begin{goal}[The ME method]
	Given a prior distribution $\varphi$, and a set of possible posterior distributions $\rho_1,\rho_2,\dots$, we seek to design a functional which ranks each posterior with respect to some order of preference.
\end{goal}
The design goal imposes that the functional $S[\rho,\varphi]$ be scalar valued so that one can rank various posteriors with respect to the prior.  The order of preference in the ranking is determined by setting the various design criteria.  A general functional form for $S[\rho,\varphi]$ can be written as
\begin{equation}
S[\rho,\varphi] = S[p(x),q(x);p(x'),q(x');\dots],\label{Sbare}
\end{equation}  
where $S[\rho,\varphi]$ depends on each value of the probabilities $p(x)$ and $q(x)$ for every $x \in \mathbf{X}$.  To constrain the functional form of (\ref{Sbare}) further, we impose various design criteria (DC).  There are several ways to proceed by choosing different DC.  Whichever approach one takes \cite{ShoreJohnson,Skilling,Jaynes,CatichaBook,VansletteThesis}, there are some defining characteristics of the entropy functional that are common among all of them.  These criteria have to do with the behavior of $S[\rho,\varphi]$ when certain sectors of the underlying space $\mathbf{X}$ are independent.  Independence in $\mathbf{X}$ comes in two flavors, \textit{subdomain} and \textit{subsystem} independence.  The former refers to subsets of $\mathbf{X}$ while the latter to subspaces.  In Vanslette \cite{VansletteThesis}, these two criteria on independence are enough to specify $S[\rho,\varphi]$ for the purposes of ranking.  The DC are imposed on the basis of the PMU.  While simple, the PMU is incredibly constraining.  By stating when one is \textit{not} supposed to change their beliefs it is operationally unique, since there are infinitely many choices for how one \textit{ought} to change them.  Consequently, the PMU imposes a level of objectivity into $S[p,q]$.

The DC we adopt in this thesis are the same as those given by Caticha in \cite{CatichaBook,Caticha_towards,Caticha_Giffin_1,Caticha_Giffin_2,Caticha_Entropy_1}.  The first DC concerns local, subdomain, updates to the probability distribution $p(x)$.
\begin{criteria}[Subdomain Independence]
	Local information has local effects.
\end{criteria}  

Assume that new information tells us we should update our beliefs about a particular subdomain $\mathcal{D} \subset \mathbf{X}$, but does not give us any information about the complement, $\mathcal{D}^c = \mathbf{X} \backslash \mathcal{D}$.  According to the PMU, we should not update our beliefs about probabilities which concern the complement domain, hence we design the inference procedure so that $q(x|\mathcal{D}^c)$ is not updated, i.e.
\begin{equation}
p(x|\mathcal{D}^c) = q(x|\mathcal{D}^c).\label{dc11}
\end{equation}  
Eq. (\ref{dc11}) is an example of the objectivity that is demanded by the PMU.  Since no new information about $\mathcal{D}^c$ is given, the prior distribution takes presidence.  This point is summed up nicely in \cite{CatichaBook},

\begin{quotation}
	``We emphasize: the point is not that we make the unwarranted assumption
	that keeping $q(x|\mathcal{D}^c)$ unchanged is guaranteed to lead to correct inferences. It
	need not; induction is risky. The point is, rather, that in the absence of any
	evidence to the contrary there is no reason to change our minds and the prior
	information takes priority.'' (Caticha, 2020, Section 6.2.3)
\end{quotation}
This criteria appears in all four of the design derivations discussed throughout\footnote{In Shore and Johnson's approach to relative entropy \cite{ShoreJohnson}, axiom four is analogous to DC 1, which states on page 27 ``IV. \textit{Subset Independence}: It should not matter whether one treats an independent subset of system states in terms of a separate conditional density or in terms of the full system density.''  In Skilling's approach \cite{Skilling} DC 1 appears as axiom one which, like Shore and Johnson's axioms, is called \textit{Subset Independence} and is justified with the following statement on page 175, ``Information about one domain should not affect the reconstruction in a different domain, provided there is no constraint directly linking the domains.''  In Caticha \cite{Caticha_Entropy_1} the axiom called \textit{Locality} and is written on page four as ``\textbf{Criterion 1: Locality}. \textit{Local information has local effects}.''  Finally, in Vanslette's work \cite{Vanslette,VansletteThesis}, the subset independence criteria is stated on page three as follows, ``Subdomain Independence:  When information is received about one set of propositions, it should not effect or change the state of knowledge (probability distribution) of the other propositions (else information was also received about them too).''}. DC1 constrains the functional form of the relative entropy to be additive in non-overlapping subdomains,
\begin{equation}
S[\rho,\varphi] \xrightarrow{\mathbf{DC1}} \int dx\, F\left(p(x),q(x),x\right),\label{SDC1}
\end{equation}  
where $F(p(x),q(x),x)$ is some undetermined function of the probabilities and possibly the coordinates.

\begin{criteria}[Coordinate invariance]
	The system of coordinates carries no information.
\end{criteria}
Because we can choose whatever labels we want for the points $x \in \mathbf{X}$, any particular choice cannot have an impact on the ranking of the probability distributions.  While in some situations certain coordinates may be preferred over others, it is not an assumption that we should make unless explicitly stated.  As a result of DC2 the functional form of the entropy becomes,
\begin{equation}
S[\rho,\varphi] \xrightarrow{\mathbf{DC2}} \int dx\, q(x)\Phi\left(\frac{p(x)}{q(x)}\right).\label{dc2a}
\end{equation}    
DC2 has replaced the function $F$, which had three arguments, by a function $\Phi$ which only has one.  To see the effects of DC2, consider a generic coordinate transformation in which $\mathbf{X} \rightarrow \mathbf{X}' = f(\mathbf{X})$ so that,
\begin{equation}
p(x) \rightarrow p'(x') = p(x)\gamma(x') \quad \mathrm{and} \quad q(x) \rightarrow q'(x') = q(x)\gamma(x').
\end{equation}
This particular transformation of the densities $p(x)$ and $q(x)$, and the coordinates $\mathbf{X}$, will always leave integrals of the form (\ref{SDC1}) invariant,
\begin{align}
S[\rho,\varphi] &= \int dx\, F(p(x),q(x),x)\nonumber\\
&= \int dx'\, \gamma(x')F\left(p'(x')\gamma(x')^{-1},q'(x')\gamma(x')^{-1},f(x)\right).\label{dc1b}
\end{align}
The above expression is valid for any change of variables, and hence it imposes no constraints on the functional form of $F$.  Imposing DC2 amounts to saying that we could have started with the coordinates $x' = f(x)$ and it would have made no difference, i.e. the ranking provided by
\begin{equation}
S[\rho',\varphi'] = \int dx'\, F(p'(x'),q'(x'),x'),
\end{equation}
is the same as the one given in (\ref{dc1b}).

\begin{criteria}[Subsystem Independence]
	When two systems are a priori believed to be independent and we receive
	independent information about each then it should not matter if one is
	included in the analysis of the other or not (and vice versa).
\end{criteria}
Like DC1, DC3 imposes the PMU by saying what not to do, rather than what one should do.  In particular, one should not introduce correlations between independent subsystems unless the constraints suggest otherwise\footnote{In Shore and Johnson's approach \cite{ShoreJohnson}, axiom three concerns subsystem independence and is stated on page 27 as ``III. \textit{System Independence}: It should not matter whether one accounts for independent information about independent systems separately in terms of different densities or together in terms of a joint density.''  In Skillings approach \cite{Skilling}, the axiom concerning subsystem independence is given by axiom three on page 179 and provides the following comment on page 180 about its consequences ``This is the crucial axiom, which reduces S to the entropic form. The basic point is that when we seek an uncorrelated image from marginal data in two (or more) dimensions, we need to multiply the marginal distributions. On the other hand, the variational equation tells us to add constraints through their Lagrange multipliers. Hence the gradient $\delta S/\delta f$ must be the logarithm.''  In Caticha's design derivation \cite{Caticha_Entropy_1}, axiom three concerns subsystem independence and is written on page 5 as ``\textbf{Criterion 3: Independence}.  \textit{When systems are known to be independent it should not matter whether they are treated separately or jointly.}'' Finally, in Vanslette \cite{Vanslette,VansletteThesis} on page 3 we have ``Subsystem Independence:  When two systems are a priori believed to be independent and we only receive information about one, then the state of knowledge of the other system remains unchanged.''}.  Consider that $\mathbf{X}$ forms a composite system, $\mathbf{X} = \mathbf{X}_1\times\mathbf{X}_2$, and that our initial information suggests that the two systems be independent.  This amounts to the prior being the product $q(x_1,x_2) = q(x_1)q(x_2)$.  Now suppose we obtain information about $\mathbf{X}_1$ such that the prior $q(x_1)$ is by itself updated to the posterior $p(x_1)$.  Suppose that we also obtain information that the prior $q(x_2)$ is by itself updated to $p(x_2)$.  DC3 imposes that the joint prior $q(x_1)q(x_2)$ be updated to the joint posterior $p(x_1)p(x_2)$, since in the absence of information to the contrary, inferences about one subsystem do not affect inferences about the other.

Imposing DC3 leads to the final functional form of the relative entropy,
\begin{equation}
S[\rho,\varphi] \xrightarrow{\mathbf{DC3}} -\int dx\, p(x)\log\frac{p(x)}{q(x)},\label{dc3a}
\end{equation}  
up to terms which do not affect the updating scheme.  As we have expressed before, the point of DC3 is not to impose some type of consistency requirement between the prior and some new information which perhaps tells us that the variables remain independent.  On the contrary, the role of DC3 is to impose that in the absence of new information about correlations, the prior must necessarily take presidence.  This point is summed up nicely in \cite{CatichaBook},

\begin{quotation}
	``We emphasize that the point is not that when we have no evidence for correlations we draw the conclusion that the systems must necessarily be independent. Induction involves risk; the systems might in actual fact be correlated. The point is rather that if the joint prior reflected independence and the new evidence is silent on the matter of correlations, then the only evidence we actually have -- namely, the prior -- takes precedence and there is no reason to change our minds. As before, a feature of the probability distribution -- in this case, independence -- will not be updated unless the evidence requires it.'' (Caticha, 2020)
\end{quotation}

\paragraph{Summary of the results ---}
The functional in (\ref{dc3a}) is constructed to serve the purposes of the design goal, namely, to allow one to update a prior distribution $q(x)$ to a posterior distribution $p(x)$ when new information comes in the form of constraints $\mathcal{C}$.  In this sense, we have extended the method of MaxEnt to include arbitrary priors.  It is also the case that the functional in (\ref{dc3a}) requires no further interpretation -- it does not represent heat, disorder or an amount of information -- it is simply a tool for updating.  Again, to quote Caticha \cite{CatichaBook}, ``we do not need to know what entropy means; we only need to know how to use it.''  

If one wishes to compare with other approaches to entropic inference, the work by Csisz\'{a}r \cite{Csiszar} provides a nice summary of the various axioms used by many authors (including Azc\'{e}l \cite{Azcel}, Shore and Johnson \cite{ShoreJohnson} and Jaynes \cite{Jaynes1}) in their definitions of information theoretic measures\footnote{A list is given on page 3 of \cite{Csiszar} which includes the following for conditions on an entropy function $H(P)$; (1) Positivity ($H(P) \geq 0$), (2) Expansibility (``expansion'' of $P$ by a new component equal to $0$ does not change $H(P)$, i.e. embedding in a space in which the probabilities of the new propositions are zero), (3) Symmetry ($H(P)$ is invariant under permutation of the probabilities), (4) Continuity ($H(P)$ is a continuous function of $P$), (5) Additivity ($H(P\times Q) = H(P) + H(Q)$), (6) Subadditivity ($H(X,Y) \leq H(X) + H(Y)$), (7) Strong additivity ($H(X,Y) = H(X) + H(Y|X)$), (8) Recursivity ($H(p_1,\dots,p_n) = H(p_1+p_2,p_3,\dots,p_n) + (p_1+p_2)H(\frac{p_1}{p_1+p_2},\frac{p_2}{p_1+p_2})$) and (9) Sum property ($H(P) = \sum_{i=1}^ng(p_i)$ for some function $g$).}.  One could associate the design criteria in this work to some of the common axioms enumerated in \cite{Csiszar}, although some of them will appear as consequences of imposing a specific design criterion, rather than as an ansatz.  For example, the \textit{strong additivity} condition is the result of imposing DC1 and DC3.  \textit{Continuity} of $S[\rho,\varphi]$ with respect to $\rho$ is imposed through the design goal, and \textit{symmetry} is a consequence of DC1.  In summary, \textbf{Design Goal}$\rightarrow$\textit{continuity}, DC1$\rightarrow$\textit{symmetry}, (DC1 + DC2)$\rightarrow$\textit{strong additivity}.  As was shown by Shannon \cite{Shannon_Book} and others \cite{Csiszar,Azcel}, various combinations of these axioms, as well as the ones mentioned in footnote 4, are enough to characterize entropic measures.

\subsection{Proof of the main result}
We begin by considering first a discrete universe of discourse $\mathcal{X} = \{x_1,\dots,x_N\}$ consisting of $N$ variables and $|\mathcal{X}|$ possible values $x_i \in \mathcal{X}$ where $i = \{1,\dots,|\mathcal{X}|\}$.  We wish to update a prior distribution $Q(x_i)$ to a posterior $P(x_i)$ when new information becomes available.  At this point, in accordance with (\ref{Sbare}), the entropy functional $S[P,Q]$ takes the general form,
\begin{equation}
S[P,Q] = S[P(x_1),Q(x_1);P(x_2),Q(x_2);\dots;P(x_N),Q(x_N)].\label{disc}
\end{equation}

\paragraph{DC1 Implementation ---}
Consider that the space $\mathcal{X}$ can be broken up into two disjoint domains of interest, $\mathcal{D} \subseteq \mathcal{X}$ and its complement $\mathcal{D}^c = \mathcal{X}\backslash \mathcal{D}$, so that $\mathcal{D} \cap \mathcal{D}^c = \emptyset$ and $\mathcal{D} \cup \mathcal{D}^c = \mathcal{X}$.  Consider also that we obtain some information that refers strictly to the subdomain $\mathcal{D}$, but not to its complement.  By DC1 we impose that we will \textit{not} update our beliefs relative to the complement domain so that the posterior probabilities $P(x_i|\mathcal{D}^c)$ must follow,
\begin{equation}
P(x_i|\mathcal{D}^c) = Q(x_i|\mathcal{D}^c).
\end{equation}
To determine the consequences of this situation, we impose a set of local constraints.  Consider the constraint on $\mathcal{D}$,
\begin{equation}
\sum_{x_i \in \mathcal{D}} P(x_i)A(x_i) = A_{\mathcal{D}},\label{proof1}
\end{equation}
where $A_{\mathcal{D}}$ is a constant.  DC1 imposes that the constraint (\ref{proof1}) shall have no influence on the conditional probabilities $p(x_i|\mathcal{D}^c)$, however it could introduce a scaling of the probabilities $P(x_i)$ over $\mathcal{D}^c$.  To deal with this, consider that we also impose constraints on the normalization of each subdomain
\begin{equation}
P(\mathcal{D}) = \sum_{x_i \in \mathcal{D}} P(x_i) \qquad \mathrm{and} \qquad P(\mathcal{D}^c) = \sum_{x_i \in \mathcal{D}^c}P(x_i),\label{dc1norm}
\end{equation}
such that $P(\mathcal{D}) + P(\mathcal{D}^c) = 1$.  To search the space of possible posterior distributions $P(x_i)$ which satisfy (\ref{dc1norm}) and (\ref{proof1}), we maximize the entropy functional subject to the constraints in (\ref{dc1norm}) and (\ref{proof1}) using the Lagrange multiplier method,
\begin{align}
\delta\left[\vphantom{\sum_{x_i \in \mathcal{D}}}S[P,Q]\right. &- \lambda_{\mathcal{D}}\left(\sum_{x_i \in \mathcal{D}}P(x_i) - P(\mathcal{D})\right) - \lambda_{\mathcal{D}^c}\left(\sum_{x_i \in \mathcal{D}^c}P(x_i) - P(\mathcal{D}^c)\right)\nonumber\\
&\left.- \mu_{\mathcal{D}}\left(\sum_{x_i \in \mathcal{D}}P(x_i)A(x_i) - \langle A\rangle_{\mathcal{D}}\right)\right] = 0,
\end{align}
which gives the set of $N+3$ differential equations -- one for each $P(x_i)$ and the Lagrange multipliers,
\begin{align}
\frac{\delta S[P,Q]}{\delta P(x_i)} &= \lambda_{\mathcal{D}^c}, \quad \forall\, x_i \in \mathcal{D}^c\\
\frac{\delta S[P,Q]}{\delta P(x_i)} &= \lambda_{\mathcal{D}} + \mu_{\mathcal{D}}A(x_i), \quad \forall\, x_i \in \mathcal{D}.
\end{align}
Given the general form of $S$ in (\ref{disc}), the derivative $\delta S/\delta P(x_i)$ could in principle depend on all $2N$ variables,
\begin{equation}
\frac{\delta S[P,Q]}{\delta P(x_i)} = f_i\left(P(x_1),Q(x_1);P(x_2),Q(x_2);\dots;P(x_N),Q(x_N)\right).
\end{equation}
However, this would violate DC1 since any arbitrary change to the constraint over $\mathcal{D}$ would necessarily influence the probabilities in $\mathcal{D}^c$.  In order to prevent the influence of the constraints on the probabilities in the subdomain $\mathcal{D}^c$, we must impose that the variations in a particular subdomain must only be influenced by the probabilies in that subdomain,
\begin{equation}
\frac{\delta S[P,Q]}{\delta P(x_i|\mathcal{D})} = f_i\left(P(x_1|\mathcal{D}),\dots,P(x_N|\mathcal{D});Q(x_1),\dots,Q(x_N)\right).\label{proof2}
\end{equation}
Since the above (\ref{proof2}) must hold for arbitrary partitions of $\mathcal{X}$, then in the most restrictive case, $\delta S/\delta P(x_i)$ can only depend on the individual $P(x_i)$ and possibly all the $Q(x_i)$,
\begin{equation}
\frac{\delta S[P,Q]}{\delta P(x_i)} = f_i\left(P(x_i);Q(x_1),Q(x_2),\dots,Q(x_N)\right).\label{proof3}
\end{equation}
We can restrict the form of (\ref{proof3}) even further by examining what happens when we make an arbitrary change to the prior outside of $\mathcal{D}$.  Consider that we change $Q(x_j)$ to $Q(x_j) + \delta Q(x_j)$ with $x_j \in \mathcal{D}^c$.  This would suggest that (\ref{proof3}) becomes,
\begin{equation}
\frac{\delta S[P,Q]}{\delta P(x_i)} = f_i\left(P(x_i);Q(x_1),\dots,Q(x_j) + \delta Q(x_j),\dots,Q(x_N)\right).\label{proof4}
\end{equation}
This demonstrates that $P(x_i)$ with $x_i \in \mathcal{D}$ could be influenced by information in $\mathcal{D}^c$, unless $\delta S/\delta P(x_i)$ is independent of all prior information in $\mathcal{D}^c$
\begin{equation}
\frac{\delta S[P,Q]}{\delta P(x_i)} = f_i\left(P(x_i);Q(x_1|\mathcal{D}),Q(x_2|\mathcal{D}),\dots,Q(x_N|\mathcal{D})\right).\label{proof5}
\end{equation}
Since again (\ref{proof5}) must hold for any arbitrary partition of $\mathcal{X}$, we must have
\begin{equation}
\frac{\delta S[P,Q]}{\delta P(x_i)} = f_i(P(x_i),Q(x_i)), \quad \forall\, x_i \in \mathcal{X}.\label{proof6}
\end{equation} 
Integrating over each $x_i \in \mathcal{X}$ we find,
\begin{equation}
S[P,Q] \xrightarrow{\mathbf{DC1}} \sum_{i=1}^{|\mathcal{X}|}F_i\left(P(x_i),Q(x_i)\right),
\end{equation}
for some undetermined function $F$ which could in principle depend on the point $x_i \in \mathcal{X}$.  Taking the continuum limit, $\mathcal{X} \rightarrow \mathbf{X}$, amounts to sending $F_i(P(x_i),Q(x_i))$ to $F(p(x),q(x),x)$ and replacing the sum by an integral,
\begin{equation}
S[\rho,\varphi] = \int dx\, F(p(x),q(x),x),
\end{equation}
which is eq. (\ref{SDC1}).

\paragraph{DC2 Implementation ---}  In order to impose coordinate invariance, we can write (\ref{SDC1}) in a form which explicitly incorporates the density $q(x)$ as a volume measure,
\begin{align}
S[\rho,\varphi] &= \int dx\, q(x)\frac{1}{q(x)}F\left(\frac{p(x)}{q(x)}q(x),q(x),x\right)\nonumber\\
&= \int dx\, q(x)\Phi\left(\frac{p(x)}{q(x)},q(x),x\right),
\end{align}
where we introduced the function $\Phi$,
\begin{equation}
\Phi\left(\frac{p(x)}{q(x)},q(x),x\right) \stackrel{\mathrm{def}}{=} \frac{1}{q(x)}F\left(\frac{p(x)}{q(x)}q(x),q(x),x\right).
\end{equation}
We can restrict the functional form of $\Phi$, and equivalently $F$, by appealing to special cases.  Consider the coordinate transformation,
\begin{equation}
x \rightarrow x' = f(x) \quad \Rightarrow \quad p(x)dx = p'(x')dx' \quad \mathrm{and}\quad q(x)dx = q'(dx').
\end{equation}
This amounts to sending $\Phi$ to 
\begin{equation}
\Phi\left(\frac{p'(x')}{q'(x')},q'(x'),x'\right) = \Phi\left(\frac{p(x)}{q(x)},q(x)\gamma(x'),x'\right).
\end{equation}
Consider the special case in which the Jacobian $\gamma(x') = 1$, which gives,
\begin{equation}
\Phi\left(\frac{p'(x')}{q'(x')},q'(x'),x'\right) = \Phi\left(\frac{p(x)}{q(x)},q(x),x'\right).
\end{equation}
To impose DC2 is to require that the right hand side be equivalent to,
\begin{equation}
\Phi\left(\frac{p(x)}{q(x)},q(x),x'\right) \stackrel{\mathbf{DC2}}{=}\Phi\left(\frac{p(x)}{q(x)},q(x),x\right)
\end{equation}
with the original coordinate.  Thus, it must be the case that $\Phi$, and also $F$, be independent of the coordinates $\mathbf{X}$,
\begin{equation}
\Phi \xrightarrow{\mathbf{DC2}} \Phi\left(\frac{p(x)}{q(x)},q(x)\right).
\end{equation}
We can restrict even further by considering a coordinate transformation with an arbitrary Jacobian,
\begin{equation}
\Phi\left(\frac{p'(x')}{q'(x')},q'(x')\right) = \Phi\left(\frac{p(x)}{q(x)},q(x)\gamma(x')\right).
\end{equation}
Again, to impose DC2 amounts to forcing $\Phi\left(\frac{p(x)}{q(x)},q(x)\gamma(x')\right)$ to be equivalent to $\Phi\left(\frac{p(x)}{q(x)},q(x)\right)$ and hence $\Phi$ must also be independent of its second argument,
\begin{equation}
\Phi \xrightarrow{\mathbf{DC2}} \Phi\left(\frac{p(x)}{q(x)}\right).
\end{equation}
As a result of DC2 we have arrived at (\ref{dc2a}),
\begin{equation}
S[\rho,\varphi] \xrightarrow{\mathbf{DC2}} \int dx\, q(x)\Phi\left(\frac{p(x)}{q(x)}\right).\label{dc2result}
\end{equation}

\paragraph{DC3 Implementation ---}
Consider the situation in which the universe of discource consists of two subspaces, $\mathbf{X} = \mathbf{X}_1\times\mathbf{X}_2$, which are constrained to be independent, $q(x_1,x_2) \xrightarrow{I_q} q(x_1)q(x_2)$.  Next, consider that new information is given for each subspace individually such that $q(x_1)\rightarrow p(x_1)$ and $q(x_2)\rightarrow p(x_2)$.  Naturally, if the systems were treated individually we would find the posteriors $p(x_1)$ and $p(x_2)$ from the individual priors.  According to DC3 on the otherhand, if we were to treat the systems jointly then we should find the joint posterior $q(x_1)q(x_2)\rightarrow p(x_1)p(x_2)$, since no new information has introducted correlations.

At this point we take some inspiration from the derivation by Vanslette \cite{Vanslette,VansletteThesis} and examine several special cases.    
\paragraph{Case 1: Updating one subsystem ---}
Consider the special case in which we constrain the distribution over $\mathbf{X}_1$ to take the known value,
\begin{flalign}
\mathbf{C1}&&\int dx_2\, p(x_1,x_2) = p(x_1).&&\label{c1}
\end{flalign}
The eq. (\ref{c1}) refers to an infinite number of constraints, one for each value of $x_1 \in \mathbf{X}_1$.  Hence, it requires the introduction of a Lagrange multiplier that is a function of $\mathbf{X}_1$ such that by maximizing the entropy,
\begin{equation}
\delta\left[S[\rho,\varphi] - \int dx_1\, \lambda_1(x_1)\left(\int dx_2\, p(x_1,x_2) - p(x_1)\right)\right] = 0,
\end{equation}
we are left with the following differential equation,
\begin{equation}
\frac{\partial \Phi}{\partial p(x_1,x_2)} = \lambda_1(x_1),\label{proof8}
\end{equation}
where we used the fact that for fixed $q(x)$, $\delta S = \int dx\, \partial\Phi/\partial p(x) \delta p(x)$.  The left hand side can be recast in a form similar to (\ref{proof6}),
\begin{equation}
\frac{\delta S[\rho,\varphi]}{\delta p(x_1,x_2)} = \frac{\partial \Phi}{\partial p(x_1,x_2)} = \phi(p(x_1,x_2),q(x_1,x_2)),
\end{equation}
where $\phi$ is a function of the posterior and the prior.  To examine the functional form of $\phi$ with respect to (\ref{proof8}) we first make use of the information in the prior, $I_q$, -- namely, that the prior is independent,
\begin{equation}
\phi \xrightarrow{I_q} \phi(p(x_1,x_2),q(x_1)q(x_2)).
\end{equation}
Now, since the constraint \textbf{C1} in (\ref{c1}) does not introduce any correlations between $\mathbf{X}_1$ and $\mathbf{X}_2$, then by DC3 we impose that the posterior must be a product also,
\begin{equation}
\phi \xrightarrow{\mathbf{DC3}} \phi(p(x_1)p(x_2),q(x_1)q(x_2)) = \lambda_1(x_1).\label{proof9}
\end{equation}
Besides being silent about any correlations, the constraint \textbf{C1} in (\ref{c1}) is also silent about any new information concerning subsystem two, and hence by the PMU the distribution over $\mathbf{X}_2$ is not updated,
\begin{equation}
q(x_2) \xrightarrow{\mathbf{PMU}} p(x_2) = q(x_2).
\end{equation}
This gives for the function $\phi$,
\begin{equation}
\phi \xrightarrow{\mathbf{PMU}} \phi(p(x_1)q(x_2),q(x_1)q(x_2)) = \lambda_1(x_1).\label{proof10}
\end{equation}
We can simplify this expression some more by recognizing that the right hand side is independent of $\mathbf{X}_2$.  This should suggest that the left hand side is also independent of $\mathbf{X}_2$, which we've already determined from DC2.  Going back to DC1 however we can argue for something more -- that the left hand side of (\ref{proof10}) is independent of $q(x_2)$ entirely.  To see this, we are reminded that (\ref{proof10}) represents a single equation for each value of $x_1$ and $x_2$.  This means that for a given $x_1 \in \mathbf{X}_1$, both sides of (\ref{proof10}) must be independent of arbitrary variations with respect to $x_2$.  These include variations with respect to the prior $q(x_2)$, which, with all other quantities held fixed leads to,
\begin{equation}
\delta \phi = \frac{\partial \phi}{\partial q(x_2)}\frac{\partial q(x_2)}{\partial x_2}\delta x_2 = 0,
\end{equation}
which means that for arbitrary $\delta x_2$ we must have $\partial \phi/\partial q(x_2) = 0$.  Thus, (\ref{proof10}) must be independent of $q(x_2)$,
\begin{equation}
\phi \xrightarrow{\mathbf{DC1}}\phi(p(x_1),q(x_1)) = \lambda_1(x_1).\label{proof11}
\end{equation}
In the special case where no constraints are implemented, eq. (\ref{proof11}) reduces to,
\begin{equation}
\phi(p(x_1),q(x_1)) \xrightarrow{\lambda=0} \phi(q(x_1),q(x_1)) = 0,\label{noconstraints}
\end{equation}
where the right hand side is zero since there are no constraints and hence no Lagrange multipliers.  Through a similar set of arguments involving a constraint of the form,
\begin{flalign}
\mathbf{C2}&&\int dx_1\, p(x_1,x_2) = p(x_2).&&\label{c2},
\end{flalign}
one can arrive at an equation analagous to (\ref{proof11}),
\begin{equation}
\phi \xrightarrow{\mathbf{C2}} \phi(p(x_2),q(x_2)) = \lambda_2(x_2).\label{proof12}
\end{equation}

Since the right hand side of (\ref{proof9}) is independent of $x_2$, then 
\paragraph{Case 2: Updating both independent subsystems ---}  Let us now consider the special case in which we update both subsystems according to constraints \textbf{C1} and \textbf{C2},
\begin{flalign}
\mathbf{C1,C2}&&\int dx_1\, p(x_1,x_2) = p(x_2) \quad \mathrm{and} \quad \int dx_2\,p(x_1,x_2) = p(x_1).&&\label{proof7}
\end{flalign}
This again constitutes an infinite number of constraints, one for each $x_1$ and $x_2$, and so requires an infinite number of Lagrange multipliers, which we label as $\lambda_1(x_1)$ and $\lambda_2(x_2)$.  Maximizing the relative entropy with respect to the constraints (\ref{proof7}),
\begin{align}
\delta\left[\vphantom{\int}S[\rho,\varphi] \right.&- \int dx_1\, \lambda_1(x_1)\left(\int dx_2\, p(x_1,x_2) - p(x_1)\right)\nonumber\\
&\left. - \int dx_2\, \lambda_2(x_2)\left(\int dx_1\, p(x_1,x_2) - p(x_2)\right)\right] = 0,
\end{align}
leads to the differential equation,
\begin{equation}
\frac{\partial \Phi}{\partial p(x_1,x_2)} = \lambda_1(x_1) + \lambda_2(x_2).\label{diffeq}
\end{equation}
Recasting the derivative of $\Phi$ as the function $\phi$ we get the functional equation,
\begin{equation}
\phi(p(x_1,x_2),q(x_1,x_2)) = \lambda_1(x_1) + \lambda_2(x_2).
\end{equation}
We will examine this equation in the same way as with case one and two.  First, using the fact that the prior contains no correlations between $\mathbf{X}_1$ and $\mathbf{X}_2$ we have,
\begin{equation}
\phi \xrightarrow{I_q} \phi(p(x_1,x_2),q(x_1)q(x_2)).
\end{equation}
The constraints \textbf{C1} and \textbf{C2} are also silent about and new information concerning correlations between $\mathbf{X}_1$ and $\mathbf{X}_2$ and hence, by DC3 we impose that the posterior must also be independent,
\begin{equation}
\phi\xrightarrow{\mathbf{DC3}} \phi(p(x_1)p(x_2),q(x_1)q(x_2)) = \lambda_1(x_1) + \lambda_2(x_2).\label{proof13}
\end{equation}
Unlike with (\ref{proof11}) and (\ref{proof12}), both sides of this equation are dependent on variations with respect to the distributions over both $\mathbf{X}_1$ and $\mathbf{X}_2$ so that the left hand side is general.  We can however identify some limiting cases to constrain the right hand side to a more manageable form.  First, consider rewriting (\ref{proof13}) in terms of $\lambda_1(x_1)$,
\begin{equation}
\lambda_1(x_1) = \phi(p(x_1)p(x_2),q(x_1)q(x_2)) - \lambda_2(x_2).\label{proof14}
\end{equation}
Since the left hand side of (\ref{proof14}) is independent of $x_2$, then the right hand side must be as well.  If we again consider arbitrary variations with respect to $x_2$ in (\ref{proof14}) we find,
\begin{equation}
\frac{\partial\phi}{\partial p(x_2)}\frac{\partial p(x_2)}{\partial x_2} + \frac{\partial \phi}{\partial q(x_2)}\frac{\partial q(x_2)}{\partial x_2} = \frac{\partial \lambda_2}{\partial x_2},
\end{equation}
which means that for arbitrary variations of $x_2$, the individual variations of $p(x_2)$, $q(x_2)$ and the Lagrange multiplier $\lambda_2(x_2)$ must cancel out.  Whatever the relationship between $\phi(p(x_1)p(x_2),q(x_1)q(x_2))$ and $\lambda_2(x_2)$ is, it must have the effect of removing dependence not only of $x_2$ but of the distributions $p(x_2)$ and $q(x_2)$ as well.  To see this, consider as in (\ref{proof4}) that we were to pick a different prior for $x_2$ which differed by some amount $q(x_2) \rightarrow q'(x_2) = q(x_2) + \delta q_2$.  We would then have that (\ref{proof14}) becomes,
\begin{align}
\lambda_1(x_1) \xrightarrow{\delta q_2}\lambda_1'(x_1) &= \phi(p(x_1)p(x_2),q'(x_1)q(x_2)) - \lambda_2(x_2)\nonumber\\
&= \lambda_1(x_1) + \sum_{n=0}^{\infty}\frac{\partial^n\phi}{\partial q_2^n}(\delta q_2)^n \neq \lambda_1(x_1).
\end{align}
Thus the right hand side of (\ref{proof14}) must be independent of $q(x_2)$.  For similar reasons it must also be independent of $p(x_2)$.  Thus, $\lambda_1(x_1)$ must be equal to some unknown function of $p(x_1)$ and $q(x_1)$ only,
\begin{equation}
\lambda_1(x_1) \rightarrow \psi_1(p(x_1),q(x_1)).
\end{equation}
A similar argument can be made for the subsystem $\mathbf{X}_2$ such that,
\begin{equation}
\lambda_2(x_2) \rightarrow \psi_2(p(x_2),q(x_2)).
\end{equation}
This gives (\ref{proof13}) in the form of a functional equation for $p(x_1)$, $p(x_2)$, $q(x_1)$ and $q(x_2)$,
\begin{equation}
\phi(p(x_1)p(x_2),q(x_1)q(x_2)) = \psi_1(p(x_1),q(x_1)) + \psi_2(p(x_2),q(x_2)).
\end{equation}
Which is a form of \textit{Pexider's equation} \cite{Pexider1,Pexider2,Aczel} which is itself a special case of \textit{Cauchy's equation}.  To construct a solution, first consider the rescaling of the probabilities via the logarithm,
\begin{equation}
p(x) \rightarrow \ell_p(x) = \log p(x) \quad \mathrm{and} \quad q(x) \rightarrow \ell_q(x) = \log q(x).
\end{equation}
Then, by scaling $\phi$, $\psi_1$ and $\psi_2$ via the exponential function we find,
\begin{align}
\phi'(\ell_p(x_1) + \ell_p(x_2),\ell_q(x_1) + \ell_q(x_2)) &= \psi_1'(\ell_p(x_1),\ell_q(x_1))\nonumber\\
&+ \psi_2'(\ell_p(x_2),\ell_q(x_2)),\label{proof15}
\end{align}
where $\phi' = \phi \circ \exp$, $\psi_1' = \psi_1\circ \exp$ and $\psi_2' = \psi_2\circ \exp$.  This brings us to the following theorem from Pexider \cite{Pexider1,Pexider2,Aczel},
\begin{theorem}[Pexider's functional equation]\label{pexider}
	The most general solution to the many-variable Pexider's equation,
	\begin{equation}
	f(x_1 + y_1,x_2 + y_2,\dots,x_n + y_n) = g(x_1,x_2,\dots,x_n) + h(y_1,y_2,\dots,y_n),\label{pexider1}
	\end{equation}
	is,
	\begin{equation}
	f(t) = \xi(t) + a + b, \quad g(t) = \xi(t) + a \quad \mathrm{and} \quad h(t) = \xi(t) + b,\label{pexider2}
	\end{equation}
	where $\xi(t)$ is an arbitrary solution of Cauchy's equation and $a$ and $b$ are constants.
\end{theorem}
A proof of \hyperref[pexider]{theorem 3.1} is given in the \hyperref[proofa1]{appendix}.  From (\ref{pexider1}) and (\ref{pexider2}) it follows that (\ref{proof15}) can be written as,
\begin{align}
\xi(\ell_p(x_1) + \ell_p(x_2),\ell_q(x_1) + \ell_q(x_2)) &= \xi(\ell_p(x_1),\ell_q(x_1))\nonumber\\
&+ \xi(\ell_p(x_2),\ell_q(x_2)),\label{proof16}
\end{align}
which is Cauchy's basic equation in two variables \cite{Cauchy}.  The general solution for a many-variable version of the basic Cauchy equation is given from the following theorem,
\begin{theorem}[Acz\'{e}l]\label{azcel2}
	The general solution of the functional equation,
	\begin{equation}
	f(x_1 + y_1,x_2 + y_2,\dots,x_n + y_n) = f(x_1,x_2,\dots,x_n) + f(y_1,y_2,\dots,y_n),
	\end{equation}
	with $f(0,\dots,0,x_k,0,\dots,0)$ continuous for $k = \{1,2,\dots,n\}$ is
	\begin{equation}
	f(x_1,x_2,\dots,x_n) = c_1x_1 + c_2x_2 + \dots + c_nx_n,
	\end{equation}
	where $c_k \in \mathbb{R}$.
\end{theorem}    
A proof of \hyperref[azcel2]{theorem 3.2} is also given in the \hyperref[proofa2]{appendix}.  From here, the solution for (\ref{proof16}) is given by,
\begin{align}
\xi(\ell_p(x_1) + \ell_p(x_2),\ell_q(x_1) + \ell_q(x_2)) &= c_p(\ell_p(x_1) + \ell_p(x_2))\nonumber\\
&+ c_q(\ell_q(x_1) + \ell_q(x_2)).
\end{align}
Switching back to the original variables and substituting $\xi$ for the original $\phi$ we have,
\begin{align}
\phi(p(x_1)p(x_2),q(x_1)q(x_2)) &= c_p(\ell_p(x_1) + \ell_p(x_2))\nonumber\\
&+ c_q(\ell_q(x_1) + \ell_q(x_2)) + c_{\xi}\nonumber\\
&= c_p\log(p(x_1)p(x_2))\nonumber\\
&+ c_q\log(q(x_1)q(x_2)) + c_{\xi},
\end{align}
where $c_{\xi}$ is an arbitrary constant.  To fix $c_{\xi}$ and the other two constants, $c_p$ and $c_q$, consider the special case for when there are no constraints so that $p(x_1)p(x_2) = q(x_1)q(x_2)$.  Then we must have according to (\ref{noconstraints}),
\begin{equation}
\phi(q(x_1)q(x_2),q(x_1)q(x_2)) = (c_p + c_q)\log(q(x_1)q(x_2)) + c_{\xi} = 0,
\end{equation}   
which means that in general,
\begin{equation}
\frac{c_{\xi}}{(c_p + c_q)} = -\log(q(x_1)q(x_2)).
\end{equation} 
The left hand side however is a constant, independent of the argument of the logarithm on the right hand side.  The only reasonable solution is to set $c_{\xi} = 0$, which leads to,
\begin{equation}
\phi(q(x_1)q(x_2),q(x_1)q(x_2)) = (c_p + c_q)\log(q(x_1)q(x_2)) = 0,
\end{equation}
so that $c_q = -c_p$.  Thus, for general $p(x_1,x_2)$ and $q(x_1,x_2)$ we have 
\begin{equation}
\phi(p(x_1,x_2),q(x_1,x_2)) = |C|\log\left(\frac{p(x_1,x_2)}{q(x_1,x_2)}\right),
\end{equation}
where $|C|$ is an overall constant.

\paragraph{The general solution ---}
We can rewrite the expression for the differential equation (\ref{diffeq}) using the solution to the functional equation for an arbitrary prior $\varphi$ and posterior distribution $\rho$,
\begin{equation}
\frac{\delta S[\rho,\varphi]}{\delta \rho} = |C|\log\left(\frac{p(x)}{q(x)}\right).
\end{equation}
This equation can be integrated with respect to $p(x)$ to give,
\begin{equation}
S[\rho,\varphi] = |C|\int dx\, p(x)\log\left(\frac{p(x)}{q(x)}\right) + |B|\int dx\, p(x) + |A|, \label{gensol}
\end{equation}
where $|A|$ and $|B|$ are constants.  Since $S[p,q]$ is designed for ranking distributions $p(x)$ with respect to the prior $q(x)$, the overall constant $|A|$ can be set to zero, since it only introduces an overall scale which is independent of the probabilities.  Since a normalization constraint will always be imposed, the constant $|B|$ can simply be absorbed into whatever the normalization is set to.  Thus the term $|B|\int dx\, p(x)$ can also be dropped.  This leaves the relative entropy in the form,
\begin{equation}
S[\rho,\varphi] = |C|\int dx\, p(x)\log\left(\frac{p(x)}{q(x)}\right).
\end{equation}  
Comparison with eq. (\ref{dc2result}) shows that the function $\Phi$ is equivalent to\footnote{One can find this solution by simply differentiating (\ref{diffeq}) with respect to $x_1$ and $x_2$ which leads to the equation $y\Phi'''(y) + \Phi''(y) = 0$ where the primes denote the derivative with respect to $y$.  The solution to this differential equation is $\Phi(y) = ay\log y + by + c$, which is equivalent to (\ref{gensol}) \cite{CatichaBook}.},
\begin{equation}
\Phi\left(\frac{p(x)}{q(x)}\right) = \frac{p(x)}{q(x)}\log\left(\frac{p(x)}{q(x)}\right).
\end{equation}
The arbitrary constant $|C|$ can be set to $\pm 1$ for convenience depending on whether one wishes to \textit{maximize} or \textit{minimize} the relative entropy.  Since $x \log (ax)$ for $a > 0$ is a strictly convex function, it has a unique minimum.  This means that $-x\log(ax)$ is a strictly concave function with a unique maximum.  Thus, picking the convention to be that of maximizing the relative entropy, we choose $|C| = -1$ to get the final functional form,
\begin{equation}
S[\rho,\varphi] = -\int dx\, p(x)\log\left(\frac{p(x)}{q(x)}\right).
\end{equation} 

\paragraph{Bayes' rule as a special case ---}
It was first shown Williams \cite{Williams,Diaconis} that Bayes' rule can be derived as a special case of MaxEnt, however the result went unappreciated.  Caticha and Giffin \cite{Caticha_Giffin_1,Caticha_Giffin_2} demonstrated the same result from the perspective of the ME method by associating data with constraint information, which ultimately allows for many generalizations.  The problem is as follows: consider that you wish to conduct inference on a set of parameters $\theta \in \Theta$ so that you can update a prior distribution $q(\theta)$ to the appropriate posterior $p(\theta)$ when information comes in the form of data $x \in \mathbf{X}$.  The relation between $\Theta$ and $\mathbf{X}$ is determined through a likelihood function $q(x|\theta)$ which is also given as prior information.  

The key insight from \cite{Caticha_Giffin_1,Caticha_Giffin_2} is to recognize that before the data $x \in \mathbf{X}$ is specified, its value is unknown and hence a generic inference requires that the universe of discourse be the joint space $\Theta \times \mathbf{X}$ and not just the parameter space $\Theta$.  It is therefore the joint distribution $q(x,\theta)$ which needs to be updated and not simply $q(\theta)$.  With this in mind, consider that we now collect data which takes the value $x' \in \mathbf{X}$.  Then, the posterior $p(x)$ is constrained to reflect the known value,
\begin{equation}
p(x) = \int d\theta\, p(x,\theta) = \delta(x - x').\label{bayesconstraint}
\end{equation}
While this data constrains the form of $p(x,\theta)$, it is not sufficient to determine $p(\theta|x)$ since any choice will satisfy,
\begin{equation}
p(x,\theta) = p(x)p(\theta|x) = \delta(x - x')p(\theta|x).
\end{equation}
To determine the full posterior, we maximize the relative entropy with respect to the constraint (\ref{bayesconstraint})
\begin{align}
0 = \delta \left[\vphantom{\int}S[\rho,\varphi] \right.&- \alpha\left(\int dxd\theta\,p(x,\theta) - 1\right)\nonumber\\
 &\left.- \int dx\, \lambda(x)\left(\int d\theta\, p(x,\theta) - \delta(x - x')\right)\right],
\end{align}
which leads to the joint posterior
\begin{equation}
p(x,\theta) = q(x,\theta)\frac{e^{\lambda(x)}}{Z},
\end{equation}
where $Z = \int dxd\theta\, q(x,\theta)\exp(\lambda(x))$ is a normalization factor.  The Lagrange multiplier $\lambda(x)$ can be determined from (\ref{bayesconstraint}),
\begin{equation}
\int d\theta\, p(x,\theta) = \int d\theta\,q(x)q(\theta|x)\frac{e^{\lambda(x)}}{Z} = q(x)\frac{e^{\lambda(x)}}{Z} = \delta(x - x'),
\end{equation}
so that the joint posterior can be written
\begin{equation}
p(x,\theta) = q(x,\theta)\frac{\delta(x - x')}{q(x)} = q(\theta|x)\delta(x - x').
\end{equation}
Thus, the posterior $p(\theta)$ is,
\begin{equation}
p(\theta) = \int dx\, p(x,\theta) = q(\theta|x') = q(\theta)\frac{q(x'|\theta)}{q(x')},
\end{equation}
which is Bayes' rule.  Further generalizations to Bayes' rule are also discussed in \cite{Caticha_Giffin_1,Caticha_Giffin_2}.

\subsection{Summary}
In this chapter we designed an entropic inference for the purposes of updating our beliefs in the presence of new information that comes in the form of constraints.  By imposing certain design criteria which adhere to the \textit{principle of minimal updating}, we were able to use eliminative induction to find a general functional form of the relative entropy (\hyperref[section42]{Section 4.2}).  We also designed a tool for ranking the correlations present within a probability distribution called the \textit{total correlation}.  This was achieved by a similar set of design criteria used for designing the relative entropy with the exception of \textit{split coordinate invariance} (\hyperref[corollary431]{Corollary 4.3.1}).  The results of \hyperref[section43]{Section 4.3} produce an entire family of correlation quantifiers including the mutual information (\hyperref[section433]{Section 4.3.3}).  In the next section we will employ a practical use of the mutual information in machine learning.

\bibliography{bigbib.bib}

\begin{thebibliography}{100}

\bibitem{Kleene}
Stephen~Cole Kleene.
\newblock {\em Introduction to Meta-Mathematics}.
\newblock Ishi Press International, 1950.

\bibitem{Chomsky}
Noam Chomsky.
\newblock {\em Syntactic Structures}.
\newblock Walter de Gruyter, 1957.

\bibitem{Davidson}
Donald Davidson.
\newblock Truth and meaning.
\newblock {\em Synthese}, 17(3):304--323, 1967.

\bibitem{Tarski}
Alfred Tarski.
\newblock The concept of truth in formalized languages.
\newblock {\em Logic, Semantics, Metamathematics}, pages 152--278.

\bibitem{Tomassi}
Paul Tomassi.
\newblock {\em Logic}.
\newblock Routledge, 1999.

\bibitem{Barwise}
J.~Barwise.
\newblock {\em Handbook of Mathematical Logic}.
\newblock Studies in Logic and the Foundations of Mathematics, 1989.

\bibitem{Aristotle}
Harold P. Cooke~(trans.) Aristotle.
\newblock {\em Aristotle, Vol. 1}.
\newblock Loeb Classical Library, William Heinemann, 1938.

\bibitem{Hamilton}
9th~Baronet Sir William~Hamilton.
\newblock {\em Lecturs on Metaphysics and Logic}.
\newblock Gould and Lincoln, 1860.

\bibitem{Russell}
Bertrand Russell.
\newblock {\em The Problems of Philosophy}.
\newblock Oxford University Press, 1912.

\bibitem{Priest}
Graham Priest.
\newblock {\em In Contradiction}.
\newblock Oxford University Press, 1987.

\bibitem{GarfieldPriest}
Jay~L. Garfield and Graham Priest.
\newblock N\={a}g\={a}rjuna and the limits of thought.
\newblock {\em Philosophy East and West}, 53(1):1--21, 2003.

\bibitem{Priest2}
Graham Priest.
\newblock What is so bad about contradictions?
\newblock {\em Journal of Philosophy}, 95(8):410--426, 1998.

\bibitem{paraconsistent}
Graham Priest.
\newblock {\em Paraconsistent Logic in Handbook of Philosophical Logic Volume
  6}.
\newblock Kluwer Academic Publishers, 2002.

\bibitem{Russell2}
Bertrand Russell.
\newblock On denoting.
\newblock {\em Mind}, 14:479–493, 1905.

\bibitem{Zadeh}
L.A. Zadeh.
\newblock Fuzzy sets.
\newblock {\em Information and Control}, 8(3):338--353, 1965.

\bibitem{Marek}
Wiktor Marek and Miroslaw Truszczynski.
\newblock Autoepistemic logic.
\newblock {\em Journal of the Association for Computing Machinery},
  38(3):588--619, 1991.

\bibitem{Brouwer}
L.E.J. Brouwer.
\newblock {\em Brouwer's Cambridge lectures on intuitionsim}.
\newblock Cambridge University Press, 1981.

\bibitem{Heyting}
A.~Heyting.
\newblock Die formalen regeln der intuitionistischen logik.
\newblock {\em Sitzungsberichte der preußischen Akademie der Wissenschaften.
  phys.-math. Klasse}, 1930.

\bibitem{LewisLangford}
Clarence~Irving Lewis and Cooper~Harold Langford.
\newblock {\em Symbolic Logic}.
\newblock The Century Philosophy series, 1932.

\bibitem{Hasse}
A.~Garrett~Birkhoff Church.
\newblock {\em Lattice theory. Revised edition.}
\newblock American Mathematical Society Colloquium publications, vol. 25.
  American Mathematical Society, New York, 1948.

\bibitem{Stalnaker}
Robert Stalnaker.
\newblock {\em Inquiry}.
\newblock MIT Press, 1984.

\bibitem{Chellas}
Brian~F. Chellas.
\newblock Basic conditional logic.
\newblock {\em Journal of Philosophical Logic}, 4(2):133--153, 1975.

\bibitem{Lewis}
David Lewis.
\newblock Probabilities of conditionals and conditional probabilities.
\newblock {\em The Philosophical Review}, 85(3):297--315, 1976.

\bibitem{Paris}
J.B. Paris.
\newblock The uncertain reasoner's companion: A mathematical perspective.
\newblock 1994.

\bibitem{CatichaBook}
Ariel Caticha.
\newblock {\em Entropic Inference and the Foundations of Physics}.
\newblock (EBEB 2012) Sao Paulo, Brazil).

\bibitem{Jaynes}
E.T. Jaynes.
\newblock {\em Probability Theory: The Logic of Science}.
\newblock Cambridge University Press, 2003.

\bibitem{Cox}
R.T. Cox.
\newblock {\em The algebra of probable inference}.
\newblock The Johns Hopkins University Press, 1961.

\bibitem{KnuthSkilling}
Kevin~H. {Knuth} and John {Skilling}.
\newblock {Foundations of Inference}.
\newblock {\em arXiv e-prints}, page arXiv:1008.4831, August 2010.

\bibitem{Popper}
Karl Popper.
\newblock {\em The Logic of Scientific Discovery}.
\newblock Basic Books, 1959.

\bibitem{Horvitz}
David E.~Heckerman Eric J.~Horvitz and Curtis~P. Langlotz.
\newblock A framework for comparing alternative formalisms for plausible
  reasoning.
\newblock {\em AAAI-86 Proceedings}, pages 210--214, 1986.

\bibitem{VansletteThesis}
Kevin {Vanslette}.
\newblock {The Inferential Design of Entropy and its Application to Quantum
  Measurements}.
\newblock {\em arXiv e-prints}, page arXiv:1804.09142, April 2018.

\bibitem{Norton}
John~D. Norton.
\newblock Probability disassembled.
\newblock {\em British Journal for the Philosophy of Science}, 58:141, 2007.

\bibitem{Cox2}
R.T. Cox.
\newblock Probability, frequency and reasonable expectation.
\newblock {\em American journal of physics}, 14(1):1--13, 1946.

\bibitem{Ling}
C.H. Ling.
\newblock Representation of associative functions.
\newblock {\em Publicationes Mathematicae}, 12:189--212, 1995.

\bibitem{Knuth1}
Kevin~H. {Knuth}.
\newblock {Measuring on Lattices}.
\newblock In Paul~M. {Goggans} and Chun-Yong {Chan}, editors, {\em American
  Institute of Physics Conference Series}, volume 1193 of {\em American
  Institute of Physics Conference Series}, pages 132--144, December 2009.

\bibitem{Knuth2}
Kevin~H. {Knuth}.
\newblock {Deriving Laws from Ordering Relations}.
\newblock In Gary~J. {Erickson} and Yuxiang {Zhai}, editors, {\em Bayesian
  Inference and Maximum Entropy Methods in Science and Engineering}, volume 707
  of {\em American Institute of Physics Conference Series}, pages 204--235,
  April 2004.

\bibitem{KlainRota}
Daniel~A. Klain and Gian-Carlo Rota.
\newblock {\em Introduction to Geometric Probability}.
\newblock Cambridge University Press, 1997.

\bibitem{Aczel}
J\'{a}nos Acz\'{e}l.
\newblock {\em Lectures on Functional Equations and their Applications}.
\newblock Dover publications, 2006.

\bibitem{Halpern}
Joseph Halpern.
\newblock A counterexample to the theorems of cox and fine.
\newblock {\em Journal of Artificial Intelligence Research}, 10:67--85, 1999.

\bibitem{Colyvan}
Mark Colyvan.
\newblock The philosophical significance of cox's theorem.
\newblock {\em International journal of approximate reasoning}, 37:71--85,
  2004.

\bibitem{Bayes}
Thomas Bayes and Richard Price.
\newblock An essay towards solving a problem in the doctrine of chance. by the
  late rev. mr. bayes, communicated by mr. price, in a letter to john canton,
  a. m. f. r. s.
\newblock {\em Philosophical Transactions of the Royal Society of London.},
  53:370–418, 1763.

\bibitem{Ariel1}
Ariel {Caticha}.
\newblock {Quantifying Rational Belief}.
\newblock In Paul~M. {Goggans} and Chun-Yong {Chan}, editors, {\em American
  Institute of Physics Conference Series}, volume 1193 of {\em American
  Institute of Physics Conference Series}, pages 60--68, December 2009.

\bibitem{Ariel2}
Ariel {Caticha}.
\newblock {Lectures on Probability, Entropy, and Statistical Physics}.
\newblock {\em arXiv e-prints}, page arXiv:0808.0012, July 2008.

\bibitem{Fine}
T.L. Fine.
\newblock {\em Theories of Probability}.
\newblock Academic Press, New York, 1973.

\bibitem{Jaynes1}
E.T. Jaynes.
\newblock How does the brain do plausible reasoning?
\newblock {\em Tech. Rep. 421, Stanford University Microwave Laboratory}, 1957.

\bibitem{SmithErickson}
C.~Smith and J.~Erickson.
\newblock Probability theory and the associativity equation.
\newblock {\em In P. F. Fougère, editor, Maximum Entropy and Bayesian Methods.
  Kluwer, Dordrecht}, 1990.

\bibitem{Tribus}
M.~Tribus.
\newblock {\em Rational Descriptions, Decisions and Designs}.
\newblock Pergamon Press, 1969.

\bibitem{VanHorn}
Kevin S.~Van Horn.
\newblock From propositional logic to plausible reasoning : A uniqueness
  theorem.
\newblock {\em International journal of approximate reasoning}, 88:309--332,
  2017.

\bibitem{VanHorn2}
Kevin S.~Van Horn.
\newblock Constructing a logic of plausible inference: a guide to cox's
  theorem.
\newblock {\em International journal of approximate reasoning}, 34:3--24, 2003.

\bibitem{Ramsey}
Frank Ramsey.
\newblock General propositions and causality.
\newblock {\em Originally an unpublished manuscript which was later
  posthumously published in two collections: Foundations: Essays in Philosophy,
  Logic, Mathematics and Economics, edited by D. H. Mellor, London: Routledge
  and Kegan Paul, 1978, pp. 133-51; and in F. P. Ramsey: Philosophical Papers,
  edited by D. H. Mellor, Cambridge: Cambridge Unversity Press, 1990, pp.
  145-63.}

\bibitem{Savage}
L.J. Savage.
\newblock {\em The Foundations of Statistics}.
\newblock John Wiley and Sons, 1954.

\bibitem{vonNeumannMorgenstern}
John von Neumann and Oskar Morgenstern.
\newblock {\em Theory of Games and Economic Behavior}.
\newblock Princeton University Press, 1953.

\bibitem{Fuchs}
Christopher~A. {Fuchs}.
\newblock {Notwithstanding Bohr, the Reasons for QBism}.
\newblock {\em arXiv e-prints}, page arXiv:1705.03483, May 2017.

\bibitem{Kolmogorov}
A.~Kolmogorov.
\newblock {\em Foundations of the Theory of Probability}.
\newblock Chelsea Publishing Com- pany, 2nd english edition, 1956.

\bibitem{BanachTarski}
Stefan Banach and Alfred Tarski.
\newblock Sur la décomposition des ensembles de points en parties
  respectivement congruentes.
\newblock {\em Fundamenta Mathematicae}, 6:244–277, 1924.

\bibitem{Fisher}
R.~Fisher.
\newblock On the mathematical foundations of theoretical statistics.
\newblock {\em Philosophical Transactions of the Royal Society}, 222:309--368,
  1922.

\bibitem{Jaynes2}
E.T. Jaynes.
\newblock Prior probabilities.
\newblock {\em IEEE Transactions on Systems Science and Cybernetics},
  4(3):227--241, 1968.

\bibitem{Savage1}
Leonard~J. Savage.
\newblock The foundations of statistics reconsidered.
\newblock {\em Proceedings of the Fourth Berkely Symposium on Mathematics and
  Probability}, 1961.

\bibitem{Feynman}
R.P. Feynman.
\newblock Space-time approach to non-relativistic quantum mechanics.
\newblock {\em Review of Modern Physics}, 20(2), 1948.

\bibitem{Koopman}
B.O. Koopman.
\newblock Quantum theory and the foundations of probability.
\newblock {\em Applied Probability}, pages 97--102, 1955.

\bibitem{Ballentine}
L.E. Ballentine.
\newblock Probability theory in quantum mechanics.
\newblock {\em American Journal of Physics}, 54(10), 1986.

\bibitem{Jaynes3}
E.T. Jaynes.
\newblock The relation of bayesian and maximum entropy methods.
\newblock {\em Maximum-Entropy and Bayesian Methods in Science and Engineering
  - G. J. Erickson and C. R. Smith (eds.), Kluwer, Dordrecht}, page~25.

\bibitem{CoverThomas_Book}
T.~M. Cover and Joy~A. Thomas.
\newblock {\em Elements of information theory}.
\newblock Wiley-Interscience, 2006.

\bibitem{Shannon_Book}
Claude~Elwood Shannon and Warren Weaver.
\newblock {\em The mathematical theory of communication}.
\newblock University of Illinois Press, 1999.

\bibitem{Renyi}
A.~Renyi.
\newblock On measures of entropy and information.
\newblock {\em Proc. 4th Berkeley Symposium on Mathematical Statistics and
  Probability}, 1:547, 1961.

\bibitem{Tsallis1}
C.~Tsallis.
\newblock Possible generalization of boltzmann-gibbs statis-tics.
\newblock {\em J. Stat. Phys.}, 52:479, 1988.

\bibitem{ShoreJohnson}
J.~E. Shore and R.~W. Johnson.
\newblock {Axiomatic Derivation of the Principle of Maximum Entropy and the
  Principle of Minimum Cross-Entropy}.
\newblock {\em IEEE Transactions of Information Theory}, IT-26(1), 1980.

\bibitem{Skilling}
J.~Skilling.
\newblock {The Axioms of Maximum Entropy}.
\newblock {\em Maximum-Entropy and Bayesian Methods in Science and Engineering,
  G. J. Erickson and C. R.Smith (eds.) (Kluwer, Dordrecht, 1988).}

\bibitem{Caticha_Giffin_1}
Ariel {Caticha} and Adom {Giffin}.
\newblock {Updating Probabilities}.
\newblock In Ali {Mohammad-Djafari}, editor, {\em Bayesian Inference and
  Maximum Entropy Methods In Science and Engineering}, volume 872 of {\em
  American Institute of Physics Conference Series}, pages 31--42, Nov 2006.

\bibitem{Caticha_Entropy_3}
Ariel {Caticha}.
\newblock {Relative Entropy and Inductive Inference}.
\newblock 707:75--96, April 2004.

\bibitem{Caticha_Entropy_4}
Ariel {Caticha}.
\newblock {Information and Entropy}.
\newblock 954:11--22, November 2007.

\bibitem{Caticha_Entropy_1}
Ariel {Caticha}.
\newblock {Entropic Inference}.
\newblock In Ali {Mohammad-Djafari}, Jean-Fran{\c{c}}ois {Bercher}, and Pierre
  {Bessi{\'e}re}, editors, {\em American Institute of Physics Conference
  Series}, volume 1305 of {\em American Institute of Physics Conference
  Series}, pages 20--29, Mar 2011.

\bibitem{Caticha_Entropy_2}
Ariel {Caticha}.
\newblock {Entropic inference: Some pitfalls and paradoxes we can avoid}.
\newblock In Udo {von Toussaint}, editor, {\em American Institute of Physics
  Conference Series}, volume 1553 of {\em American Institute of Physics
  Conference Series}, pages 200--211, Aug 2013.

\bibitem{CarraraVanslette}
Nicholas {Carrara} and Kevin {Vanslette}.
\newblock {The Design of Global Correlation Quantifiers and Continuous Notions
  of Statistical Sufficiency}.
\newblock {\em Entropy}, 22(3):357, March 2020.

\bibitem{Watanabe}
S.~Watanabe.
\newblock Information theoretical analysis of multivariate correlation.
\newblock {\em IBM Journal of Research and Development}, 4:66–82, 1960.

\bibitem{Landauer1}
R.~Landauer.
\newblock Information is physical.
\newblock {\em Physics Today}, 1991.

\bibitem{Landauer2}
Rolf Landauer.
\newblock Information is a physical entity.
\newblock {\em Physica A: Statistical Mechanics and its Applications},
  263(1):63 -- 67, 1999.
\newblock Proceedings of the 20th IUPAP International Conference on Statistical
  Physics.

\bibitem{Wheeler}
John~Archibald Wheeler.
\newblock Information, physics, quantum: The search for links.
\newblock {\em Proc. 3rd Int. Symp. Foundations of Quantum Mechanics, Tokyo},
  pages 354--368, 1989.

\bibitem{Caticha_towards}
Ariel {Caticha}.
\newblock {Towards an Informational Pragmatic Realism}.
\newblock {\em arXiv e-prints}, Dec 2014.

\bibitem{Vanslette}
K.~Vanslette.
\newblock Entropic updating of probabilities and density matrices.
\newblock {\em Entropy}, 19(12):664, 2017.

\bibitem{Csiszar}
I.~Csiszár.
\newblock Axiomatic characterizations of information measures.
\newblock {\em Entropy}, 10:261--273, 2008.

\bibitem{Azcel}
J.~Azc\'{e}l and Z.~Dar\'{o}czy.
\newblock {\em On meausre of information and their characterizations,
  Mathematics in Science and Engineering (Volume 115)}.
\newblock Academic Press, 1975.

\bibitem{Pexider1}
J.V. Pexider.
\newblock Eine studie \"{u}ber die funktionalgleichenugen (czech.).
\newblock {\em \v{C}asopis P\v{e}st. Mat.}, 29:153--195, 1900.

\bibitem{Pexider2}
J.V. Pexider.
\newblock Notiz \"{u}ber funktionaltheoreme.
\newblock {\em Monatsh. Math. Phys.}, 14:293--301, 1903.

\bibitem{Cauchy}
A.L. Cauchy.
\newblock Cours d'analyse de l'\'{E}cole polytechnique.
\newblock {\em Analyse alg\'{e}brique}, 1, 1821.

\bibitem{Williams}
P.M.Williams.
\newblock Bayesian conditionalisation and the principle of minimum information.
\newblock {\em The British Journal for the Philosophy of Science}, 31:131--144,
  1980.

\bibitem{Diaconis}
P.~Diaconis and S.~Zabell.
\newblock Updating subjective probability.
\newblock {\em Journal of the American Statistical Association},
  77(380):822--830, 1982.

\bibitem{Quine}
W.V. Quine.
\newblock {\em Philosophy of Logic}.
\newblock Harvard University Press, 1970.

\bibitem{Bool}
George Boole.
\newblock {\em An Investigation of the Laws of Thought}.
\newblock 1854.

\bibitem{Bool2}
George Boole.
\newblock {\em The Mathematical Analysis of Logic}.
\newblock Printed in England by Henerson and Spalding, 1847.

\bibitem{Kripke}
Saul~A. Kripke.
\newblock Semantical analysis of model logic i: Normal model propositional
  calculi.
\newblock {\em Zeitschr. J. Math. Logik und Grundlagen d. Math}, 9:67--96,
  1963.

\bibitem{LewisBook2}
David Lewis.
\newblock {\em On the Plurality of Worlds}.
\newblock Wiley-Blackwell, 2001.

\bibitem{vonNeumann2}
John Birkhoff, Garrett; von~Neumann.
\newblock The logic of quantum mechanics.
\newblock {\em Annals of Mathematics}, 37(4):823–843, 1936.

\bibitem{DeWittGraham}
Bryce~S. DeWitt and Neill Graham.
\newblock {\em The Many Worlds Interpretation of Quantum Mechanics}.
\newblock Princeton Series in Physics (63), 2015.

\bibitem{Peirce}
Charles~Sanders Peirce.
\newblock A boolian[sic] algebra with one constant.
\newblock {\em in Hartshorne, C. and Weiss, P., eds., (1931–35) Collected
  Papers of Charles Sanders Peirce, Vol. 4: 12–20, Cambridge: Harvard
  University Press.}, 1880.

\bibitem{Sheffer}
H.~M. Sheffer.
\newblock A set of five independent postulates for boolean algebras, with
  application to logical constants.
\newblock {\em Transactions of the American Mathematical Society},
  14:481–488, 1913.

\bibitem{Post}
E.~L. Post.
\newblock The two-valued iterative systems of mathematical logic.
\newblock {\em Annals of Mathematics studies, no. 5, Princeton University
  Press, Princeton}, page 122, 1941.

\bibitem{Stone}
Marshall~H. Stone.
\newblock The theory of representation for boolean algebras.
\newblock {\em Transactions of the American Mathematical Society},
  40(1):37--111, 1936.

\bibitem{GivantHalmos}
Steven Givant and Paul Halmos.
\newblock {\em Introduction to Boolean Algebras}.
\newblock Springer, 2009.

\bibitem{Wittgenstein}
Ludwig Wittgenstein.
\newblock {\em Tractatus Logico-Philosophicus}.
\newblock Harcourt, Brace and Company, Inc.; Accessed via open-source at
  \url{http://people.umass.edu/klement/tlp/}, 1922.

\bibitem{Frege}
Gottlob Frege.
\newblock {\em Begriffsschrift: eine der arithmetischen nachgebildete
  Formelsprache des reinen Denkens}.
\newblock Halle, 1879.

\bibitem{McGee}
Vann McGee.
\newblock A counterexample to modus ponens.
\newblock {\em Journal of Philosophy}, (9):462--471, 1985.

\bibitem{Pierce}
Charles~S. Pierce.
\newblock {\em Reasoning and the Logic of Things: The Cambridge Conferences
  Lectures of 1898}.
\newblock Harvard University Press, 1992.

\bibitem{Goldblatt}
Robert Goldblatt.
\newblock {\em Mathematical Modal Logic: A View of its Evolution}.
\newblock In Handbook of the history of logic: Volume 7, Editors: Dov M. Gabbay
  and John Woods, 2006.

\bibitem{ChellasBook}
Brian Chellas.
\newblock {\em Modal Logic}.
\newblock Cambridge University Press, 1980.

\bibitem{OhrstromHasle}
Peter~\O hrstr\o m and Per Hasle.
\newblock {\em Temporal Logic: From Ancient Ideas to Artificial Intelligence}.
\newblock Springer: Studies in Linguistics and Philosophy, 1995.

\bibitem{vonWright}
G.H. von Wright.
\newblock Deontic logic.
\newblock {\em Mind}, 60(237):1--15, 1951.

\bibitem{Smullyan}
Raymond Smullyan.
\newblock Logicians who reason about themselves.
\newblock {\em Proceedings of the 1986 conference on Theoretical aspects of
  reasoning about knowledge, Monterey (CA), Morgan Kaufmann Publishers Inc.,
  San Francisco (CA)}, pages 341--352, 1986.

\bibitem{Hume}
David Hume.
\newblock Diaglogues concerning natural religion: Part 9.
\newblock 1776.

\bibitem{Stanford}
James Garson.
\newblock Modal logic.
\newblock In Edward~N. Zalta, editor, {\em The Stanford Encyclopedia of
  Philosophy}. Metaphysics Research Lab, Stanford University, fall 2018
  edition, 2018.

\bibitem{Pacuit}
Horacio Arl\'{o}-Costa and Eric Pacuit.
\newblock First-order classical modal logic.
\newblock {\em Studia Logica}, 68, 2001.

\bibitem{LewisBook}
David Lewis.
\newblock {\em Counterfactuals}.
\newblock Wiley-Blackwell, 1973.

\bibitem{Goodman}
Nelson Goodman.
\newblock {\em Fact, Fiction and Forecast}.
\newblock Harvard University Press, 1979.

\bibitem{Adams}
Ernest Adams.
\newblock On the logic of conditionals.
\newblock {\em Inquiry}, 8:166--197, 1965.

\bibitem{Bradley}
Richard Bradley.
\newblock A defence of the ramsey test.
\newblock {\em Mind}, 116(461):1--21, 2007.

\bibitem{deFinetti1}
Bruno de~Finetti.
\newblock La logique de la probabilitie.
\newblock 1935.

\bibitem{Grice}
H.~P. Grice.
\newblock {\em Studies in the Way of Words}.
\newblock Cambridge MA: Harvard University Press, 1989.

\bibitem{Stalnaker1}
Robert Stalnaker.
\newblock A theory of conditionals.
\newblock {\em Studies in Logical Theory, American Philosophical Quarterly},
  pages 98--112.

\bibitem{Stalnaker2}
Robert Stalnaker.
\newblock Probability and conditionals.
\newblock {\em Philosophy of Science}, 32(1):64--80, 1970.

\bibitem{Stalnaker3}
Robert~C. Stalnaker and Richmond~H. Thomason.
\newblock A semantic analysis of conditional logic.
\newblock {\em Theoria}, 36(1):23--42, 1970.

\bibitem{JeffreyEdgington}
Richard Jeffrey and Dorothy Edgington.
\newblock Matter-of-fact conditionals.
\newblock {\em Proceedings of the Aristotelian Society, Supplementary Volumes},
  65, 1991.

\end{thebibliography}

\begin{appendix}

\section{Rules for Deductive Inference}\label{logicappendix}

In this Appendix we will construct the simple system $\mathcal{D} = \left\{\mathcal{A},\Omega,\mathcal{I},\mathcal{Z},\mathcal{T}\right\}$, where $\mathcal{A}$ is a set of well-formed formulas called \hyperref[propositions]{propositions}, $\Omega$ is a collection of disjoint subsets of \hyperref[connectives]{logical connectives} (relations), $\mathcal{I}$ is a set of distinguished propositions called \hyperref[axioms]{axioms}, which include the three laws of thought, and $\mathcal{Z}$ is a set of \hyperref[inferencerules]{inference rules} concerning the logical connectives.  The object $\mathcal{T}$ contains any additional structures we choose to add to our framework.  To build the system $\mathcal{D}$, we will first discuss the construction of $\mathcal{A}$.

\subsection{The space of propositions $\mathcal{A}$}\label{propositions}
The universe of discourse, or \textit{well formed formulas} of $\Sigma^*$, of interest will be the space of propositions, which we label as $\mathcal{A}$.  An element of $\mathcal{A}$ is a single proposition\footnote{We recognize that often in foundational mathematics, the calculus of logic is developed before any notion of sets is constructed.  Here we take some liberties with this and assume that the definition of the epsilon relation $\in$ and the axioms of ZF set theory are understood in advance.}, $a \in \mathcal{A}$, which is defined as,
\begin{define}\label{proposition}
	A proposition, $a \in \mathcal{A}$, is a statement, or variable, which is either true or false (labeled $\mathbf{T}$ and $\mathbf{F}$ respectively).  The set $\mathcal{A}$ is defined as the countable collection,
	\begin{equation}
	\mathcal{A} \stackrel{\mathrm{def}}{=} \left\{\frac{}{}a\,\middle|\, (\nu(a)\Leftrightarrow \mathbf{T})\vee (\nu(a) \Leftrightarrow \mathbf{F})\right\}.
	\end{equation}
\end{define}
As pointed out by Quine \cite{Quine}, the ``meaning'' of propositions can be vague even if their ``truth'' value as it relates to the real world is well defined.  At this level of discourse, we only assume that whatever propositions are, they have a well defined truth value.  In this sense, all theorems are propositions in themselves, however statements which fail to be valid for reasons of incompleteness are not propositions by definition.

The definition \hyperref[proposition]{above} establishes the \textit{semantics} as consisting of a \textit{bivalent}, two-valued or \textit{Boolean} logic \cite{Bool,Bool2}.  This fact is often imposed as a \textit{principle of bivalence} \cite{Tomassi}.  While there are only two options, true or false, this principle is not necessarily the same as the \hyperref[lawsoflogic]{law of excluded middle}. One can choose to impose the law of excluded middle while excluding the principle of bivalence\footnote{For example, in a three valued logic where one has three possible truth values, $\mathbf{A},\mathbf{B}$ and $\mathbf{C}$, one no longer has a principle of bivalence however one can always still impose the law of excluded middle by demanding that all propositions \textit{must} take one of the three possible truth values, i.e. if a proposition is not $\mathbf{A}$ or $\mathbf{B}$ then it \textit{must} be $\mathbf{C}$.} or vice versa \cite{Tomassi}.  

The function $\nu$ is called a \textit{valuation}, which is a map from the set $\mathcal{A}$ to its truth value.
\begin{define}\label{valuation}
	A valuation $\nu:\mathcal{A} \rightarrow \left\{\mathbf{T},\mathbf{F}\right\}$ is a function which returns the truth value of a proposition.
\end{define}
As was discussed in the preface, the meaning of ``truth'' is still a subject of debate among philosophers.  Here we simply identify that true and false are two mutually exclusive properties of propositions, but we do not offer any explanation as to how any particular proposition obtains its ``truth'' value, which would require a metatheory.  Instead we take the truth values of propositions as given and only worry about developing the calculus for how to manipulate them.  

For any particular universe of discourse $\mathcal{A}$, we can collect all of the associated truth values for each $a \in \mathcal{A}$ into a set which is called a \textit{world}.  A world is simply an identification of which propositions $a \in \mathcal{A}$ are true, and which are false\footnote{There are several ways in which to define a world $w$ \cite{Kripke}.  One is the specification described here in which $w$ is a collection of propositions and their associated truth values.  Another possibility is to simply have $w$ contain only propositions which are true, i.e. if $a$ is true in $w$ then $a \in w$, and if $b$ is false in $w$ then $\neg b \in w$.}.  We require in the very least that any such world be internally consistent, i.e. that it not contain contradictions.  As we will see in conditional logic, in order to address propositions about the future one must extend classical logic to include what are typically called \textit{possible worlds}. 
\begin{define}\label{possibleworlds}
	A world $w$, is a collection of truth values \cite{Kripke} associated to an underlying set of propositions $\mathcal{A}$.  The collection of all possible worlds $w$ over the set of propositions $\mathcal{A}$ is called $W$.
\end{define}
For each $a \in \mathcal{A}$, a world $w \in W$ records a truth value, so that $|w| = |\mathcal{A}|$, i.e. the number of elements in $w$ is the same as in $\mathcal{A}$.  Instead of defining $w$ in terms of truth values, one could instead define it as the collection of propositions from $\mathcal{A}$ that are true in $w$.  Due to the law of excluded middle, this representation gives $|w| = |\mathcal{A}|/2$.  Some uses of possible worlds allow for worlds $w \in W$ to only have truth values for a subset of $\mathcal{A}$, however we will not consider that direction here.

Possible worlds introduce an element of subjectivity into logical inference, since different agents can each assign a possible world as a candidate for the \textit{actual world}.  The actual world, $\tilde{w} \in W$, is the one that corresponds to ultimate reality and hence contains in it solely objective facts.  Some philosophers take an extreme position on the existence of possible worlds.  David Lewis \cite{LewisBook2} in particular pushed for an agenda of \textit{modal realism}, in which each possible world is just as real as the actual world.  This is similar to positions held by many of the \textit{many worlds approach to quantum mechanics} crowd\footnote{Theories such as quantum mechanics present a significant challenge to traditional logic since they describe a world which is necessarily \textit{indeterministic}, i.e. the truth value of propositions about the future are not only unknown but also unknowable according to the theory.  Theories which attempt to address these issues fall under the category of \textit{quantum logic}, a subject first initiated by von Neumann \cite{vonNeumann2}.} \cite{DeWittGraham}.  We do not adopt either of these positions here, since there must exist some ultimate reality, and that ultimate reality is what should be associated to the \textit{actual world}.  

The definition of a world allows us to generalize the valuation function (\ref{valuation}) to one which takes a world as an argument,
\begin{define}
	Let $w$ be a possible world over the propositions $\mathcal{A}$.  The valuation of a statement $a \in \mathcal{A}$ with respect to the world $w$ is
	\begin{equation}
	v:\mathcal{A}\times W \rightarrow \left\{\mathbf{T},\mathbf{F}\right\}.\label{truthworld}
	\end{equation}
\end{define}
Semantically, if $a$ is true in $w$, then 
\begin{equation}
w \vDash a,
\end{equation}
or, ``$w$ proves $a$.''  Another notation is to write the semantic consequence with the subscript of the world in which it is true, $\vDash_w a$.  The use of possible worlds will become important when we discuss modal logic in \hyperref[modallogic]{Section 8}

A world forms an \textit{interconnected web of beliefs} of complete knowledge \cite{CatichaBook}, which, as a rule of thumb, must necessarily be internally consistent.  Thus, any consistent world must adhere to the laws of logic.  If a possible world $w \in W$ contains a contradiction, then by the \hyperref[explosion]{principle of explosion}, every proposition in it is necessarily true.  This world is sometimes labeled $\lambda \in W$.  

We can form an equivalence between subsets of $W$ and the space of propositions $\mathcal{A}$.  Consider labeling a subset of $W$ in which $a$ is true as\footnote{We say that the subset of worlds $W_a \subseteq W$ in which $a$ is true is \textit{generated} by the proposition $a$.},
\begin{equation}
W_a \stackrel{\mathrm{def}}{=} \left\{\frac{}{}w\,\middle|\, \nu(a,w)\right\}\subseteq W.\label{setofsubsetsworlds}
\end{equation}
We then identify the set of subsets of (\ref{setofsubsetsworlds}) as $\mathcal{W} = \left\{W_a \in 2^{W}\right\}$, which is isomorphic\footnote{One can show that each set $W_a \in \mathcal{W}$ uniquely corresponds to a proposition $a \in \mathcal{A}$ so that there exists some bijection $f:\mathcal{A}\rightarrow\mathcal{W}$.} to the space of propositions $\mathcal{A}$,
\begin{equation}
\mathcal{W} \cong_{\mathrm{set}} \mathcal{A}.
\end{equation}

\subsection{Logical connectives $\Omega$}\label{connectives}
One can define maps from one or many propositions to another.  These are called $n$-ary connectives, of which the unary and binary ones will be of importance.  The set of $n$-ary connectives can be written as the union of $n$ disjoint subsets,
\begin{equation}
\Omega = \bigcup_{i\geq 0}\Omega_i = \Omega_0\cup\Omega_1\cup\Omega_2\cup\dots\cup\Omega_k\cup\dots\cup\Omega_n.
\end{equation}
The sets $\Omega_1$ and $\Omega_2$ are the \textit{unary} and \textit{binary} connectives respectively.  Unary connectives are defined as,
\begin{define}
	A unary connective, $u\in\Omega_1:\mathcal{A}\rightarrow\mathcal{A}$, is a map from a single proposition to another, i.e. $a\mapsto u(a) \in \mathcal{A}$.
\end{define}
Since the number of possible truth values is two, there are $2^2 = 4$ unique unary connectives on proposition space.  They are called, \textit{tautology}, \textit{contradiction, identity} and \textit{negation}.  The \textit{tautology} ($\top$) and \textit{contradiction} ($\bot \Leftrightarrow \neg \top$)\footnote{Often times the tautology and contradiction connectives are defined as elements of $\Omega_0$.} connectives map every proposition to either true or false\footnote{In other words, no matter what the propositions $a$ is, $\forall\, a : \top a \Leftrightarrow \mathbf{T}$ and $\forall\, a : \bot a \Leftrightarrow \mathbf{F}$.}, regardless of their actual truth value.  The identity connective ($\mathrm{Id}$) simply returns the original proposition (i.e. $\forall\, a : \mathrm{Id}(a) \Leftrightarrow a$), and the negation map ($\neg$) returns the complement of the proposition (i.e. $\forall\, a : \neg(a) \Leftrightarrow \neg a$) so that $\nu(\neg a) \Leftrightarrow \neg\nu(a)$.  A convenient construction for determining the truth values of $n$-ary connectives are truth tables.  The following is the truth table for the unary connectives,
\begin{table}[H]
	\centering
	\begin{tabular}{|c|c|c|c|c|}
		\hline
		$a$ & $\top$ & $\bot$ & $\mathrm{Id}$ & $\neg$\\
		\hline
		\hline
		$\tru$ & $\tru$ & $\fals$ & $\tru$ & $\fals$\\
		$\fals$ & $\tru$ & $\fals$ & $\fals$ & $\tru$\\
		\hline
	\end{tabular}
	\caption{The four unary operators acting on a proposition $a$.}
	\label{unary}
\end{table}

The binary connectives are defined similarly,
\begin{define}
	A binary connective, $w\in\Omega_2:\mathcal{A}\times\mathcal{A}\rightarrow \mathcal{A}$, is a map from two propositions to another, i.e. $(a,b) \mapsto w(a,b) \in \mathcal{A}$.
\end{define}
Since for any pair of propositions, $a$ and $b$, there are four possible truth values, $(\tru,\tru),(\tru,\fals),(\fals,\fals)$ and $(\fals,\tru)$, there are $2^{2^2} = 2^4 = 16$ unique binary connectives.  These include the \textit{tautology} and \textit{contradiction} from the unary connectives, two identities ($\mathrm{Id}_a$) and ($\mathrm{Id}_b$) and the two negations $(\neg a)$ and $(\neg b)$.  In addition we have the \textit{and} ($\wedge$), or \textit{conjunction}, and \textit{or} ($\vee$), or \textit{disjunction} connectives, and their negations, ($ \neg \wedge \Leftrightarrow\uparrow$)\footnote{The negated and, or \textit{nand}, operator is known to be \textit{functionally complete} \cite{Peirce,Sheffer}, meaning it can be used to generate all other $n$-ary operators and is often used as the basis for logic gates in digital hardware.  The symbol $\uparrow$ is called the \textit{Sheffer stroke}, named after Sheffer who first showed the completeness of nand \cite{Sheffer}.  The negated \textit{or}, or \textit{nor}, connective can also be used to generate the entire set of connectives and is sometimes called the Pierce stroke after Charles S. Pierce.} and ($\neg \vee \Leftrightarrow\downarrow$).  There is also \textit{exclusive or} ($\dot{\vee}$), implication arrows, ($\Rightarrow$) and ($\Leftarrow$), their negations, ($\notright \Leftrightarrow \neg \Rightarrow$) and ($\notleft\Leftrightarrow \neg \Leftarrow$), and finally the equivalence ($\Leftrightarrow$), or ``if and only if.'' Their truth values are given in the following table,

\begin{table}[H]
	\makebox[\textwidth][c]{
		\begin{tabular}{|c|c||c|c|c|c|c|c|c|c|c|c|c|c|c|c|c|c|}
			\hline
			$a$ & $b$ & $\top$ & $\bot$ & $\mathrm{Id}_a$ & $\mathrm{Id}_b$ & $\neg a$& $\neg b$& $\wedge$& $\vee$& $\uparrow$& $\downarrow$&$\dot{\vee}$&$\Leftrightarrow$&$\Rightarrow$&$\Leftarrow$&$\notright$&$\notleft$\\
			\hline
			\hline
			$\tru$ & $\tru$ &$\tru$&$\fals$&$\tru$&$\tru$&$\fals$&$\fals$&$\tru$&$\tru$&$\fals$&$\fals$&$\fals$&$\tru$&$\tru$&$\tru$&$\fals$&$\fals$\\
			$\tru$ & $\fals$ &$\tru$&$\fals$&$\tru$&$\fals$&$\fals$&$\tru$&$\fals$&$\tru$&$\tru$&$\fals$&$\tru$&$\fals$&$\fals$&$\tru$&$\tru$&$\fals$\\
			$\fals$ & $\tru$ &$\tru$&$\fals$&$\fals$&$\tru$&$\tru$&$\fals$&$\fals$&$\tru$&$\tru$&$\fals$&$\tru$&$\fals$&$\tru$&$\fals$&$\fals$&$\tru$\\
			$\fals$ & $\fals$ &$\tru$&$\fals$&$\fals$&$\fals$&$\tru$&$\tru$&$\fals$&$\fals$&$\tru$&$\tru$&$\fals$&$\tru$&$\tru$&$\tru$&$\fals$&$\fals$\\
			\hline
	\end{tabular}}
	\caption{The $2^4 = 16$ binary operators acting on two propositions $a$ and $b$.}
	\label{binaryconnectives}
\end{table}

Often in applications of Boolean algebra, the disjunction is written as a plus sign, $a\vee b \rightarrow a$``$+$''$b$, conjunction is omitted, $a\wedge b \rightarrow ab$, and the equivalence connective is written as an equals sign, $a\Leftrightarrow b \rightarrow a$``$=$''$b$.  We shall not use this notation however, since it leads to confusion with the algebra over fields.  

If the statement $a \Leftrightarrow b$ is true, then the two propositions $a$ and $b$ are \textit{equivalent in their truth value}, but they are not necessarily ``equal'' as statements.  To say that two propositions are ``equal'' is perhaps misleading.  An example of this can best be seen in the following,
\begin{equation}
a \Leftrightarrow (a \Leftrightarrow \mathbf{T}).\label{aatrue}
\end{equation}
Translated into English, this means that any statement such as ``a is true,'' is \textit{equivalent} to the proposition ``a.'' This may seem strange, since they are different statements, but the key here is that they are only equivalent with respect to their ``truth'' values, thus the use of the equivalence connective $\Leftrightarrow$ in (\ref{aatrue}) rather than an equals sign.   
\paragraph{Functional completeness ---}The set of connectives $\{\wedge,\vee,\neg\}$ are \textit{functionally complete} \cite{Post}, which means that they act as generators for all other $n$-ary connectives on the bivalent valued propositions, i.e. they provide a \textit{basis} for the Boolean algebra of propositions.  This can be seen through the following set of definitions,
\paragraph{The connectives $\wedge$,$\vee$ and $\neg$ as generators}
\begin{enumerate}
	\item \textbf{Not and} - $\forall a,b \in \mathcal{A}: a \uparrow b \Leftrightarrow \neg(a \wedge b)$.
	\item \textbf{Not or} - $\forall a,b \in \mathcal{A}: a\downarrow b \Leftrightarrow \neg(a \vee b)$.
	\item \textbf{Implication} - $\forall a,b\in \mathcal{A}: a \Rightarrow b \Leftrightarrow \neg a \vee b$.
	\item \textbf{Exclusive or} - $\forall a,b \in \mathcal{A}: a\dot{\vee} b \Leftrightarrow (a\vee b)\wedge \neg(a \wedge b)$.
	\item \textbf{Equivalence} - $\forall a,b \in \mathcal{A}: (a \Leftrightarrow b) \Leftrightarrow \neg(a \vee b) \vee (a \wedge b)$.
\end{enumerate}
While the set $\{\wedge,\vee,\neg\}$ can generate all other $n$-ary connectives, it is not \textit{minimal}, since \textit{or} can be written using $\wedge$ and $\neg$, and similarly \textit{and} can written in terms of $\vee$ and $\neg$ as in De Morgan's laws from (\ref{table3}),
\begin{align}
\forall a,b \in \mathcal{A} : a \vee b &\Leftrightarrow \neg(\neg a \wedge \neg b),\label{orand}\\
\forall a,b \in \mathcal{A} : a \wedge b &\Leftrightarrow \neg(\neg a \vee \neg b).\label{andor} 
\end{align}
This establishes the correspondence,
\begin{figure}[H]
	\centering
	\begin{tikzcd}
	\mathcal{A}\times\mathcal{A} \arrow[r,"{(\neg,\neg)}", shift left=.5ex] \arrow[d,"\wedge"']& \mathcal{A}\times\mathcal{A} \arrow[d,"\vee"]\arrow[l,shift left=.5ex]\\
	\mathcal{A}\arrow[r,shift left=.5ex]& \mathcal{A}\arrow[l,"\neg",shift left=.5ex]
	\end{tikzcd}
	\caption{Commutative diagram between, $\neg \circ \wedge\Leftrightarrow \vee\circ(\neg,\neg)$, and, $\neg\circ \vee\Leftrightarrow\wedge\circ(\neg,\neg)$.}
\end{figure}
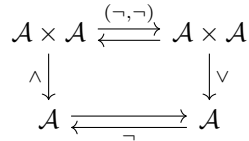
The most minimal functionally complete connectives are the \textit{not and} and \textit{not or} operators.  Typically, \textit{not and}, $\uparrow$, is used in the study of digital circuits, from which the three connectives $\{\wedge,\vee,\neg\}$ can be defined as follows,
\paragraph{The connective $\uparrow$ as a minimal generator}
\begin{enumerate}
	\item \textbf{Negation} - $\forall a \in \mathcal{A} : \neg a \Leftrightarrow a \uparrow a$.
	\item \textbf{And} - $\forall a,b\in \mathcal{A}: a\wedge b \Leftrightarrow \neg (a \uparrow b) \Leftrightarrow (a \uparrow b) \uparrow (a \uparrow b)$.
	\item \textbf{Or} - $\forall a,b \in \mathcal{A}: a \vee b \Leftrightarrow (a \uparrow a) \uparrow (b \uparrow b)$.
\end{enumerate}
Once one has these three generators, one can define the entire set of connectives, and hence $\uparrow$ is functionally complete.  A similar set of connectives holds for the \textit{not or} connective, $\downarrow$, which one gets by replacing the $\uparrow$ in the above list with $\downarrow$ and $\wedge$ with $\vee$.  Another popular basis is the \textit{ring basis}, which consists of the connectives $\{\wedge,\dot{\vee}\}$ \cite{Stone,GivantHalmos}.  The sets of functionally complete connectives were characterized by E. Post \cite{Post}.

\subsection{The algebra of propositions $\mathcal{I}$}\label{rulesofinference}
We will now turn our attention to constructing a set of axioms $\mathcal{I}$ for our logical system.  From the properties of the connectives (\ref{unary}) and (\ref{binaryconnectives}), one can write the laws of logic as a set of \textit{tautologies},
\begin{flalign}
\text{(Law of identity)}&& (a \Leftrightarrow a) &&\\
\text{(Law of non-contradiction)}&& \neg (a \wedge \neg a) &&\\
\text{(Law of excluded middle)}&& (a\vee \neg a)&&\label{laws}
\end{flalign}
In our formal system, the negation operator satisfies the useful \textit{double negative tautology}\footnote{The double negation equivalence is not accepted in intuitionist logic.},
\begin{equation}
\neg(\neg a) \Leftrightarrow a.
\end{equation}
Algebraically, the operators $\wedge$, $\vee$ and $\neg$ have several interesting properties.  These properties will be exploited in the construction of probability theory and so it will be useful to highlight them here.  

One can form a \textit{unital} algebra from the space of propositions, which is identified as a \textit{Boolean} algebra \cite{Bool}.  The \textit{or} relation $\vee$ acts as the sum on $\mathcal{A}$, and the \textit{and} relation, $\wedge$, acts as the product.  Absolute false, $\fals$, is the additive identity, while absolute truth, $\tru$, is the multiplicative identity.  The following list highlights the algebraic properties of $\wedge$ and $\vee$ which are considered as axioms in $\mathcal{I}$,

\paragraph{Algebraic properties of $\wedge$ and $\vee$ ---}
The conjunction and disjunction have the following algebraic properties,
\begin{enumerate}\label{algebraicproposition}
	\item \textbf{Commutativity of $\vee$} - The \textit{or} relation is commutative,
	\begin{equation}
	\forall a,b \in \mathcal{A} : a \vee b \Leftrightarrow b \vee a.\label{commutativeor}
	\end{equation}
	\item \textbf{Associativity of $\vee$} - The \textit{or} relation is associative,
	\begin{equation}
	\forall a,b,c \in \mathcal{A} : a\vee(b\vee c) \Leftrightarrow (a\vee b)\vee c \Leftrightarrow a\vee b \vee c.\label{associativeor}
	\end{equation}
	\item \textbf{Additive identity} - The \textit{or} relation has an additive identity, which is the absolute false $\fals$,
	\begin{equation}
	\forall a \in \mathcal{A} : a\vee \fals \Leftrightarrow a.\label{unitalor}
	\end{equation}
	\item \textbf{Multiplicative identity} - The \textit{and} relation has a multiplicative identity,which is the absolute true $\tru$,
	\begin{equation}
	\forall a \in \mathcal{A} : a \wedge \tru \Leftrightarrow a.\label{unitaland}
	\end{equation}
	\item \textbf{Commutativity of $\wedge$} - The \textit{and} relation is commutative,
	\begin{equation}
	\forall a,b \in \mathcal{A} : a \wedge b \Leftrightarrow b \wedge a.\label{commutativeand}
	\end{equation}
	\item \textbf{Associativity of $\wedge$} - The \textit{and} relation is associative,
	\begin{equation}
	\forall a,b,c \in \mathcal{A}: a\wedge (b \wedge c) \Leftrightarrow (a \wedge b)\wedge c \Leftrightarrow a \wedge b \wedge c.\label{associativeand}
	\end{equation}
	\item \textbf{Distributivity of $\vee$ over $\wedge$} - The \textit{or} relation distributes over the \textit{and} relation,
	\begin{equation}
	\forall a,b,c \in \mathcal{A} : a \vee (b \wedge c) \Leftrightarrow (a \vee b)\wedge (a \vee c).\label{distributiveone}
	\end{equation}
	\item \textbf{Distributivity of $\wedge$ over $\vee$} - The \textit{and} relation distributes over the \textit{or} relation,
	\begin{equation}
	\forall a,b,c \in \mathcal{A} : a\wedge (b\vee c) \Leftrightarrow (a \wedge b) \vee (a \wedge c).\label{distributivetwo}
	\end{equation}
\end{enumerate}

There are several other interesting algebraic relations between $\wedge$, $\vee$ and the use of the negation relation $\neg$, which are highlighted in the following table of axioms,
\begin{table}[H]
	\makebox[\textwidth][c]{
		\begin{tabular}{|c|c|c|}
			\hline
			Name & Proposition & Dual\\
			\hline
			\hline
			Idempotence &$a \wedge a \Leftrightarrow a$ & $a \vee a \Leftrightarrow a$\\ 
			De Morgan's Laws&$\neg (a \wedge b) \Leftrightarrow \neg a \vee \neg b$ & $\neg (a \vee b) \Leftrightarrow \neg a \wedge \neg b$\\
			Absorption&$(a \wedge b)\vee b \Leftrightarrow b$ & $(a \vee b)\wedge b \Leftrightarrow b$\\
			Exclusivity&$(a \dot{\vee} b) \Leftrightarrow (a \vee b)\wedge \neg(a \wedge b)$ & $(a \Leftrightarrow b) \Leftrightarrow (a \wedge b)\vee \neg(a \vee b)$\\
			\hline
	\end{tabular}}
	\caption{Various consequences of the algebra of propositions.  The left and right columns are dual in the sense of replacing all $\wedge$'s with $\vee$'s.  Cox also provides a nice summary of relations in \cite{Cox}.}
	\label{table3}
\end{table}
Together, the space of propositions $\mathcal{A}$ with the connectives $\{\wedge,\vee,\neg\}$ form the algebra of propositions $(\mathcal{A},\wedge,\vee,\neg)$ whose properties were listed in (\ref{algebraicproposition}).  Since the space of propositions $\mathcal{A}$ is isomorphic to $\mathcal{W}$ (\ref{setofsubsetsworlds}), the algebra of propositions can be brought into correspondence with the algebra of sets over $\mathcal{W}$,
\begin{equation}
(\mathcal{A},\wedge,\vee,\neg) \cong_{\mathrm{hom}} (\mathcal{W},\cap,\cup,\backslash \mathcal{W}),
\end{equation}
where $\cap$ is the intersection, $\cup$ is the union and $\backslash \mathcal{W}$ is the complement.  This identifies $\mathbf{F}$ with the empty set $\emptyset$ and $\mathbf{T}$ with the entire set $\mathcal{W}$\footnote{Given the set of worlds $W_a$ and $W_b$ in which $a$ and $b$ are true respectively, it is easy to show that,
	\begin{align}
	W_{a\wedge b} &= W_a\cap W_b\\
	W_{a\vee b} &= W_a\cup W_b.
	\end{align}   
	This result is only possible when one adopts the formalism of possible worlds described in \hyperref[propositions]{Section 2}.  If one allows for a definition of worlds that only contain subsets of $\mathcal{A}$, then the correspondence cannot be made since $\mathcal{W}$ and $\mathcal{A}$ will not be one-to-one as sets.}.  The homomorphism established here is an example of \textit{Stone's theorem} \cite{Stone}, which says that any Boolean algebra is isomorphic to a field of sets.  The field of sets in particular is a \textit{compact} and \textit{totally disconnected Hausdorff} space, which is called a \textit{Stone space}.  
\begin{theorem}[Stone's representation theorem]\label{stone}
	Every Boolean algebra is isomorphic to the dual algebra of its associated Stone space. 
\end{theorem}

\subsection{The material implication connective $\Rightarrow$}\label{materialimplication}
The material implication connective\footnote{Often the material implication connective is referred to as a conditional, however we will avoid this language here as to not cause confusion by what we mean by conditionals in \hyperref[conditionallogic]{Section 9}.} $\Rightarrow$ is often overlooked in Boolean algebra in favor of studying the \textit{and} and \textit{or} relations.  It is much more powerful however, since it leads to several important consequences in logical inference \cite{Wittgenstein,Frege}.  We will highlight some of those ideas here.  First, the material implication connective can be written in terms of \textit{or} (or \textit{and}) and negation,
\begin{equation}
\forall a,b \in \mathcal{A} : a \Rightarrow b \Leftrightarrow \neg a \vee b \Leftrightarrow \neg(a \wedge \neg b).\label{implicationandor}
\end{equation}
The semantics for $a\Rightarrow b$ is usually written as ``if $a$, then $b$.''  In general it obeys the following axioms,
\paragraph{Algebraic properties of $\Rightarrow$}
\begin{enumerate}\label{algebraicimplication}
	\item \textbf{Distributivity of $\Rightarrow$} - The material implication connective distributes over itself,
	\begin{equation}
	\forall a,b,c \in \mathcal{A}: \left\{\left[a \Rightarrow (b \Rightarrow c)\right] \Rightarrow \left[(a \Rightarrow b) \Rightarrow (a \Rightarrow c)\right]\right\}.\label{distributiveimplication}
	\end{equation}
	\item \textbf{Transitivity of $\Rightarrow$} - The material implication connective is transitive,
	\begin{equation}
	\forall a,b,c \in \mathcal{A}:\left\{\left[(a\Rightarrow b)\Rightarrow(b \Rightarrow c)\right] \Rightarrow (a\Rightarrow c)\right\}.\label{transitiveimplication}
	\end{equation}
	\item \textbf{Reflexivity of $\Rightarrow$} - The material implication connective is reflexive,
	\begin{equation}
	\forall a \in \mathcal{A} : (a\Rightarrow a).
	\end{equation}
	\item \textbf{Totality of $\Rightarrow$} - The material implication connective obeys the totality condition,
	\begin{equation}
	\forall a,b \in \mathcal{A} : \left\{(a\Rightarrow b)\vee(b\Rightarrow a)\right\}.
	\end{equation}
\end{enumerate}
Interestingly enough, the propositions $a\wedge (a\Rightarrow b)$, or $b\wedge(b\Rightarrow a)$, contain the equivalent information contained in the proposition $a\wedge b$,
\begin{align}
a \wedge (a \Rightarrow b) &\Leftrightarrow a \wedge (\neg a \vee b)\nonumber\\
&\Leftrightarrow (a \wedge \neg a)\vee (a \wedge b) \Leftrightarrow a\wedge b.
\end{align}
Below is the corresponding truth table for $a \wedge (a \Rightarrow b)$ and $b \wedge (b\Rightarrow a)$,
\begin{table}[H]
	\makebox[\textwidth][c]{
		\begin{tabular}{|c|c|c|c|c|c|c|}
			\hline
			$a$ & $b$ & $a\Rightarrow b$ & $b\Rightarrow a$ & $a\wedge (a\Rightarrow b)$ & $b\wedge (b\Rightarrow a)$ & $a\wedge b$\\
			\hline
			\hline
			$\tru$ & $\tru$ & $\tru$ & $\tru$ & $\tru$ &$\tru$& $\tru$\\
			$\tru$ & $\fals$ & $\fals$ & $\tru$ & $\fals$& $\fals$ & $\fals$\\
			$\fals$ & $\tru$ & $\tru$ & $\fals$ & $\fals$& $\fals$ & $\fals$\\
			$\fals$ & $\fals$ & $\fals$ & $\tru$ & $\fals$& $\fals$ & $\fals$\\
			\hline
	\end{tabular}}
	\caption{Conditional relations between propositions $a$ and $b$.}
\end{table}
A theorem which will be important later is the following,
\begin{theorem}\label{theorem1}
	Let $\mathcal{A}$ be the space of propositions.  For any propositions, $a,b,c \in \mathcal{A}$, the distributive property (\ref{distributiveimplication}) of the material implication connective (\ref{implicationandor}) can be written as, 
	\begin{align}
	\forall a,b,c \in \mathcal{A}: [a \Rightarrow (b \Rightarrow c)] &\Leftrightarrow (a \wedge b) \Rightarrow c.
	\end{align}
\end{theorem}
\begin{proof}
	This is easily proved using (\ref{implicationandor}) and (\ref{orand}),
	\begin{align}
	[a \Rightarrow (b \Rightarrow c)] &\Leftrightarrow a \Rightarrow (\neg b \vee c)\nonumber\\
	&\Leftrightarrow \neg a \vee (\neg b \vee c)\nonumber\\
	&\Leftrightarrow (\neg a \vee \neg b)\vee c\nonumber\\
	&\Leftrightarrow \neg (a \wedge b) \vee c \Leftrightarrow [(a \wedge b)\Rightarrow c].
	\end{align}
\end{proof}
Given the commutative nature of the \textit{and} relation, the implications can also be rearranged,
\begin{equation}
\forall a,b,c\in\mathcal{A}: [a\Rightarrow (b\Rightarrow c)] \Leftrightarrow [b\Rightarrow (a\Rightarrow c)].
\end{equation}
From \hyperref[theorem1]{theorem 1} we can write a more general result,
\begin{theorem}\label{theorem2}
	Let $\mathcal{A}$ be the space of propositions.  Let $a_1,a_2,\dots,a_N$ be a sequence of propositions.  The nested implication over the sequence $\{a_i\}_{i=1}^N$ is equivalent to,
	\begin{equation}
	[a_1\Rightarrow(a_2\Rightarrow \dots \Rightarrow (a_{N-1}\Rightarrow a_N) \cdots )] \Leftrightarrow [(a_1\wedge\dots\wedge a_{N-1})\Rightarrow a_N].
	\end{equation}
\end{theorem}
\begin{proof}
	From \hyperref[theorem1]{theorem 1} we know that the three inner most terms in the sequence satisfy,
	\begin{equation}
	a_{N-2}\Rightarrow (a_{N-1} \Rightarrow a_N) \Leftrightarrow (a_{N-2}\wedge a_{N-1})\Rightarrow a_{N}.\label{eq1}
	\end{equation}
	Consider the next implication, $a_{N-3}$, outside of these three.  Due to the above (\ref{eq1}), it is true that,
	\begin{equation}
	a_{N_3} \Rightarrow (a_{N-2}\Rightarrow (a_{N-1} \Rightarrow a_N)) \Leftrightarrow a_{N-3}\Rightarrow \left[(a_{N-2}\wedge a_{N-1})\Rightarrow a_{N}\right].
	\end{equation}
	The right hand side however has the same form, $a \Rightarrow (b \Rightarrow c)$, from (\ref{eq1}) since $a_{N-2}\wedge a_{N-1}$ is also a proposition.  Thus it must also be true that,
	\begin{align}
	a_{N_3} \Rightarrow (a_{N-2}\Rightarrow (a_{N-1} \Rightarrow a_N)) &\Leftrightarrow a_{N-3}\Rightarrow \left[(a_{N-2}\wedge a_{N-1})\Rightarrow a_{N}\right]\nonumber\\
	&\Leftrightarrow \left[a_{N-3}\wedge(a_{N-2}\wedge a_{N-1})\Rightarrow a_{N}\right].
	\end{align}
	Due to the associative (\ref{associativeand}) and commutative (\ref{commutativeand}) properties of $\wedge$, by induction we have,
	\begin{equation}
	[a_1\Rightarrow(a_2\Rightarrow \dots \Rightarrow (a_{N-1}\Rightarrow a_N) \cdots )] \Leftrightarrow [(a_1\wedge\dots\wedge a_{N-1})\Rightarrow a_N].
	\end{equation}
\end{proof}
McGee \cite{McGee} calls this property the \textit{import-export principle}.  The material implication defines a partial order on $\mathcal{A}$, which we will discuss in \hyperref[partialordersection]{Section 10.1}.

\paragraph{Principle of explosion ---} The material implication connective, $\Rightarrow$, leads to an important principle called \textit{the principle of explosion} \cite{LewisLangford}, or \textit{ex falso quodlibet}, which means ``from falsehood, anything can be proven.''  To see this, consider that the contradiction $(a \wedge \neg a)$ is true.  Then, using \textit{simplification}\footnote{\textit{Simplifcation} is the ability to pull individual propositions out of $a\wedge b$ and identify them as true whenever the proposition $a\wedge b$ is true itself.}, $a$ and $\neg a$ are also true on their own.  Then, given any statement $b$, the relation $(a \vee b)$ must be true, due to \textit{disjuntive introduction}\footnote{\textit{Disjuntive introduction} is the consequence that, given any true statement $a$, one can always add any other statement $b$ through disjunction, $a\vee b$, and the resulting proposition will also be true, regardless if $b$ is true or false}.  But, because $\neg a$ we have the \textit{disjuntive syllogism}\footnote{\textit{Disjunctive syllogism}, or \textit{modus tollendo ponens}, is a consequence of the axiom of excluded middle, which says that if one has a disjunction, $a \vee b$, and $\neg a$ is known to be true, then $b$ must be true, otherwise $a\vee b$ is false.},
\begin{equation}
\forall a,b \in \mathcal{A} : (a \vee b) \wedge \neg a \Leftrightarrow b.
\end{equation}  
Therefore, from the contradiction $(a\wedge\neg a)$, we've established that any $b$ is true, i.e. $(a\wedge \neg a)\Rightarrow b$.  The term \textit{explosion} signifies the fact that, once one contradiction is allowed in $\mathcal{A}$, anything can be proven to be true.

\subsection{Deductive inference rules $\mathcal{Z}$}\label{inferencerules}
In our simple system there are at most $24$ valid inference rules in $\mathcal{Z}$ which are distinct \cite{Aristotle}.  These are often called \textit{syllogisms} which constitute the basic form of a logical argument that were first organized by Aristotle.  There are $256$ distinct syllogisms in total, meaning that $232$ of them are \textit{formal fallacies}, which constitute invalid forms of reasoning.  Below is a table of some commonly used rules of inference,
\begin{table}[H]
	\makebox[\textwidth][c]{
		\begin{tabular}{|c|l|c|}
			\hline
			Connective &Inference Rule & Semantics ($\forall a,b,c \in \mathcal{A}$)\\
			\hline
			\hline
			\multirow{2}{*}{$\Rightarrow$} &Modus ponens & $a,a\Rightarrow b \vdash b$\\
			&Modus tollens & $\neg b, a\Rightarrow b \vdash \neg a$\\
			\hline
			\multirow{3}{*}{$\wedge$} &Conjunction introduction & $a,b \vdash a\wedge b$\\
			&Conjunction elimination & $a \vdash a \vee b$\\
			&Simplification & $a\wedge b,\vdash a$\\
			\hline
			\multirow{3}{*}{$\vee$} &Disjunction introduction & $a \vdash a \vee b$\\
			&Disjunction elimination & $a\Rightarrow c, b \Rightarrow c, a \vee b \vdash c$\\
			&Modus tollendo ponens & $a \vee b, \neg a \vdash b$\\
			\hline
			$\Rightarrow,\vee$&Constructive dilemma & $a\Rightarrow b,c\Rightarrow d,a\vee c \vdash b \vee d$\\
			\hline
	\end{tabular}}
	\caption{Common inference rules for propositional calculus}
\end{table}
These and other inference rules allow one to construct proofs for stated theorems.  Proofs are examples of logical arguments, which are a sequence of propositions $\Gamma = \{a_1,a_2,\dots,a_n\}$, such that each proposition $a_i \in \Gamma$ is either an axiom from $\mathcal{I}$, or it is a tautology, such as the laws (\ref{laws}), or it is an instance of \textit{modus ponens}.

\paragraph{Implication elimination ---}\label{modusponens} The material implication connective gives rise to a standard rule of inference called \textit{modus ponens} (more completely, \textit{modus ponendo ponens}), or ``mode that by affirming affirms'' \cite{Barwise}.  This is essentially the statement that if a proposition $a$ implies $b$, i.e. $a \Rightarrow b$, and $a$ is true, then $b$ must also be true.  A standard form of the argument is as follows,
\begin{quote}
	If Werner is a physicist, then he understands quantum mechanics.\\
	Werner is a physicist.\\
	Therefore, Werner understands quantum mechanics.
\end{quote}     
Identifying ``Werner is a physicst'' as ``$a$'' and ``Werner understands quantum mechanics'' as ``$b$'', the semantic version of the above is written,
\begin{flalign}
\mathbf{MP}&&\forall a,b \in \mathcal{A} : a \Rightarrow b, a\vdash b&&.
\end{flalign}

\paragraph{Reductio ad absurdum ---} From the properties of the implication arrow we can write the following theorem, which is called formally a \textit{contraposition} \cite{Barwise},
\begin{equation}
\forall a,b \in \mathcal{A} : (a \Rightarrow b) \Leftrightarrow (\neg b \Rightarrow \neg a).
\end{equation}
This situation allows us to prove propositions by contradiction, which is a special case of what is known as \textit{reductio ad absurdum}, or ``argument from absurdity.''  If we first assume that $a$ is true, then the statement $a\Rightarrow b$ is only true whenever $b$ is true as well.  Instead, if we assume that $b$ is false, and this leads to $\neg b \Rightarrow \neg a$ being false, then it must be the case that $b$ is actually true.

\paragraph{Sequents as tautologies ---}
Proofs, such as $\Gamma \vdash b$, can be written as tautologies by first writing the antecedent as a conjunction,
\begin{equation}
\Gamma \Leftrightarrow a_1\wedge a_2\wedge\dots\wedge a_n,
\end{equation}
which together with a theorem $b$ forms the statement $\Gamma \Rightarrow b$.  These statements can be collected into the set of theorems for a given proposition space $\mathcal{O}$.  

\paragraph{Summary of propositional calculus ---}
In this section we have briefly explored the structure $\mathcal{D} = \left\{\mathcal{A},\Omega,\mathcal{I},\mathcal{Z}\right\}$ that we sought to define in the introduction.  Additional structure such as the set of all possible worlds $W$ will be appended to a set $\mathcal{T}$ which can then be added to $\mathcal{D} = \left\{\mathcal{A},\Omega,\mathcal{I},\mathcal{Z},\mathcal{T}\right\}$.  Proofs are special propositions which exist as a collection of subsets of $\mathcal{A}$.

\subsection{First-order predicate logic}\label{predicatelogic}
While the algebra of propositions serves as a useful foundation, in physics we are often interested in the use of \textit{predicates}.  First-order predicate logic introduces the notion of a predicate, together with two \textit{quantors}, $\forall$ and $\exists$.  It was first formally developed by Frege in his \textit{Begriffsschrift} \cite{Frege}.  It was also independently developed by Pierce \cite{Pierce}.  The concepts of quantification are closely related with work that goes back to Aristotle in his \textit{De interpretatione} \cite{Aristotle}.  

A predicate is defined as,
\begin{define}\label{predicate}
	A predicate, $A(x)$, is a proposition valued function for some variable $x$.
\end{define}
In the broadest sense, the variable $x$ can be assigned to any element of the space of words $\Sigma^*$.  Likewise, a predicate is also an element of the space of words $A(x) \in \Sigma^*$, in which the variable $x$ appears as part of the word.  We can then identify the \textit{space of predicates} as some suitable subspace of the space of words, $\mathcal{P} \subset \Sigma^*$.

It is important to note that a predicate such as $A(x)$ is \textit{not} a proposition.  It only becomes a proposition once the argument $x$ has been specified, otherwise $x$ remains \textit{free} in $A(x)$ \cite{Tomassi}.  For example, the predicate $A(x)$ could be, ``$x$ is mortal.''  Depending on what we substitute for $x$ the predicate $A(x)$ will inherit a meaning.  If for $x$ we put ``Socrates,'' then $A(x)$ becomes the proposition ``Socrates is mortal.''  If however we put for $x$, ``the Heisenberg uncertainty principle,'' then $A(x)$ is nonsense.  Thus, there is a range of applicability of $A(x)$ with respect to $x$ \cite{LewisLangford}.

Besides inserting a singular value for $x$ in $A(x)$, there are several other ways in which one can construct a proposition from a predicate.  Two of them involve the use of the \textit{universal} and \textit{existential} quantors, $\forall$ and $\exists$, that were introduced in the algebraic axioms for Boolean algebra.  The universal quantor can be defined as,
\begin{define}\label{universal}
	Let $A(x)$ be a predicate for the variable $x$.  The universal quantor $\forall$ is defined as,
	\begin{equation}
	\forall : \Sigma^*\times\mathcal{P} \rightarrow \mathcal{A}.
	\end{equation}
	It does this by quantifying the variable $x$,
	\begin{equation}
	\forall x : A(x) \Leftrightarrow [A(a)\wedge A(b)\wedge A(c) \wedge\dots],\label{forallwedge}
	\end{equation}
	which reads, ``for all $x$, $A(x)$ is true,'' i.e. $A$ is true independent of $x$.  
\end{define}
Likewise, the existential quantor $\exists$ is defined as,
\begin{define}\label{existential}
	Let $A(x)$ be a predicate for the variable $x$.  The existential quantor $\exists$ is defined as,
	\begin{equation}
	\exists : \Sigma^*\times\mathcal{P}\rightarrow \mathcal{A}.
	\end{equation}
	We can quantify $x$ through the statement,
	\begin{equation}
	\exists x : A(x) \Leftrightarrow [A(a)\vee A(b) \vee A(c) \vee\dots],\label{existsvee}
	\end{equation}
	which reads, ``there exists at least one $x$, such that $A(x)$ is true.''
\end{define}
Like propositions, one can generate new predicates from old ones using the binary relations in \hyperref[binaryconnectives]{table 2}.  From these relations, the existential quantor can be defined in terms of the universal quantor as,
\begin{equation}
\exists x: A(x) \Leftrightarrow \neg\left(\forall x : \neg A(x)\right).
\end{equation}
The universal quantor is then,
\begin{equation}
\forall x: \neg A(x) \Leftrightarrow \neg\left(\exists x : A(x)\right).
\end{equation}
One can also define the uniqueness quantor $\exists!$, which when applied to a predicate like, $\exists !x : A(x)$ reads ``there exists one and only one $x$ such that $A(x)$ is true.''  The uniqueness quantor can be expressed as,
\begin{equation}
\exists !x: A(x) \Leftrightarrow \exists x [A(x) \wedge \neg \exists y : [A(y) \wedge y \neq x ]],
\end{equation}
i.e., ``there exists an $x$ such that $A(x)$ is true and there exists no $y$ such that $A(y)$ is true when $y$ is not equal to $x$.''

The formalism can easily be extended to predicates of several variables, $A(x_1,x_2,\dots,x_n)$.  For predicates of several variables, such as $A(x,y)$, one can generate a new predicate through the use of the quantors.  For example, the predicate $A(x)$ can be generated from $A(x,y)$ by specifying $y$ through either $\forall$ or $\exists$,
\begin{align}
A(x) \Leftrightarrow [\exists y : A(x,y)] \Leftrightarrow [A(x,a)\vee A(x,b) \vee A(x,c) \vee\dots].
\end{align}
Here the variable $y$ is called \textit{bound}, while the $x$ remains \textit{free}.  In practice one should be careful with the order of the use of quantors, since in general,
\begin{equation}
\forall x : \exists y : A(x,y) \not\Leftrightarrow \exists y : \forall x : A(x,y).
\end{equation}
The above propositions are necessarily different since,
\begin{align}
\forall x : \exists y : A(x,y) \Leftrightarrow & [A(a,a)\vee A(a,b)\vee A(a,c)\vee\dots]\nonumber\\
\wedge &\, [A(b,a)\vee A(b,b)\vee A(b,c)\vee\dots]\nonumber\\
\wedge &\, [A(c,a)\vee A(c,b)\vee A(c,c)\vee\dots]\nonumber\\
\wedge &\, \dots,
\end{align}  
where as,
\begin{align}
\exists y : \forall x : A(x,y) \Leftrightarrow & [A(a,a)\wedge A(b,a)\wedge A(c,a)\wedge\dots]\nonumber\\
\vee &\, [A(a,b)\wedge A(b,b)\wedge A(c,b)\wedge\dots]\nonumber\\
\vee &\, [A(a,c)\wedge A(b,c)\wedge A(c,c)\wedge\dots]\nonumber\\
\vee &\, \dots,
\end{align}  
There is however an implication given by switching the quantor,
\begin{equation}
\exists y : \forall x : A(x,y) \Rightarrow \forall y : \exists x : A(x,y).
\end{equation}
This can easily be proved by using the form of (\ref{forallwedge}) and (\ref{existsvee}).  There are also the useful axioms involving conjunction and disjunction,
\begin{flalign}
\mathbf{PC}&&\forall x : A(x) \wedge B(x) &\Leftrightarrow [\forall x : A(x)] \wedge [\forall x : B(x)]&&\label{predicateconj}\\
\mathbf{PD}&&[\forall x : A(x)] \vee [\forall x : B(x)] &\Rightarrow \forall x : A(x) \vee B(x).&&\label{predicatedis}
\end{flalign}
Given the definition of a predicate (\ref{predicate}), the quantors (\ref{universal}), (\ref{existential}), and any axioms such a $\mathbf{PC}$ and $\mathbf{PD}$, we establish what is called a \textit{first-order} predicate logic.  It is of first-order since it does not concern predicate valued predicates, which would be a predicate logic of \textit{second-order}.  One can construct higher-order predicate logics from first-order ones \cite{Russell}, but we will not discuss this further.

One can also adopt a set of new inference rules concerning the quantors.  Some examples of these are the following,
\begin{table}[H]
	\makebox[\textwidth][c]{
		\begin{tabular}{|c|l|c|}
			\hline
			Quantor &Inference Rule & Semantics ($\forall A(x),B(x),C(x) \in \mathcal{P}$)\\
			\hline
			\hline
			\multirow{2}{*}{$\forall$} &Universal generalization & $A(a/x) \Rightarrow \forall x : A(x)$\\
			&Universal instantiation & $\forall x : A(x) \Rightarrow A(a/x)$\\
			\hline
			\multirow{2}{*}{$\exists$} &Existential generalization & $A(a/x) \Rightarrow \exists x : A(x)$\\
			&Existential instantiation & $\exists x : A(x) \Rightarrow A(a/x)$\\
			\hline
	\end{tabular}}
	\caption{Common inference rules for first-order logic.}
	\label{predicateinferencerules}
\end{table}
The notation $A(a/x)$ means to replace the free variable $x$ everywhere in $A(x)$ with $a$.  There are restrictions on how one must analyze these formulas due to the behavior of free and bound variables\footnote{For example, if one has the predicate $A(x,y)$ over the two variables $x$ and $y$, one cannot use an inference rule such as $\exists x : A(x,y) \Rightarrow A(a/x,y)$, and replace $a$ with some function of $y$.  This will leave open the possibility for contradiction.  In general, one cannot substitute free variables with bound ones when using these rules.} 
\paragraph{The summary thus far ---}\label{predicatesummary}
In a sense, the space of predicates, $\mathcal{P} \subset \Sigma^*$, together with the quantifiers $\forall$ and $\exists$ is more general than the space of propositions $\mathcal{A}$.  Thus, we can organize our predicate logic as an extension of the system $\mathcal{D}$ by incorporating the space of predicates, $\mathcal{P} \subset \Sigma^*$, together with the quantors $\mathcal{Q} = \{\forall,\exists\}$.  The axioms such as $\mathbf{PC}$ and $\mathbf{PD}$ can be collected as a set $\mathcal{I}_{\forall,\exists}$, which can be appended to our list of propositional axioms $\mathcal{I}$ from \hyperref[rulesofinference]{Section 4} and \hyperref[materialimplication]{2.5} so that $\mathcal{I}\rightarrow \mathcal{I}_p = \mathcal{I}\cup\mathcal{I}_{\forall,\exists}$.  Likewise the set of inference rules from \hyperref[inferencerules]{Section 6} can be appended with the set $\mathcal{Z}_{\forall,\exists}$, which could contain rules from (\ref{predicateinferencerules}), so that we define $\mathcal{Z}\rightarrow\mathcal{Z}_p=\mathcal{Z}\cup\mathcal{Z}_{\forall,\exists}$.  Finally, any additional structure introduced with respect to the predicate calculus we label $\mathcal{T}_{\forall,\exists}$, which can then be appended to the structures in $\mathcal{D}$, $\mathcal{T}\rightarrow\mathcal{T}_p = \mathcal{T}\cup\mathcal{T}_{\forall,\exists}$.  We call this system $\mathcal{D}_p = \left\{\mathcal{P},\mathcal{Q},\mathcal{A},\Omega,\mathcal{I}_p,\mathcal{Z}_p,\mathcal{T}_p\right\}$.  

The semantics that one could employ with this framework is rather limited.  In the next few sections we will attempt to extend the system $\mathcal{D}$ in order to incorporate a more complex semantics.  This involves the incorporation of \textit{modalities} as well as \textit{conditional logic}.

\subsection{Modal logic}\label{modallogic}
In this section we will briefly survey some of the basic properties of modal logic.  Modal logic \cite{LewisLangford,Goldblatt,ChellasBook} dates back to Aristotle, who studied the semantics associated to words like ``necessary'' and ``possible.''  This can be formalized into an extension of classical logic by identifying a \textit{modality}, which is a term used in linguistics that typically refers to the \textit{intention} of the speaker.  

A modality is a phrase that, when applied to a statement, generates a new statement that addresses the \textit{mode of truth} of the original statement.  Depending on the modality one is using, the subject can range from temporal logic \cite{OhrstromHasle} to ethics and morality (Deontic logic \cite{vonWright}) or to the logic concerning beliefs (Doxastic logic \cite{Smullyan}).  Each of these programs has in common a new set of logical operators which have the character of the \textit{universal} and \textit{existential} quantors, (\ref{universal}) and (\ref{existential}), and are given by the symbols $\square$ and $\diamondsuit$.

While the inference framework should ultimately be agnostic with respect to what modalities one is using, we will develop the theory of conditionals using the \textit{alethic} approach in which $\square$ stands for ``necessary'' and $\diamondsuit$ stands for ``possibly.''  In particular, the \textit{necessary} and \textit{possibly} modalities generate statements such as, $\square p$, which reads ``necessarily, $p$,'' and, $\diamondsuit p$, which reads ``possibly, $p$.''

There are several axiomatic systems that one can impose concerning the modal operators.  Most of these systems were developed by Lewis and Langford in their book ``Symbolic Logic'' \cite{LewisLangford}.  The most basic of these axioms identifies a dual relationship between $\square$ and $\diamondsuit$,

\begin{define}\label{modals}
	Let $\mathcal{A}$ be a space of propositions.  The modal operators $\square$ and $\diamondsuit$ represent the modalities ``necessary'' and ``possibly,'' and are dual to each other, i.e.,
	\begin{align}
	\forall a \in \mathcal{A} : \square a \Leftrightarrow \neg \diamondsuit \neg a\\
	\forall a \in \mathcal{A} : \diamondsuit a \Leftrightarrow \neg \square \neg a.
	\end{align}
\end{define}
According to the definition (\ref{modals}), the statement, ``the particle will \textit{necessarily} go through slit A or slit B,'' is identical to the statement, ``it is not \textit{possible} that the particle will not go through slit A or slit B.''  Likewise, the statement, ``it is possible that the particle will go through slit A,'' is identical to the statement, ``it is not necessarily true that the particle will not go through slit A.''  Both of these equivalences seem intuitive, and are hardly controversial.  As we construct more complicated modal logics however, the statements will become less intuitive, and perhaps controversial.

Modal theory incorporates the \hyperref[possibleworlds]{possible worlds semantics} through what are called \textit{accessibility relations}. An accessibility relation is a relation between two worlds, $w$ and $w'$ that obeys the following axioms,
\begin{define}\label{accessibilityrelation}
	Let $\mathcal{A}$ be a countable set of propositions and let $W$ be the set of possible worlds over $\mathcal{A}$.  An accessibility relation $R\subseteq W\times W\rightarrow W$, is a binary relation between two worlds such that, if $\square a$ is true in world $w_1$ and $a$ is true in $w_2$, then we say that $w_2$ is accessible to $w_1$, i.e., 
	\begin{equation}
	\forall a \in \mathcal{A} : \forall w_1 : [\nu(\square a,w_1) \Leftrightarrow \forall w_2 : R(w_1,w_2) \Rightarrow \nu(a,w_2)].  
	\end{equation}
\end{define}
The proposition $R(w_1,w_2)$ defines the statement ``$w_2$ is accessible to $w_1$,'' and is sometimes written as $R(w_1,w_2)\Leftrightarrow w_1Rw_2$.  Given the \textit{actual world}, denoted by $\tilde{w}$, the statement that a proposition $a \in \mathcal{A}$ is necessarily true, $\square a$, means that it is true in all possible worlds.  This can be written as,
\begin{equation}
\forall a \in \mathcal{A} : \square a \Leftrightarrow [\forall w \in W : \tilde{w}Rw\Rightarrow \nu(a,w)].
\end{equation}
Likewise, we say that $a$ is possible, $\diamondsuit a$, if there exists a world $w \in W$ such that,
\begin{equation}
\forall a \in \mathcal{A} : \diamondsuit a \Leftrightarrow [\exists w \in W : \tilde{w}Rw \Rightarrow \nu(a,w)].
\end{equation}
From here, different modal logics can be formed by constructing additional axioms on the behavior of the universal and existential operators $\square$ and $\diamondsuit$.  Some are more controversial than others among logicians and we will introduce them in an order of their suggested controversy.  Among the different families of modal logics are the ones first developed by Lewis \cite{LewisLangford}. 

The simplest modal logic is typically labeled by $\mathbf{K}$, which is named after Saul Kripke \cite{Kripke}.  Together with the duality axiom (\ref{modals}), it adopts the following additional axioms to standard propositional logic,
\begin{axiom}[The logic $K$]\label{K}
	Let $\mathcal{A}$ be the space of propositions and let $\square$ and $\diamondsuit$ be the modal operators defined in (\ref{modals}).  We call a modal logic of type $K$ if it obeys the following distribution axiom,
	\begin{flalign}
	\mathbf{K}&&\forall a,b, \in \mathcal{A}: \square(a \Rightarrow b) \Rightarrow (\square a \Rightarrow \square b),&&\label{aK}
	\end{flalign}
	together with the necessitation rule,
	\begin{equation}
	\vdash A \vDash \square A.
	\end{equation}
\end{axiom} 
The necessitation rule reads, ``if $A$ is a theorem, then $\square A$ is also a theorem.''  The proposition $\square(a \Rightarrow b)$ is called the \textit{strict implication}.  An additional set of axioms that are typically adopted are given in the following list,
\begin{axioms}[The axioms $\mathbf{D}$,$\mathbf{M}$,$\mathbf{S}_4$ and $\mathbf{S}_5$]\label{modalaxioms}
	Let $\mathcal{A}$ be the space of propositions and let $\square$ and $\diamondsuit$ be the modal operators defined in (\ref{modals}).  The following axioms are represented by the letters in the left-hand column,
	\begin{flalign}
	\mathbf{D}&&\forall a \in \mathcal{A} : \square a \Rightarrow \diamondsuit a.&&\label{aD}\\
	\mathbf{M}&&\forall a \in \mathcal{A}: \square a \Rightarrow a.&&\label{aM}\\
	\mathbf{S}_4&&\forall a \in \mathcal{A} : \square a \Rightarrow \square\square a.&&\label{aS4}\\
	\mathbf{S}_5&&\forall a \in \mathcal{A}: \diamondsuit a \Rightarrow \square\diamondsuit a.&&\label{aS5}\\
	\mathbf{B}&&\forall a \in \mathcal{A}: a \Rightarrow \square\diamondsuit a.&&\label{aB}
	\end{flalign}
\end{axioms}
Axiom $\mathbf{D}$ is hardly controversial, since for any proposition $a$ it implies that ``if $a$ is necessarily true, then $a$ is possible.''  In ethical circles the universal and existential operators take the modalities ``obligatory'' and ``permissible,'' hence axiom (\ref{aD}) entails statements such as, ``if $a$ is obligatory, then $a$ is permissible.''  Axiom $\mathbf{M}$ implies\footnote{In Deontic logic axiom $\mathbf{M}$ is not accepted since it leads to an ``ought,is'' fallacy, i.e. ``if $a$ ought to be true, then it is true.''  Instead, axiom $\mathbf{D}$, which is much weaker, is chosen.}, ``if $a$ is necessarily true, then $a$ is true.'' The system $\mathbf{M}$ contains $\mathbf{D}$ as a subsystem due to the following corollary,
\begin{flalign}
\mathbf{M}'&&\forall a\in \mathcal{A} : a \Rightarrow \diamondsuit a.&&\label{aM2}
\end{flalign}
Due to (\ref{modals}) there is also an associated corollary to axiom $\mathbf{S}_4$,
\begin{flalign}
\mathbf{S}_4'&&\forall a \in \mathcal{A}: \diamondsuit\diamondsuit a \Rightarrow \diamondsuit a.&&\label{aS42}
\end{flalign}
The definitions (\ref{aS4}) and (\ref{aS42}) allow us to reduce any statement which contains a string of either $\square$'s or $\diamondsuit$'s to a single one.  Likewise, one can extend a single $\square$ or $\diamondsuit$ to a string of them. 

Axiom $\mathbf{S}_5$ (\ref{aS5}) is the subject of some controversy, since it allows one to form variants of the ontological argument of Plantinga\footnote{The ontological argument is one example of \textit{a priori} reasoning which, as Hume argued \cite{Hume}, is invalid and unscientific.  In science we cannot \textit{prove} that things exist, we can only demonstrate that theories are inconsistent.}.  There are plenty of reasons however to reject axiom $\mathbf{S}_5$.  One argument involves the use of axiom $\mathbf{B}$ (\ref{aB}).  Axiom $\mathbf{B}$ defines a sort of ``in-between'' for the axioms $\mathbf{S}_4$ and $\mathbf{S}_5$, since the system ${S}_5$ can be constructed by adding $\mathbf{B}$ to $\mathbf{S}_4$.  Axiom $\mathbf{B}$ reads, ``if $a$ is true, then $a$ is necessarily possible,'' which seems reasonable enough.  The corollary to axiom $\mathbf{B}$ however can be written,
\begin{flalign}
\mathbf{B}'&&\forall a \in \mathcal{A} : \diamondsuit\square a \Rightarrow a.&&\label{aB2}
\end{flalign}
The corollary (\ref{aB2}) states that, ``if $a$ is possibly necessary, then $a$ is true.''  This seems to invoke a fallacy, however it is simply the result of accepting axiom $\mathbf{B}$ (\ref{aB}) which seems rather reasonable.  The problem is that axiom $\mathbf{B}$ is really a bad way of writing the following,
\begin{flalign}
\mathbf{K}+\mathbf{M}&&\forall a \in \mathcal{A} : \square(a\Rightarrow \diamondsuit a).&&\label{aMK}
\end{flalign}
Together, axioms $\mathbf{K}$ (\ref{aK}) and $\mathbf{M}$ (\ref{aM}) imply the statement, ``it is necessarily true that if $a$ is true, then $a$ is possible.''  As pointed out in the Stanford encyclopedia of philosophy \cite{Stanford}, the confusion between axiom $\mathbf{B}$ and the above statement (\ref{aMK}) is due to the syntactical ambiguities used in the English language.  In order to avoid accepting $\mathbf{B}$, and ultimately $\mathbf{S}_5$, we should instead evaluate the merits of $\mathbf{B}$ using the corollary $\mathbf{B}'$, which is easily rejected.

It is also often useful to adopt analogous axioms to those of the predicate quantors over conjunction (\ref{predicateconj}) and disjunction (\ref{predicatedis}),
\begin{flalign}
\mathbf{NC}&& \forall a,b \in \mathcal{A} : \square(a \wedge b) \Rightarrow (\square a \wedge \square b),&&\label{NC}\\
\mathbf{CN}&& \forall a,b \in \mathcal{A} : (\square a \wedge \square b) \Rightarrow \square(a \wedge b),&&\label{NCp}\\
\mathbf{ND}&&\forall a,b \in \mathcal{A} : (\square a \vee \square b)\Rightarrow \square (a \vee b).&&\label{ND}
\end{flalign}
Adoption of $\mathbf{NC}$ and $\mathbf{CN}$ together imply the equivalence in (\ref{predicateconj}).  Some programs of modal logic however will adopt $\mathbf{NC}$ but not the converse $\mathbf{CN}$\footnote{For a history of first-order modal logics see \cite{Pacuit}.}.

Together the axioms (\ref{aK}), (\ref{aD}), (\ref{aM}), (\ref{aS4}) and (\ref{aS5}) determine the set of \textit{normal modal logics}.  Depending on the system adopted, the accessibility relation (\ref{accessibilityrelation}) accumulates certain properties.  These are the following,
\begin{flalign}
\text{(Reflexive)}&& \forall w \in W : wRw.&&\\
\text{(Symmetric)}&& \forall w,v \in W: wRv\Rightarrow vRw.&&\\
\text{(Transitive)}&& \forall w,v,u : (wRv\wedge vRu) \Rightarrow wRu.&&\\
\text{(Serial)}&& \forall w \in W: \exists v \in W : wRu.&&\\
\text{(Euclidean)}&& \forall w,v,u \in W: (wRv\wedge wRu) \Rightarrow vRu.&&
\end{flalign}

For one, if $\mathbf{M}$ is adopted, then $R$ becomes reflexive.  If $\mathbf{D}$ is adopted then $R$ becomes \textit{serial}.  In $\mathbf{S}_4$, $R$ is reflexive and transitive and in $\mathbf{S}_5$, $R$ is reflexive and Euclidean.

Typically there is only one additional rule of inference that is added to modal systems.  This is an extension of \textit{modus ponens} (\ref{modusponens}) which can be called a \textit{modal modus ponens} and has the following semantics,
\begin{equation}
\square(a \Rightarrow b), \square a \vdash \square b.\label{modalmodus}
\end{equation}

Like we did with first-order predicate logic, we can compile the set of basic properties of modal logic into our \hyperref[predicatesummary]{deductive system} in the following way.  First, we can collect the modal operators into a set $\Omega_{\square,\diamondsuit} = \{\square,\diamondsuit\}$ which can be added to the connectives $\Omega$ such that, $\Omega \rightarrow \Omega_m = \Omega \cup \{\square,\diamondsuit\}$.  Any set of modal axioms, such as the ones in (\ref{modalaxioms}), can be collected into a set $\mathcal{I_{\square,\diamondsuit}}$.  These can then be added to the axioms $\mathcal{I}_p$ so that, $\mathcal{I}_p\rightarrow\mathcal{I}_{m} = \mathcal{I}_p\cup\mathcal{I}_{\square,\diamondsuit}$.  We can also collect any rules of inference, including the modal modus ponens (\ref{modalmodus}), into a set $\mathcal{Z}_{\square,\diamondsuit}$ which can be added to $\mathcal{Z}_p$, $\mathcal{Z}_p \rightarrow \mathcal{Z}_m = \mathcal{Z}_p\cup\mathcal{Z}_{\square,\diamondsuit}$. 

Some modal systems introduce different structures in the form of \textit{frames}, \textit{models} and various maps such as the accessibility relation.  We will not attempt to classify all of these different approaches and instead refer to these additional structures as belonging to some set $\mathcal{T}_{\square,\diamondsuit}$ which gets appended to $\mathcal{T}$, $\mathcal{T}\rightarrow\mathcal{T}_m = \mathcal{T}\cup\mathcal{T}_{\square,\diamondsuit}$.  Then, our modal-predicate system is given by the tuple, 
\begin{equation}
\mathcal{D}_{m} = \{\mathcal{P},\mathcal{Q},\mathcal{A},\Omega_m,\mathcal{I}_{m},\mathcal{Z}_m,\mathcal{T}_m\}.
\end{equation}

There are some technical, and perhaps philosophical, difficulties with applying a predicate calculus to modal statements.  This was first raised as a point by Quine \cite{Quine} and has since been a continuing problem.  These problems can be compounded depending on the axioms that are accepted as part of $\mathcal{I}_{\forall,\exists}$ and $\mathcal{I}_{\square,\diamondsuit}$.  We will not attempt to survey the approaches to these issues here and instead assume that whatever structures are ultimately imposed in $\mathcal{T}$, that they be consistent with those axioms adopted in $\mathcal{I}_{\forall,\exists}$ and $\mathcal{I}_{\square,\diamondsuit}$.    

\paragraph{Which axioms in $\mathcal{I}_{\square,\diamondsuit}$ should we accept as part of our inference framework? ---}
The subject matter presented in this section is only a small introduction to the vast array of different approaches to incorporating modalities.  At this point, we leave the sets $\mathcal{I}_{\square,\diamondsuit}$ and $\mathcal{T}_{\square,\diamondsuit}$ open so that the inference framework can accommodate any additional axioms or structures one wishes to use.  In order to perhaps narrow the list of axioms, we will study the approaches for how to add general conditional statements to our system.

\subsection{Conditional logic}\label{conditionallogic}
Our goal in this section is to develop a modern theory of conditionals as it pertains to deductive inference.  Conditional statements are common in English, and typically come in three forms; \textit{indicative} (or \textit{matter-of-fact}), \textit{subjunctive} and \textit{counter-factual}.  An \textit{indicative} conditional is one which can be stated in terms of objective reality as it is.  Subjunctive conditionals express desires such as, ``if only it would $a$, then it would $b$.''  Of course the adverb ``would'' suggests that $a$ and $b$ refer to things in the future, which is already a problematic situation to deal with.  

Counterfactuals on the other hand are not statements which reflect actual reality \cite{LewisBook}.  They typically deal with hypothetical statements about the world as it could possibly be if certain objective facts were actually false.  For example,
\begin{quotation}
	``If Heisenberg hadn't hated wave mechanics, no one would believe the Copenhagen interpretation.''
\end{quotation}
It is not obvious how one would access the \textit{truth} of such a statement.

Attempts to construct a well defined theory of conditionals in deductive inference have mostly led to confusion \cite{Goodman,Lewis,Popper,Adams}.  It is still not a settled question in deductive logic.  There are several reasons for this, but the main one, and coincidentally the reason why such an idea cannot be fully defined in deductive inference, will be discussed at the end of this section.  Despite these difficulties, there are certain situations in which a special class of conditional statements can be well defined in a deductive setting, which we will attempt to classify. 

Most discussions of conditionals begin with a reference to F. Ramsey's ``test'' \cite{Ramsey}, where he specifies a criterion that should hold for any notion of a conditional, 
\begin{quote}
	``If two people are arguing ``if $p$ will $q$?'' and are both in doubt as to $p$, they are adding $p$ hypothetically to their stock of knowledge and arguing on that basis about $q$; so that in a sense ``if $p$,$q$'' and ``if $p$,$\neg q$'' are contradictories.''(Ramsey, 1929)
\end{quote}
This passage points out that the negation with respect to conditional statements, ``if $p$,$q$'', should only affect the consequent $q$, and not the antecedent $p$.  This is because the statement ``if $p$ will $q$?'' is already assuming that $p$ is true.  This condition is obvious, which is why Ramsey only mentioned it in passing.  Any system which attempts to construct a conditional must in the very least adhere to this requirement. 

One might assume that the material implication connective $\Rightarrow$ is the type of conditional which properly encapsulates matter-of-fact conditionals and counterfactuals and which eventually leads to conditional probabilities.  This however is incorrect, since it automatically violates the Ramsey test\footnote{For some examples see Bradley \cite{Bradley}.},
\begin{equation}
\forall a,b \in \mathcal{A} : \neg(a\Rightarrow b) \not\Leftrightarrow a \Rightarrow \neg b.
\end{equation}
It also causes inconsistencies with the sum and product rules developed for probabilities in the next chapter.  One can see this by examining the consequences of the equivalences in (\ref{implicationandor}).  Despite these problems, several authors have attempted to justify this idea, including \cite{deFinetti1,Adams,Lewis,Grice}, however the definitions only coincide in special cases.

In \cite{Stalnaker1,Stalnaker2,Stalnaker3}, R. Stalnaker et. al. develops an extension of modal logic to include a new connective (or relation) which he calls the \textit{corner}, and is written with a greater than symbol, i.e. $a > b$.  This new connective is meant to represent what we would intuitively understand as a conditional statement, ``assuming $a$ is true, what do we believe about the truth of $b$?'' Stalnaker gives a criterion for how one would access the truth of a proposition when it is conditioned on another,
\begin{quote}
	``First, add the antecedent (hypothetically) to your stock of beliefs; second, make whatever adjustments are required to maintain consistency (without modifying the hypothetical belief in the antecedent); finally, consider whether of not the consequent is true.'' (Stalnaker, 1970)
\end{quote}  
One may refer to the above construction as a \textit{principle of minimal distortion} of a world $w \in W$.  One procedure for evaluating $b > a$ is as follows.  First we define a model \cite{Chellas,JeffreyEdgington},
\begin{define}\label{frames}
	A model, $M = \left\langle W,f,P\right\rangle$, is a tuple containing the set of all possible worlds $W$, together with a function $f:\mathcal{A}\times \mathcal{P}(W) \rightarrow \mathcal{P}(W)$ which is called a selection function, and a mapping $P:\mathbb{N}\rightarrow \mathcal{P}(W)$.
\end{define}
The selection function defined above takes in a proposition $a \in\mathcal{A}$ and a possible world $w \in W$ and returns either the empty set $\emptyset$, if no worlds in $W$ are such that $a$ is true, or it returns the set $\{w\}$ if $a$ is true in $w$, or it returns a set of worlds which are \textit{altered versions} of $w$ in the most \textit{minimal} way possible.\footnote{As an example of a definition for the function $f$ when $a$ is not true in $w$, but is true in some worlds of $W$, consider the following function $\bar{\nu}_a:\mathcal{A}\times W\times W\rightarrow \{0,1\}$, defined in terms of (\ref{truthworld}),
	\begin{equation}
	\bar{\nu}_a(w,w') \stackrel{\mathrm{def}}{=} \left\{\begin{matrix}
	1 &\Leftrightarrow& \nu(a,w) \Leftrightarrow \nu(a,w')\\
	0 &\Leftrightarrow& \nu(a,w) \not\Leftrightarrow \nu(a,w')
	\end{matrix}\right..
	\end{equation}
	The above function compares the truth value of the proposition $a\in\mathcal{A}$ and returns $1$ if the truth value is the same in both worlds $w$ and $w'$, or returns zero if it is not.  We can compare the differences in all truth values between two worlds $w,w' \in W$ by summing over all propositions, $\bar{\nu}(w,w')\stackrel{\mathrm{def}}{=} \sum_{a\in\mathcal{A}}\bar{\nu}_a(w,w').$  Such a function maps to the natural numbers, $\bar{\nu}:W\times W\rightarrow \mathbb{N}$.  Now we can define the world $w_a$ as the one which minimizes $\bar{\nu}$ over the subset $W_a$,
	\begin{equation}
	w_a(w) \stackrel{\mathrm{def}}{=} \min_{w' \in W_a}\bar{\nu}(w,w').\label{minimizeworld}
	\end{equation}
	There may be more than one world $w \in W$ which satisfies this criteria.  We call the collection of worlds which satisfy (\ref{minimizeworld}) $W_a^w \subseteq W$, and are the result of the function $f(a,w)$ when $\nu(a,w)$ is false but there exist some worlds $w'\in W$ in which $\nu(a,w')$ is true.}  

Any sentence $\Gamma \subseteq \mathcal{A}$ which is true in the world $w \in W$ can be written $\vDash^{M}_w \Gamma$.  If $\Gamma$ is true in all possible worlds we write $\vDash^{M}\Gamma$, which is the same as the statement $\square \Gamma$.  The truth of material implication with respect to the world $w \in W$ can be written,
\begin{equation}
\vDash^{M}_w (a\Rightarrow b) \Leftrightarrow [(\vDash^{M}_w a)\Rightarrow (\vDash^{M}_w b)].
\end{equation}
This is a result we would expect, since the material implication is well-defined within any world.  For the conditional however, we find the following,
\begin{equation}
\vDash^{M}_w (a > b) \Leftrightarrow f(a,W_a) \subseteq W_b.\label{truthconditional}
\end{equation}
This demands that whatever worlds in which $a$ is true must be a subset of the worlds in which $b$ is true, otherwise we allow ex falso quodlibet.  

A conditional as defined by Stalnaker \cite{Stalnaker} is the following,
\begin{define}\label{conditionalrelation}
	Let $\mathcal{A}$ be a set of propositions.  For any propositions, $a$ and $b$ in $\mathcal{A}$, we define the conditional relation $a > b$, which reads ``$b$ given $a$.''  The conditional relation is defined in terms of the universal and existential modal operators as,
	\begin{align}
	\forall a \in \mathcal{A} : \square a &\Leftrightarrow \neg a > a\\
	\forall a\, \mathcal{A} : \diamondsuit a &\Leftrightarrow \neg(a > \neg a).
	\end{align}
\end{define}
The first condition may seem strange intuitively, however it is simply a proposition which states that $a$ is necessarily true, if it is still true in the face of its contrary.  This represents a state of complete certainty with respect to beliefs about the proposition $a$.  What's more important is the compliment relation, which says that $a$ is possible, if the conditional statement $\neg a|a$ is false.  For all situations in which the truth of $a$ is not absolutely certain, $a$ will always be possible.  The definitions in (\ref{conditionalrelation}) are easily seen to be consistent with (\ref{truthconditional}) since for $\square a$, $\neg a$ is not true in any world and hence $f(\neg a, W) = \emptyset$.  For $\diamondsuit a$, $a$ must be true in at least one possible world, so that $f(a,W) \neq \emptyset$.

In order to avoid confusion with the greater than symbol applied to numbers, we will write conditionals using the \textit{solidus} ``$|$'' so that the statement $a > b$ is written $b|a$.
Depending on the type of conditional logic one wishes to develop, one can adopt some subset of the following axioms,
\begin{axioms}
	The following list gives possible axioms for a system of conditional logic.  The names in parenthesis correspond to the convention given by Chellas \cite{Chellas}.  The ones which begin with $\mathbf{S}$ are defined by Stalnaker \cite{Stalnaker1},
	\begin{flalign}
	\mathbf{SN} &&  \forall a,b\in\mathcal{A} : \square(a \Rightarrow b) \Rightarrow (b | a)&&\\
	\mathbf{SP} &&  \forall a\in\mathcal{A} : \diamondsuit a \Rightarrow [b|a \Rightarrow \neg(\neg b|a)]&&\\
	\mathbf{SD} &&  \forall a,b,c\in\mathcal{A} : b\vee c|a \Rightarrow (b|a)\vee (c|a)&&\\
	\mathbf{SC} &&  \forall a,b,c \in \mathcal{A} : [(b|a)\wedge (a|b)] \Rightarrow [(c|a) \Rightarrow (c|b)]&&\\
	(\mathbf{RCEA}) && \forall a,b,c \in \mathcal{A} : [a \Leftrightarrow b] \Rightarrow [(c|a)\Leftrightarrow (c|b)]  &&\\
	(\mathbf{RCEC}) && \forall a,b,c \in \mathcal{A} : [a \Leftrightarrow b] \Rightarrow [(c|a)\Leftrightarrow (c|b)] &&\\
	(\mathbf{RCM}) &&  \forall a,b,c \in \mathcal{A} : (a \Rightarrow b)\Rightarrow [(a|c) \Rightarrow (b|c)]&&\\
	(\mathbf{RCR}) &&  \forall a,b,c,d \in \mathcal{A} : [(a\wedge b)\Rightarrow c] \Rightarrow [((a|d)\wedge(b|d))\Rightarrow (c|d)]&&\\
	(\mathbf{RCK}) && [(a_1\wedge\dots\wedge a_n)\Rightarrow b] \Rightarrow [(a_1|c)\wedge\dots\wedge(a_n|c)]\Rightarrow (b|c)\\
	(\mathbf{MP}) &&  \forall a,b\in \mathcal{A} : b|a \Rightarrow a\Rightarrow b&&\\
	(\mathbf{CM}) && \forall a,b,c\in\mathcal{A} : [a\wedge b|c]\Rightarrow [(a|c)\wedge(b|c)] &&\\
	(\mathbf{CC}) &&  \forall a,b,c \in \mathcal{A} : [(a|c)\wedge(b|c)] \Rightarrow [a\wedge b|c]&&
	\end{flalign}
\end{axioms}
Stalnaker defines his conditional logic \cite{Stalnaker1} through the axioms $\mathbf{MP}$ and $\mathbf{SN}+\mathbf{SP}+\mathbf{SD}+\mathbf{SC}$.
Chellas \cite{Chellas} showed that the class of conditional logics $CK$ defined by axioms $\mathbf{RCEA}$ and $\mathbf{RCK}$ are decideable and are determined by the class of standard frames (\ref{frames}).

One important take away of these axioms is the sequence of entailment with respect to the strict implication and the normal implication,
\begin{flalign}
\mathbf{SN}+\mathbf{MP}&&\forall a,b\in\mathcal{A} : \square(a \Rightarrow b) \Rightarrow (b|a) \Rightarrow (a \Rightarrow b)&&.
\end{flalign}
This demonstrates that our intuition about the implication arrow representing conditional statements is not totally off base, however they are not the same.  The conditional $b|a$ implies the implication $a\Rightarrow b$, however the antecedant $a$ in $b|a$ is always assumed to be true, whereas in the implication $a$ could be false in an example of \textit{ex falso quodlibet}.  

While there are plenty of axioms to choose from, the inference framework we adopt here will only need a few.  The framework should at least include axioms $\mathbf{SD}$, $\mathbf{RCR}$, $\mathbf{CM}$ and $\mathbf{CC}$.  These properties are all that we will need to construct our inductive inference of the next chapter. 

We can collect the solidus into our set of connectives, $\Omega_m \rightarrow \Omega_{c} = \Omega_m\cup\{|\}$.  In total we have, $\Omega_m = \Omega\cup\{\square,\diamondsuit\}\cup\{|\}$.  We can also collect any axioms concerning ``$|$'' into a set  $\mathcal{I}_{|}$, as well as any additional structure such as the model $M$ (\ref{frames}) into a set $\mathcal{T}_{|}$.  Our complete list of axioms is then, $\mathcal{I}_{c} = \mathcal{I}\cup \mathcal{I}_{\forall,\exists}\cup\mathcal{I}_{\square,\diamondsuit}\cup\mathcal{I}_{|}$.
\paragraph{Summary of our deductive system ---}
Our deductive system can be written as the tuple, $\mathcal{D}_{c} = \left\{\mathcal{P},\mathcal{Q},\mathcal{A},\Omega_{c},\mathcal{I}_{c},\mathcal{Z}_c,\mathcal{T}_{c}\right\}$ whose elements are the following,
\begin{enumerate}
	\item (\textbf{Predicates}) - The set $\mathcal{P}\subset\Sigma^*$ of predicates for a particular universe of discourse.
	\item (\textbf{Quantors}) - The set of quantification operators $\{\forall,\exists\}$ (\ref{universal}),(\ref{existential}).
	\item (\textbf{Propositions}) - The set of all \hyperref[propositions]{propositions} $\mathcal{A}$ for a particular universe of discourse.
	\item (\textbf{Connectives}) - The set of \hyperref[connectives]{connectives} $\Omega$ containing the set of generators $\{\wedge,\vee\,\neg\}$, the \hyperref[modallogic]{modal operators} $\{\square,\diamondsuit\}$ and the \hyperref[conditionallogic]{conditional operator} $|$.
	\item (\textbf{Axioms}) - The set of \hyperref[rulesofinference]{propositional axioms} $\mathcal{I}$ together with the axioms involving the modal operators $\mathcal{I}_{\square,\diamondsuit}$, the axioms concerning the quantors $\mathcal{I}_{\forall,\exists}$ and the axioms concerning the conditional operator $\mathcal{I}_{|}$.
	\item (\textbf{Inference Rules}) - The set of \hyperref[inferencerules]{inference rules} $\mathcal{Z}$ for propositional calculus, together with any inference rules adopted for predicate calculus, $\mathcal{Z}_{\forall,\exists}$, and inference rules for the modal operators $\mathcal{Z}_{\square,\diamondsuit}$ and the conditional operator $\mathcal{Z}_{|}$.
	\item (\textbf{Additional structure}) - The set $\mathcal{T}$ of any additional structures such as the set of possible worlds (\ref{possibleworlds}), frames, models (\ref{frames}), accessibility relations (\ref{accessibilityrelation}), etc.
\end{enumerate}

\section{Proof of \hyperref[aczel]{theorem 1.1}}\label{proofaczel}
The proof follows the one given in \cite{Aczel} with some similar notation.  Let us first simplify the notation of the elements of $\Xi$ by writing,
\begin{equation}
\xi(a|\Gamma) \stackrel{\mathrm{def}}{=} a, \quad \xi(b|\Gamma) \stackrel{\mathrm{def}}{=} b \quad \mathrm{and} \quad \xi(c|\Gamma) \stackrel{\mathrm{def}}{=} c,
\end{equation}
so that (\ref{associativeor}) is written,
\begin{equation}
f_{\vee}(a,f_{\vee}(b,c)) = f_{\vee}(f_{\vee}(a,b),c).
\end{equation}
To simplify further, we will sometimes let,
\begin{equation}
f_{\vee}(a,b) \stackrel{\mathrm{def}}{=} a \circ b,
\end{equation}
which then reduces (\ref{associativeor}) to,
\begin{equation}
(a \circ b) \circ c = a \circ (b \circ c).
\end{equation}
Since $f_{\vee}$ is a map to $\Xi$, $f_{\vee}:\Xi\times\Xi\rightarrow\Xi$, it suffices to determine its behavior in several different situations.  This includes when both arguments are the same, $f_{\vee}(\ell,\ell)$, which we will investigate first.  We will then analyze the limits of the interval $\langle \xi_{\fals},\xi_{\tru}\rangle$, to determine if it is open or closed on either end.  
\paragraph{Behavior of $f_{\vee}(\ell,\ell)$ ---}
There are three possible types of behavior of $f_{\vee}$ for each individual element $\ell$,
\begin{flalign}
\mathbf{G}&&f_{\vee}(\ell,\ell) > \ell,&&\label{G}\\
\mathbf{E}&&f_{\vee}(\ell,\ell) = \ell,&&\label{E}\\
\mathbf{L}&&f_{\vee}(\ell,\ell) < \ell.&&\label{L}
\end{flalign}
Each element $\ell \in \Xi$ will satisfy either of these equations for $f_{\vee}$, i.e. each element will either return itself in disjunction, or an element $\ell' \in \Xi$ which is a smaller or larger degree of belief than $\ell$.  For each of these situations, we have the following consequences,
\begin{flalign}
\mathbf{G}'&&\forall a \in \Xi : f_{\vee}(\ell,a) > a,&&\label{Gp}\\
\mathbf{E}'&&\forall a \in \Xi : f_{\vee}(\ell,a) = a,&&\label{Ep}\\
\mathbf{L}'&&\forall a \in \Xi : f_{\vee}(\ell,a) < a.\label{Lp}
\end{flalign}
For the disjunction function $f_{\vee}$, we necessarily have the situation $\mathbf{E}$ for every element in $\Xi$\footnote{This fact is quite subtle and one should use caution when analyzing $f_{\vee}$.  Since $f_{\vee}$ is defined as the \textit{representation} of $\xi(a\vee b|\Gamma)$, its behavior necessarily depends on the behavior of $\vee$.  For example, we have the identity $a\vee a = a$, thus it should be the case that for every $a \in \tilde{\mathcal{A}}$, $\xi(a\vee a) = \xi(a)$ and hence $f_{\vee}(\xi(a),\xi(a)) = \xi(a)$.  However, $a$ is not mutually exclusive from itself, so the result $f_{\vee}(\xi(a),\xi(a)) = \xi(a)$ only follows from the fact that the other two arguments, $\xi(a|a) = \xi_{\tru}$, are constants and hence $f_{\vee}(\xi(a),\xi(a),\xi(a|a),\xi(a|a)) \rightarrow f_{\vee}(\xi(a),\xi(a),\xi_{\tru},\xi_{\tru})$ reduces in form to the mutually exclusive case.}.  It will suffice however, to discuss the behavior of the other two situations.  The first, $\mathbf{G}'$ is easily shown from the associative property of $f_{\vee}$ and $\mathbf{G}$,
\begin{align}
\forall a \in \Xi : f_{\vee}(\ell,f_{\vee}(\ell,a)) &= f_{\vee}(f_{\vee}(\ell,\ell),a)\nonumber\\
&> f_{\vee}(\ell,a), 
\end{align}
so that $f_{\vee}(\ell,a) > a$\footnote{As shown in \cite{Aczel}, one can determine $\mathbf{G}$ and likewise $\mathbf{G}'$ if there exists at least one element $a_0 \in \Xi$ such that $f_{\vee}(\ell,a_0) > a_0$.  To see this, consider the associative property,
	\begin{equation}
	f_{\vee}(f_{\vee}(\ell,\ell),a_0) = f_{\vee}(\ell,f_{\vee}(\ell,a_0)) > a_0.\label{proofG}
	\end{equation}
	This implies that $f_{\vee}(\ell,\ell)>\ell$ which is $\mathbf{G}$.  Hence, a single $a_0$ with this property is enough to impose $\mathbf{G}$ and $\mathbf{G}'$.}.  Property $\mathbf{L}'$ follows by analogy.  The property $\mathbf{E}$ identifies $\ell = e$ as the identity element with respect to $f_{\vee}$.  Thus, $\mathbf{E}'$ can only be true for one element of $\Xi$,
\begin{flalign}
\mathbf{E}'&&\forall a \in \Xi : !\exists e \in \Xi : f_{\vee}(e,a) = a&&.
\end{flalign}  
\paragraph{Specifying the left limit $\xi_{\fals}$ ---}
Depending on the existence of elements which have either property $\mathbf{G}$ or $\mathbf{L}$, determines the type of interval $\langle \xi_{\fals},\xi_{\tru}\rangle$ is with respect to its end points.  Both end points will either be open or closed.  First we study the behavior of the left side of the interval, for which there are three possible behaviors,
\begin{enumerate}
	\item\textbf{Case 1} - The interval $[\xi_{\fals},\xi_{\tru}\rangle$ is closed on the left.  Since $f_{\vee}(\xi_{\fals},\xi_{\fals}) \in [\xi_{\fals},\xi_{\tru}\rangle$, the situation $\mathbf{L}$ is impossible for $\xi_{\fals}$.  Thus, we must either have,
	\begin{flalign}
	\mathbf{C}&&f_{\vee}(\xi_{\fals},\xi_{\fals}) = \xi_{\fals} \quad \mathrm{or} \quad f_{\vee}(\xi_{\fals},\xi_{\fals}) > \xi_{\fals}.&&\label{C}
	\end{flalign}
	\item \textbf{Case 2} - The interval $(\xi_{\fals},\xi_{\tru}\rangle$ is open on the left, however we have the left hand limit\footnote{The existence of the left hand limit of course depends on whether $f_{\vee}$ remains continuous and associative, however it must also allow for the reduction of $f_{\vee}(a,\ell) = f_{\vee}(c,\ell)$ to $a = c$.  This is easily seen by the continuity of $f_{\vee}$,
		\begin{align}
		f_{\vee}(a,\lim_{\ell\rightarrow \xi_{\fals}}f_{\vee}(\ell,b)) &= \lim_{\ell\rightarrow \xi_{\fals}}f_{\vee}(f_{\vee}(a,\ell),b) = f_{\vee}(\lim_{\ell\rightarrow \xi_{\fals}}f_{\vee}(a,\ell),b)\nonumber\\
		&= f_{\vee}(\lim_{\ell\rightarrow \xi_{\fals}}f_{\vee}(c,\ell),b) = \lim_{\ell\rightarrow \xi_{\fals}}f_{\vee}(c,f_{\vee}(\ell,b))\nonumber\\
		&= f_{\vee}(c,\lim_{\ell\rightarrow \xi_{\fals}}f_{\vee}(\ell,b)).
		\end{align}},
	\begin{flalign}
	\mathbf{OG}&&\forall a \in \Xi : \lim_{\ell\rightarrow \xi_{\fals}}f_{\vee}(\ell,a) &\geq a&&\label{OG1}\\
	&&&>\xi_{\fals}.&&\label{OG2}
	\end{flalign}
	\item \textbf{Case 3} - The interval $(\xi_{\fals},\xi_{\tru}\rangle$ is open on the left, however we have the left hand limit\footnote{In this case, the left hand limit $\xi_{\fals}$ plays the role of zero in multiplication.  This comes immediately from $\mathbf{OL}$ since,
		\begin{equation}
		\forall a \in \Xi : \lim_{a\rightarrow \xi_{\fals}}\left[\lim_{\ell\rightarrow\xi_{\fals}}f_{\vee}(\ell,a)\right] \leq \lim_{a\rightarrow \xi_{\fals}}a = \xi_{\fals},
		\end{equation}
		which means that,
		\begin{equation}
		\forall a,b \in \Xi : \lim_{a\rightarrow \xi_{\fals}}\lim_{b\rightarrow \xi_{\fals}}f_{\vee}(a,b) = \xi_{\fals}.\label{proofOL}
		\end{equation}
		It is easy to show that for any element $a \in \Xi$ we must then have,
		\begin{equation}
		\forall a \in \Xi : \lim_{\ell\rightarrow \xi_{\fals}}f_{\vee}(\ell,b) = \xi_{\fals}.
		\end{equation}
		If one allowed for example, $\lim_{\ell\rightarrow\xi_{\fals}}f_{\vee}(\ell,c) < a$, then we could conclude from $\mathbf{OL}$ and (\ref{proofOL}),
		\begin{align}
		a > \lim_{\ell\rightarrow\xi_{\fals}}f_{\vee}(\ell,a) &= \lim_{a\rightarrow \xi_{\fals}}f_{\vee}(\ell,\lim_{b\rightarrow \xi_{\fals}}f_{\vee}(b,c)) = \lim_{a\rightarrow \xi_{\fals}}\lim_{b\rightarrow \xi_{\fals}}f_{\vee}(\ell,f_{\vee}(b,c))\nonumber\\
		&=\lim_{a\rightarrow \xi_{\fals}}\lim_{b\rightarrow \xi_{\fals}}f_{\vee}(f_{\vee}(\ell,b),c) = \lim_{d\rightarrow\xi_{\fals}}f_{\vee}(d,c) = a,
		\end{align}
		which is obviously a contradiction.},
	\begin{flalign}
	\mathbf{OL}&&\forall a \in \Xi : \lim_{\ell\rightarrow \xi_{\fals}}f_{\vee}(\ell,a) < a.&&\label{OL}
	\end{flalign}
\end{enumerate}
In the first two cases, $\mathbf{C}$ and $\mathbf{OG}$, there are no inverse elements, however there can exist an identity.  In $\mathbf{OG}$ there is only an identity element if,
\begin{equation}
f_{\vee}(\xi_{\fals},\xi_{\fals}) = \xi_{\fals},\label{proofI}
\end{equation} 
so that $\xi_{\fals} = e$.  If this fails so that $f_{\vee}(\xi_{\fals},\xi_{\fals}) > \xi_{\fals}$, then by the corollary in (\ref{proofG}) we have for all $a \in \Xi$, $f_{\vee}(\xi_{\fals},a) > a$ and hence for all $b \in \Xi$, $f_{\vee}(a,b) > a$.  Thus there can be no identity for the system with $f_{\vee}(\xi_{\fals},\xi_{\fals}) > \xi_{\fals}$.  If on the other hand (\ref{proofI}) holds, then we have,
\begin{align}
\forall a,b \in \Xi : f_{\vee}(a,b) \geq f_{\vee}(\xi_{\fals},b) &= b\nonumber\\
&>\xi_{\fals},
\end{align}
which precludes the existence of inverse elements even for (\ref{proofI}).  For case three on the other hand, we can have inverses so that $f_{\vee}(a,a^{-1}) = e$.  This requires $e \neq \xi_{\fals}$ so that,
\begin{equation}
\exists a,c \in \Xi : a < e < c.
\end{equation}
It is obvious that this cannot be the case for (\ref{solutionor}) since we have the identity $\fals \vee \fals = \fals$, which invokes (\ref{proofI}) must be true.  Therefore, with respect to (\ref{solutionor}) there can be no inverses of any element $\xi(a|\Gamma) \in \Xi$ with respect to the disjunction $f_{\vee}$, however $\xi_{\fals}$ plays the role of the identity.  This also precludes case three, since we have the identity $a \vee \fals = a$, therefore the inequality $\lim_{\ell\rightarrow \xi_{\fals}}f_{\vee}(\ell,a) < a$ cannot be satisfied.  For the remainder of the proof, one can either adopt case one or two.

\paragraph{The recursive functions $\varphi_m$ and $\psi$ ---}
To continue, let us define the recursive functions $\varphi_m(a)$ where,
\begin{equation}
\forall a \in \Xi : \varphi_1(a) = a \quad \mathrm{and}\quad \varphi_{m+1}(a) = f_{\vee}(a,\varphi_m(a)),\label{phidef}
\end{equation}
where $m \in \mathbb{N}$.  This evaluates a set of $m$-nested disjunctions with respect to the element $a \in \Xi$. From the definition in (\ref{phidef}) and the associativity property (\ref{associativeor}), we have,
\begin{equation}
\forall a \in \Xi : f_{\vee}(\varphi_{m_1}(a),\varphi_{m_2}(a)) = f_{\vee}(\varphi_{m_2}(a),\varphi_{m_1}(a)) = \varphi_{m_1+m_2}(a),
\end{equation}
which using the $\circ$ notation is simply,
\begin{equation}
\forall a \in \Xi : \varphi_{m_1 + m_2}(a) = \varphi_{m_1}(a)\circ\varphi_{m_2}(a) = \varphi_{m_2}(a)\circ\varphi_{m_1}(a).
\end{equation}
Likewise we have,
\begin{equation}
\forall a \in \Xi: \varphi_m\left[\varphi_n(a)\right] = \varphi_n\left[\varphi_m(a)\right] = \varphi_{mn}(a),
\end{equation}
which means that $\varphi_m(a)$ plays the role of multiplication with respect to $f_{\vee}$.  If we have for each $a \in \Xi$, $f_{\vee}(a,a) > a$, then it follows that,
\begin{equation}
\varphi_{m+1}(a) > \varphi_m(a),
\end{equation}
so long as $a \neq e$.  Since we are excluding the possibility of $f_{\vee}(a,a) < a$, then we must have,
\begin{equation}
\varphi_{m+1}(a) \geq \varphi_m(a),\label{varphiincrease}
\end{equation}
in general.  Thus we can take the limit to the right part of the interval $[\xi_{\fals},\xi_{\tru}\rangle$ so that,
\begin{equation}
\lim_{a\rightarrow \xi_{\tru}}\varphi_m(a) = \xi_{\tru}.\label{phiproof2}
\end{equation}
Since in general we have that $f_{\vee}$ is increasing in both of its arguments,
\begin{equation}
\lim_{a\rightarrow \xi_{\tru}}f_{\vee}(a,b) \geq \lim_{a\rightarrow \xi_{\tru}} a = \xi_{\tru},
\end{equation}
which means that we must have,
\begin{equation}
\lim_{a\rightarrow \xi_{\tru}}f_{\vee}(a,b) = \lim_{a\rightarrow \xi_{\tru}}\lim_{b\rightarrow\xi_{\tru}}f_{\vee}(a,b) = \xi_{\tru}.
\end{equation}
This shows that (\ref{phiproof2}) must hold since,
\begin{align}
\forall a \in \Xi: \lim_{a\rightarrow \xi_{\tru}}\varphi_m(a) &= \lim_{a\rightarrow \xi_{\tru}}a\circ\varphi_{m-1}(a))\nonumber\\
&= \dots = \lim_{a\rightarrow \xi_{\tru}}a\circ\dots\circ \varphi_1(a) = \xi_{\tru}.
\end{align}
We also have for $\xi_{\fals}$,
\begin{equation}
\varphi_m(\xi_{\fals}) = \xi_{\fals} \circ \varphi_{m-1}(\xi_{\fals}) = \dots = \xi_{\fals}.\label{falsevar}
\end{equation}
Thus, due to (\ref{falsevar}) and (\ref{phidef}), for every $a \in \Xi$ the function $\varphi_m(a)$ has a unique inverse,
\begin{equation}
\forall a \in \Xi : \forall m \in \mathbb{N} : \exists b \in \Xi : b = \varphi_m(a) \Rightarrow a = \varphi^{-1}(b),
\end{equation}
where the inverse satisfies,
\begin{equation}
\varphi_m(a) \leq b < \xi_{\tru}.\label{varphiinvincrease}
\end{equation}
Thus the inverse function $b \in \Xi$ is continuous and increasing in $a$.  Now, let us associate to the function $a \in \Xi$ the function $\psi$ whose input is a rational number $q\in\mathbb{Q}$,
\begin{equation}
\psi\left(\frac{m}{n}\right) \stackrel{\mathrm{def}}{=}\varphi_n^{-1}\left[\varphi_m(a)\right],\label{phivarphi}
\end{equation}
where $m,n \in \mathbb{N}$\footnote{If one adopts the situation $\mathbf{G}$ (\ref{G}), then the definition (\ref{phivarphi}) only applies for $m > n$.}.  Since we have in particular property $\mathbf{E}$, $f_{\vee}(a,a) = a$ for all $a \in \Xi$, the action of $\varphi$ and its inverse $\varphi^{-1}$ commute so that the above can be rewritten as,
\begin{equation}
\psi\left(\frac{m}{n}\right) =\varphi_m\left[\varphi_n^{-1}(a)\right].
\end{equation}
Naturally, the value $m=1$ returns the original element,
\begin{equation}
\forall a \in \Xi : \psi(1) = a.
\end{equation}
Since $\Xi$ contains the identity $\xi_{\fals}$, we can further define,
\begin{equation}
\psi(0) \stackrel{\mathrm{def}}{=} \xi_{\fals}.\label{psizero}
\end{equation}
From (\ref{varphiincrease}) and (\ref{varphiinvincrease}) the function $\psi$ is strictly monotonically increasing,
\begin{equation}
\psi\left(\frac{m+1}{n}\right) = \varphi_{m+1}\left[\varphi_n^{-1}(a)\right] > \varphi_m\left[\varphi_n^{-1}(a)\right] = \psi\left(\frac{m}{n}\right).\label{varprop}
\end{equation}
\paragraph{Specifying the right limit $\xi_{\tru}$ ---}
It suffices to discuss now the right side of the interval $[\xi_{\fals},\xi_{\tru}\rangle$.  In \cite{Aczel} it is assumed that the end point $\xi_{\tru}$ is obtainable by an at most infinite limit,
\begin{equation}
\forall a \in \Xi : \lim_{m\rightarrow\infty}\psi(m) = \xi_{\tru}.\label{limittru}
\end{equation}
In particular, if this limit exists, the point $\xi_{\tru}$ cannot be within the interval.  Thus, the right hand side must be open, $\langle\xi_{\fals},\xi_{\tru})$.  One can easily show this by assuming the opposite.  If for some $a \neq \xi_{\fals}$ we have,
\begin{align}
\xi_{\tru} < f_{\vee}(a,\xi_{\tru}) &= f_{\vee}\left(a,\lim_{m\rightarrow\infty}\psi(m)\right) = \lim_{m\rightarrow \infty}f_{\vee}(a,\varphi(a))\nonumber\\
&= \lim_{m\rightarrow \infty}\varphi_{m+1}(a) = \lim_{m\rightarrow\infty}\psi(m+1) = \xi_{\tru}.
\end{align}
Thus we must have $f_{\vee}(a,\xi_{\tru}) = \xi_{\tru}$, which is certainly true since $\forall a \in \tilde{\mathcal{A}} : a\vee \tru = \tru$.  While this certainly demonstrates that $\xi_{\tru}$ cannot be in the interval $\langle \xi_{\fals},\xi_{\tru}\rangle$ if it exists as a limit point for repeated application of the disjunction, the limit in (\ref{limittru}) is only true for $a = \xi_{\tru}$.  Since $f_{\vee}(\xi_{\tru},a)$ is never greater than $a$ or $\xi_{\tru}$, the point $\xi_{\tru}$ can belong to the interval without generating a contradiction.  Thus we have that the interval can be closed on both sides without loss of generality, $\langle \xi_{\fals},\xi_{\tru}\rangle \rightarrow [\xi_{\fals},\xi_{\tru}]$.

\paragraph{The general solution ---}
Given the nature of (\ref{varprop})\footnote{As discussed in \cite{Aczel}, the equation (\ref{proofalmost}) also holds for the property $\mathbf{G}$, however the derivation is more difficult.}, we can show the following,
\begin{align}
f_{\vee}\left(\frac{}{}\psi\left(\frac{m_1}{n}\right),\psi\left(\frac{m_2}{n}\right)\right) &= \varphi_{m_1}\left[\varphi_{n}^{-1}(a)\right]\circ \varphi_{m_2}\left[\varphi_n^{-1}(a)\right]\nonumber\\
&= \varphi_{m_1+m_2}\left[\varphi_n^{-1}(a)\right] = \psi\left(\frac{m_1 + m_2}{n}\right).\label{proofalmost}
\end{align}
The function $\psi$ is necessarily monotonic for all positive rational values $q \in \mathbb{Q}^+$ and hence possesses left and right limits for any positive real number $x \in \mathbb{R}^+$\footnote{To see this, first consider that the converse is true via the limits being different,
	\begin{equation}
	r = \psi(x - 0) < \psi(x + 0) = s,\label{proofdefs}
	\end{equation}
	such that we have,
	\begin{equation}
	f_{\vee}(r,r) < f_{\vee}(r,s) < f_{\vee}(s,s), \quad \mathrm{or}\quad r < f_{\vee}(r,s) < s.
	\end{equation}
	Since $f_{\vee}(r,s)$ is increasing, and due to the strict inequalities above, we can always find an $\varepsilon$ such that $f_{\vee}(r,r) + \varepsilon < f_{\vee}(r,s) < f_{\vee}(s,s) - \varepsilon$.  Then, using continuity again we have the associated $\delta$'s to the $\varepsilon$ that define $|r' - r|<\delta$ and $|s'-s|<\delta$ so that,
	\begin{equation}
	f_{\vee}(r,s) - \varepsilon < f_{\vee}(r',s') < f_{\vee}(r,s) + \varepsilon.
	\end{equation}
	Using the definition of $r$ and $s$ in (\ref{proofdefs}), we could always define $r' = \psi(m_r)$ and $s' = \psi(m_s)$ in terms of the rational numbers $m_r$ and $m_s$ where $m_r < x < m_s$ and $\frac{m_r+m_s}{2} > x$.  However, this would suggest that,
	\begin{align}
	\varphi_2\left[\psi\left(\frac{m_r + m_s}{2}\right)\right] &= \varphi_2\left[\varphi_2^{-1}\left[\varphi_{m_r+m_s}(a)\right]\right] = \psi(m_r + m_s) = f_{\vee}(\psi(m_r),\psi(m_s))\nonumber\\
	&= f_{\vee}(r',s') < f_{\vee}(r,s) + \varepsilon < f_{\vee}(s,s) = \varphi_2\left[\psi(x+0)\right],
	\end{align}
	which is a contradiction since $\varphi$ is increasing.}  Thus, the function $\psi(x)$ for $x \in \mathbb{R}^+$ must be continuous.  Finally, from (\ref{phidef}) we have in the limit,
\begin{equation}
\forall x,y \in [\xi_{\fals},\xi_{\tru}] : \psi(x + y) = f_{\vee}(\psi(x),\psi(y)).
\end{equation}
Substituting $a = \psi(x)$ and $b = \psi(y)$ we find,
\begin{equation}
f_{\vee}(a,b) = \psi\left(\psi^{-1}(a) + \psi^{-1}(b)\right).\label{solution}
\end{equation}
Identifying $\phi = \psi^{-1}$ reduces (\ref{solution}) to (\ref{solutionor}).

\section{Proof of \hyperref[Cauchy]{theorem 1.2}}\label{cauchyproof}
We will first show the solution for rational numbers $q \in\mathbb{Q}$.  Consider that $\Xi$ maps to the rational numbers, $\xi:\tilde{\mathcal{A}}\rightarrow \mathbb{Q}$, so that the function $f_{\wedge}:\Xi\times\Xi\rightarrow\Xi$ is a map from $\mathbb{Q}^2$ to $\mathbb{Q}$.  For ease of notation, let
\begin{equation}
\xi(a|\Gamma) \stackrel{\mathrm{def}}{=} a, \quad \xi(b|\Gamma\wedge a) \stackrel{\mathrm{def}}{=} b \quad \mathrm{and}\quad \xi(c|\Gamma\wedge a) \stackrel{\mathrm{def}}{=} c,
\end{equation}
so that (\ref{distributiveand}) is written,
\begin{equation}
f_{\wedge}(a,b+c) = f_{\wedge}(a,b) + f_{\wedge}(a,c).
\end{equation}
Since $\Xi$ is bounded by $[0,\xi_{\tru}]$, it suffices to show that $f_{\wedge}$ is linear in its second argument for two cases, $(b+c) = \xi_{\fals} = 0$ and $(b+c) > 0$.  For the first case, consider $c = \xi_{\fals}$,
\begin{align}
f_{\wedge}(a,b+\xi_{\fals}) &= f_{\wedge}(a,b) + f_{\wedge}(a,\xi_{\fals})\nonumber\\
&= f_{\wedge}(a,b).
\end{align}
Thus we have that $f_{\wedge}(a,\xi_{\fals}) = f_{\wedge}(a,0) = 0$.  Likewise, consider the repeated application of the Cauchy equation so that,
\begin{equation}
f_{\wedge}(a,b+\dots + b) = f_{\wedge}(a,nb) = nf_{\wedge}(a,b),\label{proof1}
\end{equation}
where $n \in \mathbb{N}$.  Multiplying the right hand side of the above by $\frac{m}{n}$, where $m \in \mathbb{N}$, and substituting for $b$, $\frac{b}{n}$, we find,
\begin{equation}
mf_{\wedge}\left(a,\frac{b}{n}\right) = \frac{m}{n}f_{\wedge}(a,b).
\end{equation}
Using the definition in (\ref{proof1}), the left hand side of the above becomes,
\begin{equation}
f_{\wedge}\left(a,\frac{m}{n}b\right) = \frac{m}{n}f_{\wedge}(a,b).
\end{equation}
Since $\frac{m}{n}$ is rational, we have that,
\begin{equation}
f_{\wedge}(a,qb) = qf_{\wedge}(a,b),
\end{equation}
for $q \in \mathbb{Q}$.  Thus, for $b = 1$ we have,
\begin{equation}
f_{\wedge}(a,q) = qf_{\wedge}(a,1) = \lambda qg_{\wedge}(a),
\end{equation}  
where $\lambda \in \mathbb{R}$.  Now, since $\Xi$ is continuous, then the result will hold for any real number.  Assume that $b$ can be achieved as the limit of a sequence of rational numbers,
\begin{equation}
b = \lim_{k\rightarrow \infty}q_k.
\end{equation}    
Then, since $f_{\wedge}$ and $\Xi$ are continuous, the limit commutes between continuous maps,
\begin{align}
f_{\wedge}(a,b) = f_{\wedge}(a,\lim_{k\rightarrow\infty}q_k) &= \lim_{k\rightarrow\infty}f_{\wedge}(a,q_k)\nonumber\\
 &= \lambda\lim_{k\rightarrow\infty}q_kf_{\wedge}(a,1) = \lambda bf_{\wedge}(a,1).
\end{align}
Hence we have that,
\begin{equation}
\forall a,b \in \mathbb{R} : f_{\wedge}(a,b) = \lambda g_{\wedge}(a)b.
\end{equation}

\section{Proof of \hyperref[pexider]{theorem 3.1}}\label{proofa1}
\begin{proof}
	This proof follows the one given in Acz\'{e}l \cite{Aczel} (see pgs. 141-142).  We wish to show that the many-variable Pexider equation,
	\begin{equation}
	f(x_1 + y_1,x_2 + y_2,\dots,x_n + y_n) = g(x_1,x_2,\dots,x_n) + h(y_1,y_2,\dots,y_n),
	\end{equation}
	reduces to a Cauchy type equation.  First consider the special case when
	\begin{equation}
	y_1 = y_2 = \dots = y_n = 0,
	\end{equation}
	so that
	\begin{equation}
	f(x_1,x_2,\dots,x_n) = g(x_1,x_2,\dots,x_n) + h(0,0,\dots,0).
	\end{equation}
	The value $h(0,0,\dots,0) = a$ with $a$ a constant.  Likewise for
	\begin{equation}
	x_1 = x_2 = \dots = x_n = 0,
	\end{equation}
	we have,
	\begin{equation}
	f(y_1,y_2,\dots,y_n) = g(0,0,\dots,0) + h(y_1,y_2,\dots,y_n),
	\end{equation}
	where $g(0,0,\dots,0) = b$ with $b$ another constant.  Thus we have,
	\begin{align}
	g(x_1,x_2,\dots,x_n) &= f(x_1,x_2,\dots,x_n) - a,\\
	h(y_1,y_2,\dots,y_n) &= f(y_1,y_2,\dots,y_n) - b.
	\end{align}
	The many-variable Pexider equation can then be written as,
	\begin{align}
	f(x_1 + y_1,x_2 + y_2,\dots,x_n + y_n) &= f(x_1,x_2,\dots,x_n)\nonumber\\
	&+ f(y_1,y_2,\dots,y_n) - a - b.\label{pexider3}
	\end{align}
	Letting
	\begin{equation}
	\xi(t) = f(t) - a - b,
	\end{equation}
	eq. (\ref{pexider3}) becomes,
	\begin{equation}
	\xi(x_1+y_1,\dots,x_n+y_n) = \xi(x_1,\dots,x_n)  + \xi(y_1,\dots,y_n),
	\end{equation}
	which is the many-variable Cauchy equation.
\end{proof}

\end{appendix}

\end{document}